\documentclass[11pt]{article}

\usepackage{float}
\usepackage{xr}
\usepackage{microtype}
\usepackage{graphicx}
\usepackage{subfigure}
\usepackage{booktabs} 
\usepackage{amssymb, amsthm, mathrsfs} 
\usepackage{dsfont}
\usepackage{commath}
\usepackage{comment}
\usepackage{url}
\usepackage{bbm}
\usepackage{upgreek}
\usepackage{multirow}
\usepackage{mathtools}
\usepackage[inline]{enumitem}
\usepackage[usenames,dvipsnames]{xcolor}

\usepackage{geometry}
\RequirePackage{algorithm}
\RequirePackage{algorithmic}
\usepackage{environ}
\newcommand{\acksection}{\section*{Acknowledgments and Disclosure of Funding}}
\NewEnviron{ack}{%
  \acksection
  \BODY
}

\usepackage{natbib}

\usepackage[textwidth=1.8cm]{todonotes}
\usepackage{xargs}
\usepackage{nicefrac}

\usepackage{aliascnt}
\usepackage{thmtools}
\usepackage{amsmath}

\usepackage{titletoc}
\definecolor{mydarkblue}{rgb}{0,0.08,0.45}
\usepackage[colorlinks=true,
    linkcolor=mydarkblue,
    citecolor=mydarkblue,
    filecolor=mydarkblue,
    urlcolor=mydarkblue]{hyperref}%
\usepackage{cleveref,nameref}
\usepackage[page,header]{appendix}

\makeatletter
\newtheorem{theorem}{Theorem}
\crefname{theorem}{theorem}{Theorems}
\Crefname{Theorem}{Theorem}{Theorems}

\newaliascnt{lemma}{theorem}
\newtheorem{lemma}[theorem]{Lemma}
\aliascntresetthe{lemma}
\crefname{lemma}{lemma}{lemmas}
\Crefname{Lemma}{Lemma}{Lemmas}

\newaliascnt{proposition}{theorem}
\newtheorem{proposition}[theorem]{Proposition}
\aliascntresetthe{proposition}
\crefname{proposition}{proposition}{propositions}
\Crefname{Proposition}{Proposition}{Propositions}

\newaliascnt{corollary}{theorem}

\aliascntresetthe{corollary}
\crefname{corollary}{corollary}{corollaries}
\Crefname{Corollary}{Corollary}{Corollaries}

\newtheorem{assumption}{\textbf{H}\hspace{-3pt}}
\Crefname{assumption}{\textbf{H}\hspace{-3pt}}{\textbf{H}\hspace{-3pt}}
\crefname{assumption}{\textbf{H}}{\textbf{H}}

\Crefname{assumptionSC}{\textbf{SC}\hspace{-3pt}}{\textbf{SC}\hspace{-3pt}}
\crefname{assumptionSC}{\textbf{SC}}{\textbf{SC}}

\Crefname{assumptionF}{\textbf{F}\hspace{-3pt}}{\textbf{F}\hspace{-3pt}}
\crefname{assumptionF}{\textbf{F}}{\textbf{F}}

\Crefname{assumptionQC}{\textbf{QC}\hspace{-3pt}}{\textbf{QC}\hspace{-3pt}}
\crefname{assumptionQC}{\textbf{QC}}{\textbf{QC}}

\Crefname{assumptionBR}{\textbf{BR}\hspace{-3pt}}{\textbf{BR}\hspace{-3pt}}
\crefname{assumptionBR}{\textbf{BR}}{\textbf{SC}}

\Crefname{assumptionR}{\textbf{R}\hspace{-3pt}}{\textbf{R}\hspace{-3pt}}
\crefname{assumptionR}{\textbf{R}}{\textbf{R}}

\Crefname{assumptionA}{\textbf{A}\hspace{-3pt}}{\textbf{A}\hspace{-3pt}}
\crefname{assumptionA}{\textbf{A}}{\textbf{A}}

\Crefname{assumptionMD}{\textbf{MD}\hspace{-3pt}}{\textbf{MD}\hspace{-3pt}}
\crefname{assumptionMD}{\textbf{MD}}{\textbf{MD}}

\Crefname{assumptionMA}{\textbf{MA}\hspace{-3pt}}{\textbf{MA}\hspace{-3pt}}
\crefname{assumptionMA}{\textbf{MA}}{\textbf{MA}}

\Crefname{assumptionB}{\textbf{B}\hspace{-3pt}}{\textbf{B}\hspace{-3pt}}
\crefname{assumptionB}{\textbf{B}}{\textbf{B}}

\definecolor{darkgreen}{RGB}{0,128,0}

\newcommand{\B}[1]{\mathbf{#1}}
 
\newcommand{\btheta}{\boldsymbol{\theta}}
\newcommand{\tbtheta}{\tilde{\btheta}}
\newcommand{\bz}{\B{z}}
\newcommand{\Zrm}{\mathrm{Z}}
\newcommand{\bx}{\B{x}}
\newcommand{\by}{\B{y}}
\newcommand{\dd}{\mathrm{d}}
\newcommand{\Rd}{\mathbb{R}^d}

\newcommand{\grad}{\nabla}
\newcommand{\E}{\mathbb{E}}
\newcommand{\PP}{\mathbb{P}}
\newcommand{\bthetastar}{\boldsymbol{\theta^{\star}}}

\newcommand{\1}{\mathds{1}}

\newcommand{\R}{\mathbb{R}}

\newcommand{\pr}[1]{\left({#1}\right)}
\newcommand{\prn}[1]{({#1})}
\newcommand{\prbig}[1]{\big({#1}\big)}
\newcommand{\prBig}[1]{\Big({#1}\Big)}
\newcommand{\prbigg}[1]{\bigg({#1}\bigg)}
\newcommand{\prBigg}[1]{\Bigg({#1}\Bigg)}
\newcommand{\br}[1]{\left[{#1}\right]}
\newcommand{\brn}[1]{[{#1}]}
\newcommand{\brbig}[1]{\big[{#1}\big]}
\newcommand{\brBig}[1]{\Big[{#1}\Big]}
\newcommand{\brbigg}[1]{\bigg[{#1}\bigg]}

\newcommand{\ac}[1]{\left\{{#1}\right\}}
\newcommand{\acn}[1]{\{{#1}\}}
\newcommand{\acbig}[1]{\big\{{#1}\big\}}
\newcommand{\acBig}[1]{\Big\{{#1}\Big\}}
\newcommand{\acbigg}[1]{\bigg\{{#1}\bigg\}}

\newcommand{\normn}[1]{\|{#1}\|} 
\newcommand{\normbig}[1]{\big\|{#1}\big\|} 
\newcommand{\normBig}[1]{\Big\|{#1}\Big\|} 
\newcommand{\normbigg}[1]{\bigg\|{#1}\bigg\|}
\newcommand{\absn}[1]{\textstyle|{#1}|}

\newcommand{\nofrac}[2]{{#1}/{#2}}
\newcommand{\nosqrt}[1]{{#1}^{\half}}

\def\eqsp{\;}

\renewcommandx{\iint}[2][1=]{\ifthenelse{\equal{#1}{}}{[#2]}{[#1:#2]}}

\def\bfx{\mathbf{x}}

\def\msa{\mathsf{A}}

\def\msz{\mathsf{Z}}

\def\msb{\mathsf{B}}

\def\msx{\mathsf{X}}
\def\msy{\mathsf{Y}}

\def\mcbb{\mathcal{B}}  
\newcommand{\mcb}[1]{\mathcal{B}(#1)}

\def\mcy{\mathcal{Y}}
\def\mcx{\mathcal{X}}

\def\rset{\mathbb{R}}

\def\nset{\mathbb{N}}
\def\nsets{\mathbb{N}^*}

\def\rmd{\mathrm{d}}
\def\rmZ{\mathrm{Z}}

\def\rme{\mathrm{e}}

\newcommandx{\functionspace}[2][1=+]{\mathbb{F}_{#1}(#2)}

\newcommand{\argmin}{\operatorname*{arg\,min}}

\newcommandx{\VarDeux}[3][3=]{\operatorname{Var}^{#3}_{#1}\left\{#2 \right\}}

\newcommand{\LeftEqNo}{\let\veqno\@@leqno}

\newcommand{\N}{\ensuremath{\mathbb{N}}}

\newcommand{\PE}{\mathbb{E}}

\newcommand{\absLigne}[1]{\vert #1 \vert}

\newcommandx{\Vnorm}[2][1=V]{\| #2 \|_{#1}}
\newcommandx{\VnormEq}[2][1=V]{\left\| #2 \right\|_{#1}}
\newcommandx{\normLigne}[2][1=]{\ifthenelse{\equal{#1}{}}{\Vert #2 \Vert}{\Vert #2\Vert^{#1}}}

\newcommand{\parenthese}[1]{\left(#1 \right)}
\newcommand{\parentheseLigne}[1]{(#1 )}

\newcommand{\parentheseDeuxLigne}[1]{[ #1 ]}
\newcommand{\defEns}[1]{\left\lbrace #1 \right\rbrace }
\newcommand{\defEnsLigne}[1]{\lbrace #1 \rbrace }

\newcommand{\ps}[2]{\left\langle#1,#2 \right\rangle}
\newcommand{\psLigne}[2]{\langle#1,#2 \rangle}

\newcommandx\probaMarkovTilde[2][2=]
{\ifthenelse{\equal{#2}{}}{{\widetilde{\mathbb{P}}_{#1}}}{\widetilde{\mathbb{P}}_{#1}\left[ #2\right]}}

\newcommand{\expe}[1]{\PE \left[ #1 \right]}

\newcommand{\bigO}{\ensuremath{\mathcal{O}}}

\newcommand{\plusinfty}{+\infty}

\def\ie{\textit{i.e.}}

\def\eqsp{\;}

\newcommand{\ooint}[1]{\left(#1\right)}

\newcommandx{\weight}[2][2=n]{\omega_{#1,#2}^N}

\newcommandx\sequence[3][2=,3=]
{\ifthenelse{\equal{#3}{}}{\ensuremath{\{ #1_{#2}\}}}{\ensuremath{\{ #1_{#2}, \eqsp #2 \in #3 \}}}}
\newcommandx\sequenceD[3][2=,3=]
{\ifthenelse{\equal{#3}{}}{\ensuremath{\{ #1_{#2}\}}}{\ensuremath{( #1)_{ #2 \in #3} }}}

\newcommandx{\sequencen}[2][2=n\in\N]{\ensuremath{\{ #1_n, \eqsp #2 \}}}
\newcommandx\sequenceDouble[4][3=,4=]
{\ifthenelse{\equal{#3}{}}{\ensuremath{\{ (#1_{#3},#2_{#3}) \}}}{\ensuremath{\{  (#1_{#3},#2_{#3}), \eqsp #3 \in #4 \}}}}
\newcommandx{\sequencenDouble}[3][3=n\in\N]{\ensuremath{\{ (#1_{n},#2_{n}), \eqsp #3 \}}}

\def\eg{e.g.}

\newcommand{\opnorm}[1]{{\left\vert\kern-0.25ex\left\vert\kern-0.25ex\left\vert #1 
    \right\vert\kern-0.25ex\right\vert\kern-0.25ex\right\vert}}

\newcommandx{\CPE}[3][1=]{{\mathbb E}_{#1}\left[#2 \left \vert #3 \right. \right]} 
\newcommandx{\CPVar}[3][1=]{\mathrm{Var}^{#3}_{#1}\left\{ #2 \right\}}
\newcommand{\CPP}[3][]
{\ifthenelse{\equal{#1}{}}{{\mathbb P}\left(\left. #2 \, \right| #3 \right)}{{\mathbb P}_{#1}\left(\left. #2 \, \right | #3 \right)}}

\newcommandx{\osc}[2][1=]{\mathrm{osc}_{#1}(#2)}

\def\b{b}

\def\tbfx{\tilde{\bfx}}

\def\tX{\tilde{X}}

\def\Yb{\bar{Y}}

\newcommand\coupling[2]{\Gamma(\mu,\nu)}

\newcommand{\half}{{\nicefrac{1}{2}}}
\newcommandx{\wasserstein}[2][1=2]{\ifthenelse{\equal{#1}{}}{W}{W_{#1}}}
\newcommandx{\wassersteinLigne}[3][1=\distance,3=]{\mathbf{W}_{#1}^{#3}(#2)}
\newcommandx{\wassersteinD}[1][1=\distance]{\mathbf{W}_{#1}}
\newcommandx{\wassersteinDLigne}[1][1=\distance]{\mathbf{W}_{#1}}

\def\varespilon{\varepsilon}

\def\bgamma{\bar{\gamma}}

\def\tZ{\tilde{Z}}
\def\tP{\tilde{P}}
\def\tQ{\tilde{Q}}

\def\bfz{\mathbf{z}}
\def\zbf{\bfz}
\def\tbfz{\tilde{\mathbf{z}}}

\def\tR{\tilde{R}}

\def\thetabf{\btheta}
\def\bftheta{\btheta}

\def\Abf{\mathbf{A}}

\def\bfgamma{\boldsymbol{\gamma}}
\def\gammabf{\bfgamma}
\def\bfrho{\boldsymbol{\rho}}
\def\rhobf{\boldsymbol{\rho}}
\def\bfN{\boldsymbol{N}}

\def\bfOne{\boldsymbol{1}}

\def\Yb{Y}
\def\Yc{\tilde{Y}}
\def\Ybr{\mathrm{Y}}
\def\Ycr{\tilde{\mathrm{Y}}}
\def\Zb{Z}
\def\Zc{\tilde{Z}}

\newcommand{\txts}{\textstyle}
\def\loiGauss{\mathrm{N}}
\def\tPi{\tilde{\Pi}}
\def\tildeU{V}

\def\bfOne{\boldsymbol{1}}

\usepackage{authblk}

\title{DG-LMC: A Turn-key and Scalable Synchronous Distributed MCMC Algorithm via Langevin Monte Carlo within Gibbs}

\author[1,2,*]{\textbf{Vincent Plassier}}
\author[1,*]{\textbf{Maxime Vono}}
\author[3,*]{\textbf{Alain Durmus}}
\author[2]{\textbf{Eric Moulines}}
 
\affil[1]{Lagrange Mathematics and Computing Research Center, Paris, France}
\affil[2]{Ecole Polytechnique, Universit\'e Paris-Saclay, Palaiseau, France}
\affil[3]{Ecole Normale Sup\'erieure Paris-Saclay, Gif-sur-Yvette, France}
\affil[*]{Equal contribution}

\begin{document}

\maketitle

\begin{abstract}
  Performing reliable Bayesian inference on a big data scale is becoming a keystone in the modern era of machine learning.
  A workhorse class of methods to achieve this task are Markov chain Monte Carlo (MCMC) algorithms and their design to handle distributed datasets has been the subject of many works.
  However, existing methods are not completely either reliable or computationally efficient.
  In this paper, we propose to fill this gap in the case where the dataset is partitioned and stored on computing nodes within a cluster under a master/slaves architecture.
  We derive a user-friendly centralised distributed MCMC algorithm with provable scaling in high-dimensional settings.
  We illustrate the relevance of the proposed methodology on both synthetic and real data experiments.
\end{abstract}

\section{Introduction}
\label{sec:introduction}

In the current machine learning era, data acquisition has seen
significant progress due to rapid technological advances which now
allow for more accurate, cheaper and faster data storage and
collection.  This data quest is motivated by modern machine learning
techniques and algorithms which are now well-proven and have become
common tools for data analysis.  In most cases, the empirical success
of these methods are based on a very large sample size
\citep{Bardenet17,Bottou_SIREV_2018}.  This need for data is also
theoretically justified by data probabilistic modelling which asserts
that under appropriate conditions, the more data can be processed, the
more accurate the inference can be performed.  However, in recent
years, several challenges have emerged regarding the use and access to
data in mainstream machine learning methods.  Indeed, first the amount
of data is now so large that it has outpaced the increase in
computation power of computing resources
\citep{Survey_distributed_ML}.  Second, in many modern applications,
data storage and/or use are not on a single machine but shared across
several units \citep{AstroPortal,Bernstein2009}.  
Third, life privacy is becoming a prime concern for many users of machine learning applications who are therefore asking for methods preserving data anonymity \citep{Shokri15,Abadi16}. 
Distributed machine learning aims at tackling these issues.
One of its popular paradigms, referred to as data-parallel approach, is to consider that the training data are divided across multiple machines.  
Each of these units constitutes a worker node of a computing network and can perform a \textit{local} inference based on the data it has access.   
Regarding the choice of the network, several options and frameworks have been considered.
We focus here on the master/slaves architecture where the worker nodes communicate with each other through a device called the \emph{master} node.\mbox{}\\

\noindent Under this framework, we are interested in carrying Bayesian inference about a parameter $\boldsymbol{\theta}\in\mathbb{R}^d$ based on observed data $\{\by_k\}_{k=1}^n \in \msy^n$ \citep{Robert94}.
The dataset is assumed to be partitioned into $S$ \emph{shards} and stored on $S$ machines among a collection of $b$ worker nodes.
The subset of observations associated to worker $i \in [b]$ is denoted by $\boldsymbol{y}_i$, with potentially $\boldsymbol{y}_i = \{\emptyset\}$ if $i \in \iint{S+1:b}$, $b > S$.
The posterior distribution of interest is assumed to admit a density w.r.t. the $d$-dimensional Lebesgue measure which factorises across workers, \emph{i.e.,}  
\begin{equation}
  \label{eq:target_density}
  \pi(\boldsymbol{\theta}\mid\B{y}_{1:n}) = Z_{\pi}^{-1}\, \prod_{i=1}^b\mathrm{e}^{-U_{\boldsymbol{y}_i}(\B{A}_i\boldsymbol{\theta})}  \eqsp,
\end{equation}
where $Z_{\pi} = \int_{\Rd} \prod_{i=1}^b\mathrm{e}^{-U_{\boldsymbol{y}_i}(\B{A}_i\boldsymbol{\theta})}\,\dd \btheta$ is a normalisation constant and $\B{A}_i \in \mathbb{R}^{d_i \times d}$ are matrices that might act on the parameter of interest.
For $i \in [b]$, the potential function $U_{\boldsymbol{y}_i}: \mathbb{R}^{d_i} \rightarrow \mathbb{R}$ is assumed to depend only on the subset of observations $\boldsymbol{y}_i$. Note that for $i \in \iint{S+1:b}$, $b > S$, $U_{\boldsymbol{y}_i}$ does not depend on the data but only on the prior.
For the sake of brevity, the dependency of $\pi$ w.r.t. the observations $\{\boldsymbol{y}_i\}_{i=1}^b$ is notationally omitted and for $i \in [b]$, $U_{\boldsymbol{y}_i}$ is simply denoted by $U_i$.\mbox{}\\

\noindent To sample from $\pi$ given by \eqref{eq:target_density} in a distributed fashion, a large number of approximate methods have been proposed in the past ten years \citep{Neiswanger2014,Ahn14,2015_variational_consensus,Scott2016,nemeth2018,pmlr-v84-chowdhury18a,Rendell2020}.
Despite multiple research lines, to the best of authors' knowledge, none of these proposals has been proven to be satisfactory.
Indeed, the latter are not completely either computationally efficient in high-dimensional settings, reliable or theoretically grounded \citep{Jordan2019}.\mbox{}\\

\noindent This work is an attempt to fill this gap.
To this purpose, we follow the data augmentation approach introduced in \citet{Vono_AXDA_2019} and referred to as asymptotically exact data augmentation (AXDA). 
Given a tolerance parameter $\bfrho$, the main idea behind this methodology is to consider a joint distribution $\Pi_{\bfrho}$ on the extended state space $\rset^d \times \prod_{i=1}^b \rset^{d_i}$ such that $\Pi_{\bfrho}$ has a density w.r.t. the Lebesgue measure of the form $(\bftheta,\zbf_{1:b}) \mapsto \prod_{i=1}^b \Pi_{\bfrho}^{i}(\bftheta,\bz_i)$, with $\bftheta \in \rset^d$ and $\zbf_i \in \rset^{d_i}$, $i \in [b]$. 
$\Pi_{\bfrho}$ is carefully designed so that its marginal w.r.t. $\bftheta$, denoted by $\pi_{\bfrho}$, is a proxy of \eqref{eq:target_density} for which quantitative approximation bounds can be derived and are controlled by $\bfrho$. 
In addition, for any $i \in [b]$,  $\Pi_{\bfrho}^{i}(\bftheta,\bz_i)$ only depends on the data $\boldsymbol{y}_i$, and therefore plays a role similar to the local posterior $\pi^{i}(\btheta) \propto \mathrm{e}^{-U_i(\B{A}_i\btheta)}$ in popular embarrassingly parallel approaches \citep{Neiswanger2014,Scott2016}.
However, compared to this class of methods, AXDA does not seek for each worker to sample from $\Pi_{\bfrho}^{i}$.
Following a data augmentation strategy based on Gibbs sampling, AXDA instead requires each worker to sample from the conditional distribution $\Pi_{\bfrho}(\zbf_i \mid\bftheta)$ and to communicate its sample to the master.
$\Pi_{\bfrho}$ is generally chosen such that sampling from $\Pi_{\bfrho}(\bftheta \mid \zbf_{1:b})$ is easy and does not require to access to the data.
However, two main challenges remain: one has to sample efficiently from the conditional distribution $\Pi_{\bfrho}(\zbf_i \mid\bftheta)$ for $i \in [b]$ and avoid too frequent communication rounds on the master.
Existing AXDA-based approaches unfortunately do not fulfill these important requirements \citep{Vono_Paulin_Doucet_2019,Rendell2020}.
In this work, we leverage these issues by considering the use of the Langevin Monte Carlo (LMC) algorithm to approximately sample from $\Pi_{\bfrho}(\zbf_i \mid\bftheta)$ \citep{rossky:doll:friedman:1978,Roberts1996}.\mbox{}\\

\noindent Our contributions are summarised in what follows.
\begin{enumerate*}[label=(\arabic*)]
\item We introduce in \Cref{sec:pres_DG_LMC} a new methodology called Distributed Gibbs using Langevin Monte Carlo (DG-LMC).
\item Importantly, we provide in \Cref{sec:theoretical_analysis_DG_LMC} a  detailed quantitative analysis of the induced bias and show explicit convergence results.
This stands for our main contribution and to the best of authors' knowledge, this theoretical study is one of the most complete among existing works which focused on distributed Bayesian machine learning with a master/slaves architecture. 
In particular, we discuss the complexity of our algorithm, the choice of hyperparameters, and provide practitioners with simple prescriptions to tune them. Further, we provide a thorough comparison of our method with existing approaches in \Cref{sec:comparison}.
\item Finally, in \Cref{sec:experiments}, we show the benefits of the proposed sampler over popular and recent distributed MCMC algorithms on several numerical experiments.
\end{enumerate*}
All the proofs are postponed to the Appendices.\\[1em]

\noindent\textbf{Notations and conventions.} 
The Euclidean norm on $\mathbb{R}^d$ is denoted by $\|\cdot\|$.
For $n \ge 1$, we refer to $\{1,\ldots,n\}$ with the notation $[n]$ and for $i_1,i_2 \in \nset$, $i_1 \le i_2$, $\{i_1,\ldots,i_2\}$ with the notation $[i_1:i_2]$.
For $0\le i < j$ and ($\B{u}_k; k \in \{i,\cdots,j\}$), we use the notation $\B{u}_{i:j}$ to refer to the vector $[\B{u}_i^{\top},\cdots,\B{u}_{j}^{\top}]^{\top}$.
We denote by $\loiGauss(\B{m},\B{\Sigma})$ the Gaussian distribution with mean vector $\B{m}$ and covariance matrix $\B{\Sigma}$.
For a given matrix $\B{M} \in \mathbb{R}^{d \times d}$, we denote its smallest eigenvalue by $\lambda_{\mathrm{min}}(\B{M})$.
We denote by $\mathcal{B}(\mathbb{R}^d)$ the Borel $\sigma$-field of $\mathbb{R}^d$.
We define the Wasserstein distance of order $2$ for any probability measures $\mu,\nu$ on $\Rd$ with finite $2$-moment by $\wasserstein{} (\mu, \nu) = (\inf_{\zeta \in \mathcal{T}(\mu,\nu)} \int_{\mathbb{R}^d \times \mathbb{R}^d}\|\btheta-\btheta'\|^2\dd\zeta(\btheta,\btheta'))^{\half}$, where $\mathcal{T}(\mu, \nu)$ is the set of transference plans of $\mu$ and $\nu$.

\section{Distributed Gibbs using Langevin Monte Carlo (DG-LMC)}
\label{sec:pres_DG_LMC}

In this section, we present the proposed methodology which is based on the AXDA statistical framework and the popular LMC algorithm.  

AXDA relies on the decomposition of the target distribution $\pi$ given in \eqref{eq:target_density} to introduce an extended distribution which enjoys favorable properties for distributed computations. 
This distribution is defined on the state space $\rset^d \times \msz$, $\msz = \prod_{i=1}^b \rset^{d_i}$, and admits a density w.r.t. the Lebesgue measure given, for any $\bftheta \in \rset^d$, $\bfz_{1:b} \in \msz$, by
\begin{equation}
  \label{eq:joint_density_AXDA}
  \Pi_{\bfrho}(\boldsymbol{\theta},\bz_{1:b}) \propto \prod_{i=1}^b \tPi_{\bfrho}^{i}(\bftheta,\bfz_i) \eqsp,
\end{equation}
where $\tPi_{\bfrho}^{i}(\bftheta,\bfz_i) = \exp(-U_i(\B{z}_i) - \nicefrac{\norm{\bz_i - \B{A}_i\btheta}^2}{2\rho_i})$ and $\bfrho = \{\rho_i\}_{i=1}^b \in \rset_+^b$ is a sequence of positive tolerance parameters. 
Note that $\tPi_{\bfrho}^{i}$ is not necessarily a probability density function. 
Actually, for $\Pi_{\bfrho}$ to define a proper probability density, \ie~$\int_{\rset^d \times \msz} \prod_{i=1}^b \tPi_{\bfrho}^{i}(\bftheta,\bfz_i)  \rmd \bftheta \rmd \bfz_{1:b} < \infty$, some conditions are required.
\begin{assumption}
  \label{ass:well_defined_density}
There exists $b' \in [b-1]$ such that the following conditions hold: $\min_{i \in [b']} \inf_{\bz_{i}\in \mathbb{R}^{d_i}} U_i(\bz_{i})>-\infty$, and $\max_{i \in [b'+1:b]} \int_{\rset^{d_i}} \rme^{-U_i(\bz_i)} \rmd \bz_i < \infty$.
  In addition, $\sum_{j=b'+1}^b\B{A}_j^{\top}\B{A}_j$ is invertible.
\end{assumption}
The next result shows that these mild assumptions are sufficient to guarantee that the extended model \eqref{eq:joint_density_AXDA} is well-defined.
\begin{proposition}
  \label{prop:pi_rho_proper}
 Assume \Cref{ass:well_defined_density}. Then, for any $\bfrho \in \rset_+^b$, $\Pi_{\bfrho}$ in \eqref{eq:joint_density_AXDA} is a proper density.
\end{proposition}
The data augmentation scheme \eqref{eq:joint_density_AXDA} is approximate in the sense that the $\btheta$-marginal defined by
\begin{equation}
  \label{eq:theta_marginal}
  \pi_{\bfrho}(\btheta ) = \int_{\msz} \Pi_{\bfrho}(\boldsymbol{\theta},\bz_{1:b}) \dd \bz_{1:b} \eqsp,
\end{equation}
coincides with \eqref{eq:target_density} only in the limiting case $\max_{i \in [b]} \rho_i \downarrow 0$ \citep{Scheffe1947}.
For a fixed $\bfrho$, quantitative results on the induced bias in total variation distance can be found in \citet{Vono_Paulin_Doucet_2019}.
The main benefit of working with \eqref{eq:joint_density_AXDA} is that conditionally upon $\btheta$, auxiliary variables $\{\bz_i\}_{i=1}^b$ are independent.
Therefore, they can be sampled in parallel within a Gibbs sampler.
For $i \in [b]$, the conditional density of $\bz_i$ given $\btheta$ writes
\begin{equation}
  \label{eq:conditional_zi_given_theta}
  \Pi_{\bfrho}(\bz_i\mid\boldsymbol{\theta}) \propto \exp\big(-U_i(\B{z}_i) - \textstyle\frac{\norm{\bz_i - \B{A}_i\btheta}^2}{2\rho_i}\big)\eqsp.
\end{equation}
On the other hand, the conditional distribution of $\btheta$ given $\bz_{1:b}$ is a Gaussian distribution
\begin{equation}
  \label{eq:def:Pi_rho_cond}
  \Pi_{\bfrho}(\boldsymbol{\theta} \mid \bz_{1:b}) = \loiGauss(\boldsymbol{\mu}(\bz_{1:b}),\B{Q}^{-1})\eqsp,
\end{equation}
with precision matrix $\B{Q} = \sum_{i=1}^b\B{A}_i^{\top}\B{A}_i/\rho_i$ and mean vector $\boldsymbol{\mu}(\bz_{1:b}) = \B{Q}^{-1}\sum_{i=1}^b\B{A}_i^{\top}\B{z}_i/\rho_i$.
Under \textbf{H}\ref{ass:well_defined_density}, note that $\B{Q}$ is invertible and therefore this conditional Gaussian distribution is well-defined.
Since sampling from high-dimensional Gaussian distributions can be performed efficiently \citep{Vono2020_gaussian}, this Gibbs sampling scheme is interesting as long as sampling from \eqref{eq:conditional_zi_given_theta} is cheap.
\citet{Vono_Paulin_Doucet_2019} proposed the use of a rejection sampling step requiring to set $\rho_i = \mathcal{O}(1/d_i)$. 
When $d_i \gg 1$, this condition unfortunately leads to prohibitive computational costs and hence prevents its practical use for general Bayesian inference problems.
Instead of sampling exactly from \eqref{eq:conditional_zi_given_theta}, \citet{Rendell2020} rather proposed to use Metropolis-Hastings algorithms.
However, it is not clear whether this choice indeed leads to efficient sampling schemes.
\begin{algorithm}
   \caption{Distributed Gibbs using LMC (DG-LMC)}
   \label{algo:ULAwSG}
  \begin{algorithmic}
     \STATE {\bfseries Input:} burn-in $T_{\mathrm{bi}}$; for $i \in [b]$, tolerance parameters $\rho_i > 0$, step-sizes $\gamma_i \in (0,\rho_i/(1+\rho_i M_i)]$, local LMC steps $N_i \ge 1$.
     \STATE Initialise $\btheta^{(0)}$ and $\bz_{1:b}^{(0)}$.
     \FOR{$t=0$ {\bfseries to} $T-1$}
      \STATE \COMMENT{Sampling from $\Pi_{\bfrho}(\bfz_{1:b}|\bftheta)$}
        \FOR{$i=1$ {\bfseries to} $b$ \COMMENT{In parallel on the $b$ workers}}
        \STATE $\B{u}_i^{(0)} = \bz_i^{(t)}$ 
          \FOR{$k=0$ {\bfseries to} $N_i-1$ \COMMENT{$N_i$ local LMC steps}}
            \STATE $\boldsymbol{\xi}_i^{(k,t)} \sim \loiGauss\pr{\B{0}_{d_i},\B{I}_{d_i}}$
            \STATE $\B{g}_i = \big(1-\frac{\gamma_i}{\rho_i}\big)\B{u}_i^{(k)} + \frac{\gamma_i}{\rho_i}\B{A}_i\btheta^{(t)} - \gamma_i\grad U_i\big(\B{u}_i^{(k)}\big)$
            \STATE $\B{u}_i^{(k+1)} = \B{g}_i + \sqrt{2\gamma_i}\boldsymbol{\xi}_i^{(k,t)}$ \COMMENT{See \eqref{eq:conditional_zi_given_theta}}
          \ENDFOR
          \STATE $\B{z}_i^{(t+1)} = \B{u}_i^{(N_i)}$
        \ENDFOR
        \STATE \COMMENT{Sampling from $\Pi_{\bfrho}(\bftheta |\bfz_{1:b})$}
        \STATE $\btheta^{(t+1)} \sim \loiGauss\Big(\boldsymbol{\mu}\big(\bz_{1:b}^{(t+1)}\big),\B{Q}^{-1}\Big)$ \COMMENT{See \eqref{eq:def:Pi_rho_cond}}
     \ENDFOR
     \STATE {\bfseries Output:} samples $\{\btheta^{(t)}\}_{t=T_{\mathrm{bi}}-1}^{T}$.
  \end{algorithmic}
\end{algorithm}
%
To tackle these issues, we propose to build upon LMC to end up with a distributed MCMC algorithm which is both simple to implement, efficient and amenable to a theoretical study.
LMC stands for a popular way to approximately generate samples from a given distribution based on the Euler-Maruyama discretisation scheme of the overdamped Langevin stochastic differential equation \citep{Roberts1996}.
At iteration $t$ of the considered Gibbs sampling scheme and given a current parameter $\btheta^{(t)}$, LMC applied to \eqref{eq:conditional_zi_given_theta} considers, for $i \in [b]$, the recursion
\begin{equation*}
  \bz_i^{(t+1)} = \big(1-\textstyle\frac{\gamma_i}{\rho_i}\big)\bz_i^{(t)} + \frac{\gamma_i}{\rho_i}\B{A}_i\btheta^{(t)} - \gamma_i\grad U_i\big(\B{z}_i^{(t)}\big) + \sqrt{2\gamma_i}\boldsymbol{\xi}_i^{(t)}
\end{equation*}
where $\gamma_i >0$ is a fixed step-size and $(\boldsymbol{\xi}_i^{(k)})_{k \in \mathbb{N},i\in [b]}$ a sequence of independent and identically distributed (i.i.d.) $d$-dimensional standard Gaussian random variables.
Only using a single step of LMC on each worker might incur important communication costs.
To mitigate the latter while increasing the proportion of time spent on exploring the state-space, we instead allow each worker to perform $N_i \ge 1$ LMC steps \citep{2019_Dieuleveut,Rendell2020}.
Letting $N_i$ varies across workers prevents \Cref{algo:ULAwSG} to suffer from a significant block-by-the-slowest delay in cases where the response times of the workers are unbalanced \citep{Ahn14}.
The proposed algorithm, coined Distributed Gibbs using Langevin Monte Carlo (DG-LMC), is depicted in \Cref{algo:ULAwSG} and illustrated in \Cref{fig:DULAG_schema}.
\begin{figure}
\begin{center}
\centerline{\includegraphics[scale=0.5]{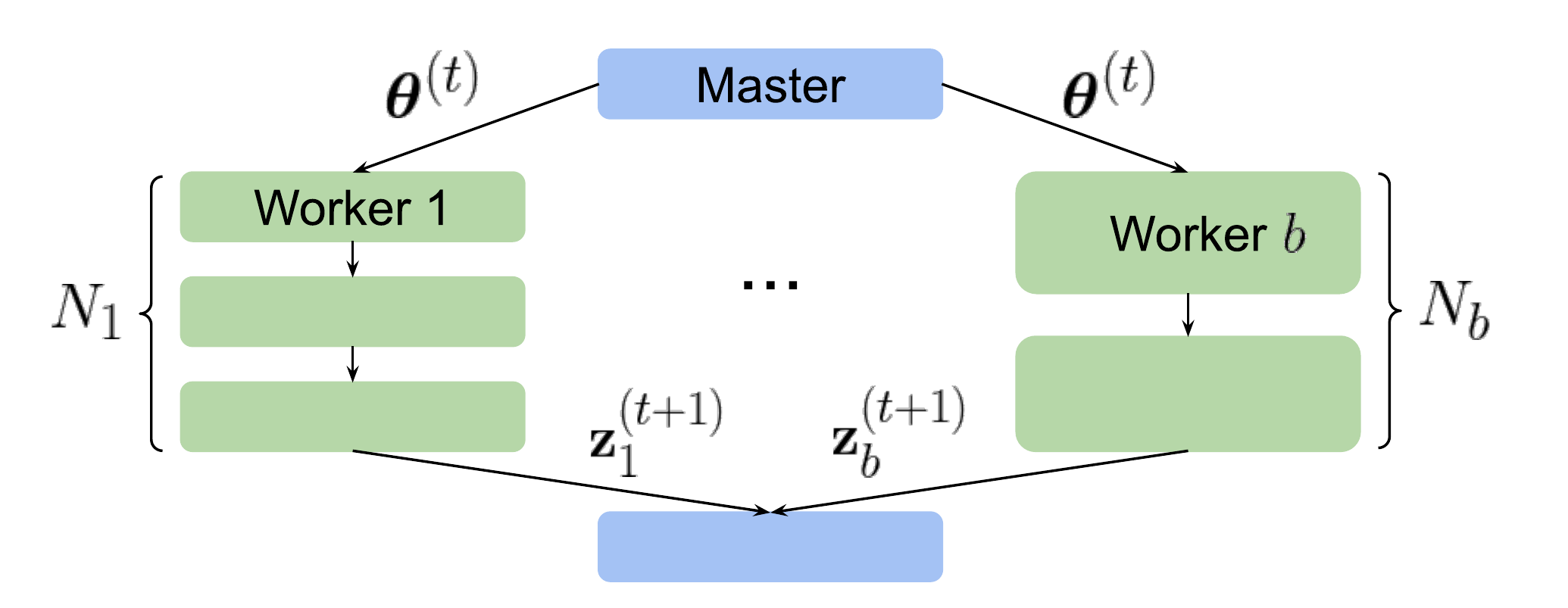}}
\caption{Illustration of one global iteration of \Cref{algo:ULAwSG}. For each worker, the width of the green box represents the amount of time required to perform one LMC step.}
\label{fig:DULAG_schema}
\end{center}
\vskip -0.4in
\end{figure}

\section{Detailed analysis of DG-LMC}
\label{sec:theoretical_analysis_DG_LMC}

In this section, we derive quantitative bias and convergence results for DG-LMC and show that its mixing time only scales quadratically w.r.t. the dimension $d$.
We also discuss the choice of hyperparameters and provide guidelines to tune them.

\subsection{Non-Asymptotic Analysis}
\label{subsec:non_asymptotic_analysis}

The scope of our analysis will focus on smooth and strongly log-concave target posterior distributions $\pi$.
While these assumptions may be restrictive in practice, they allow for a detailed theoretical study of the proposed algorithm.
\begin{assumption}
  \label{ass:supp_fort_convex}
  \begin{enumerate}[wide, labelwidth=!, labelindent=0pt,label=(\roman*),noitemsep,nolistsep]

    \item \label{ass:1} For any $i\in [b]$, $U_i$ is twice continuously differentiable and $\sup_{\bz_i\in\R^{d_i}} \|\nabla^2 U_i(\bz_{i})\| \le M_i$.

    \item \label{ass:2} For any $i\in [b]$, $U_i$ is $m_i$-strongly convex: there exists $m_i >0$ such that $m_i\B{I}_{d_i}\preceq \nabla ^2 U_i$.
    \end{enumerate}
  \end{assumption}
Under these assumptions, it is shown in \Cref{lem:U_stronglyconvex} in the Appendix that $-\log \pi$ is strongly convex with constant
\begin{equation}
\label{eq:mu}
m_U = \textstyle\lambda_{\min}(\sum_{i=1}^b m_i \B{A}_i^{\top} \B{A}_i)\eqsp.
\end{equation}
Behind the use of LMC, the main motivation is to end up with a simple hybrid Gibbs sampler amenable to a non-asymptotic theoretical analysis based on previous works \citep{durmus2018high,Dalalyan2019}.
In the following, this study is carried out using the Wasserstein distance of order 2.

\subsubsection{Convergence Results} 

DG-LMC introduced in \Cref{algo:ULAwSG} defines a homogeneous Markov chain $(V_t)_{t \in \mathbb{N}} = (\theta_t,Z_t)_{t \in \mathbb{N}}$ with realisations $(\btheta^{(t)},\bz_{1:b}^{(t)})_{t \in \mathbb{N}}$.
We denote by $P_{\bfrho,\bfgamma,\bfN}$ the Markov kernel associated with $(V_t)_{t \in \mathbb{N}}$.
Since no Metropolis-Hastings step is used in combination with LMC, the proposed algorithm does not fall into the class of Metropolis-within-Gibbs samplers \citep{roberts2006}.
Therefore, a first step is to show that $P_{\bfrho,\bfgamma,\bfN}$ admits an unique invariant distribution and is geometrically ergodic.
We proceed via an appropriate synchronous coupling which reduces the convergence analysis of $(V_t)_{t \in\nset}$ to that of the marginal process $(Z_t)_{t \in\nset}$.
While the proof of the convergence of $(Z_t)_{t \in\nset}$ shares some similarities with LMC \citep{durmus2018high}, the analysis of $(Z_t)_{t \in\nset}$ is much more involved and especially in the case $\max_{i \in [b]}N_i > 1$.
We believe that the proof techniques we developed to show the next result can be useful to the study of other MCMC approaches based on LMC.
\begin{proposition}
  \label{prop:convergence_rho_gamma}
  Assume \Cref{ass:well_defined_density}-\Cref{ass:supp_fort_convex} and let $c>0$ and $\bfgamma = \{\gamma_i\}_{i=1}^b$  $\bfN = \{N_i\}_{i=1}^b$ satisfying $\max_{i\in[b]}\gamma_i \leq \bgamma$,  $\min_{i\in[b]}\{N_i\gamma_i\}/\max_{i\in[b]}\{N_i\gamma_i\} \ge c$ and $\max_{i\in[b]}\{N_i\gamma_i\} \le C_1$ where $\bgamma,C_1$ are explicit constants only depending on $(m_i,M_i,\rho_i)_{i \in [b]}$\footnote{When $\bfN = \mathbf{1}_b$, $C_1 =\bar{\gamma} =  1/\txts\max_{i\in[b]}\{ M_i + \rho_i^{-1}\}$.}\footnote{When $\max_{i \in [b]} N_i > 1$, $C_1$ is of order $\min_{i \in [b]} \rho_i^2$ when $\max_{i \in [b]} \rho_i \rightarrow 0$, see \Cref{lem:contrac} in the Appendix.}.  Then, there exists a probability measure $\Pi_{\rhobf,\bfgamma,\bfN}$ such that  $\Pi_{\rhobf,\bfgamma,\bfN}$ is invariant for $P_{\bfrho,\bfgamma,\bfN}$.
  Moreover there exists $C_2 > 0$ such that for any integer $t\ge 0$ and $\B{v} = (\btheta,\bz) \in \mathbb{R}^d \times \mathsf{Z}$, we have
  \begin{align*}
  \wasserstein{}(\updelta_{\B{v}} P_{\bfrho,\bfgamma,\bfN}^{t}, \Pi_{\bfrho,\bfgamma,\bfN}) &\le C_2 \cdot (1-\min_{i\in[b]}\{N_i\gamma_im_i\}/2)^t \cdot \wasserstein{}(\updelta_{\B{v}}, \Pi_{\bfrho,\bfgamma,\bfN})\eqsp.
  \end{align*}
  Explicit expressions for $C_1$ and $C_2$ are given in \Cref{cor:convergence_rho_gamma} in the Appendix.
  Finally, if $\bfN = N \bfOne_b$ for $N \geq 1$, then  $\Pi_{\rhobf,\bfgamma,\bfN} =  \Pi_{\rhobf,\bfgamma,\bfOne_b}$.
\end{proposition}
We now discuss \Cref{prop:convergence_rho_gamma}.
If we set, for any $i \in [b]$, $N_i = 1$, the convergence rate in \Cref{prop:convergence_rho_gamma} becomes equal to $1 - \min_{i\in[b]} \{\gamma_i m_i\}/2$.
In this specific case, we show in \Cref{prop:convergence_N_1} in the Appendix that DG-LMC actually admits the tighter convergence rate $1 - \min_{i\in[b]} \{\gamma_i m_i\}$ which simply corresponds to the rate at which the slowest LMC conditional kernel converges.
On the other hand, when $\max_{i \in [b]}N_i > 1$, the convergence of $P_{\bfrho,\bfgamma,\bfN}$ towards $\Pi_{\bfrho,\bfgamma,\bfN}$ only holds if $\max_{i\in[b]}\{N_i\gamma_i\}$ is sufficiently small.
This condition is necessary to ensure a contraction in $W_2$ and can be understood intuitively as follows in the case where $\bfN = N \mathbf{1}_b$ and $\bfgamma = \gamma \mathbf{1}_b$.
Given two vectors $(\theta_k,\theta_k')$ and an appropriate coupling $(Z_{k+1},Z_{k+1}')$, we can show that $Z_{k+1}-Z_{k+1}'$ involves two competing terms: one keeping $Z_{k+1}-Z_{k+1}'$ close to $Z_{k}-Z_{k}'$ and another one driving $Z_{k+1}-Z_{k+1}'$ away from $\theta_k-\theta_k'$ (and therefore of $Z_{k}-Z_{k}'$) as $N$ increases.
This implies that $N$ stands for a trade-off and the product $N\gamma$ cannot be arbitrarily chosen.
Finally, it is worth mentioning that the tolerance parameters $\{\rho_i\}_{i \in [b]}$ implicitly drive the convergence rate of DG-LMC.
In the case $N_i=1$, a sufficient condition on the step-sizes to ensure a contraction is $\gamma_i \le 2/(M_i+m_i+1/\rho_i)$. 
We can denote that the smaller $\rho_i$, the smaller $\gamma_i$ and the slower the convergence.

Starting from the results of \Cref{prop:convergence_rho_gamma}, we can analyse the convergence properties of DG-LMC.
We specify our result to the case where we take for the specific initial distribution
\begin{equation}
\label{eq:_def_mu_star_supp_2main} 
\mu_{\rhobf}^{\star} =   \updelta_{\B{z}^{\star}}  \otimes \Pi_{\rhobf}(\cdot|\bfz^{\star}) \eqsp,
\end{equation}
where $\bz^{\star} = ([\B{A}_1\btheta^{\star}]^{\top} , \cdots,[\B{A}_b\btheta^{\star}]^{\top})^{\top}$, $\btheta^{\star} = \argmin\{-\log\pi\}$ and $\Pi_{\rhobf}(\cdot|\bfz^{\star})$ is defined in \eqref{eq:def:Pi_rho_cond}.
Note that sampling from $\mu_{\rhobf}^{\star}$ is straightforward  and simply consists in setting $\zbf^{(0)} = \bfz^{\star}$ and drawing $\bftheta^{(0)}$ from $\Pi_{\bfrho}(\cdot \mid \bfz^{\star})$.
For $t \ge 1$, we consider the marginal law of $\theta_t$ initialised at $\B{v}^\star$ with distribution $\mu_{\rhobf}^{\star}$ and denote it $\Gamma_{\B{v}^\star}^t$.
As mentioned previously, the proposed approach relies on two approximations which both come with some bias we need to control.
This naturally brings us to consider the following inequality based on the triangular inequality and the definition of the Wasserstein distance:
\begin{align}
\wasserstein{}(\Gamma_{\B{v}^{\star}}^t,\pi) 
&\le \wasserstein{}(\mu^{\star}_{\rhobf} P_{\bfrho,\bfgamma,\bfN}^{t}, \Pi_{\bfrho,\bfgamma,\bfN}) + \wasserstein{} (\Pi_{\bfrho,\bfgamma,\bfN},\Pi_{\bfrho}) + \wasserstein{}(\pi_{\bfrho},\pi)\eqsp,
\label{eq:error_decomp}
\end{align}
where $\Pi_{\bfrho,\bfgamma,\bfN}$, $\Pi_{\bfrho}$ and $\pi_{\bfrho}$ are defined in \Cref{prop:convergence_rho_gamma}, \eqref{eq:joint_density_AXDA} and \eqref{eq:theta_marginal}, respectively.
In \Cref{cor:convergence_rho_gamma_star_v2} in the Appendix, we provide an upper bound on the first term on the right hand side based on \Cref{prop:convergence_rho_gamma}.
In the next section, we focus on controlling the last two terms on the right hand side. 

\subsubsection{Quantitative Bounds on the Bias} 
\begin{table*}
\caption{For the specific initialisation $\B{v}^{\star}$ with distribution $\mu_{\rhobf}^{\star}$ given in \eqref{eq:_def_mu_star_supp_2main}, dependencies w.r.t. $d$ and $\varepsilon$ of the parameters involved in \Cref{algo:ULAwSG} and of $t_{\operatorname{mix}}(\varepsilon;\B{v}^{\star})$ to get a $\wasserstein{}$-error of at most $\varepsilon$.}
\vskip 0.15in
\begin{center}
 {\renewcommand{\arraystretch}{1.5}
\begin{tabular}{llccccc}
\toprule
Assumptions & & $\rho_{\varepsilon}$ & $\gamma_{\varepsilon}$ & $N_{\varepsilon}$ & $t_{\operatorname{mix}}(\varepsilon ; \B{v}^{\star})$ & \text{Gradient evaluations} \\
\midrule
\multirow{2}{*}{\Cref{ass:well_defined_density}, \Cref{ass:supp_fort_convex}} &$d$ & $\mathcal{O}(d^{-1})$ & $\mathcal{O}(d^{-3})$ & $\mathcal{O}(d)$ & $\mathcal{O}(d^2\log(d))$ & $\mathcal{O}(d^3\log(d))$ \\
&$\varepsilon$ & $\mathcal{O}(\varepsilon)$ & $\mathcal{O}(\varepsilon^{4})$ & $\mathcal{O}(\varepsilon^{-2})$ & $\mathcal{O}(\varepsilon^{-2}|\log(\varepsilon)|)$ & $\mathcal{O}(\varepsilon^{-4}|\log(\varepsilon)|)$ \\
\hline
\multirow{2}{*}{\Cref{ass:well_defined_density}, \Cref{ass:supp_fort_convex}, \Cref{ass:hessian_lipschitz}} &$d$ & $\mathcal{O}(d^{-1})$ & $\mathcal{O}(d^{-2})$ & $\mathcal{O}(1)$ & $\mathcal{O}(d^2\log(d))$ & $\mathcal{O}(d^2\log(d))$ \\
&$\varepsilon$ & $\mathcal{O}(\varepsilon)$ & $\mathcal{O}(\varepsilon^{2})$ & $\mathcal{O}(1)$ & $\mathcal{O}(\varepsilon^{-2}|\log(\varepsilon)|)$ & $\mathcal{O}(\varepsilon^{-2}|\log(\varepsilon)|)$ \\
\bottomrule
\end{tabular}}
\end{center}
\vskip -0.2in
\label{table:mixing_time}
\end{table*}
The error term $W_2(\pi_{\bfrho},\pi)$ in \eqref{eq:error_decomp} is related to the underlying AXDA framework which induces an approximate posterior representation $\pi_{\bfrho}$.
It can be controlled by the sequence of positive tolerance parameters $\{\rho_i\}_{i=1}^b$.
By denoting $\bar{\rho} = \max_{i \in [b]} \rho_i$, \Cref{prop:bound_bias_rho} shows that this error can be quantitatively assessed and is of order $\mathcal{O}(\bar{\rho})$ for sufficiently small values of this parameter.
\begin{proposition}
  \label{prop:bound_bias_rho}
  Assume \Cref{ass:well_defined_density}, \Cref{ass:supp_fort_convex}.
  In addition, let $\B{A} = [\B{A}_1^\top,\ldots,\B{A}_b^\top]^{\top}$ and denote $\sigma_U^2=\|\B{A}^{\top}\B{A}\|\max_{i \in [b]}\{M_i^2\}/m_U$, where $m_U$ is defined in \eqref{eq:mu}. 
  Then, for any $\bar{\rho} \le \sigma_U^2/12$,
  $$
  W_2(\pi_{\bfrho},\pi) \le \sqrt{2/m_U}\max(A_{\bfrho},B_{\bfrho})\eqsp,
  $$
  where $A_{\bfrho} = d\mathcal{O}(\bar{\rho})$ and $B_{\bfrho} = d^{\half}\mathcal{O}(\bar{\rho})$ for $\bar{\rho} \downarrow 0$.
  Explicit expressions for $A_{\bfrho}, B_{\bfrho}$ are given in \Cref{sec:proof_prop1} in the Appendix.
\end{proposition}
In the case where $\pi$ is Gaussian, the approximate distribution $\pi_{\boldsymbol{\rho}}$ admits an explicit expression and is Gaussian as well (e.g. when $b=1$, the mean is the same and the covariance matrix is inflated by a factor $\rho\mathbf{I}_d$), see for instance \citet[Section S2]{Rendell2020} and \citet[Section 5.1]{Vono_AXDA_2019}.
Hence, an explicit expression for $W_2(\pi_{\boldsymbol{\rho}},\pi)$ can be derived.
Based on this result, we can check that the upper bound provided by \Cref{prop:bound_bias_rho} matches the same asymptotics as $\rho \rightarrow 0$ and $d \rightarrow \infty$.

The second source of approximation error is induced by the use of LMC within \Cref{algo:ULAwSG} to target the conditional distribution $\Pi_{\bfrho}(\bz_{1:b}\mid\boldsymbol{\theta})$ in \eqref{eq:conditional_zi_given_theta}.
The stationary distribution of $P_{\bfrho,\bfgamma,\bfN}$ whose existence is ensured in \Cref{prop:convergence_rho_gamma} differs from $\Pi_{\bfrho}$.
The associated bias is assessed quantitatively in \Cref{prop:bias_gamma}.

\begin{proposition}
\label{prop:bias_gamma}
Assume \Cref{ass:well_defined_density}-\Cref{ass:supp_fort_convex}. For any $i \in [b]$, define $\tilde{M}_i = M_i + 1/\rho_i$ and let $\gammabf\in (\rset_+^*)^b$, $\bfN \in (\N^*)^b$ such that for any $i \in [b]$, 
\begin{align}
  \label{eq:gamma_bound}
  \gamma_i &\le \frac{m_i}{40\tilde{M}_i^2}\min_{i\in[b]}(m_i/\tilde{M}_i)^2/\max_{i\in[b]}(m_i/\tilde{M}_i)^2\eqsp, \\
  \label{eq:eq:good_choice_Nmain}
  N_{i} &= \big\lfloor m_{i}\min_{i\in[b]}\{m_{i}/\tilde{M}_{i}\}^2/(20\gamma_{i}\tilde{M}_{i}^2\max_{i\in[b]}\{m_{i}/\tilde{M}_{i}\}^2)\big\rfloor\eqsp.
\end{align}
Then, we have
\begin{equation*}
\wasserstein{}^{2}\pr{\Pi_{\bfrho,\bfgamma,\bfN},\Pi_{\bfrho}}
\le C_3\sum_{i=1}^b d_i\gamma_i \tilde{M}_i^2\eqsp,
\end{equation*}
where $C_3>0$ only depends of $(m_i,M_i,\B{A}_i,\rho_i)_{i=1}^{b}$ and is explicitly given in \Cref{cor:bias_pi_rho_pirho_gamma} in the Appendix.
\end{proposition}
With the notation $\bar{\gamma} = \max_{i \in [b]} \gamma_i$, \Cref{prop:bias_gamma} implies that $\wasserstein{}(\Pi_{\bfrho}, \Pi_{\bfrho,\bfgamma,\bfN}) \le \mathcal{O}(\bar{\gamma}^{\half})(\sum_{i=1}^b d_i)^{\half}$ for $\bar{\gamma} \downarrow 0$.
Note that this result is in line with \citet[Corollary 7]{durmus2018high} and can be improved under further regularity assumptions on $U$, as shown below.

\begin{assumption}
  \label{ass:hessian_lipschitz}
  $U$ is three times continuously differentiable and there exists $L_i > 0$ such that for all $\bz_i,\bz_i'\in\R^{d_i}$, $\|\nabla^2 U_i(\bz_i) - \nabla^2 U_i(\bz_i')\|\le L_{i}\|\bz_i-\bz_i'\|$.
\end{assumption}
\begin{proposition}
\label{prop:bias_gamma_bis}
Assume \Cref{ass:well_defined_density}-\Cref{ass:supp_fort_convex}-\Cref{ass:hessian_lipschitz}.
For any $i \in [b]$, define $\tilde{M}_i = M_i + 1/\rho_i$ and let $\gammabf\in (\rset_+^*)^b$, $\bfN\in (\N^*)^b$ such that for any $i \in [b]$, \eqref{eq:gamma_bound} and \eqref{eq:eq:good_choice_Nmain} hold.
Then, we have 
\begin{equation*}
\wasserstein{}^2\pr{\Pi_{\bfrho,\bfgamma,\bfN},\Pi_{\bfrho}}\le C_4 \sum_{i\in[b]}d_i\gamma_i(1/\tilde{M}_i^2 + \gamma_i\tilde{M}_i^2)\eqsp,
\end{equation*}
where $C_4>0$ only depends on $(m_i,M_i,L_i,\B{A}_i,\rho_i)_{i=1}^{b}$ and is explicitly given in \Cref{cor:bias_pi_rho_pirho_gamma_alternative} in the Appendix.
\end{proposition}

\subsubsection{Mixing Time with Explicit Dependencies} 
Based on explicit non-asymptotic bounds shown in Propositions \ref{prop:convergence_rho_gamma}, \ref{prop:bound_bias_rho} and \ref{prop:bias_gamma} and the decomposition \eqref{eq:error_decomp}, we are now able to analyse the scaling of \Cref{algo:ULAwSG} in high dimension.
Given a prescribed precision $\varepsilon>0$ and an initial condition $\B{v}^\star$ with distribution $\mu_{\rhobf}^{\star}$ given in \eqref{eq:_def_mu_star_supp_2main}, we define the $\varepsilon$-mixing time associated to $\Gamma_{\B{v}^\star}$ by
\[
t_{\operatorname{mix}}(\varepsilon ; \B{v}^\star)=\min \big\{t \in \mathbb{N} : W_2\big(\Gamma_{\B{v}^\star}^{t}, \pi\big) \le \varepsilon\big\}\eqsp.
\]
This quantity stands for the minimum number of DG-LMC iterations such that the $\btheta$-marginal distribution is at most at an $\varepsilon$ $W_2$-distance from the initial target $\pi$.
Under the condition that  $b \max_{i \in [b]}d_i = \mathcal{O}(d)$ and by assuming for simplicity that for any $i \in [b]$, $m_i=m, M_i=M, L_i=L, \rho_i=\rho,\gamma_i=\gamma$ and $N_i=N$, \Cref{table:mixing_time} gathers the dependencies w.r.t. $d$ and $\varepsilon$ of the parameters involved in \Cref{algo:ULAwSG} and of $t_{\operatorname{mix}}(\varepsilon;\B{v}^{\star})$ to get a $W_2$-error of at most $\varepsilon$.
Note that the mixing time of \Cref{algo:ULAwSG} scales at most quadratically (up to polylogarithmic factors) in the dimension.
When \Cref{ass:hessian_lipschitz} holds, we can see that the number of local iterations becomes independent of $d$ and $\varepsilon$ which leads to a total number of gradient evaluations with better dependencies w.r.t. to these quantities.
Up to the authors' knowledge, these explicit results are the first among the centralised distributed MCMC literature and in particular give the dependency w.r.t. $d$ and $\varepsilon$ of the number of local LMC iterations on each worker. 
Overall, the proposed approach appears as a scalable and reliable alternative for high-dimensional and distributed Bayesian inference.

\subsection{DG-LMC in Practice: Guidelines for Practitioners}
\label{subsec:discussion}

We now discuss practical guidelines for setting the values of hyperparameters involved in \Cref{algo:ULAwSG}.
Based on \Cref{prop:convergence_rho_gamma}, we theoretically show an optimal choice of order $N_i\gamma_i \asymp m_i\rho_i^2/(\rho_iM_i+1)^2$.
Ideally, within the considered distributed setting, the optimal value for $(N_i,\gamma_i)_{i \in [b]}$ would boil down to optimise the value of $\max_{i \in [b]}\{N_i\gamma_i\}$ under the constraints derived in \Cref{prop:convergence_rho_gamma} combined with communication considerations.
In particular, this would imply a comprehensive modelling of the communication costs including I/O bandwiths constraints.
These optimisation tasks fall outside the scope of the present paper and therefore we let the search of optimal values for future works.
Since our aim here is to provide practitioners with simple prescriptions, we rather focus on general rules involving tractable quantities.

\subsubsection{Selection of $\boldsymbol{\gamma}$ and $\bfrho$}
\label{subsubsec:selection_rho_gamma}

%
From \citet{Durmus2017} and references therein, a simple sufficient condition on step-sizes $\bfgamma = \{\gamma_i\}_{i=1}^b$ to guarantee the stability of LMC is $\gamma_i \le \rho_i/(\rho_iM_i + 1)$ for $i \in [b]$.
Both the values of $\gamma_i$ and $\rho_i$ are subject to a bias-variance trade-off.
More precisely, large values yield a Markov chain with small estimation variance but high asymptotic bias.
Conversely, small values produce a Markov chain with small asymptotic bias but which requires a large number of iterations to obtain a stable estimator.
We propose to mitigate this trade-off by setting $\gamma_i$ to a reasonably large value, that is for $i \in [b]$, $\gamma_i \in [0.1\rho_i/(\rho_iM_i + 1),0.5\rho_i/(\rho_iM_i + 1)]$.
Since $\gamma_i$ saturates to $1/M_i$ when $\rho_i \rightarrow \infty$, there is no computational advantage to choose very large values for $\rho_i$.
Based on several numerical studies, we found that setting $\rho_i$ of the order of $1/M_i$ was a good compromise between computational efficiency and asymptotic bias.

\subsubsection{$\bfN$: A Trade-Off between Asymptotic Bias and Communication Overhead}

In a similar vein, the choice of $\bfN = \{N_i\}_{i=1}^b$ also stands for a trade-off but here between asymptotic accuracy and communication costs.
Indeed, a large number of local LMC iterations reduces the communication overhead but at the expense of a larger asymptotic bias since the master parameter is not updated enough.
\citet{Ahn14} proposed to tune the number of local iterations $N_i$ on a given worker based on the amount of time needed to perform one local iteration, denoted here by $\uptau_i$.
Given an average number of local iterations $N_{\mathrm{avg}}$, the authors set $N_i = q_iN_{\mathrm{avg}}b$ with $q_i = \uptau_i^{-1}/\sum_{k=1}^b\uptau_k^{-1}$ so that $b^{-1}\sum_{i=1}^bN_i = N_{\mathrm{avg}}$.
As mentioned by the aforementioned authors, this choice allows to keep the block-by-the-slowest delay small by letting fast workers perform more iterations in the same wall-clock time.
Although they showed how to tune $N_i$ w.r.t. communication considerations, they let the choice of $N_{\mathrm{avg}}$ to the practitioner.
Here, we propose a simple guideline to set $N_{\mathrm{avg}}$ such that $N_i$ stands for a good compromise between the amount of time spent on exploring the state-space and communication overhead.
As highlighted in the discussion after \Cref{prop:convergence_rho_gamma}, as $\gamma_i$ becomes smaller, more local LMC iterations are required to sufficiently explore the latent space before the global consensus round on the master.
Assuming for any $i \in [b]$ that $\gamma_i$ has been chosen following our guidelines in \Cref{subsubsec:selection_rho_gamma}, this suggests to set $N_{\mathrm{avg}} = \lceil(1/b)\sum_{i\in[b]}\rho_i/(\gamma_i[\rho_iM_i + 1])\rceil$.

\section{Related work}
\label{sec:comparison}

As already mentioned in \Cref{sec:introduction}, hosts of contributions have focused on deriving distributed MCMC algorithms to sample from \eqref{eq:target_density}.
This section briefly reviews the main existing research lines and draws a detailed comparison with the proposed methodology.

\subsection{Existing distributed MCMC methods}

\begin{table}[t]
\caption{Synthetic overview of the main existing distributed MCMC methods under a master-slave architecture. The column \emph{Exact} means that the Markov chain defined by the MCMC sampler admits \eqref{eq:target_density} as invariant distribution.
The column \emph{Comm.} reports the communication frequency. A value of $1$ means that the sampler communicates after every iteration. $T$ stands for the total number of iterations and $N < T$ is a tunable parameter to mitigate communication costs.
The acronym D-SGLD stands for distributed stochastic gradient Langevin dynamics.}
\label{table:overview}
\vskip 0.15in
\begin{center}
\begin{footnotesize}
\begin{sc}
\begin{tabular}{lccccc}
\toprule
Method & Type & Exact & Comm. & Bias bounds & Scaling\\
\midrule
\citet{Wang2013} & one-shot & $\times$ & $1/T$ & $\surd$ & $\mathcal{O}(\mathrm{e}^d)$  \\
\citet{Neiswanger2014} & one-shot & $\times$ & $1/T$ & $\times$ & $\mathcal{O}(\mathrm{e}^d)$ \\
\citet{Minsker2014} & one-shot & $\times$ & $1/T$ & $\surd$ & unknown \\
\citet{Srivastava2015} & one-shot & $\times$ & $1/T$ & $\times$ & unknown \\
\citet{Wang2015} & one-shot & $\times$ & $1/T$ & $\surd$ & $\mathcal{O}(\mathrm{e}^d)$ \\
\citet{Scott2016} & one-shot & $\times$ & $1/T$ & $\times$ & unknown \\
\citet{nemeth2018} & one-shot & $\times$ & $1/T$ & $\times$ & unknown \\
\citet{Jordan2019} & one-shot & $\times$ & $1/T$ & $\surd$ & unknown \\
\citet{Ahn14} & D-SGLD & $\times$ & $1/N$ & $\times$ & unknown \\
\citet{Chen16_stale_gradients} & D-SGLD & $\times$ & $1$ & $\surd$ & unknown  \\
\citet{2020_conducive_gradients} & D-SGLD & $\times$ & $1/N$ & $\surd$ & unknown \\
\citet{2015_variational_consensus} & g. consensus & $\times$ & $1/N$ & $\times$ & unknown \\
\citet{pmlr-v84-chowdhury18a} & g. consensus & $\surd$ & $1$ & N/A & unknown \\
\citet{Rendell2020} & g. consensus & $\times$ & $1/N$ & $\surd$ & unknown \\
 This paper & g. consensus & $\times$ & $1/N$ & $\surd$ & $\mathcal{O}(d^2\log(d))$ \\
\bottomrule
\end{tabular}
\end{sc}
\end{footnotesize}
\end{center}
\vskip -0.1in
\end{table}
Existing methodologies are mostly approximate and can be loosely speaking divided into three groups: \emph{one-shot}, distributed stochastic gradient MCMC and \emph{global consensus} approaches.
To ease the understanding, a synthetic overview of their main characteristics is presented in \Cref{table:overview}.

One-shot approaches stand for communication-efficient schemes where workers and master only exchange information at the very beginning and the end of the sampling task; similarly to MapReduce schemes \citep{DeanGehmawat_MapReduce2014}.
Most of these methods assume that the posterior density factorises into a product of local posteriors and launch independent Markov chains across workers to target them.
The local posterior samples are then combined through the master node using a single final aggregation step. This step turns to be the milestone of one-shot approaches and was the topic of multiple contributions \citep{Wang2013,Neiswanger2014,Minsker2014,Srivastava2015,Scott2016,nemeth2018}.
Unfortunately, the latter are either infeasible in high-dimensional settings or have been shown to yield inaccurate posterior representations empirically, if the posterior is not near-Gaussian, or if the local posteriors differ significantly \citep{Wang2015,dai_pollock_roberts_2019,Rendell2020}.
Alternative schemes have been recently proposed to tackle these issues but their theoretical scaling w.r.t. the dimension $d$ is currently unknown \citep{Jordan2019,mesquita20}.

Albeit popular in the machine learning community, distributed stochastic gradient MCMC methods \citep{Ahn14} suffer from high variance when the dataset is large because of the use of stochastic gradients \citep{Brosse2018}.
Some surrogates have been recently proposed to reduce this variance such as the use of \emph{stale} or \emph{conducive}  gradients \citep{Chen16_stale_gradients,2020_conducive_gradients}.
However, these variance reduction methods require an increasing number of workers for the former and come at the price of a prohibitive pre-processing step for the latter.
In addition, it is currently unclear whether these methods are able to generate efficiently accurate samples from a given target distribution.

Contrary to aforementioned distributed MCMC approaches, global consensus methods periodically share information between workers by performing a consensus round between the master and the workers \citep{2015_variational_consensus,pmlr-v84-chowdhury18a,Vono2019,Rendell2020}.
Again, they have been shown to perform well in practice but their theoretical understanding is currently limited. 

\subsection{Comparison with the proposed methodology}
\label{subsec:method_complexity}

\Cref{table:overview} compares \Cref{algo:ULAwSG} with existing approaches detailed previously.
In addition to having a simple implementation and guidelines, it is worth noticing that DG-LMC appears to benefit from favorable convergence properties compared to the other considered methodologies.

We complement this comparison with an informal discussion on the computational and communication complexities of \Cref{algo:ULAwSG}.
Recall that the dataset is assumed to be partitioned into $S$ shards and stored on $S$ workers among a collection of $b$ computing nodes. 
Suppose that the $s$-th shard has  size $n_s$, and let $T$ be the number of total MCMC iterations and $c_{\mathrm{com}}$ the communication cost. 
In addition, denote by $c_{\mathrm{eval}}^{(i)}$ the approximate wall-clock time required to evaluate $U_i$ or its gradient.  
For the ease of exposition, we do not discuss the additional overhead due to bandwidth restrictions and assume similar computation costs, \emph{i.e.}, $Nc_{\mathrm{eval}}\simeq N_ic_{\mathrm{eval}}^{(i)}$, to perform each local LMC step at each iteration of \Cref{algo:ULAwSG}.  
Under these assumptions, the total complexity of \Cref{algo:ULAwSG} is $\mathcal{O}(T[2c_{\mathrm{com}} + Nc_{\mathrm{eval}}])$. 
Following the same reasoning,  distributed stochastic gradient Langevin dynamics (D-SGLD) and one-shot approaches admit complexities of the order $\mathcal{O}(T[2c_{\mathrm{com}} +
Nc_{\mathrm{eval}}n_{\mathrm{mb}}/n_s])$ and
$\mathcal{O}(Tc_{\mathrm{eval}} + 2c_{\mathrm{com}})$, respectively.
The integer $n_{\mathrm{mb}}$ stands for the
mini-batch size used in D-SGLD.  
Despite their very low communication overhead,
existing one-shot approaches are rarely reliable and therefore not necessarily efficient to sample from $\pi$ given a prescribed computational budget, see \citet{Rendell2020} for a recent overview.
D-SGLD seems to enjoy a lower complexity than
\Cref{algo:ULAwSG} when $n_{\mathrm{mb}}$ is small.  
Unfortunately, this choice comes with two main shortcomings: (i) a larger number of iterations $T$ to achieve the same precision because of higher variance of gradient estimators, and (ii) a smaller amount of time spent on exploration compared to communication latency. 
By falling into the global consensus class of methods, the proposed methodology hence appears as a good compromise between one-shot and D-SGLD algorithms in terms of both computational complexity and accuracy.
\Cref{sec:experiments} will enhance the benefits of \Cref{algo:ULAwSG} by showing experimentally better convergence properties and posterior approximation.

\section{Experiments}
\label{sec:experiments}

This section compares numerically DG-LMC with the most popular and recent centralised distributed MCMC approaches namely D-SGLD and the global consensus Monte Carlo (GCMC) algorithm proposed in \citet{Rendell2020}.
Since all these approaches share the same communication latency, this feature is not discussed here.

\begin{figure}[h]
\vskip -0.1in
\begin{center}
\mbox{{\includegraphics[scale=0.7]{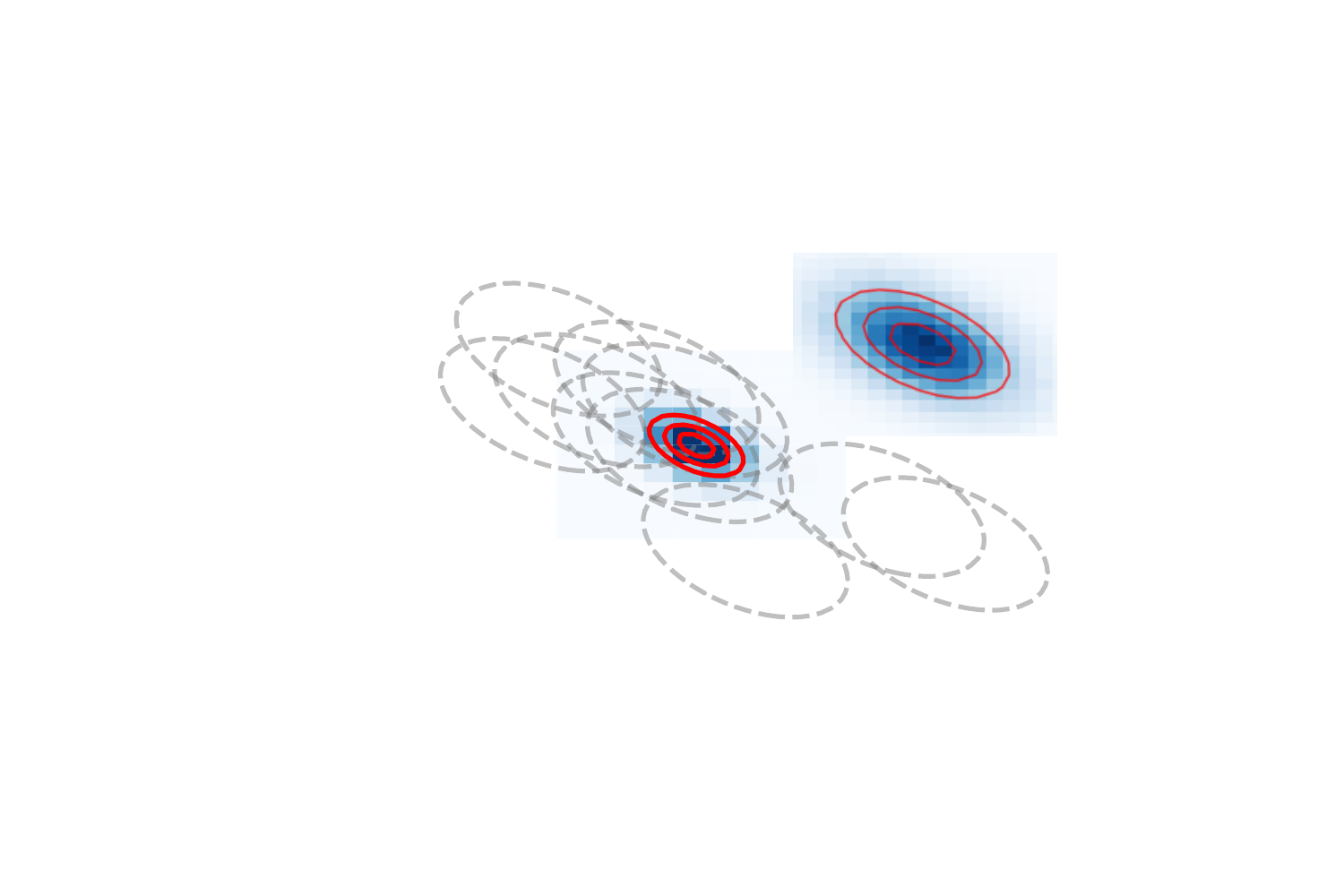}}}
\mbox{{\includegraphics[scale=0.7]{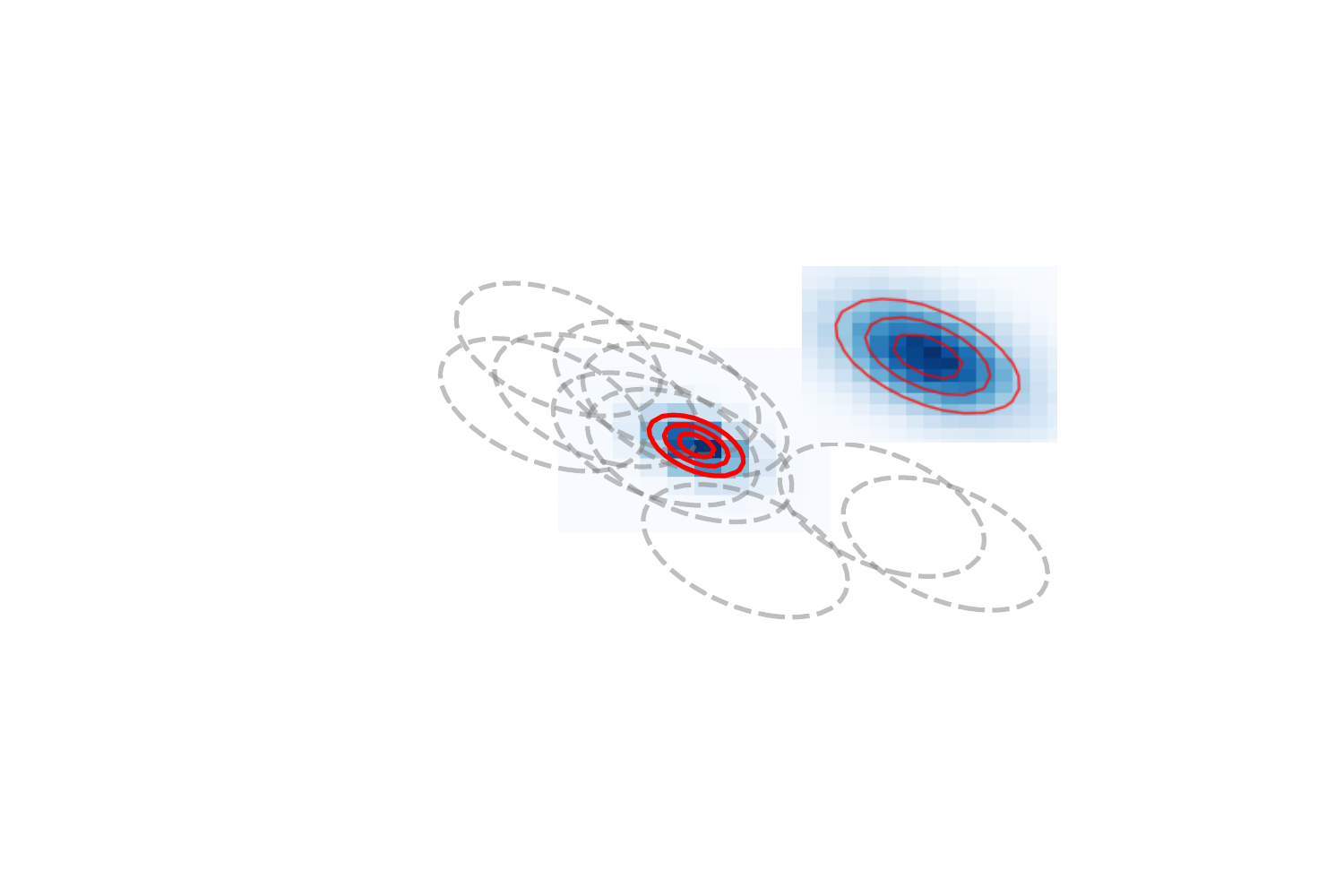}}}
\mbox{{\includegraphics[scale=0.7]{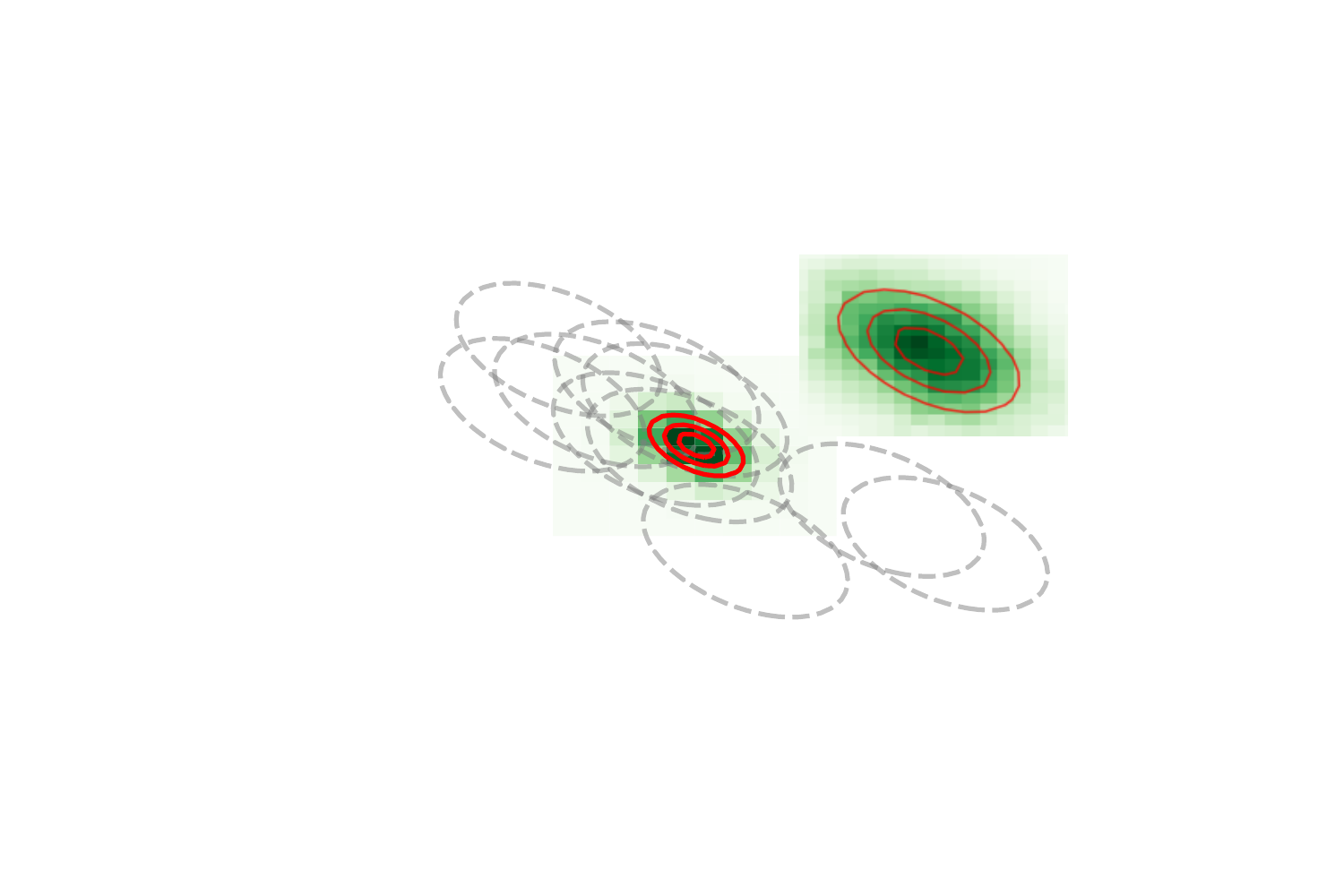}}}
\mbox{{\includegraphics[scale=0.7]{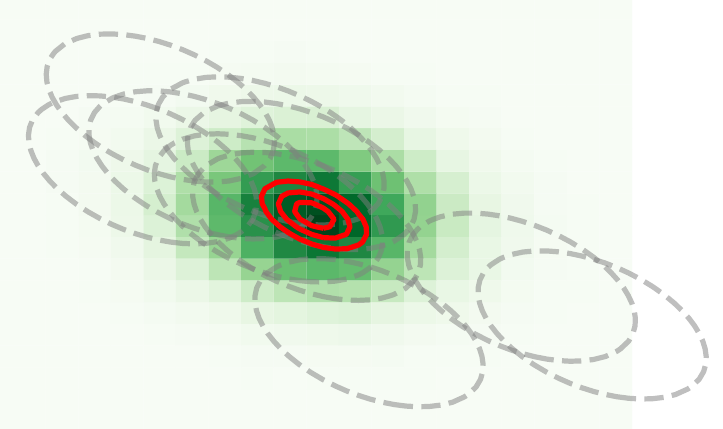}}}
\mbox{{\includegraphics[scale=0.33]{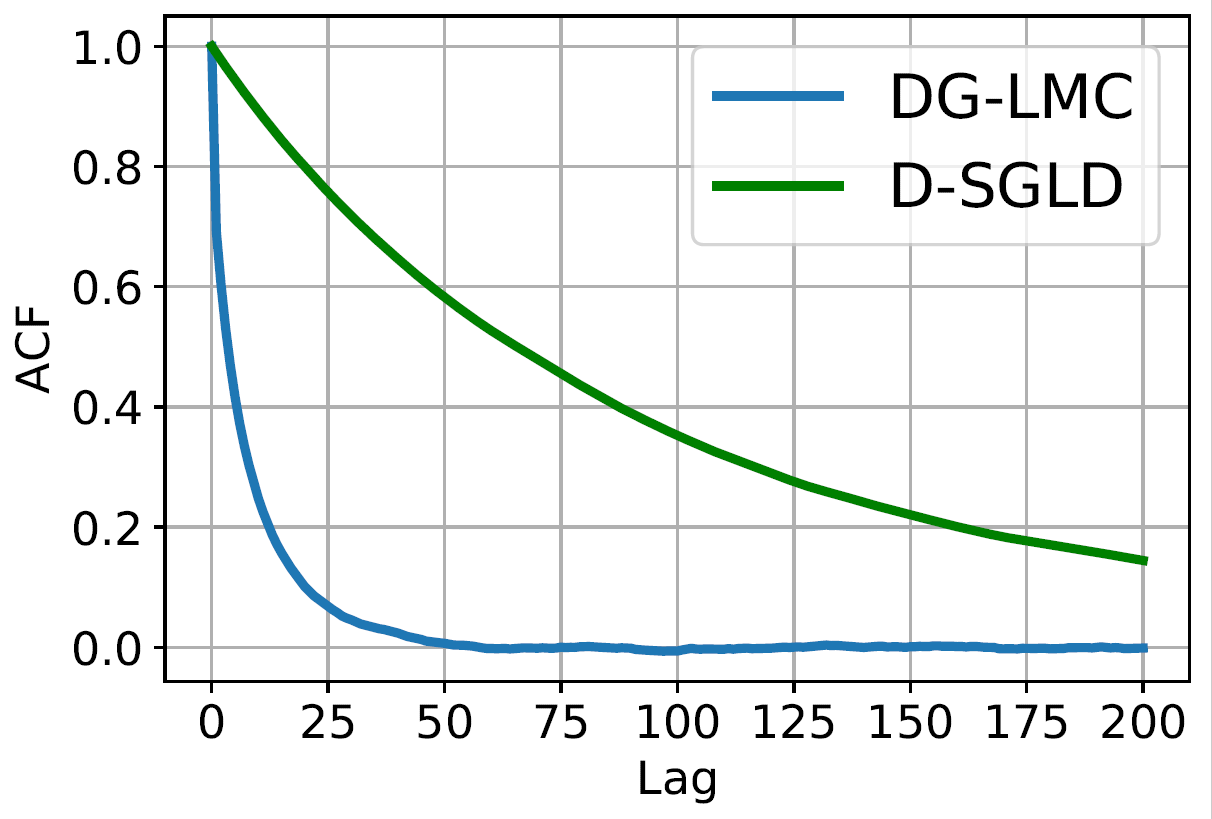}}}
\mbox{{\includegraphics[scale=0.5]{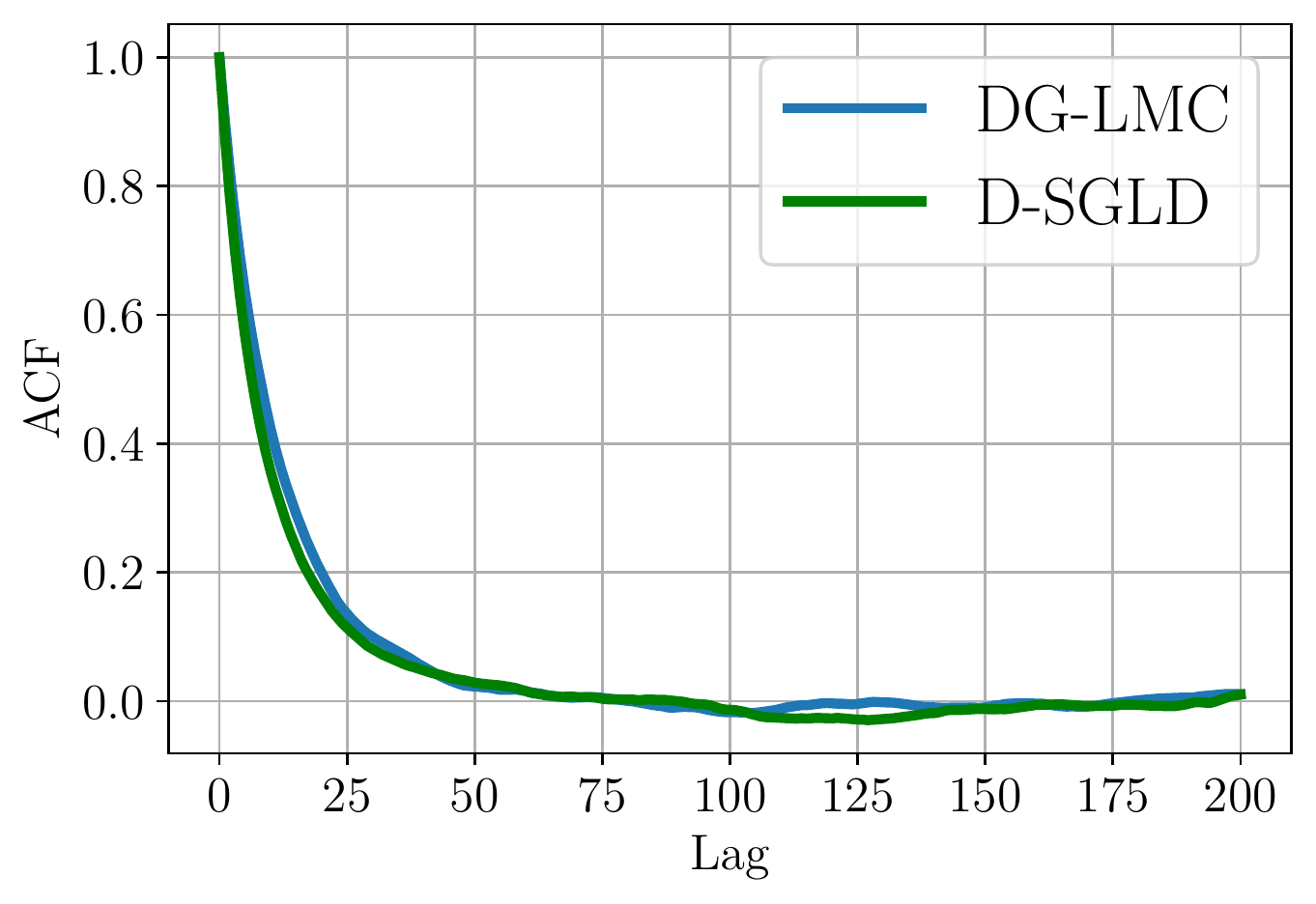}}}
\caption{Toy Gaussian experiment. (left) $N=1$ local iterations and (right) $N=10$. (top) DG-LMC, (middle) D-SGLD and (bottom) ACF comparison between DG-LMC and D-SGLD.}
\label{fig:exp1}
\end{center}
\vskip -0.2in
\end{figure}
\begin{figure*}
\begin{center}
\mbox{{\includegraphics[scale=0.5]{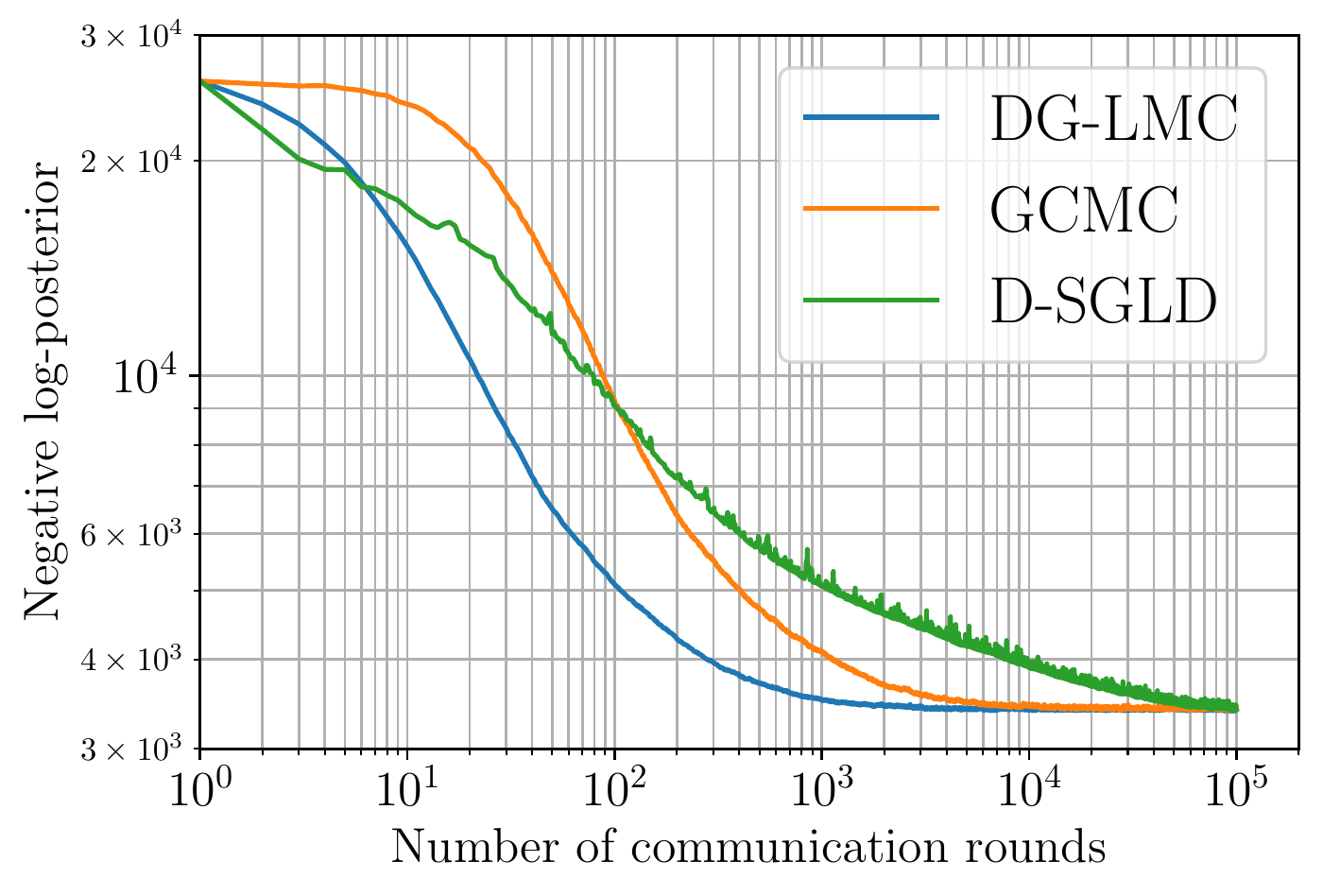}}}
\mbox{{\includegraphics[scale=0.5]{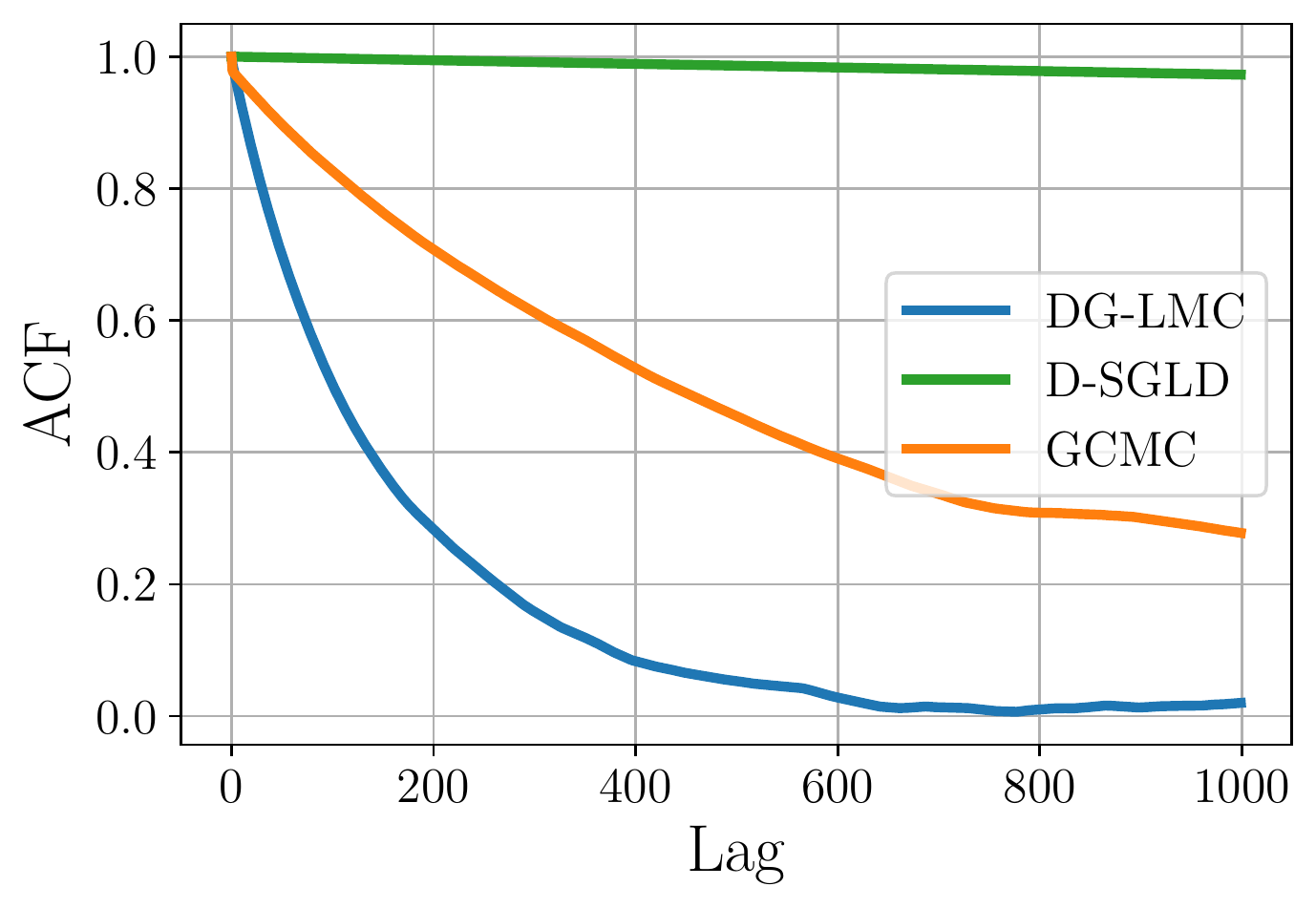}}}
\mbox{{\includegraphics[scale=0.5]{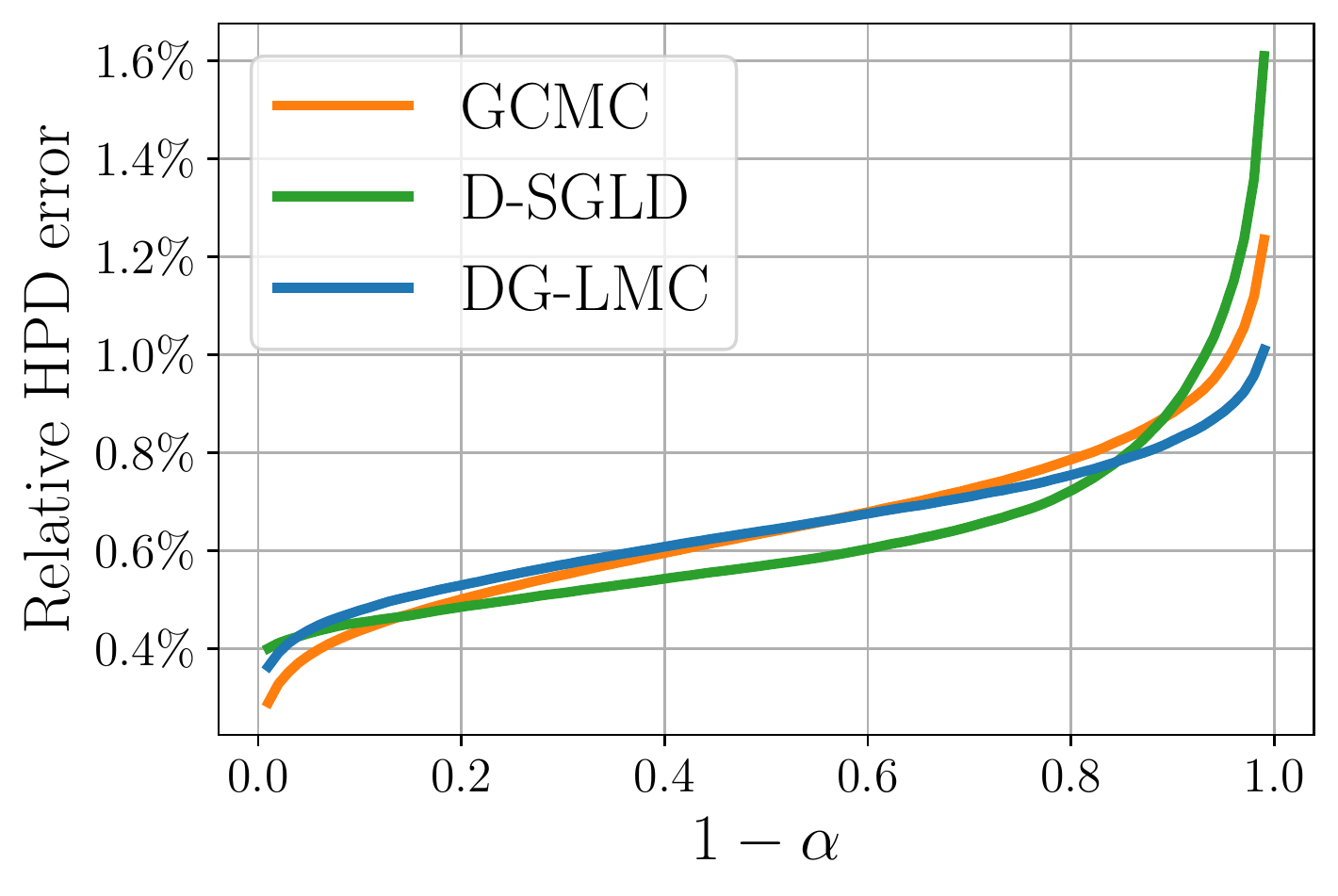}}}
\mbox{{\includegraphics[scale=0.5]{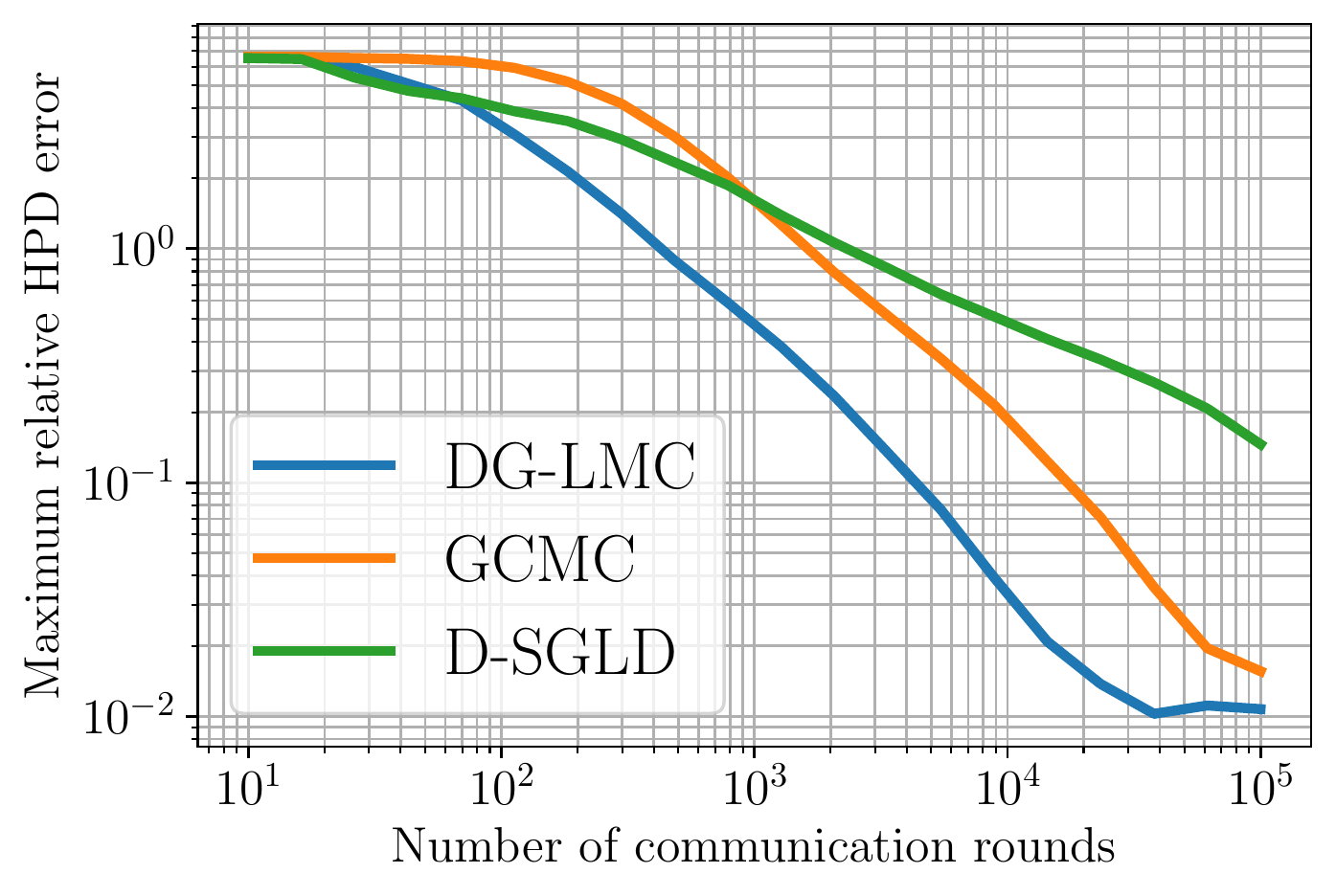}}}
\caption{Logistic regression. From left to right: negative log-posterior, ACF, HPD relative error after and during the sampling procedure.}
\label{fig:exp2}
\end{center}
\vskip -0.2in
\end{figure*}

\subsection{Toy Gaussian Example}
\label{subsec:toy_gaussian_example}

In this toy example, we first illustrate the behavior of DG-LMC w.r.t. the number of local iterations which drives the communication overhead.
We consider the conjugate Gaussian model $\pi(\btheta|\by_{1:n}) \propto \loiGauss(\btheta|\B{0}_d,\B{\Sigma_0})\prod_{i=1}^n\loiGauss(\B{y}_i|\btheta,\B{\Sigma_1})$, with positive definite matrices $\B{\Sigma_0},\B{\Sigma_1}$.
We set $d=2$, allocate $n=20,000$ observations to a cluster made of $b=10$ workers and compare DG-LMC with D-SGLD.
Both MCMC algorithms have been run using the same number of local iterations $N$ per worker and for a fixed budget of $T=100,000$ iterations including a burn-in period equal to $T_{\mathrm{bi}} = 10,000$.
Regarding DG-LMC, we follow the guidelines in \Cref{subsubsec:selection_rho_gamma} and set for all $i \in [b]$, $\B{A}_i = \B{I}_d$, $\rho_i = 1/(5M_i)$ and $\gamma_i =  0.25\rho_i/(\rho_i M_i +1)$.
On the other hand, D-SGLD has been run with batch-size $n/(10b)$ and a step-size chosen such that the resulting posterior approximation is similar to that of DG-LMC for $N=1$.
\Cref{fig:exp1} depicts the results for $N=1$ and $N=10$ on the left and right columns, respectively.
The top row (resp. middle row) shows the contours of the $b$ local posteriors in dashed grey, the contours of the target posterior in red and the 2D histogram built with DG-LMC (resp. D-SGLD) samples in blue (resp. green). 
When required, a zoomed version of these figures is depicted at the top right corner.
It can be noted that DG-LMC exhibits better mixing properties while achieving similar performances as shown by the autocorrelation function (ACF) on the bottom row.
Furthermore, its posterior approximation is robust to the choice of $N$ in contrast to D-SGLD, which needs further tuning of its step-size to yield an accurate posterior representation.
This feature is particularly important for distributed computations since $N$ is directly related to communication costs and might often change depending upon the hardware architecture.

\subsection{Bayesian Logistic Regression}
\begin{figure}
\begin{center}
\mbox{{\includegraphics[scale=0.35]{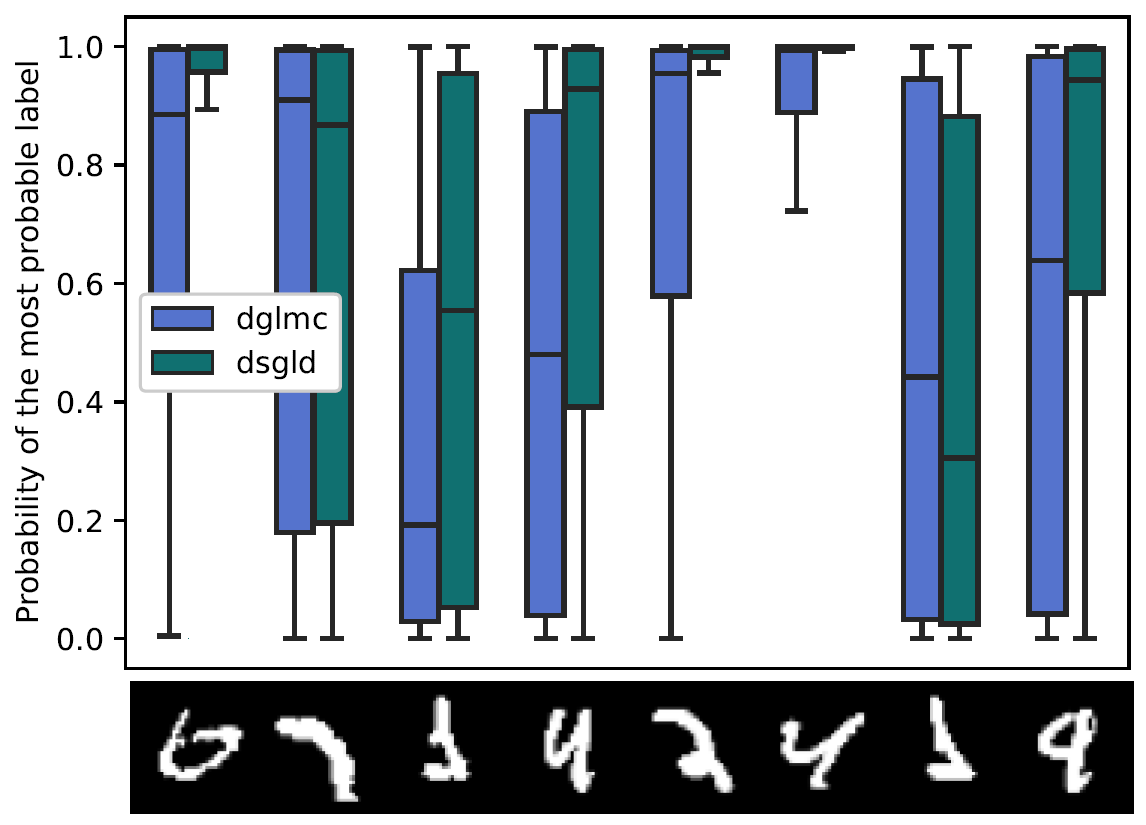}}}
\mbox{{\includegraphics[scale=0.35]{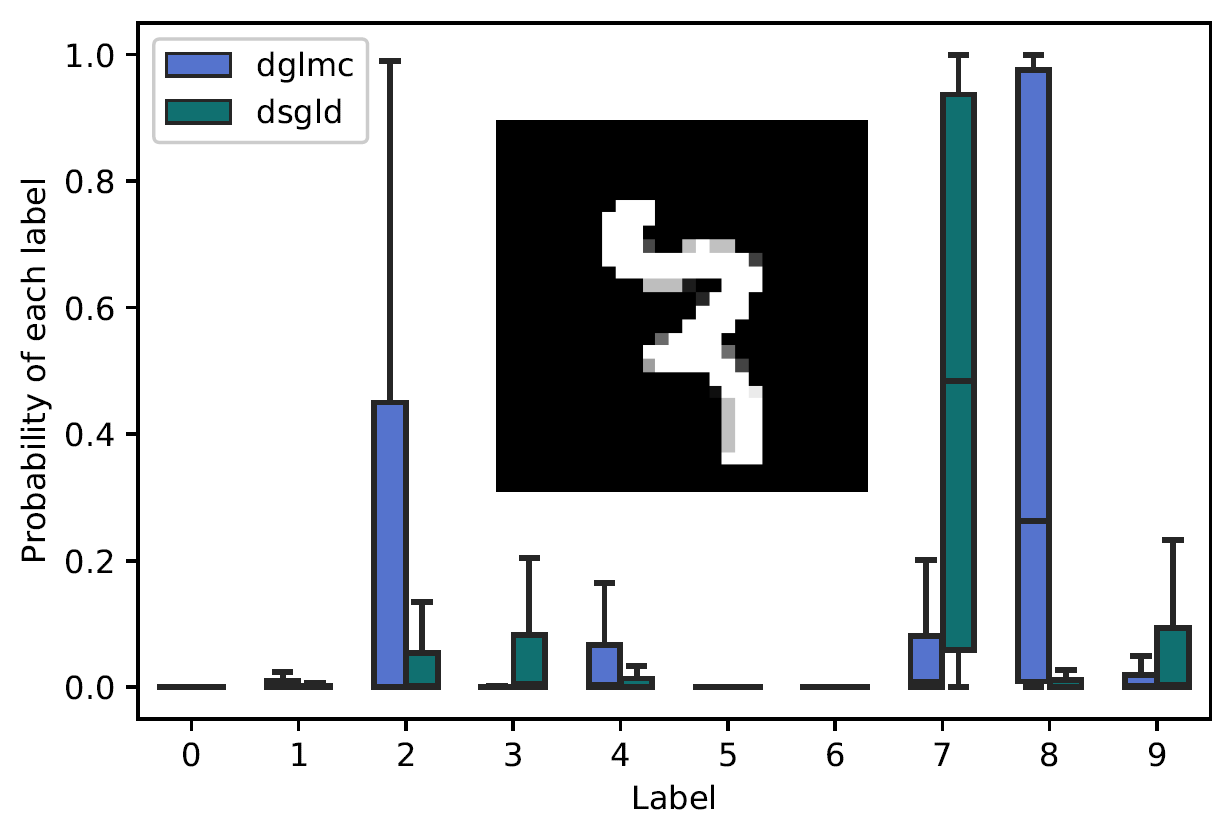}}}
\caption{Bayesian neural network. (left) probability of the most probable label for 8 examples and (right) probability of each label for a single example.}
\label{fig:exp3}
\end{center}
\vskip -0.2in
\end{figure}
This second experiment considers a more challenging problem namely Bayesian logistic regression.
We use the \textit{covtype}\footnote{\scriptsize\url{www.csie.ntu.edu.tw/~cjlin/libsvmtools/datasets}} dataset with $d=54$ and containing $n$ = 581,012 observations partitioned into $b=16$ shards.
We set $N=10$, $T=200,000$, $T_{\mathrm{bi}}=T/10$ for all approaches, and again used the guidelines in \Cref{subsubsec:selection_rho_gamma} to tune DG-LMC.
Under the Bayesian paradigm, we are interested in performing uncertainty quantification by estimating highest posterior density (HPD) regions.
For any $\alpha \in (0,1)$, define $\mathcal{C}_{\alpha} = \{\btheta \in \Rd ; -\log \pi(\btheta|\by_{1:n}) \le \eta_{\alpha}\}$ where $\eta_{\alpha} \in \mathbb{R}$ is chosen such that $\int_{\mathcal{C}_{\alpha}}\pi(\btheta|\by_{1:n})\dd\btheta = 1-\alpha$.
For the three approximate MCMC approaches, we computed the relative HPD error based on the scalar summary $\eta_{\alpha}$, \emph{i.e.} $|\eta_{\alpha}-\eta_{\alpha}^{\mathrm{true}}|/\eta_{\alpha}^{\mathrm{true}}$ where $\eta_{\alpha}^{\mathrm{true}}$ has been estimated using the Metropolis adjusted Langevin algorithm.
The parameters of GCMC and D-SGLD have been chosen such that all MCMC algorithms achieve similar HPD error.
\Cref{fig:exp2} shows that this error is reasonable and of the order of 1\%.
Nonetheless, one can denote that DG-LMC achieves this precision level faster than GCMC and D-SGLD due to better mixing properties.
This confirms that the proposed methodology is indeed efficient and reliable to perform Bayesian analyses compared to existing popular methodologies.

\subsection{Bayesian Neural Network}

Up to now, both our theoretical and experimental results focused on the strongly log-concave scenario and showed that even in this case, DG-LMC appeared as a competitive alternative.
In this last experiment, we propose to end the study of DG-LMC on an open note without ground truth by tackling the challenging sampling problem associated to Bayesian neural networks. 
We consider the MNIST training dataset consisting of $n=60,000$ observations partitioned into $b=50$ shards and such that for any $i \in [n]$ and $k \in [10]$, $\mathbb{P}(y_i = k | \btheta,\B{x}_i) = \beta_k$ where $\beta_k$ is the $k$-th element of $\sigma(\sigma(\B{x}_i^{\top}\B{W}_1 + \B{b}_1)\B{W}_2 + \B{b}_2)$, $\sigma(\cdot)$ is the sigmoid function, $\B{x}_i$ are covariates, and $\B{W}_1$, $\B{W}_2$, $\B{b}_1$ and $\B{b}_2$ are matrices of size 784$\times$128, 128 $\times$ 10, 1$\times$128 and 1$\times$10, respectively. 
We set normal priors for each weight matrix and bias vector, $N=10$ and ran DG-LMC with constant hyperparameters across workers $(\rho,\gamma) = (0.02,0.005)$ and D-SGLD using a step-size of $10^{-5}$.
Exact MCMC approaches are too computationally costly to launch for this experiment and therefore no ground truth about the true posterior distribution is available.
To this purpose, \Cref{fig:exp3} only compares the credibility regions associated to the posterior predictive distribution.
Similarly to previous experiments, we found that D-SGLD was highly sensitive to hyperparameters choices (step-size and mini-batch size).
Except for a few testing examples, most of conclusions given by DG-LMC and D-SGLD regarding the predictive uncertainty coincide.
In addition, posterior accuracies on the test set given by both algorithms are similar.

\section{Conclusion}

In this paper, a simple algorithm coined DG-LMC has been introduced for distributed MCMC sampling.  In addition, it has been established that this method inherits favorable convergence properties and numerical illustrations support our claims. 

\section*{Acknowledgements}

The authors acknowledge support of the Lagrange Mathematics and Computing Research Center.

\bibliography{../../bibliography/biblio}
\bibliographystyle{chicago}

\newpage

\appendix

\newtheorem{unlemma}{Lemma S}
\newtheorem{unproposition}{Proposition S}
\newtheorem{uncorollary}{Corollary S}
\newtheorem{untheorem}{Theorem S}

\setcounter{equation}{0}
\setcounter{figure}{0}
\setcounter{table}{0}
\setcounter{lemma}{0}
\setcounter{proposition}{0}
\setcounter{corollary}{0}
\makeatletter
\renewcommand{\theequation}{S\arabic{equation}}
\renewcommand{\thefigure}{S\arabic{figure}}
\renewcommand{\thetheorem}{S\arabic{theorem}}
\renewcommand{\thelemma}{S\arabic{lemma}}
\renewcommand{\thesection}{S\arabic{section}}
\renewcommand{\theproposition}{S\arabic{proposition}}
\renewcommand{\thecorollary}{S\arabic{corollary}}


\begin{center}
\textbf{\large APPENDIX \\ \vspace{.15cm}
DG-LMC: A Turn-key and Scalable Synchronous Distributed MCMC Algorithm via Langevin Monte Carlo within Gibbs}
\end{center}
 \vspace{0.4cm}



\noindent\textbf{Notations and conventions.}
We denote by $\mathcal{B}(\mathbb{R}^d)$ the Borel $\sigma$-field of $\mathbb{R}^d$, $\mathbb{M}(\mathbb{R}^d)$ the set of all Borel measurable functions $f$ on $\mathbb{R}^d$, $\norm{f}_{\infty} = \sup_{\bx \in \mathbb{R}^d}|f(\bx)|$ and $\norm{\cdot}$ the Euclidean norm on $\mathbb{R}^d$.
For $\mu$ a probability measure on $(\Rd,\mathcal{B}(\Rd))$ and $f \in \mathbb{M}(\mathbb{R}^d)$ a $\mu$-integrable function, denote by $\mu(f)$ the integral of $f$ with respect to (w.r.t.) $\mu$.
Let $\mu$ and $\nu$ be two sigma-finite measures on $(\Rd,\mathcal{B}(\Rd))$. 
Denote by $\mu \ll \nu$ if $\mu$ is absolutely continuous w.r.t. $\nu$ and $\dd \mu/\dd \nu$ the associated density. 
Let $\mu$, $\nu$ be two probability measures on $(\Rd,\mathcal{B}(\Rd))$. 
Define the Kullback-Leibler (KL) divergence
of $\mu$ from $\nu$ by
\begin{equation*}
\mathrm{KL} (\mu|\nu) = 
\begin{cases}
  \int_{\Rd} \frac{\dd \mu}{\dd \nu}(\bx)\log\pr{\frac{\dd \mu}{\dd \nu}(\bx)}\,\dd \nu(\bx)\eqsp, & \text{if $\mu \ll \nu$}\\
  +\infty & \text{otherwise.}
\end{cases}  
\end{equation*}
In addition, define the Pearson $\chi^2$-divergence of $\mu$ from $\nu$ by
\begin{equation*}
\chi^2 (\mu|\nu) = 
\begin{cases}
  \int_{\Rd} \pr{\frac{\dd \mu}{\dd \nu}(\bx) - 1}^2\,\dd \nu(\bx)\eqsp, & \text{if $\mu \ll \nu$}\\
  +\infty & \text{otherwise.}
\end{cases}  
\end{equation*}
We say that $\zeta$ is a transference plan of $\mu$ and $\nu$ if it is a probability measure on $(\Rd \times \Rd, \mathcal{B}(\Rd \times \Rd))$ such that for all measurable set $\mathsf{A}$ of $\Rd$, $\zeta(\mathsf{A} \times \Rd) = \mu(\mathsf{A})$
and $\zeta(\Rd \times \mathsf{A}) = \nu(\mathsf{A})$.
 We denote by $\mathcal{T}(\mu, \nu)$ the set of transference plans of $\mu$ and $\nu$.
 In addition, we say that a couple of $\mathbb{R}^d$-random variables $(X,Y)$ is a coupling of $\mu$ and $\nu$ if there exists $\zeta \in \mathcal{T}(\mu, \nu)$ such that $(X,Y)$ are distributed according to $\zeta$.
 Let $\B{M}$ be a $d\times d$ symmetric positive definite matrix. Denote $\left\langle, \right\rangle_{\B{M}}$ the scalar product corresponding to $\B{M}$, defined for any $\B{x}, \B{y} \in\Rd$ by $\langle\B{x},\B{y} \rangle_{\B{M}} = \B{x}^{\top} \B{M} \B{y}$. Denote $\|\cdot\|_{\B{M}}$ the corresponding norm. 
We denote by $\mathcal{P}_2(\Rd)$ the set of probability measures with finite $2$-moment: for all $\mu \in \mathcal{P}_2(\Rd), \int_{\Rd} \|\B{x}\|^2\,\dd\mu(\B{x}) < \infty$. 
We define the Wasserstein distance of order $2$ associated with $\|\cdot\|_{\B{M}}$ for any probability measures $\mu,\nu \in \mathcal{P}_2(\Rd)$ by
\[
W_{\B{M}}^2 (\mu, \nu) = \inf_{\zeta \in \mathcal{T}(\mu,\nu)} \int_{\mathbb{R}^d \times \mathbb{R}^d}\|\B{x}-\B{y}\|_{\B{M}}^2\,\dd\zeta(\B{x},\B{y}) \eqsp.
\]
In the case when $\B{M} = \B{I}_d$, we will denote the Wasserstein distance of order $2$ by $\wasserstein{}$. 
By \citet[Theorem 4.1]{Villani2008}, for all $\mu$, $\nu$ probability measures on $\Rd$, there exists a transference
plan $\zeta^{\star} \in \mathcal{T}(\mu,\nu)$ such that for any coupling $(X,Y)$ distributed according to $\zeta^{\star}$, $W_{\B{M}}(\mu, \nu) = \mathbb{E}[\|\B{x}-\B{y}\|_{\B{M}}^2]^{1/2}$. 
This kind of transference plan (respectively coupling) will be called an optimal transference plan (respectively optimal coupling) associated with $W_{\B{M}}$. 
By \citet[Theorem 6.16]{Villani2008}, $\mathcal{P}_2(\Rd)$ equipped with the
Wasserstein distance $W_{\B{M}}$ is a complete separable metric space.
The total variation norm between two probability measures $\mu$ and $\nu$ on $(\mathbb{R}^d,\mathcal{B}(\mathbb{R}^d))$ is defined by 
\begin{equation*}
    \norm{\mu - \nu}_{\mathrm{TV}} = \sup_{f \in \mathbb{M}(\mathbb{R}^d), \norm{f}_{\infty} \le 1} \left|\int_{\mathbb{R}^d}f(\bx)\,\dd\mu(\bx) - \int_{\mathbb{R}^d}f(\bx)\,\dd\nu(\bx) \right|\eqsp.
\end{equation*}
For the sake of simplicity, with little abuse, we shall use the same notations for
a probability distribution and its associated probability density function.
For a Markov chain with transition kernel $P$ on $\mathbb{R}^d$ and invariant distribution $\pi$, we define the $\varepsilon$-mixing time associated to a statistical distance $D$, precision $\varepsilon > 0$ and initial distribution $\nu$, by
\begin{equation*}
t_{\mathrm{mix}}(\varepsilon;\nu) = \min \ac{t \ge 0 \ \big\vert \ D(\nu P^t,\pi) \le \varepsilon}\eqsp,
\end{equation*}
which stands for the minimum number of steps of the Markov chain such that its distribution is at most at an $\varepsilon$ $D$-distance from the invariant distribution $\pi$. 
For $n \ge 1$, we refer to the set of integers between $1$ and $n$ with the notation $[n]$.
The $d$-multidimensional Gaussian probability distribution with mean $\boldsymbol{\mu}$ and covariance matrix $\B{\Sigma}$ is denoted by $\loiGauss(\boldsymbol{\mu},\B{\Sigma})$.
When $\boldsymbol{\mu} = \B{0}_d$ and $\B{\Sigma} = \B{I}_d$, the associated probability density function is denoted by by $\upphi_d$.
Let $F:\R^d\to\R$ be a twice continuously differentiable function, denote $\vec{\Delta}$ the vector Laplacian of $F$ defined, for all $x\in\R^d$, by $\vec{\Delta}F(x)=\acn{\sum_{l=1}^{d}(\partial^2 F_k)(x)/\partial x_l^2}_{k=1}^{d}$.
For $0\le i < j$, we use the notation $\B{u}_{i:j}$ to refer to the vector $[\B{u}_{i}^{\top},\cdots,\B{u}_{j}^{\top}]^{\top}$ built by stacking $j-i+1$ vectors ($\B{u}_{k}; \, k \in \{i,\cdots,j\}$).
For a given matrix $\B{M} \in \mathbb{R}^{d \times d}$, we denote its smallest and largest eigenvalues by $\lambda_{\mathrm{min}}(\B{M})$ and $\lambda_{\mathrm{max}}(\B{M})$, respectively. 
Fix $b\in\N^*$ and let $\B{M}_{1},\ldots,\B{M}_b$ be $d$-dimensional matrices. We denote $\prod_{\ell=i}^j \B{M}_\ell = \B{M}_j \ldots \B{M}_i$ if $i\le j$ and with the convention $\prod_{\ell=i}^j \B{M}_\ell = \B{I}_d$ if $i >j$.
For any $b\in\N^*$, $(d_i)_{i\in[b]}\in(\N^*)^b$ and $(\B{M}_i)_{i\in[b]}\in \otimes_{i\in[b]} \R^{d_i\times d_i}$, we denote $\mathrm{diag}(\B{M}_1,\ldots,\B{M}_b)$ the unique matrix $\B{M}\in\R^{(\sum_i d_i)\times (\sum_i d_i)}$ satisfying for any $\B{u}=(\B{u}_1,\ldots,\B{u}_b)\in\R^{d_1}\times\cdots\times\R^{d_b}$, $\B{M}\B{u}=\sum_{i=1}^b \B{M}_i\B{u}_i$ which corresponds to
\[\renewcommand{\arraystretch}{2}
\B{M}=
\begin{pmatrix}
\B{M}_1 & \B{0}_{d_1,d_2} & \cdots & \B{0}_{d_b,d_b}\\
\B{0}_{d_2,d_1} & \ddots & \ddots & \vdots\\
\vdots & \ddots & \ddots & \B{0}_{d_{b-1}, d_{b}} \\
\B{0}_{d_b,d_1} & \cdots & \B{0}_{d_{b}, d_{b-1}} & \B{M}_{b}
\end{pmatrix}\eqsp.
\]
For any $\mathbf{v}\in\R^b$, define the block diagonal matrix
\begin{equation}\label{eq:def_diag}
  \B{D}_{\mathbf{v}} = \mathrm{diag}\pr{v_1 \cdot\B{I}_{d_1},\hdots,v_b\cdot\B{I}_{d_b}}\in\R^{p\times p}\eqsp.
\end{equation}
For any symmetric matrices $\B{S}_{1}, \B{S}_{2}\in\R^{p\times p}$, we note $\B{S}_{1}\preccurlyeq\B{S}_{2}$ if and only if, for any $\B{u}\in\R^{p}$, we have $\B{u}^\top(\B{S}_{2}-\B{S}_{1})\B{u}\ge 0$.
Let $(\msx,\mcx)$ and $(\msy,\mcy)$ be two measurable spaces, we say that a transition probability kernel on $(\msy \times \msx) \times \mcy$ is a conditional Markov kernel. One elementary step in most Gibbs samplers corresponds to a conditional Markov kernel.

\newpage
\begin{center}
\textbf{\Large Table of contents}
\end{center}
 \vspace{0.4cm}

\startcontents[sections,subsections]
\printcontents[sections,subsections]{l}{1}{\setcounter{tocdepth}{3}}

\section{Proof of \Cref{prop:pi_rho_proper}}
\label{sec:proof_proposition_1}


Let $b' \in [b-1]$, $p' = \sum_{i=b'+1}^b d_i$ and consider 
\begin{align}\label{eq:def_B_bar_B}
  &\B{B}_{b'}^{\top} = [\B{A}^{\top}_{b'+1}/\rho_{b'+1}^{\half} \cdots \B{A}^{\top}_b/\rho_b^{\half}] \in \rset^{d \times p'} \eqsp,
  &\bar{\B{B}}_{b'} = \B{B}_{b'}^{\top}\B{B}_{b'} = \sum_{i=b'+1}^b \{\B{A}_i^{\top}\B{A}_i/\rho_i\}  \in \rset^{d \times d}\eqsp.
\end{align}
Note that under  \Cref{ass:well_defined_density}, $\bar{\B{B}}_{b'}$ is invertible. Indeed, it is a  symmetric positive definite matrix since
for any $\btheta \in \rset^d$, $\ps{\bar{\B{B}}_{b'}\btheta}{\btheta} \ge [\min_{i\in\iint{b}}\rho_i^{-1}] \psLigne{\sum^b_{i=b'+1} \B{A}_i^{\top} \B{A}_i \btheta}{\btheta} >0$ using that  $\sum^b_{i=b'+1} \B{A}_i^{\top} \B{A}_i$ is invertible.
Define the orthogonal projection onto the range of $\B{B}_{b'}$ and the diagonal matrix: 
\begin{equation}\label{eq:def_projection}
\B{P}_{b'} = \B{B}_{b'} \bar{\B{B}}_{b'}^{-1} \B{B}_{b'}^{\top} \eqsp, \qquad   \B{\tilde{D}}_{b'} = \mathrm{diag}(\B{I}_{d_{b'+1}}/\rho_{b'+1},\ldots,\B{I}_{d_b}/\rho_b) \eqsp. 
\end{equation}

\subsection{Technical lemma}

\begin{lemma}
  \label{lem:cal_conditionelle}
  Assume \Cref{ass:well_defined_density}. For any $(\thetabf,\zbf_{b'+1:b}) \in \rset^d\times \rset^{p'}$, setting $\zbf = \zbf_{b'+1:b}$, we have 
  \begin{multline*}
    \sum_{i=b'+1}^b\ac{\normLigne[2]{\zbf_i-\Abf_i \thetabf}/\rho_i}
    = (\B{\tilde{D}}_{b'}^{\half} \B{z})^{\top}\{\B{I}_{p'} - \B{P}_{b'}\}(\B{\tilde{D}}_{b'}^{\half} \B{z})\\
    + (\bftheta-\bar{\B{B}}_{b'}^{-1}\B{B}_{b'}^{\top} \B{\tilde{D}}_{b'}^{\half} \zbf )^{\top} \bar{\B{B}}_{b'} (\bftheta-\bar{\B{B}}_{b'}^{-1}\B{B}_{b'}^{\top} \B{\tilde{D}}_{b'}^{\half} \zbf ) \eqsp. 
  \end{multline*}
\end{lemma}

\begin{proof}
  Setting $\B{b} = \B{B}_{b'}^{\top} \B{\tilde{D}}_{b'}^{\half} \zbf$ and using the fact that $\bar{\B{B}}_{b'}$ is symmetric, we have
  \begin{align*}
    \sum_{i=b'+1}^b\ac{\normLigne[2]{\zbf_i-\Abf_i \thetabf}/\rho_i}
    &= \bftheta^{\top} \bar{\B{B}}_{b'} \bftheta -2 \bftheta^{\top} \B{b} +\sum_{i=b'+1}^b \norm{\zbf_i}^2/\rho_i \\
    &  =  \sum_{i=b'+1}^b \norm{\zbf_i}^2/\rho_i - \B{b}^{\top} \bar{\B{B}}_{b'}^{-1} \B{b} + (\bftheta-\bar{\B{B}}_{b'}^{-1}\B{b} )^{\top} \bar{\B{B}}_{b'} (\bftheta-\bar{\B{B}}_{b'}^{-1}\B{b} ) \eqsp. 
  \end{align*}
  Using that $\B{b}^{\top} \bar{\B{B}}_{b'}^{-1} \B{b} = (\B{\tilde{D}}_{b'}^{\half}\B{z})^{\top} \B{P}_{b'} (\B{\tilde{D}}_{b'}^{\half}\B{z})$ and $\B{P}_{b'}$ is a projection, $\B{P}_{b'}^2=\B{P}_{b'}$ completes the proof. 
\end{proof}

\subsection{Proof of \Cref{prop:pi_rho_proper}}

\begin{proposition}
\label{prop:integrability}
Assume \Cref{ass:well_defined_density}. Then, the function $\psi: (\thetabf,\zbf_{1:b}) \mapsto \prod_{i=1}^b \exp\{-U_i(\zbf_i) - \normLigne[2]{\zbf_i- \Abf_i \thetabf}/(2\rho_i)\}$ is integrable on $\rset^d \times \rset^p$, where $p = \sum_{i=1}^b d_i$.
\end{proposition}
\begin{proof}
Using \Cref{ass:well_defined_density} and the Fubini theorem, there exists $C_1 > 0$ such that:
\begin{align}
  \nonumber
  &\int_{\mathbb{R}^d}\br{\prod_{i=1}^{b'}\int_{\mathbb{R}^{d_i}}\mathrm{e}^{-U_i(\B{z}_{i})} \mathrm{e}^{- \frac{\norm{\B{z}_i-\B{A}_i\btheta}^2}{2\rho_i}}\,\dd\B{z}_{i} \cdot \prod_{j=b'+1}^{b}\int_{\mathbb{R}^{d_j}}\mathrm{e}^{-U_j(\B{z}_{j})} \mathrm{e}^{- \frac{\norm{\B{z}_j-\B{A}_j\btheta}^2}{2\rho_j}}\,\dd\B{z}_{j}  }\,\dd \btheta \\
  \nonumber
  &\le C_1 \int_{\mathbb{R}^d}\br{ \prod_{i=1}^{b'}\int_{\mathbb{R}^{d_i}}\mathrm{e}^{- \frac{\norm{\B{z}_i-\B{A}_i\btheta}^2}{2\rho_i}}\,\dd\B{z}_{i} \cdot \prod_{j=b'+1}^{b}\int_{\mathbb{R}^{d_j}}\mathrm{e}^{-U_j(\B{z}_{j})} \mathrm{e}^{- \frac{\norm{\B{z}_j-\B{A}_j\btheta}^2}{2\rho_j}}\,\dd\B{z}_{j}  }\,\dd \btheta \\
  \nonumber
  &\le C_1 \prod_{i=1}^{b'} (2\uppi\rho_i)^{d_i/2} \int_{\mathbb{R}^d}\br{\prod_{j=b'+1}^{b}\int_{\mathbb{R}^{d_j}}\mathrm{e}^{-U_j(\B{z}_{j})} \exp\pr{-\norm{\B{z}_j-\B{A}_j\btheta}^2/(2\rho_j)}\,\dd\B{z}_{j}  }\,\dd \btheta \\
  &= C_1 \prod_{i=1}^{b'} (2\uppi\rho_i)^{d_i/2} \int_{\mathbb{R}^{d_{b'+1}}}\cdots \int_{\mathbb{R}^{d_b}}\br{\prod_{j=b'+1}^b \mathrm{e}^{-U_j(\B{z}_{j})}} \br{\int_{\mathbb{R}^d}\prod_{j=b'+1}^{b} \mathrm{e}^{- \frac{\norm{\B{z}_j-\B{A}_j\btheta}^2}{2\rho_j}}\,\dd \btheta}\,\dd\B{z}_{b'+1:b}\eqsp. \label{eq:integrab1}
  \end{align}
  Using \Cref{lem:cal_conditionelle} and the fact that $\B{I}_{p'} - \B{P}_{\b'}$ is positive definite, we obtain 
  \begin{align*}
    &\int_{\mathbb{R}^d}\prod_{j=b'+1}^{b} \exp\pr{- \norm{\B{z}_j-\B{A}_j\btheta}^2/(2\rho_j)}\,\dd \btheta \\
    &= \exp\pr{-(\B{\tilde{D}}_{b'}^{\half} \B{z})^{\top}\{\B{I}_{p'} - \B{P}_{\b'}\}(\B{\tilde{D}}_{b'}^{\half} \B{z})/2}\\
    &\times\int_{\mathbb{R}^d} \exp\pr{-(\bftheta-\bar{\B{B}}_{b'}^{-1}\B{B}_{b'}^{\top} \B{\tilde{D}}_{b'}^{\half} \zbf )^{\top} \bar{\B{B}}_{b'} (\bftheta-\bar{\B{B}}_{b'}^{-1}\B{B}_{b'}^{\top} \B{\tilde{D}}_{b'}^{\half} \zbf )/2}\,\dd \btheta \\
    &\le \mathrm{det}\pr{\bar{\B{B}}_{b'}}^{-1/2}(2\uppi)^{d/2}.
  \end{align*}
Then, the proof is completed by plugging this expression into \eqref{eq:integrab1} and using from \Cref{ass:well_defined_density} that $\bz_{b'+1:b} \mapsto \prod_{j=b'+1}^b \mathrm{e}^{-U_j(\B{z}_{j})}$ is integrable.
\end{proof}

\section{Proof of \Cref{prop:convergence_rho_gamma}}\label{sec:proof-proposition-2-1}


This section aims at proving \Cref{prop:convergence_rho_gamma} in the main paper. 
To ease the understanding, we dissociate the scenarios where $\max_{i \in [b]} N_i = 1$ and $\max_{i \in [b]} N_i > 1$. In addition, in all this section $\bfrho \in (\rset_+^*)^b$ is assumed to be fixed. 

\subsection{Single local LMC iteration}

In this section, we assume that a single LMC step is performed locally on each worker, that is $\max_{i \in [b]} N_i = 1$. For this, we introduce the conditional Markov transition kernel defined for any $\bfgamma = (\gamma_1,\ldots,\gamma_b)$, $\btheta \in \Rd$, $\bz = (\bz_{1},\cdots,\bz_{b}) \in \mathbb{R}^{d_1} \times \cdots \times \mathbb{R}^{d_b}$, and for $i \in [b]$, $\mathsf{B}_i \in \mathcal{B}(\mathbb{R}^{d_i})$, by
\begin{align}
Q_{\bfrho, \bfgamma}\pr{\bz,\mathsf{B}_1\times\cdots\times\mathsf{B}_b|\btheta}
=\prod_{i=1}^b R_{\rho_i,\gamma_i}(\bz_{i},\mathsf{B}_i|\btheta)\eqsp, \label{eq:Q_rho_gamma}
\end{align}
where  
\begin{equation}
\label{eq:def_R_rho_gamma}
  R_{\rho_i, \gamma_i}(\bz_{i},\mathsf{B}_i|\btheta) = \int_{\mathsf{B}_i}\exp\ac{-\frac{1}{4\gamma_i}\norm{\tilde{\bz}_{i} - \pr{1-\frac{\gamma_i}{\rho_i}}\bz_{i} - \frac{\gamma_i}{\rho_i}\B{A}_i\btheta + \gamma_i \grad U_i(\bz_{i})}^2}\frac{\dd\tilde{\bz}_{i}}{(4\uppi\gamma_i)^{d_i/2}}\eqsp.
\end{equation}
Recall that $p = \sum_{i=1}^b d_i$.
The considered Gibbs sampler in \Cref{algo:ULAwSG} defines a homogeneous Markov chain $X_{n}^{\top} = (\theta_{n}^{\top},Z_{n}^{\top})_{n \ge 1}$ where $Z_{n}^{\top} = ([Z_{n}^1]^{\top},\cdots,[Z_{n}^b]^{\top})$. Indeed, it is easy to show that for any $n \in \nset$ and measurable bounded function $f : \rset^p \to \rset_+$, $\PE[f(Z_{n+1})|X_n] = \int_{\rset^p} f(\zbf)Q_{\bfrho,\gammabf}(Z_{n},\dd \zbf|\theta_n)$ and therefore $(X_n)_{n \in\nset}$ is associated with the Markov kernel defined, for any $\bx^{\top} = (\btheta^{\top},\bz^{\top}) \in \R^{d} \times \R^{p}$ and $\msa \in \mcbb(\rset^d)$, $\msb \in \mathcal{B}( \mathbb{R}^p)$, by
\begin{equation}\label{eq:def:prop2:P_rho_gamma}
  P_{\bfrho, \bfgamma}(\bx,\msa \times \msb) = \int_{\msb} Q_{\bfrho,\bfgamma}\pr{\bz,\dd {\tbfz}|\thetabf}\int_{\msa}\Pi_{\bfrho}(\dd\tbtheta|\tbfz) \eqsp,
\end{equation}
where $\Pi_{\bfrho}(\cdot|\tbfz)$ is defined in \eqref{eq:def:Pi_rho_cond}.
Let $(\xi_n)_{n \ge 1}$ be a sequence of i.i.d. $d$-dimensional standard Gaussian random variables independent of the family of independent random variables $\{(\eta_{n}^{i})_{n\ge 1} : i\in [b]\}$ where for any $i \in [b]$ and $n \ge 1$, $\eta_{n}^{i}$ is a $d_i$-dimensional standard Gaussian random variable.
We define the stochastic processes $(X_{n},\tilde{X}_{n})_{n\ge 0}$ on $\R^p\times\R^p$ starting from $(X_0,\tilde{X}_0) = (\bx,\tilde{\bx}) = ((\btheta^{\top},\bz^{\top})^{\top},(\tilde{\btheta}^{\top},\tilde{\bz}^{\top})^{\top})$ and following the recursion for $n\ge 0$,
\begin{equation}\label{eq:def:X}
X_{n+1}=(\theta_{n+1}^{\top}, Z_{n+1}^{\top})^{\top}\eqsp, \qquad
\tilde{X}_{n+1}=(\tilde{\theta}_{n+1}^{\top}, \tilde{Z}_{n+1}^{\top})^{\top}\eqsp,
\end{equation}
where $Z_{n+1} = ([Z_{n+1}^1]^{\top},\ldots,[Z_{n+1}^b]^{\top})^{\top}, \tilde{Z}_{n+1} = ([\tilde{Z}_{n+1}^1]^{\top},\ldots,[\tilde{Z}_{n+1}^b]^{\top})^{\top}$ are defined, for any $i \in [b]$, by
\begin{align}
  \label{eq:coupling_Z}
      Z_{n+1}^{i} &= \pr{1-\gamma_i/\rho_i}Z_{n}^{i} + \pr{\gamma_i/\rho_i}\B{A}_i\theta_{n} - \gamma_i\grad U_i(Z_{n}^{i}) + \sqrt{2\gamma_i}\eta_{n+1}^i\eqsp,\\
      \nonumber
      \tilde{Z}_{n+1}^{i}& = \pr{1-\gamma_i/\rho_i}\tilde{Z}_{n}^{i} + \pr{\gamma_i/\rho_i}\B{A}_i\tilde{\theta}_{n} - \gamma_i\grad U_i(\Zc_{n}^{i}) + \sqrt{2\gamma_i}\eta_{n+1}^i\eqsp,
\end{align}
and $\theta_{n+1}, \tilde{\theta}_{n+1}$ by
\begin{equation}
\label{eq:coupling_theta}
      \theta_{n+1} = \bar{\B{B}}_0^{-1}\B{B}_0^{\top}\B{\tilde{D}}_{0}^{\half} \Zb_{n+1} + \bar{\B{B}}_0^{-\half} \xi_{n+1}\eqsp,
 \qquad 
      \tilde{\theta}_{n+1} = \bar{\B{B}}_0^{-1}\B{B}_0^{\top}\B{\tilde{D}}_{0}^{\half} \Zc_{n+1} + \bar{\B{B}}_0^{-\half} \xi_{n+1}
\eqsp,
\end{equation}
where $\bar{\B{B}}_0$, $\B{B}_0$ and $\B{\tilde{D}}_{0}$ are given in \eqref{eq:def_B_bar_B} and \eqref{eq:def_projection}, respectively.
Note that $X_{n}$ and $\tX_{n}$ are distributed according to $\updelta_{\bx}P_{\bfrho,\bfgamma}^{n}$ and $\updelta_{\tbfx}P_{\bfrho,\bfgamma}^{n}$, respectively.
Hence, by definition of the Wasserstein distance of order 2, it follows that 
\begin{equation}\label{eq:W2X_def}
    \wasserstein{}(\updelta_{\bx}P_{\bfrho,\bfgamma}^{n},\updelta_{\tbfx}P_{\bfrho,\bfgamma}^{n}) \le \mathbb{E}\br{\|X_{n}-\tX_{n}\|^{2}}^{\half}.
\end{equation}
Thus, in this section we focus on upper bounding the squared norm $ \|X_{n}-\tX_{n}\|$ from which we get an explicit bound on the Wasserstein distance thanks to the previous inequality.
\subsubsection{Supporting lemmata}
Note that \Cref{ass:well_defined_density} implies the invertibility of the matrix $\B{B}_0$ defined in \eqref{eq:def_B_bar_B} since we have the existence of $b'\in[b-1]$, such that $\sum_{i=b'+1}^b \lambda_{\min}(\B{A}_i^{\top}\B{A}_i)/\rho_i> 0$ and by the semi-positiveness of the symmetric matrices $\{\B{A}_i^{\top}\B{A}_i\}_{i \in [b]}$, we get that $\lambda_{\min}\prn{\B{B}_0}=\sum_{i=1}^{b}\lambda_{\min}\prn{\B{A}_i^{\top}\B{A}_i}/\rho_i \ge \sum_{i=b'+1}^b \lambda_{\min}(\B{A}_i^{\top}\B{A}_i)/\rho_i$.
To prove \Cref{prop:convergence_rho_gamma} in the case $\max_{i \in [b]}N_i=1$, we first upper bound \eqref{eq:W2X_def} by building upon the following two technical lemmas.
\begin{lemma}\label{lem:geo_decr}
Assume \Cref{ass:well_defined_density} and consider $(X_{n},\tilde{X}_{n})_{n \in\nset}$ defined in \eqref{eq:def:X}. Then, for any $n \in \nset$, it holds almost surely that
\[
\|X_{n+1}- \tilde{X}_{n+1} \|^2 \le (1 + \normLigne[2]{\bar{\B{B}}_0^{-1}\B{B}_0^{\top}\B{\tilde{D}}_{0}^{\half}}) \|Z_{n+1} - \tilde{Z}_{n+1} \|^2 \eqsp.
\]
\end{lemma}
\begin{proof}
Let $n\ge 0$. By \eqref{eq:coupling_theta}, we have $\theta_{n+1}-\tilde{\theta}_{n+1} 
= \bar{\B{B}}_0^{-1}\B{B}_0^{\top}\B{\tilde{D}}_{0}^{\half} (Z_{n+1} - \tilde{Z}_{n+1})$ which implies that 
\begin{equation*}
\|X_{n+1}- \tilde{X}_{n+1} \|^2 = \|\theta_{n+1} - \tilde{\theta}_{n+1} \|^2 + \|Z_{n+1} - \tilde{Z}_{n+1} \|^2 \le (1+\|\bar{\B{B}}_0^{-1}\B{B}_0^{\top}\B{\tilde{D}}_{0}^{\half}\|^2) \|Z_{n+1} - \tilde{Z}_{n+1} \|^2 \eqsp. 
\end{equation*}
\end{proof}
Define the contraction factor
\begin{equation}\label{eq:def:kappa}
  \kappa_{\bfgamma} \txts = \max_{i\in[b]} \ac{\absLigne{1 - \gamma_i m_i} \vee
  \absLigne{1 - \gamma_i (M_i+1/\rho_i)}}\eqsp.
\end{equation}
Then, the following result holds.
\begin{lemma}\label{lem:W2_dirac}
Assume \Cref{ass:well_defined_density}-\Cref{ass:supp_fort_convex} and let $\gammabf \in (\rset_+^*)^b$.
Then for any $\bx=(\bz^{\top},\btheta^{\top})^{\top},\tilde{\bx}=(\tilde{\bz}^{\top},\tilde{\btheta}^{\top})^{\top}$, with $(\btheta,\tilde{\btheta}) \in (\mathbb{R}^d)^2$ and $(\bz,\tilde{\bz}) \in (\mathbb{R}^p)^2$, for any $n\ge 1$, we have
\begin{multline*}
\wasserstein{}(\updelta_{\bx} P_{\bfrho, \bfgamma}^{n}, \updelta_{\tilde{\bx}} P_{\bfrho, \bfgamma}^{n})
\le \kappa_{\bfgamma}^{n-1} \cdot \prbigg{(1 + \|\bar{\B{B}}_0^{-1}\B{B}_0^{\top}\B{\tilde{D}}_{0}^{\half}\|^2) \cdot \frac{\max_{i \in [b]}\{\gamma_i\}}{\min_{i \in [b]}\{\gamma_i\}}}^{\half}\\
\times\br{\kappa_{\bfgamma}\|\bz-\tilde{\bz}\| + \|\B{D}_{\bfgamma/\sqrt{\bfrho}}\B{B}_0\|\|\btheta-\tilde{\btheta}\|}\eqsp,
\end{multline*}
where $\B{D}_{\bfgamma/\sqrt{\bfrho}}$ is defined as in \eqref{eq:def_diag} with $\bfgamma/\sqrt{\bfrho} = (\gamma_1/\rho_1^{\half},\ldots,\gamma_b/\rho_b^{\half})$, $\bar{\B{B}}_0$, $\B{B}_0$, $P_{\bfrho, \bfgamma}$ and $\kappa_{\bfgamma}$ are given in \eqref{eq:def_B_bar_B}, \eqref{eq:def:prop2:P_rho_gamma}, \eqref{eq:def:kappa}, respectively.
\end{lemma}
\begin{proof}
  Consider $(X_{k},\tilde{X}_{k})_{k \in\nset}$ defined in \eqref{eq:def:X}. By \eqref{eq:W2X_def} and \Cref{lem:geo_decr}, we need to bound $ ( \| Z_{k}-\tilde{Z}_{k}\|)_{k \in\nset}$. 
  Let $n \in \N^*$.
  For any $i \in [b]$, we have by \eqref{eq:coupling_Z}, that 
\begin{align}
\label{eq:zi_recursion}
  Z_{n+1}^{i}-\tilde{Z}_{n+1}^{i} &= \prBig{1-\frac{\gamma_i}{\rho_i}}(Z_{n}^{i}-\tilde{Z}_{n}^{i}) + \frac{\gamma_i}{\rho_i}\B{A}_i(\theta_{n}-\tilde{\theta}_{n}) - \gamma_i\pr{\grad U_i(Z_{n}^{i})-\grad U_i(\tilde{Z}_{n}^{i})}\eqsp. 
\end{align}
Since $U_i$ is twice differentiable, we have
\[
\grad U_i(Z_{n}^{i})-\grad U_i(\tilde{Z}_{n}^{i}) = \int_0^1 \grad^2 U_i(\tilde{Z}_{n}^{i} + t(Z_{n}^{i}-\tilde{Z}_{n}^{i}))\,\dd t \cdot (Z_{n}^{i}-\tilde{Z}_{n}^{i})\eqsp.
\]
Using $\theta_{n}-\tilde{\theta}_{n} 
= \bar{\B{B}}_0^{-1}\B{B}_0^{\top}\B{\tilde{D}}_{0}^{\half} (Z_{n} - \tilde{Z}_{n})$, it follows that
\begin{multline*}
  Z_{n+1}^{i}-\tilde{Z}_{n+1}^{i} = \pr{\Big[1-\frac{\gamma_i}{\rho_i}\Big]\B{I}_{d_i} - \gamma_i\int_0^1 \grad^2 U_i(\tilde{Z}_{n}^{i} + t(Z_{n}^{i}-\tilde{Z}_{n}^{i}))\,\dd t}(Z_{n}^{i}-\tilde{Z}_{n}^{i})\\
  + \frac{\gamma_i}{\rho_i}\B{A}_i\bar{\B{B}}_0^{-1}\B{B}_0^{\top}\B{\tilde{D}}_{0}^{\half}(Z_{n}-\tilde{Z}_{n})\eqsp. 
\end{multline*}
Consider the  $p\times p$ block diagonal matrix defined by
\begin{align*}
  &\B{D}_{U,n} = \mathrm{diag}\pr{\gamma_1\int_0^1 \grad^2 U_1(\tilde{Z}_{n}^{1} + t(Z_{n}^{1}-\tilde{Z}_{n}^{1}))\,\dd t, \cdots, \gamma_b\int_0^1 \grad^2 U_b(\tilde{Z}_{n}^{b} + t(Z_{n}^{b}-\tilde{Z}_{n}^{b}))\,\dd t}\eqsp.
\end{align*}
With the projection matrix $\B{P}_{0}$ defined in \eqref{eq:def_projection}, the difference $Z_{n+1}-\tilde{Z}_{n+1}$ can be rewritten as
\begin{align*}
  Z_{n+1}-\tilde{Z}_{n+1} &= \pr{\B{I}_p - \B{D}_{U,n} - \B{D}_{\bfgamma}^{\half}\B{D}_{\bfgamma/\bfrho}^{\half}(\B{I}_{p}-\B{P}_{0})\B{\tilde{D}}_{0}^{\half}}(Z_{n}-\tilde{Z}_{n})\eqsp,
\end{align*}
 where $\B{D}_{\bfgamma/\bfrho}$ is defined as in \eqref{eq:def_diag} with $\bfgamma/\bfrho = (\gamma_1/\rho_1,\ldots,\gamma_b/\rho_b)$.
Since $\B{D}_{U,n}$ commutes with $\B{D}_{\bfgamma}$ and $\B{P}_{0}$ is an orthogonal projection matrix, using  \Cref{ass:supp_fort_convex}-\ref{ass:1}-\ref{ass:2}, we get 
\begin{align*}
  &\| Z_{n+1}-\tilde{Z}_{n+1}\|_{\B{D}_{\bfgamma}^{-1}}\\
  &= \|\B{D}_{\bfgamma}^{-\half}(\B{D}_{\bfgamma}^{\half}\B{D}_{\bfgamma}^{-\half} - \B{D}_{\bfgamma}^{\half} \B{D}_{U,n}\B{D}_{\bfgamma}^{-\half}
- \B{D}_{\bfgamma}^{\half} \B{D}_{\bfgamma/\bfrho}^{\half} (\B{I}_{p}-\B{P}_{0}) \B{D}_{\bfgamma/\bfrho}^{\half}\B{D}_{\bfgamma}^{-\half})(Z_{n}-\tilde{Z}_{n})\| \nonumber \\
  &\le \|\B{I}_p - \B{D}_{U,n} - \B{D}_{\bfgamma/\bfrho}^{\half} \pr{\B{I}_{p} - \B{P}_{0}}\B{D}_{\bfgamma/\bfrho}^{\half}\| \|Z_{n}-\tilde{Z}_{n}\|_{\B{D}_{\bfgamma}^{-1}}\eqsp.
\end{align*}
Note that  \Cref{ass:well_defined_density} and \Cref{ass:supp_fort_convex} and the fact that $\B{P}_{0}$ is an orthogonal projector, so $  \B{0}_p\preccurlyeq\B{I}_p - \B{P}_{0}$, imply that
\begin{multline*}
  \mathrm{diag}(\{1- \gamma_1(M_1 + 1/\rho_1)\}\B{I}_{d_1}, \cdots, \{1- \gamma_b (M_b + 1/\rho_b)\} \B{I}_{d_b})
\preccurlyeq
\B{I}_p - \B{D}_{U,n} - \B{D}_{\bfgamma/\bfrho}^{\half} \pr{\B{I}_{p}-\B{P}_{0}}\B{D}_{\bfgamma/\bfrho}^{\half}
\\  \preccurlyeq
  \mathrm{diag}\pr{\left\{1-  \gamma_1 m_1 \right\}\B{I}_{d_1}, \hdots, \left\{1-  \gamma_b m_b\right\} \B{I}_{d_b}} \eqsp.
\end{multline*}
Therefore, we get
\begin{align}
  \| Z_{n+1}-\tilde{Z}_{n+1}\|_{\B{D}_{\bfgamma}^{-1}}  &\le \max_{i \in [b]}\ac{\max (|1-\gamma_i m_i|,|1-\gamma_i(M_i+1/\rho_i)|)}\|Z_{n}-\tilde{Z}_{n}\|_{\B{D}_{\bfgamma}^{-1}}\nonumber\\ 
  &= \kappa_{\bfgamma}\|Z_{n}-\tilde{Z}_{n}\|_{\B{D}_{\bfgamma}^{-1}}\eqsp. \label{eq:contraction}
\end{align}
An immediate induction shows, for any $n \ge 1$,
\begin{equation}
  \label{eq:contract_n_2}
  \|Z_{n}-\tilde{Z}_{n}\|_{\B{D}_{\bfgamma}^{-1}} \le \kappa_{\bfgamma}^{n-1} \|Z_1-\tilde{Z}_1\|_{\B{D}_{\bfgamma}^{-1}}\eqsp.
\end{equation}
In addition, by \eqref{eq:zi_recursion}, we have for any $i \in [b]$, 
\begin{align*}
  Z_{1}^{i}-\tilde{Z}_{1}^{i} &= \prBig{1-\frac{\gamma_i}{\rho_i}}(\bz_{i}-\tilde{\bz}_{i}) + \frac{\gamma_i}{\rho_i}\B{A}_i(\btheta-\tilde{\btheta}) - \gamma_i(\grad U_i(\bz_{i})-\grad U_i(\tilde{\bz}_{i}))\eqsp.
\end{align*}
It follows that $Z_{1}-\tilde{Z}_{1} = (\B{I}_p - \B{D}_{\bfgamma/\bfrho} - \B{D}_{U,0})(\bz-\tilde{\bz}) + \B{D}_{\bfgamma/\bfrho}\B{\tilde{D}}_{0}^{-\half}\B{B}_0(\btheta-\tilde{\btheta})$.
Using the triangle inequality and \Cref{ass:supp_fort_convex}  gives
\begin{align*}
\|Z_1-\tilde{Z}_1\|_{\B{D}_{\bfgamma}^{-1}} 
&\txts\le (\min_{i \in [b]} \{\gamma_i\})^{-\half}\|(\B{I}_p - \B{D}_{\bfgamma/\bfrho} - \B{D}_{U,0})(\bz-\tilde{\bz}) + (\B{D}_{\bfgamma/\sqrt{\bfrho}}\B{B}_0(\btheta-\tilde{\btheta})\| \\
&\txts\le (\min_{i \in [b]} \{\gamma_i\})^{-\half} \Big[\|\B{I}_p - \B{D}_{\bfgamma/\bfrho} - \B{D}_{U,0}\|\|\bz-\tilde{\bz}\| + \|\B{D}_{\bfgamma/\sqrt{\bfrho}}\B{B}_0\|\|\btheta-\tilde{\btheta}\|\Big] \\
&\txts \le (\min_{i \in [b]} \{\gamma_i\})^{-\half} \Big[\max_{i\in[b]}\{\absn{1-\gamma_i(m_i+1/\rho_i)},\absn{1-\gamma_i(M_i+1/\rho_i)}\}\|\bz-\tilde{\bz}\| \\
&\qquad+ \|\B{D}_{\bfgamma/\sqrt{\bfrho}}\B{B}_0\|\|\btheta-\tilde{\btheta}\|\Big]\\
&\txts \le (\min_{i \in [b]} \{\gamma_i\})^{-\half} \brBig{\kappa_{\gamma}\|\bz-\tilde{\bz}\| + \|\B{D}_{\bfgamma/\sqrt{\bfrho}}\B{B}_0\|\|\btheta-\tilde{\btheta}\|}\eqsp.
\end{align*}
Combining \eqref{eq:contract_n_2} and the previous inequality and using \Cref{lem:geo_decr}, we get for $n \ge 1$,
\begin{multline*}
\| X_{n} - \tilde{X}_{n} \|^2
\le \kappa_{\bfgamma}^{2(n-1)} \pr{1 + \|\bar{\B{B}}_0^{-1}\B{B}_0^{\top}\B{\tilde{D}}_{0}^{\half} \|^2} \frac{\max_{i \in [b]}\{\gamma_i\}}{\min_{i \in [b]}\{\gamma_i\}}\\
\times \brBig{\kappa_{\gamma}\|\bz-\tilde{\bz}\| + \|\B{D}_{\bfgamma/\sqrt{\bfrho}}\B{B}_0\|\|\btheta-\tilde{\btheta}\|}^2\eqsp.
\end{multline*}
The proof is concluded by \eqref{eq:W2X_def}.
\end{proof}
\subsubsection{Specific case of \Cref{prop:convergence_rho_gamma}}
Based on the previous lemmata, we provide in what follows a specific instance of \Cref{prop:convergence_rho_gamma} in the scenario where $\max_{i \in [b]} N_i = 1$.
\begin{proposition}\label{prop:convergence_N_1}
     Assume \Cref{ass:well_defined_density}-\Cref{ass:supp_fort_convex} and let $\bfgamma\in(\R_+^*)^{b}$ such that, for any $i\in[b]$, $\gamma_i\le 2(m_i+M_i+1/\rho_i)^{-1}$.
     Then, $P_{\bfrho,\bfgamma}$ defined in \eqref{eq:def:prop2:P_rho_gamma} admits a unique stationary distribution $\Pi_{\bfrho,\bfgamma}$ and for any $\bx=(\bz^{\top},\btheta^{\top})^{\top}$ with $\btheta \in \mathbb{R}^d$, $\bz \in \mathbb{R}^p$ and any $n \in \N^*$, we have
    \begin{align*}
      W_{2}^2(\updelta_{\B{x}} P_{\bfrho,\bfgamma}^{n}, \Pi_{\bfrho,\bfgamma}) &\le
      \prbig{1-\min_{i\in[b]}\acn{\gamma_im_i}}^{2(n-1)} \prbigg{(1 + \|\bar{\B{B}}_0^{-1}\B{B}_0^{\top}\B{\tilde{D}}_{0}^{\half}\|^2) \cdot \frac{\max_{i \in [b]}\{\gamma_i\}}{\min_{i \in [b]}\{\gamma_i\}}}\\
      &\times \int_{\rset^{d}\times \rset^p}\Big[\prn{1-\min_{i\in[b]}\acn{\gamma_im_i}}\|\bz-\tilde{\bz}\| + \|\B{D}_{\bfgamma/\sqrt{\bfrho}}\B{B}_0\|\|\btheta-\tilde{\btheta}\|\Big]^2 \dd \Pi_{\bfrho,\bfgamma}(\tilde{\bx})\eqsp,
    \end{align*}
    where $\bar{\B{B}}_0,\B{B}_0,\B{\tilde{D}}_{0}$, $P_{\bfrho, \bfgamma}$ are defined in \eqref{eq:def_B_bar_B} and \eqref{eq:def_projection}.
\end{proposition}
\begin{proof}
For any $i\in[b]$, note that the condition $0 <\gamma_i\le 2(m_i+M_i+1/\rho_i)^{-1}$ ensures that $\kappa_{\bfgamma}=1-\min_{i\in[b]}\acn{\gamma_im_i}\in(0,1)$ and the proof follows from \Cref{lem:W2_dirac} combined with \citet[Lemma 20.3.2, Theorem 20.3.4]{douc2018markov}.
\end{proof}


\subsection{Multiple local LMC iterations}

In this section, we consider the general case $\max_{i \in [b]} N_i \ge 1$. For this, we introduce the conditional Markov transition kernel defined for any $\bfgamma = (\gamma_1,\ldots,\gamma_b)$, $\bfN = (N_1,\ldots,N_b)$, $\btheta \in \Rd$, $\bz = (\bz_{1},\cdots,\bz_{b}) \in \mathbb{R}^{d_1} \times \cdots \times \mathbb{R}^{d_b}$, for $i \in [b]$ and $\mathsf{B}_i \in \mathcal{B}(\mathbb{R}^{d_i})$, by
\begin{equation}
  \label{eq:Q_rho_gamma_N}
Q_{\bfrho, \bfgamma,\bfN}\pr{\bz,\mathsf{B}_1\times\cdots\times\mathsf{B}_b|\btheta}
=\prod_{i=1}^b R_{\rho_i,\gamma_i}^{N_i}(\bz_{i},\mathsf{B}_i|\btheta)\eqsp, 
\end{equation}
where $R_{\rho_i,\gamma_i}$ is defined by \eqref{eq:def_R_rho_gamma}. 
Then, as in the case $\max_{i \in[b]} N_i = 1$, the Gibbs sampler presented in \Cref{algo:ULAwSG} defines a homogeneous Markov chain $X_{n}^{\top} = (\theta_{n}^{\top},Z_{n}^{\top})_{n \ge 1}$ where $Z_{n}^{\top} = ([Z_{n}^1]^{\top},\cdots,[Z_{n}^b]^{\top})$. Indeed, it is easy to show that for any $n \in \nset$ and measurable function $f : \rset^p \to \rset_+$, $\PE[f(Z_{n+1})|X_n] = \int_{\rset^p} f(\zbf)Q_{\bfrho,\gammabf,\bfN}(Z_{n},\dd \zbf|\theta_n)$.
Therefore, $(X_n)_{n \in\nset}$ is associated with the Markov kernel defined, for any $\bx^{\top} = (\btheta^{\top},\bz^{\top})$ and $\msa \in \mcbb(\rset^d)$, $\msb \in \mathcal{B}( \mathbb{R}^p)$, by
\begin{equation}\label{eq:P_rho_gamma_N}
  P_{\bfrho, \bfgamma,\bfN}(\bx,\msa \times \msb) = \int_{\msb} Q_{\bfrho,\bfgamma,\bfN}\pr{\bz,\dd {\tbfz}|\thetabf}\int_{\msa}\Pi_{\bfrho}(\dd \tbtheta|\tbfz) \eqsp,
\end{equation}
where $\Pi_{\bfrho}(\cdot|\tbfz)$ is defined in \eqref{eq:def:Pi_rho_cond}.
We now define a coupling between $\updelta_{\B{x}}P^n_{\bfrho,\bfgamma,\bfN}$ and $\updelta_{\B{\tilde{x}}}P^n_{\bfrho,\bfgamma,\bfN}$ for any $n \ge 1$ and $\B{x}, \B{\tilde{x}} \in \R^{d} \times \R^p$.
Let $(\xi_n)_{n \ge 1}$ be a sequence of i.i.d. $d$-dimensional standard Gaussian random variables independent of the family of independent random variables $\{(\eta_{n}^{i})_{n\ge 1} : i\in [b]\}$ where for any $i \in [b]$ and $n \ge 1$, $\eta_{n}^{i}$ is a $d_i$-dimensional standard Gaussian random variable.
Define by induction the synchronous coupling $(\theta_n, Z_n)_{n \ge 0}, (\tilde{\theta}_n, \Zc_n)_{n \ge 0}$, for any $i\in [b]$ starting from $(\theta_0,\Zb_0) = \bx = ( \btheta, \bz)$, $(\tilde{\theta}_0,\tilde{\Zb}_0) = \tilde{\bx} = ( \tilde{\btheta}, \tilde{\bz})$ and for any $n\ge 0$ by
\begin{equation}\label{eq:cont_coupling}
\begin{aligned}
	&\Zc_{n+1}^{i} = \Yc_{N_i}^{(i, n)}\eqsp,
	\qquad \tilde{\theta}_{n+1} = \bar{\B{B}}_0^{-1}\B{B}_0^{\top}\B{\tilde{D}}_{0}^{\half} \Zc_{n+1} + \bar{\B{B}}_0^{-\half} \xi_{n+1}\eqsp, \\
	&\Zb_{n+1}^{i} = \Yb_{N_i}^{(i, n)}\eqsp,
	\qquad \theta_{n+1} = \bar{\B{B}}_0^{-1}\B{B}_0^{\top}\B{\tilde{D}}_{0}^{\half} \Zb_{n+1} + \bar{\B{B}}_0^{-\half} \xi_{n+1}\eqsp,
\end{aligned}
\end{equation}
where $\bar{\B{B}}_0,\B{B}_0,\B{\tilde{D}}_{0}$ are given by \eqref{eq:def_B_bar_B}-\eqref{eq:def_projection} and $\Yc_{0}^{(i, n)} = \Zc_n^i$, $\Yb_{0}^{(i, n)} = \Zb_n^i$, and for any $k\in\N$
\begin{equation}\label{eq:coupling_process_Y}
\begin{aligned}
&\Yc_{k+1}^{(i, n)} = \Yc_{k}^{(i, n)}-\gamma_i\nabla \tildeU_i(\Yc_{k}^{(i, n)}) + (\gamma_i/\rho_i) \B{A}_i \tilde{\theta}_{n} + \sqrt{2\gamma_i}\eta_{k+1}^{(i, n)}\eqsp,\\
&\Yb_{k+1}^{(i, n)} = \Yb_{k}^{(i, n)} -\gamma_i\nabla \tildeU_i(\Yb_{k}^{(i, n)}) + (\gamma_i/\rho_i) \B{A}_i \theta_{n} + \sqrt{2\gamma_i}\eta_{k+1}^{(i, n)}\eqsp,
\end{aligned}
\end{equation}
where, for any $\bz_i \in \R^{d_i}$, $\tildeU_i$ is defined by
\begin{equation}
	\label{eq:def:cont_Utilde}
	\tildeU_i(\bz_i) = U_i(\bz_i) + (2\rho_i)^{-1}\norm{\bz_i}^2\eqsp.
\end{equation}
For any $n,k\in\N$ consider the $p\times p$ matrices defined by
\begin{align}
&\B{H}_{U,k}^{(n)} = \mathrm{diag}\bigg(\gamma_1\int_0^1 \nabla^2 U_1((1-s) \Yb_{k}^{(1, n)} + s \Yc_{k}^{(1, n)} )\,\dd s, \nonumber\\
&\qquad\qquad\qquad\qquad\qquad\qquad\qquad\qquad\hdots, \gamma_b\int_0^1 \nabla^2 U_b ((1-s) \Yb_{k}^{(b, n)} + s \Yc_{k}^{(b, n)} )\,\dd s\bigg)\eqsp, \nonumber\\
&\B{J}(k) = \mathrm{diag}\pr{\1_{[N_1]}(k+1)\cdot\B{I}_{d_1},\cdots,\1_{[N_b]}(k+1)\cdot\B{I}_{d_b}}\eqsp,\label{eq:def:J}\\
&\B{C}_{k}^{(n)} = \B{J}(k)(\B{D}_{\bfgamma/\bfrho} + \B{H}_{U,k}^{(n)})\eqsp,\label{eq:def:C}\\
&\B{M}_{k+1}^{(n)} = (\B{I}_{p} - \B{C}_{0}^{(n)} )^{-1} \ldots (\B{I}_{p} - \B{C}_{k}^{(n)} )^{-1}\eqsp, \qquad \text{ with } \B{M}_0^{(n)} = \B{I}_p \eqsp.\label{eq:def:M}
\end{align}
Under \Cref{ass:supp_fort_convex}, we have $\|\B{C}_{k}^{(n)}\|\le \max_{i\in[b]}\acn{\gamma_i(M_i+1/\rho_i)}$, thus if we suppose that for any $i\in[b], 0<\gamma_i<(M_i+1/\rho_i)^{-1}$, the matrix $(\B{I}_p-\B{C}_{k}^{(n)})$ is invertible.
In addition, for any $n\in\N, k\ge\max_{i\in[b]}\{N_i\}, \B{C}_k^{(n)}=\mathbf{0}_{p\times p}$, hence the sequence $(\B{M}_{k}^{(n)})_{k\in\N}$ is stationary and we denote its limit by $\B{M}_{\infty}^{(n)}$ which is equal to $\B{M}_{\max_{i\in[b]}\{N_i\}}^{(n)}$. 
\subsubsection{Technical lemmata}

Similarly to \Cref{lem:geo_decr}, the following result shows that it is enough to consider the marginal process $(Z_n,\tilde{Z}_n)_{n \ge 0}$ to control 
\begin{equation}\label{eq:W2X_def2}
    \wasserstein{}(\updelta_{\bx}P_{\bfrho,\bfgamma,\bfN}^{n},\updelta_{\tbfx}P_{\bfrho,\bfgamma,\bfN}^{n}) \le \mathbb{E}\br{\|X_{n}-\tX_{n}\|^{2}}^{\half}\eqsp.
\end{equation}
\begin{lemma}\label{lem:geo_decr_cont1}
Assume \Cref{ass:well_defined_density} and let $\bfN\in(\N^*)^{b}, \gammabf \in (\rset_+^*)^b$.
Then, for any $n \in \nset$, the random variables $X_{n}=(\theta_{n}^{\top}, \Zb_{n}^{\top})^{\top}, \tilde{X}_{n}=(\tilde{\theta}_{n}^{\top}, \Zc_{n}^{\top})^{\top}$ defined in \eqref{eq:cont_coupling} satisfy
\[
\|\tilde{X}_{n+1}-X_{n+1}\|^2 \le (1 + \| \bar{\B{B}}_0^{-1}\B{B}_0^{\top}\B{\tilde{D}}_{0}^{\half} \|^2) \|\Zc_{n+1} - \Zb_{n+1} \|^2\eqsp,
\]
where $\bar{\B{B}}_0,\B{B}_0,\B{\tilde{D}}_{0}$ are defined in \eqref{eq:def_B_bar_B}-\eqref{eq:def_projection}.
\end{lemma}
\begin{proof}
The proof is similar to the proof of \Cref{lem:geo_decr} and is omitted.
\end{proof}
To ease notation, for any $i\in[b]$, we consider all along this section the quantities
\begin{align}\label{eq:def_tilde_m_M}
  &\tilde{m}_i=m_i+1/\rho_i\eqsp,
  &\tilde{M}_i= M_i+1/\rho_i \eqsp.
\end{align}
The following lemma provides an explicit expression for $\|\Zc_{n+1}-\Zb_{n+1}\|$ with respect to $\|\Zc_{n}-\Zb_{n}\|$.
\begin{lemma}
\label{lem:contraction_V0}
Assume \Cref{ass:well_defined_density}-\Cref{ass:supp_fort_convex} and let $\bfN\in(\N^*)^{b}, \gammabf \in (\rset_+^*)^b$ such that, for any $i\in[b], \gamma_i < \tilde{M}_i^{-1}$.
Then, for any $n\ge 1$, we have
\begin{multline}
\|\Zc_{n+1}-\Zb_{n+1}\|_{\B{D}_{\bfN\bfgamma}^{-1}} \le \normBig{[\B{M}_{\infty}^{(n)}]^{-1}
+ \sum_{k=0}^{\infty}[\B{M}_{\infty}^{(n)}]^{-1} \B{M}_{k+1}^{(n)}\B{J}(k)\B{D}_{\bfN}^{-\half}\B{D}_{\bfgamma/\bfrho}^{\half} \B{P}_{0} \B{D}_{\bfgamma/\bfrho}^{\half}\B{D}_{\bfN}^{\half}}\\
\times\|\Zc_{n}-\Zb_{n}\|_{\B{D}_{\bfN\bfgamma}^{-1}}\eqsp,\label{eq:contration_V0}
\end{multline}
where $(\B{M}_k^{(n)})_{k \in \N}$ is defined in \eqref{eq:def:M}, $(\tilde{Z}_k, Z_k)_{k\in\N}$ in \eqref{eq:cont_coupling}, $\bfN\bfgamma = (\gamma_1 N_1,\ldots,\gamma_b N_b)$ and $\bfgamma/\bfrho = (\gamma_1 /\rho_1,\ldots,\gamma_b/\rho_b)$.
\end{lemma}
\begin{proof}
Let $n \ge 1$.
By \eqref{eq:coupling_process_Y}, for any $i\in[b], k\in\N$, we obtain
\begin{multline*}
\Yc_{k+1}^{(i, n)} - \Yb_{k+1}^{(i, n)}
=\prbigg{\B{I}_{d_i} - \gamma_i \int_0^1 \nabla^2 \tildeU_i((1-s) \Yb_{k}^{(i, n)} + s \Yc_{k}^{(i, n)})\,\dd s} (\Yc_{k}^{(i, n)} - \Yb_{k}^{(i, n)})\\
+ (\gamma_i/\rho_i) \B{A}_i (\tilde{\theta}_{n} - \theta_{n})\eqsp.
\end{multline*}
Consider the process $((\Ycr_{k}^{(n)},\Ybr_{k}^{(n)})= \{\Ycr_{k}^{(i,n)}, \Ybr_{k}^{(i,n)}\}_{i=1}^b)_{k\in\N}$ with values in $\R^p\times\R^p$ defined for any $i\in[b]$, $k\ge 0$, by
\begin{align}\label{eq:def:discrete_stoch_approx}
	&\Ycr_k^{(i,n)}=\Yc_{\min(k, N_i)}^{(i,n)}\eqsp,
	&\Ybr_k^{(i,n)}=\Yb_{\min(k, N_i)}^{(i,n)}\eqsp.
\end{align}
By \eqref{eq:cont_coupling}, we have $\B{A}_i(\tilde{\theta}_{n}-\theta_{n})=\B{A}_i\bar{\B{B}}_0^{-1}\B{B}_0^{\top}\B{\tilde{D}}_{0}^{\half}(\Zc_{n}-\Zb_{n})$. Since $\B{B}_{0}^{\top} = [\B{A}^{\top}_{1}/\rho_{1}^{\half} \cdots \B{A}^{\top}_b/\rho_b^{\half}]$ and $\B{P}_{0} = \B{B}_{0} \bar{\B{B}}_{0}^{-1} \B{B}_{0}^{\top}$ is the orthogonal projection matrix defined in \eqref{eq:def_projection}, it follows that
\begin{align}\label{eq:eq:rec_N}
\Ycr_{k+1}^{(n)} - \Ybr_{k+1}^{(n)}
=& \prbig{\B{I}_{p}-\B{C}_{k}^{(n)}}(\Ycr_{k}^{(n)} - \Ybr_{k}^{(n)})
+ \B{J}(k)\B{D}_{\bfgamma/\sqrt{\bfrho}} \B{P}_{0}\B{\tilde{D}}_{0}^{\half} (\Ycr_{0}^{(n)} - \Ybr_{0}^{(n)})\eqsp.
\end{align}
Since $\B{D}_{\bfN\bfgamma}$ commutes with $\B{C}_{k}^{(n)}$ and $\B{J}(k)$, multiplying \eqref{eq:eq:rec_N} by $\B{M}_{k+1}^{(n)}\B{D}_{\bfN\bfgamma}^{-\half}$, yields
\begin{multline}\label{eq:to_be_telescoped}
\B{M}_{k+1}^{(n)}\B{D}_{\bfN\bfgamma}^{-\half} (\Ycr_{k+1}^{(n)} - \Ybr_{k+1}^{(n)})
= \B{M}_{k}^{(n)} \B{D}_{\bfN\bfgamma}^{-\half}(\Ycr_{k}^{(n)} - \Ybr_{k}^{(n)})\\
+ \B{M}_{k+1}^{(n)}\B{J}(k)\B{D}_{\bfN}^{-\half}\B{D}_{\bfgamma/\bfrho}^{\half} \B{P}_{0}\B{\tilde{D}}_{0}^{\half} (\Ycr_{0}^{(n)} - \Ybr_{0}^{(n)})\eqsp.
\end{multline}
By definition of the processes in \eqref{eq:cont_coupling}-\eqref{eq:coupling_process_Y} and \eqref{eq:def:discrete_stoch_approx}, we have for $k \ge \max_{i\in[b]}\{N_i\}$, $(\Ycr_{k}^{(n)}, \Ybr_{k}^{(n)})=(\tilde{Z}_{n+1},Z_{n+1})$ and $\B{J}(k)=\mathbf{0}_{p\times p}$.
Therefore summing the previous equality \eqref{eq:to_be_telescoped} yields
\begin{multline*}
\txts\B{M}_{\infty}^{(n)} \B{D}_{\bfN\bfgamma}^{-\half}(\tilde{Z}_{n+1}-Z_{n+1})
= \brbig{\B{M}_{0}^{(n)}
+ \sum_{k=0}^{\infty} \B{M}_{k+1}^{(n)}\B{J}(k)\B{D}_{\bfN}^{-\half}\B{D}_{\bfgamma/\bfrho}^{\half}\B{P}_{0}\B{D}_{\bfgamma/\bfrho}^{\half}\B{D}_{\bfN}^{\half}}\\
\times\B{D}_{\bfN\bfgamma}^{-\half}(\Ycr_{0}^{(n)} - \Ybr_{0}^{(n)})\eqsp.
\end{multline*}
Multiplying this last equality by $[\B{M}_{\infty}^{(n)}]^{-1}$ and applying the norm $\| \cdot\|_{\B{D}_{\bfN\bfgamma}^{-1}}$  concludes the proof.
\end{proof}
The three following lemmata aim at providing an explicit upper bound on \eqref{eq:contration_V0}.
To this end, for $n,k\in\N$ and $i \in [b]$, consider $\B{C}_{k}^{(i, n)}$ corresponding to the $i$-th diagonal block of $\B{C}_{k}^{(n)}$ defined in \eqref{eq:def:C}, \emph{i.e.}
\begin{equation}
  \label{eq:def_B_C_i}
  \B{C}_{k}^{(i, n)}= \1_{[N_i]}(k+1)\gamma_i\ac{\rho_i^{-1}\B{I}_{d_i}+\int_0^1 \nabla^2 U_i((1-s) \Yb_{k}^{(i, n)} + s \Yc_{k}^{(i, n)} )\,\dd s}\in\R^{d_i\times d_i} \eqsp,
\end{equation}
where, for any $n \in \N$ and $i \in [b]$, $(\Yb_{k}^{(i, n)},\Yc_{k}^{(i, n)})_{k \in\N}$ is defined in \eqref{eq:coupling_process_Y}.
Thus, using the definition~\eqref{eq:def:M} of $\B{M}_{k}^{(n)}$, we can write $[\B{M}_{\infty}^{(n)}]^{-1}\B{M}_{k}^{(n)}$ as a block-diagonal matrix $\mathrm{diag}(([\B{M}_{\infty}^{(n)}]^{-1}\B{M}_{k}^{(n)})^1,\ldots,([\B{M}_{\infty}^{(n)}]^{-1}\B{M}_{k}^{(n)})^b)$ where for any $i \in [b]$, $([\B{M}_{\infty}^{(n)}]^{-1}\B{M}_{k}^{(n)})^i= \prod_{l=k}^{N_i-1}(\B{I}_{d_i}-\B{C}_{l}^{(i, n)})\in\R^{d_i\times d_i}$.
\begin{lemma}\label{lem:bound_T1_approx}
Assume \Cref{ass:well_defined_density}-\Cref{ass:supp_fort_convex} and let $\bfN\in(\R_+^*)^b$, $\bfgamma\in(\R_+^*)^b$ such that, for any $i\in[b]$, $\gamma_i<\tilde{M}_i^{-1}$. Then, for any $i \in [b]$, $n \in \nset$ and $k\in[N_i]$, we have
\[
\normn{([\B{M}_{\infty}^{(n)}]^{-1} \B{M}_{k}^{(n)})^i - \B{I}_{d_i} - \txts\sum_{l= k}^{\infty}\B{C}_{l}^{(i, n)}}
\le \exp\{(N_i-k)\gamma_i\tilde{M}_i\} - 1 - (N_i-k)\gamma_i\tilde{M}_i\eqsp,
\]
where $\B{M}_k^{(n)}, \tilde{M}_i$ are defined in \eqref{eq:def:M}, \eqref{eq:def_tilde_m_M} respectively, and $\B{M}_{\infty}^{(n)}$ is the limit of the stationnary sequence $(\B{M}_{k}^{(n)})_{k\in\N}$.
\end{lemma}
\begin{proof}
Let $n \in \N$, $i \in [b]$ and $k \in [N_i]$. The approximation error between $\prod_{l=k}^{\infty}(\B{I}_{d_i} - \B{C}_{l}^{(i, n)})$ and its linear approximation can be upper bounded as
\begin{align}
\nonumber
&\normbigg{\prod_{l=k}^{\infty}(\B{I}_{d_i} - \B{C}_{l}^{(i, n)}) - \B{I}_{d_i} - \sum_{l=k}^{\infty} \B{C}_{l}^{(i, n)}}
=\normbigg{\sum_{m=2}^{\infty}(-1)^{m}\sum_{k\le l_1<\cdots<l_m}\B{C}_{l_1}^{(i,n)}\cdots \B{C}_{l_m}^{(i,n)}} \\
\nonumber
&  \qquad \qquad \le \sum_{m=2}^{\infty}\sum_{k\le l_1<\cdots<l_m}\|\B{C}_{l_1}^{(i,n)}\|\cdots \|\B{C}_{l_m}^{(i,n)}\|
= \prod_{l= k}^{\infty}(1 +\| \B{C}_{l}^{(i, n)}\|) - 1 - \sum_{l\ge k} \|\B{C}_{l}^{(i, n)}\|\\
&\qquad \qquad \le \exp\prbigg{\sum_{l=  k}^{\infty}\|\B{C}_{l}^{(i, n)}\|} - 1 - \sum_{l=k}^{\infty}\| \B{C}_{l}^{(i, n)}\|\eqsp,\nonumber
\end{align}
where the products and the sums are well defined since for any $l\ge N_i$, we have $\B{C}_{l}^{(i, n)}=\mathbf{0}_{d_i}$.
Finally, the proof is concluded using that $x\mapsto \exp(x) - 1 - x$ is increasing on $\R$ and for $l\in \N$, $\| \B{C}_{l}^{(i, n)}\| \le \gamma_i\tilde{M}_i\1_{[N_i]}(l+1)$ from \Cref{ass:supp_fort_convex}-\ref{ass:1}.
\end{proof}
For any $\bfN=(N_1,\ldots,N_b)\in(\N^*)^b, \bfgamma=(\gamma_1,\ldots,\gamma_b)\in(\R_+^*)^b$, define the $p\times p$ block matrices
\begin{align}
\nonumber
&\B{S}_{1}=\mathrm{diag}(\{1-N_1 \gamma_1\tilde{M}_1\}\B{I}_{d_1}, \cdots, \{1-N_b \gamma_b \tilde{M}_b \} \B{I}_{d_b})\eqsp,\\
\label{eq:def:S}
&\B{S}_{2}=\B{I}_p - \sum_{l=0}^{\infty} \B{J}(l)\B{H}_{U,l}^{(n)}
- (\B{D}_{\bfN} \B{D}_{\bfgamma/\bfrho})^{\half} (\B{I}_p-\B{P}_{0}) (\B{D}_{\bfN} \B{D}_{\bfgamma/\bfrho})^{\half}\eqsp, \\
\nonumber
&\B{S}_{3}=\mathrm{diag}\pr{\left\{1- N_1 \gamma_1 m_1 \right\}\B{I}_{d_1}, \hdots, \left\{1- N_b \gamma_b m_b\right\} \B{I}_{d_b}}\eqsp,
\end{align}
where for any $i\in[b]$, $\tilde{M}_i$ is defined in \eqref{eq:def_tilde_m_M} and $\B{P}_{0}, \B{J}(l), \B{H}_{U,l}^{(n)}$ are defined in \eqref{eq:def_projection}, \eqref{eq:def:cont_J}, \eqref{eq:def:cont_H}, respectively.
\begin{lemma}\label{lem:bound:S}
Assume \Cref{ass:well_defined_density}-\Cref{ass:supp_fort_convex}. Then, for any $\bfN\in(\N^*)^{b}, \gammabf \in (\rset_+^*)^b$, we have
\[
\B{S}_{1}\preccurlyeq\B{S}_{2}\preccurlyeq\B{S}_{3}\eqsp.
\]
As a result, under the additional assumption, for any $i \in [b]$, $\gamma_iN_i \le 2/(m_i + M_i + 1/\rho_i)$, we get
\begin{equation}\label{eq:proof_contrac_N_expansion_1}
	\norm{\B{S}_2}
	\le 1- \min_{i\in[b]} \{N_i\gamma_i m_i\}\eqsp.
  \end{equation}
\end{lemma}
\begin{proof}
Since $\B{P}_{0}$ is an orthogonal projection defined in \eqref{eq:def_projection}, we have $\B{P}_{0}\preccurlyeq\B{I}_p$, therefore we easily get
\[
\mathbf{0}_{p\times p}
\preccurlyeq (\B{D}_{\bfN} \B{D}_{\bfgamma/\bfrho})^{\half} (\B{I}_p-\B{P}_{0}) (\B{D}_{\bfN} \B{D}_{\bfgamma/\bfrho})^{\half}
\preccurlyeq \B{D}_{\bfN} \B{D}_{\bfgamma/\bfrho}
\]
and \Cref{ass:supp_fort_convex}-\ref{ass:1}-\ref{ass:2} imply 
\[
\mathrm{diag}(N_1\gamma_1m_1 \B{I}_{d_1},\cdots,N_b\gamma_bm_b \B{I}_{d_b})
\preccurlyeq \sum_{l=0}^{\infty} \B{J}(l)\B{H}_{U,l}^{(n)}
\preccurlyeq \mathrm{diag}(N_1\gamma_1M_1 \B{I}_{d_1},\cdots,N_b\gamma_bM_b \B{I}_{d_b})\eqsp.
\]
Substracting these previous inequalities and adding $\B{I}_p$ complete the first part of the proof.
The additional condition, for any $i \in [b]$, $\gamma_i N_i \le 2/(m_i + M_i + 1/\rho_i)$, ensures that $\B{S}_1$ is definite-positive.
Since $\B{S}_1 \preceq \B{S}_2$, we deduce that $\B{S}_2$ is symmetric positive-definite as well.
Then, $\norm{\B{S}_2}$ is equal to the largest eigenvalue of $\B{S}_2$.
The inequality $\B{S}_2 \preceq \B{S}_3$ concludes the second part of the proof.
\end{proof}
For any $\bfN=(N_1,\ldots,N_b)\in(\N^*)^b, \bfgamma=(\gamma_1,\ldots,\gamma_b)\in(\R_+^*)^b$, define
\begin{multline}\label{eq:def:r}
r_{\bfgamma, \bfrho, \bfN}
= \max_{i\in[b]}\{N_i \gamma_i/\rho_i\}\max_{i\in[b]}\{N_i\gamma_i\tilde{M}_i\} \Big(1/2+\max_{i\in[b]}\{N_i\gamma_i\tilde{M}_i\}\Big)
+ 4\max_{i\in[b]}\{N_i\gamma_i\tilde{M}_i\}^2\eqsp,
\end{multline}
where $\tilde{M}_i$ is defined in \eqref{eq:def_tilde_m_M}.
\begin{lemma}\label{lem:bound:cont_T1}
Assume \Cref{ass:well_defined_density}-\Cref{ass:supp_fort_convex}. Let $\bfN\in(\N^*)^{b}, \gammabf \in (\rset_+^*)^b$ such that, for any $i \in [b]$,  $N_i\gamma_i\le 2/(m_i + \tilde{M}_i)$ and $\gamma_i<\tilde{M}_i^{-1}$. Then, for any $n\in\N$, we have
\begin{align*}
\| [\B{M}_{\infty}^{(n)}]^{-1}
+ \txts\sum_{k=0}^{\infty}[\B{M}_{\infty}^{(n)}]^{-1} \B{M}_{k+1}^{(n)}\B{J}(k)\B{D}_{\bfN}^{-\half}\B{D}_{\bfgamma/\bfrho}^{\half} \B{P}_{0} \B{D}_{\bfgamma/\bfrho}^{\half} \B{D}_{\bfN}^{\half}\|
&\le  1 - \min_{i\in[b]} \{N_i \gamma_i m_i\} + r_{\bfgamma, \bfrho, \bfN}\eqsp,
\end{align*}
where $\B{P}_{0}$, $\B{D}_{\bfgamma/\bfrho}$, $\B{J}(k), \B{M}_{k}^{(n)}$ and $r_{\bfgamma, \bfrho, \bfN}$ are defined in \eqref{eq:def_projection}, \eqref{eq:def:J}, \eqref{eq:def:M} and \eqref{eq:def:r}, respectively.
\end{lemma}
\begin{proof}
Let $n \in \nset$. For any $k \in \N$, define
\begin{equation}\label{eq:def:cont_R}
\B{R}_{k}^{(n)} = \prod_{l=k}^{\infty}(\B{I}_p - \B{C}_{l}^{(n)}) - \B{I}_p + \sum_{l=k}^{\infty} \B{C}_{l}^{(n)}\eqsp, \qquad   \B{R}_{k}^{(i, n)} = \prod_{l=k}^{\infty}(\B{I}_{d_i} - \B{C}_{l}^{(i, n)}) - \B{I}_{d_i} + \sum_{l=k}^{\infty} \B{C}_{l}^{(i, n)} \eqsp, \quad i \in [b]\eqsp,
\end{equation}
where $(\B{C}_{l}^{(i, n)})_{l \in \N}$ is defined in \eqref{eq:def_B_C_i} and remark that the products and the sums are well defined since for any $l\ge N_i$, we have $\B{C}_{l}^{(i, n)}=\mathbf{0}_{d_i}$.
By noting, for any $k\in [\max_{i\in[b]}N_i]$, that
$
[\B{M}_{\infty}^{(n)}]^{-1} \B{M}_{k}^{(n)}
= \prod_{l=k}^{\infty}(\B{I}_p - \B{C}_{l}^{(n)})
$,
it follows that
$
[\B{M}_{\infty}^{(n)}]^{-1} \B{M}_{k}^{(n)} = \B{I}_p - \sum_{l=k}^{\infty} \B{C}_{l}^{(n)} + \B{R}_{k}^{(n)}$.
Since for any $i\in[b], l\ge N_i$, $\B{R}_{k}^{(i, n)}=\B{0}_{d_i}$, thus we have $\B{J}(k)\B{R}_{k+1}^{(n)} =\B{R}_{k+1}^{(n)}$. In addition, using that $\B{M}_{0}^{(n)} = \B{I}_p$, $\B{C}_{l}^{(n)}=\B{J}(l)(\B{D}_{\bfgamma/\bfrho}+\B{H}_{U,l}^{(n)})$, $\B{D}_{\bfN}=\sum_{k=0}^{\infty}\B{J}(k)$, $\B{D}_{\bfN}\B{C}_{l}^{(n)}=\B{C}_{l}^{(n)}\B{D}_{\bfN}$, we get 
\begin{align}
\nonumber
&[\B{M}_{\infty}^{(n)}]^{-1}
+ \sum_{k=0}^{\infty}[\B{M}_{\infty}^{(n)}]^{-1} \B{M}_{k+1}^{(n)} \B{J}(k)\B{D}_{\bfN}^{-\half}\B{D}_{\bfgamma/\bfrho}^{\half} \B{P}_{0} \B{D}_{\bfgamma/\bfrho}^{\half}\B{D}_{\bfN}^{\half} \\
\nonumber
=& \ \B{I}_p - \sum_{l=0}^{\infty} \B{C}_{l}^{(n)}
+ \sum_{k=0}^{\infty} \B{J}(k)\B{D}_{\bfN}^{-\half}\B{D}_{\bfgamma/\bfrho}^{\half} \B{P}_{0} \B{D}_{\bfgamma/\bfrho}^{\half}\B{D}_{\bfN}^{\half} - \sum_{k=0}^{\infty} \sum_{l=k+1}^{\infty} \B{J}(k)\B{D}_{\bfN}^{-\half}\B{C}_{l}^{(n)} \B{D}_{\bfgamma/\bfrho}^{\half} \B{P}_{0} \B{D}_{\bfgamma/\bfrho}^{\half}\B{D}_{\bfN}^{\half} \\
\nonumber
&+ \B{R}_{0}^{(n)} + \sum_{k=0}^{\infty} \B{R}_{k+1}^{(n)} \B{J}(k) \B{D}_{\bfN}^{-\half}\B{D}_{\bfgamma/\bfrho}^{\half} \B{P}_{0} \B{D}_{\bfgamma/\bfrho}^{\half}\B{D}_{\bfN}^{\half} \\
\nonumber
=& \ \B{I}_p - \sum_{l=0}^{\infty}\B{J}(l)\B{H}_{U,l}^{(n)}
- \Big(\sum_{k=0}^{\infty} \B{J}(k)\Big)\B{D}_{\bfN}^{-\half} \B{D}_{\bfgamma/\bfrho}^{\half} (\B{I}_p-\B{P}_{0})\B{D}_{\bfgamma/\bfrho}^{\half}\B{D}_{\bfN}^{\half}\\
&- \sum_{l=1}^{\infty} \Big(\sum_{k=0}^{l-1}\B{J}(k)\Big)\B{D}_{\bfN}^{-\half}\B{C}_{l}^{(n)} \B{D}_{\bfgamma/\bfrho}^{\half} \B{P}_{0} \B{D}_{\bfgamma/\bfrho}^{\half}\B{D}_{\bfN}^{\half}
\nonumber
+ \B{R}_{0}^{(n)} + \sum_{k=0}^{\infty} \B{J}(k)\B{D}_{\bfN}^{-\half} \B{R}_{k+1}^{(n)} \B{D}_{\bfgamma/\bfrho}^{\half} \B{P}_{0} \B{D}_{\bfgamma/\bfrho}^{\half}\B{D}_{\bfN}^{\half} \\
\nonumber
=& \ \B{S}_2 - \sum_{l=1}^{\infty}\Big(\sum_{k=0}^{l-1}\B{J}(k)\Big)\B{D}_{\bfN}^{-1}\B{C}_{l}^{(n)} (\B{D}_{\bfN}\B{D}_{\bfgamma/\bfrho})^{\half} \B{P}_{0} (\B{D}_{\bfgamma/\bfrho}\B{D}_{\bfN})^{\half}\nonumber\\
&+ \B{R}_{0}^{(n)} + \sum_{k=1}^{\infty} \B{D}_{\bfN}^{-1} \B{R}_{k}^{(n)} (\B{D}_{\bfN}\B{D}_{\bfgamma/\bfrho})^{\half} \B{P}_{0} (\B{D}_{\bfgamma/\bfrho}\B{D}_{\bfN})^{\half}\eqsp,
\label{eq:bound:cont_terms}
\end{align}
where $\B{S}_2$ is defined in \eqref{eq:def:S}.
We now bound the different terms of \eqref{eq:bound:cont_terms} separately. 
First, using \eqref{eq:proof_contrac_N_expansion_1}, we have
\begin{equation}
	\norm{\B{S}_2} \le 1- \min_{i\in[b]} \{N_i\gamma_i m_i\}\eqsp.
  \end{equation}
By recalling $\B{R}_{0}^{(n)}$ defined in \eqref{eq:def:cont_R}, \Cref{lem:bound_T1_approx} shows that
\begin{align}
  \|\B{R}_{0}^{(n)}\| \le \max_{i \in[b]}\|\B{R}_{0}^{(i, n)}\|
&=\max_{i\in[b]}\bigg\{\Big\|\prod_{l=0}^{\infty}\big(\B{I}_{d_i}-\B{C}_{l}^{(i, n)}\big)-\B{I}_{d_i}-\sum_{l=0}^{\infty}\B{C}_{l}^{(i, n)}\Big\|\bigg\}
  \\
  &\le \max_{i\in[b]}\bigg\{\exp\Big(\sum_{l=0}^{\infty}\|\B{C}_{l}^{(i, n)}\|\Big)-1-\sum_{l=0}^{\infty}\|\B{C}_{l}^{(i, n)}\|\bigg\}\\
&\le \max_{i\in[b]}\big\{\exp\{(N_i-1)\gamma_i\tilde{M}_i\} -1-(N_i-1)\gamma_i\tilde{M}_i\big\} \\
&\le \max_{i\in[b]}\{((N_i-1)\gamma_i\tilde{M}_i)^2 \mathrm{e}^{(N_i-1)\gamma_i\tilde{M}_i}\}/2 \\
&\le 4\max_{i\in[b]}\{(N_i-1)\gamma_i\tilde{M}_i\}^2\eqsp,\label{eq:boundR_0}
\end{align}
where, in the penultimate line, we used for any $t \ge 0$, that $\exp(t) - 1 -t \le t^2\exp(t)/2$.
Regarding the second term of \eqref{eq:bound:cont_terms}, using that $\B{P}_{0}$ is an orthogonal projector, we get 
\begin{multline*}
\norm{\sum_{l=1}^{\infty}\Big(\sum_{k=0}^{l-1}\B{J}(k)\Big)\B{D}_{\bfN}^{-1}\B{C}_{l}^{(n)}(\B{D}_{\bfN}\B{D}_{\bfgamma/\bfrho})^{\half}\B{P}_{0}(\B{D}_{\bfN}\B{D}_{\bfgamma/\bfrho})^{\half}}\\
\le \max_{i\in[b]}\left(\frac{N_i\gamma_i}{\rho_i}\right)\norm{\sum_{l=1}^{\infty}\Big(\sum_{k=0}^{l-1}\B{J}(k)\Big)\B{D}_{\bfN}^{-1}\B{C}_{l}^{(n)}}\eqsp.
\end{multline*}
Combining the following upper bound
\[
\norm{\sum_{l=1}^{\infty}\Big(\sum_{k=0}^{l-1}\B{J}(k)\Big)\B{D}_{\bfN}^{-1}\B{C}_{l}^{(n)}}
\le \max_{i\in[b]}\acbigg{\frac{1}{N_i}\sum_{l=1}^{\infty}l\|\B{C}_{l}^{(i, n)}\|}
\]
with the fact, for any $i \in [b]$, that $\|\B{C}_{l}^{(i, n)}\|\le\gamma_i\tilde{M}_i\1_{[N_i]}(l+1)$, we get that
\begin{multline}\label{eq:proof_contrac_N_expansion_0}
\norm{\sum_{l=1}^{\infty}\Big(\sum_{k=0}^{l-1}\B{J}(k)\Big)\B{D}_{\bfN}^{-1}\B{C}_{l}^{(n)}(\B{D}_{\bfN}\B{D}_{\bfgamma/\bfrho})^{\half}\B{P}_{0}(\B{D}_{\bfN}\B{D}_{\bfgamma/\bfrho})^{\half}}\\
\le \max_{i\in[b]} \bigg(\frac{N_i\gamma_i}{\rho_i}\bigg) \max_{i\in[b]}\acbigg{\frac{N_i\gamma_i\tilde{M}_i}{2}}\eqsp.
\end{multline}
To upper bound the last term of \eqref{eq:bound:cont_terms}, we start from the following inequality
\begin{equation*}\label{eq:bound:DR}
\norm{\sum_{k=1}^{\infty}\B{D}_{\bfN}^{-1}\B{R}_{k}^{(n)}}
\le \max_{i\in[b]}\acbigg{\frac{1}{N_i}\sum_{k=1}^{N_i-1}\|\B{R}_{k}^{(i, n)}\|}\eqsp.
\end{equation*}
\Cref{lem:bound_T1_approx} shows that  for any $k\in [N_i-1]$ and $i \in [b]$,
$
\txts\|\B{R}_{k}^{(i, n)}\|\le \exp\{(N_i-k)\gamma_i\tilde{M}_i\}-1-(N_i-k)\gamma_i\tilde{M}_i
$. Then, for any $i\in[b]$, we have
\begin{multline}
\frac{1}{N_i}\sum_{k=1}^{N_i-1}\|\B{R}_{k}^{(i, n)}\| \le \frac{1}{N_i}\sum_{k=1}^{N_i-1}[\exp\{(N_i-k)\gamma_i\tilde{M}_i\}-1-(N_i-k)\gamma_i\tilde{M}_i] \\
\le (N_i\gamma_i\tilde{M}_i)^{-1}\int_{0}^{N_i\gamma_i\tilde{M}_i}(\rme^{t}-1-t)\,\dd t
\le \frac{(N_i\gamma_i\tilde{M}_i)^2}{12}\big(\rme^{N_i\gamma_i\tilde{M}_i}+1\big) \\
\le \max_{i \in [b]}\{(N_i\gamma_i\tilde{M}_i)^2\}\eqsp,\label{eq:bound:cont_int_exp}
\end{multline}
where we have used $\mathrm{e}^2+1\le 12$. Plugging \eqref{eq:bound:cont_int_exp}, \eqref{eq:proof_contrac_N_expansion_0}, \eqref{eq:boundR_0} into \eqref{eq:proof_contrac_N_expansion_1}, we get
\begin{align*}
\normBig{[\B{M}_{\infty}^{(n)}]^{-1}
+ \sum_{k\in\N}[\B{M}_{\infty}^{(n)}]^{-1} \B{M}_{k+1}^{(n)}\B{J}(k)\B{D}_{\bfN}^{-\half}\B{D}_{\bfgamma/\bfrho}^{\half} \B{P}_{0} \B{D}_{\bfgamma/\bfrho}^{\half}\B{D}_{\bfN}^{\half}}
\le 1 - \min_{i\in[b]} \{N_i \gamma_i m_i\} + r_{\bfgamma,\bfrho,\bfN}\eqsp,
\end{align*}
where $r_{\bfgamma,\bfrho,\bfN}$ is defined in \eqref{eq:def:r}.
\end{proof}
\begin{lemma}\label{lem:W2_dirac_contraction}
Assume \Cref{ass:well_defined_density}-\Cref{ass:supp_fort_convex}. Let $\bfN\in(\N^*)^{b}, \gammabf \in (\rset_+^*)^b$ such that, for any $i \in [b]$, $N_i\gamma_i\le 2/(m_i + \tilde{M}_i)$ and $\gamma_i<\tilde{M}_i^{-1}$.
Then, for any $\bx=(\bz^{\top},\btheta^{\top})^{\top},\tilde{\bx}=(\tilde{\bz}^{\top},\tilde{\btheta}^{\top})^{\top} \in \mathbb{R}^{p+d}$, with $(\btheta,\tilde{\btheta}) \in (\mathbb{R}^d)^2, (\bz,\tilde{\bz})\in(\R^p)^2$ and any $n\ge 1$ we have
\begin{multline*}
\wasserstein{}^2(\updelta_{\tilde{\bx}} P_{\bfrho, \bfgamma,\bfN}^{n}, \updelta_{\bx} P_{\bfrho, \bfgamma,\bfN}^{n})
\le (1-\min_{i\in[b]}\{N_i \gamma_i m_i\} + r_{\bfgamma,\bfrho,\bfN})^{2n-2}
(1+\|\bar{\B{B}}_0^{-1}\B{B}_0^{\top}\B{\tilde{D}}_{0}^{\half}\|^2) \\
\times\frac{\max_{i\in[b]}\{N_i\gamma_i\}}{\min_{i\in[b]}\{N_i\gamma_i\}}
\txts\left[\|[\B{M}_{\infty}^{(0)}]^{-1}\|\|\tilde{\bz}-\bz\|
+ (\sum_{i\in[b]}\|\B{A}_i\|/\rho_i)\|\tilde{\btheta}- \btheta\|\right]^2\eqsp,
\end{multline*}
where $\B{B}_0, \bar{\B{B}}_0, \B{\tilde{D}}_{0}, P_{\bfrho, \bfgamma,\bfN}, \B{M}_{\infty}^{(0)}, r_{\bfgamma,\bfrho,\bfN}$ are defined in \eqref{eq:def_B_bar_B}, \eqref{eq:def_projection}, \eqref{eq:P_rho_gamma_N}, \eqref{eq:def:M}, \eqref{eq:def:r}, respectively.
\end{lemma}
\begin{proof}
Combining \Cref{lem:contraction_V0} and \Cref{lem:bound:cont_T1}, we have for $n \ge 1$,
\begin{align*}
\|\Zc_{n+1}-\Zb_{n+1}\|_{\B{D}_{\bfN\bfgamma}^{-1}}
&\le  (1-\min_{i\in[b]} \{N_i \gamma_i m_i\} + r_{\bfgamma,\bfrho,\bfN})\|\Zc_{n}-\Zb_{n}\|_{\B{D}_{\bfN\bfgamma}^{-1}}\eqsp.
\end{align*}
Thereby, for any $n\ge 1$, we obtain by induction
\begin{equation}\label{eq:bound:induction_z}
\|\Zc_{n}-\Zb_{n}\|_{\B{D}_{\bfN\bfgamma}^{-1}}
\le (1-\min_{i\in[b]} \{N_i \gamma_i m_i\} + r_{\bfgamma,\bfrho,\bfN})^{n-1}\|\Zc_{1}-\Zb_{1}\|_{\B{D}_{\bfN\bfgamma}^{-1}}\eqsp.
\end{equation}
Define the process $((\Ycr_{k}^{(0)},\Ybr_{k}^{(0)})= \{\Ycr_{k}^{(i,0)}, \Ybr_{k}^{(i,0)}\}_{i=1}^b)_{k\in\N}$ with values in $\R^p\times\R^p$ defined for any $i\in[b]$, $k\ge 0$ by
\begin{align*}
&\Ycr_k^{(i,0)}=\Yc_{\min(k, N_i)}^{(i,0)}\eqsp,
&\Ybr_k^{(i, 0)}=\Yb_{\min(k, N_i)}^{(i,0)}\eqsp.
\end{align*}
By \eqref{eq:cont_coupling}, it follows that for any $i \in [b]$, $(\Zc_1^i,\Zb_1^i) = (\Yc_{N_i}^{(i, 0)},\Yb_{N_i}^{(i, 0)})$ where $(\Yc_{0}^{(i, 0)},\Yb_{0}^{(i, 0)}) = (\Zc^i_0,\Zb^i_0)$. We get by \eqref{eq:coupling_process_Y} for $k\ge 0$, 
\begin{align*}
\Ycr_{k+1}^{(0)} - \Ybr_{k+1}^{(0)}
=& (\B{I}_{p}-\B{C}_{k}^{(0)}) (\Ycr_{k}^{(0)} - \Ybr_{k}^{(0)})
+ \B{J}(k)\B{D}_{\bfgamma/\sqrt{\bfrho}}\B{B}_{0} (\tilde{\theta}_0- \theta_0)\eqsp.
\end{align*}
Hence, for $k \ge 0$, we obtain
\begin{align*}
\B{M}_{k+1}^{(0)}\B{D}_{\bfN\bfgamma}^{-\half} (\Ycr_{k+1}^{(0)} - \Ybr_{k+1}^{(0)})
= \B{M}_{k}^{(0)} \B{D}_{\bfN\bfgamma}^{-\half}(\Ycr_{k}^{(0)} - \Ybr_{k}^{(0)})
+ \B{M}_{k+1}^{(0)}\B{J}(k)\B{D}_{\bfN}^{-\half}\B{D}_{\bfgamma/\bfrho}^{\half}\B{B}_{0} (\tilde{\theta}_0- \theta_0)\eqsp.
\end{align*}
Summing the previous equality gives
\begin{align*}
\B{M}_{\infty}^{(0)} \B{D}_{\bfN\bfgamma}^{-\half}(\Ycr_{\bfN}^{(0)} - \Ybr_{\bfN}^{(0)})
&= \B{M}_{0}^{(0)}\B{D}_{\bfN\bfgamma}^{-\half}(\Ycr_{0}^{(0)} - \Ybr_{0}^{(0)})
+ \sum_{k=0}^{\infty} \B{M}_{k+1}^{(0)}\B{J}(k)\B{D}_{\bfN}^{-\half}\B{D}_{\bfgamma/\bfrho}^{\half}\B{B}_{0} (\tilde{\theta}_0- \theta_0)\eqsp.
\end{align*}
Multiplying by $[\B{M}_{\infty}^{(0)}]^{-1}$ and using the fact that $(\theta_0, \Ybr_{0}^{(0)}) = (\btheta, \bz)$, $(\tilde{\theta}_0, \Ycr_{0}^{(0)}) = (\tilde{\btheta}, \tilde{\bz})$, we get
\begin{align*}
\B{D}_{\bfN\bfgamma}^{-\half}(\Zc_{1}-\Zb_{1})
=& [\B{M}_{\infty}^{(0)}]^{-1}\B{D}_{\bfN\bfgamma}^{-\half} (\tilde{\bz}-\bz) + \sum_{k=0}^{\infty} [\B{M}_{\infty}^{(0)}]^{-1}\B{M}_{k+1}^{(0)}\B{J}(k)\B{D}_{\bfN}^{-\half}\B{D}_{\bfgamma/\bfrho}^{\half}\B{B}_{0}  (\tilde{\btheta}- \btheta)\eqsp.
\end{align*}
Plugging the result in \eqref{eq:bound:induction_z} implies for any $n\ge 1$,
\begin{multline}\label{eq:bound:diff_zc_zb}
\txts\|\Zc_{n}-\Zb_{n}\|_{\B{D}_{\bfN\bfgamma}^{-1}}
\le \txts(1-\min_{i\in[b]}\{N_i \gamma_i m_i\} + r_{\bfgamma,\bfrho,\bfN})^{n-1}\left[\|[\B{M}_{\infty}^{(0)}]^{-1}\|\|\tilde{\bz}-\bz\|_{\B{D}_{\bfN\bfgamma}^{-1}}\right.\\
+ \|\sum_{k=0}^{\infty}[\B{M}_{\infty}^{(0)}]^{-1}\B{M}_{k+1}^{(0)}\B{J}(k)\B{D}_{\bfN}^{-\half}\B{D}_{\bfgamma/\bfrho}^{\half}\B{B}_{0}\|\|\tilde{\btheta}- \btheta\|\bigg]\eqsp.
\end{multline}
By \Cref{ass:supp_fort_convex}-\ref{ass:2} and the definitions of $\B{C}_{l}^{(0)}, \B{M}_{k}^{(0)}$ given in \eqref{eq:def:C}, \eqref{eq:def:M}, we have $\|\B{I}_{d_i}-\B{C}_{l}^{(i, 0)}\|\le 1-\gamma_i\tilde{m}_i$. As a result and since $([\B{M}_{\infty}^{(0)}]^{-1}\B{M}_{k}^{(0)})^i=\prod_{l=0}^{k-1}(\B{I}_{d_i}-\B{C}_{l}^{(i, 0)})$, the triangle inequality implies
\begin{align*}
\normbigg{\sum_{k=0}^{\infty} [\B{M}_{\infty}^{(0)}]^{-1}\B{M}_{k+1}^{(0)}\B{J}(k)\B{D}_{\bfN}^{-\half}\B{D}_{\bfgamma/\bfrho}^{\half}\B{B}_{0}}
&\le \sum_{i\in[b]}\sqrt{\gamma_i/N_i}(\|\B{A}_i\|/\rho_i)\sum_{k=1}^{N_i}\|([\B{M}_{\infty}^{(0)}]^{-1}\B{M}_{k}^{(0)})^i\|\\
&\le \sum_{i\in[b]}\sqrt{\gamma_i/N_i}(\|\B{A}_i\|/\rho_i)\sum_{k=0}^{N_i-1}(1-\gamma_i\tilde{m}_i)^{k}\\
&\le \sum_{i\in[b]}\|\B{A}_i\|\sqrt{N_i\gamma_i}/\rho_i\eqsp.
\end{align*}
Plugging this result in \eqref{eq:bound:diff_zc_zb}, we get
\begin{multline*}
\|\Zc_{n}-\Zb_{n}\|_{\B{D}_{\bfN\bfgamma}^{-1}}
\le (1-\min_{i\in[b]}\{N_i \gamma_i m_i\} + r_{\bfgamma,\bfrho,\bfN})^{n-1}
\Big[\|[\B{M}_{\infty}^{(0)}]^{-1}\|\|\tilde{\bz}-\bz\|_{\B{D}_{\bfN\bfgamma}^{-1}}\\
+ \prBig{\sum_{i\in[b]}\|\B{A}_i\|\sqrt{N_i\gamma_i}/\rho_i}\|\tilde{\btheta}- \btheta\|\bigg]\eqsp.
\end{multline*}
Finally, \Cref{lem:geo_decr_cont1} gives
\begin{multline*}
\|\tilde{X}_n - X_n\|^2
\le (1-\min_{i\in[b]}\{N_i \gamma_i m_i\} + r_{\bfgamma,\bfrho,\bfN})^{2n-2}
\cdot (1+\|\bar{\B{B}}_0^{-1}\B{B}_0^{\top}\B{\tilde{D}}_{0}^{\half}\|^2) \frac{\max_{i\in[b]}\{N_i\gamma_i\}}{\min_{i\in[b]}\{N_i\gamma_i\}} \\
\times \brbigg{\|[\B{M}_{\infty}^{(0)}]^{-1}\|\|\tilde{\bz}-\bz\|
+ \prBig{\sum_{i\in[b]}\|\B{A}_i\|/\rho_i}\|\tilde{\btheta}- \btheta\|}^2 \eqsp.
\end{multline*}
Plugging this result into \eqref{eq:W2X_def2} concludes the proof.
\end{proof}
The following result gives a condition on $\max_{i\in[b]}\{N_i\gamma_i\}$ to simplify the contrating term in \Cref{lem:W2_dirac_contraction} to $1 - \min_{i\in[b]} \{N_i\gamma_i m_i\}/2$. 
To this end, define
\begin{align*}
	A_0 &= \max_{i\in [b]}\{\tilde{M}_i\}\max_{i\in [b]}\{1/\rho_i\}/2 + 4\max_{i\in [b]}\{\tilde{M}_i\}^2\eqsp, \\ 
	A_1 &= \max_{i\in [b]}\{\tilde{M}_i\}^2\max_{i\in [b]}\{1/\rho_i\}\eqsp. 
\end{align*}
\begin{lemma}
	\label{lem:contrac}
	Assume \Cref{ass:well_defined_density}-\Cref{ass:supp_fort_convex} and let $c\in\R_+^*, \bfN\in(\N^*)^{b}, \gammabf \in (\rset_+^*)^b$ such that
	\begin{equation}
			\label{eq:Ngamma}
	\begin{aligned}
		&\min_{i\in[b]}\{N_i\gamma_i\}/\max_{i\in[b]}\{N_i\gamma_i\} \ge c\eqsp, \\
		&\max_{i\in[b]}\{N_i\gamma_i\} \le \frac{c\min_{i\in[b]}\{m_i\}}{2A_0+\sqrt{2A_1c\min_{i\in[b]}\{m_i\}}} \wedge \frac{2}{\max_{i\in[b]}\{m_i + M_i + 1/\rho_i\}}\eqsp.
	\end{aligned}
	\end{equation}
	Then, $1 - \min_{i\in[b]} \{N_i\gamma_i m_i\}+ r_{\bfgamma,\bfrho,\bfN} < 1 - \min_{i\in[b]} \{N_i\gamma_i m_i\}/2 < 1$, where $r_{\bfgamma,\bfrho,\bfN}$ is defined in \eqref{eq:def:r}.
\end{lemma}
\begin{proof}
	The proof is straightforward solving a second order polynomial inequality and using for any $a,{b}\in\R_+^*$, ${a}+\frac{b^2}{2{a}+{b}}\le \sqrt{{a}^2+b^2}$.
\end{proof}
\subsubsection{Proof of \Cref{prop:convergence_rho_gamma}}
The next proposition quantifies the convergence of $\updelta_{\B{x}} P_{\bfrho, \bfgamma,\bfN}^{n}$ towards $\Pi_{\bfrho,\bfgamma}$ in $(\mathcal{P}_2(\Rd), \wasserstein{})$, where $\Pi_{\bfrho,\bfgamma}$ is the stationnary distribution derived in \Cref{prop:convergence_N_1}. In addition, it generalises and gives a more formal statement than \Cref{prop:convergence_rho_gamma}.
%
\begin{proposition}\label{cor:convergence_rho_gamma}
    Assume \Cref{ass:well_defined_density}-\Cref{ass:supp_fort_convex} and let $c>0$ and $\bfgamma = \{\gamma_i\}_{i=1}^b$, $\bfN\in(\N^*)^{b}$ such that \eqref{eq:Ngamma} is satisfied, for any $i\in[b]$, $N_i\gamma_i < 2/\txts\max_{i\in[b]}\{m_i + \tilde{M}_i\}$ and $\gamma_i<\tilde{M}_i^{-1}$.  Then, $P_{\bfrho,\bfgamma,\bfN}$  defined in \eqref{eq:P_rho_gamma_N} admits a unique invariant probability measure $\Pi_{\rhobf,\bfgamma,\bfN}$. In addition, for any $\bx=(\bz^{\top},\btheta^{\top})^{\top}$ whith $(\btheta,\bz) \in \mathbb{R}^d \times\mathbb{R}^p$, any integer $n\ge 1$, we have
\begin{multline*}
\wasserstein{}^2(\updelta_{\B{x}} P_{\bfrho, \bfgamma,\bfN}^{n}, \Pi_{\bfrho,\bfgamma})
\le (1-\min_{i\in[b]}\{N_i \gamma_i m_i\}/2)^{2n-2}
\cdot (1+\|\bar{\B{B}}_0^{-1}\B{B}_0^{\top}\B{\tilde{D}}_{0}^{\half}\|^2) \frac{\max_{i\in[b]}\{N_i\gamma_i\}}{\min_{i\in[b]}\{N_i\gamma_i\}} \\
\times \int_{\R^d\times\R^p}\brbigg{\|[\B{M}_{\infty}^{(0)}]^{-1}\|\|\tilde{\bz}-\bz\|
+ \prBig{\sum_{i\in[b]}\|\B{A}_i\|/\rho_i}\|\tilde{\btheta}- \btheta\|}^2\dd \Pi_{\bfrho,\bfgamma}(\tilde{\bx}) \eqsp,
\end{multline*}
where $\B{B}_0, \bar{\B{B}}_0, \B{M}_{\infty}^{(0)}$ are defined in \eqref{eq:def_B_bar_B}, \eqref{eq:def:M}, respectively.

Finally, if $\bfN= N (1,\ldots,1) = N \bfOne_b$ for $N \geq 1$, then  $\Pi_{\rhobf,\bfgamma,\bfN} =  \Pi_{\rhobf,\bfgamma,\bfOne_b}$.

\end{proposition}
\begin{proof}
Note that under the conditions on $\bfgamma$ and $\bfN$ stated in \Cref{cor:convergence_rho_gamma}, \Cref{lem:contrac} ensures that $1-\min_{i\in[b]}\{N_i \gamma_i m_i\}/2 < 1$. 
		Then, from \Cref{lem:W2_dirac_contraction} and \citet[Lemma 20.3.2, Theorem 20.3.4]{douc2018markov}, we deduce the existence and unicity of a stationary distribution $\Pi_{\bfrho,\bfgamma,\bfN}$ for $P_{\bfrho,\bfgamma,\bfN}$.
		The proof is concluded by using the upper bound given in \Cref{lem:W2_dirac_contraction}.

We now show the last statement and assume that $\bfN = N\bfOne_b$, for $N \geq1 $.                 By \Cref{prop:convergence_N_1}, we have the existence and unicity of a stationary distribution $\Pi_{\bfrho,\bfgamma,\bfOne_b}$ which is invariant for $P_{\bfrho,\bfgamma}$ defined in \eqref{eq:def:prop2:P_rho_gamma}. For ease of notation, we simply denote $\Pi_{\bfrho,\bfgamma,\bfOne_b}$ by $\Pi_{\bfrho,\bfgamma}$
We now show that $\Pi_{\bfrho,\bfgamma}$ is also  invariant for $P_{\bfrho,\bfgamma,\bfN}$ defined in \eqref{eq:P_rho_gamma_N}.
Using the fact that $P_{\bfrho,\bfgamma}$ defined in \eqref{eq:def:prop2:P_rho_gamma} leaves $\Pi_{\bfrho,\bfgamma}$ invariant from \Cref{prop:convergence_N_1} and Fubini's theorem, we get for any $\msa \in \mathcal{B}(\Rd)$ and $\msb \in \mathcal{B}(\R^p)$,
    \begin{align}
& \Pi_{\bfrho,\bfgamma}P_{\bfrho,\bfgamma,\bfN}(\msa \times \msb)\\
      & = \int_{\msa \times \mathsf{B}} \int_{\mathbb{R}^d \times \mathbb{R}^p} \Pi_{\bfrho,\bfgamma}(\dd \tilde{\btheta},\dd \tilde{\bz})P_{\bfrho,\bfgamma,\bfN}((\tilde{\btheta},\tilde{\bz}),(\dd \btheta,\dd \bz)) \nonumber\\
        &= \int_{\msa \times \mathsf{B}} \int_{\mathbb{R}^d \times \mathbb{R}^p} \Pi_{\bfrho,\bfgamma}(\dd \tilde{\btheta},\dd \tilde{\bz})Q_{\bfrho,\bfgamma,\bfN}(\tilde{\bz},\dd \bz | \tilde{\btheta})\Pi_{\bfrho}(\dd\btheta|\bz) \nonumber\\
        &= \int_{\msa \times \mathsf{B}} \int_{\mathbb{R}^d \times \mathbb{R}^p} \Pi_{\bfrho,\bfgamma}(\dd \tilde{\btheta},\dd \tilde{\bz})\br{\prod_{i=1}^b R_{\rho_i,\gamma_i}^{N_i}(\tilde{\bz}_{i},\dd \bz_i|\tilde{\btheta})}\Pi_{\bfrho}(\dd\btheta|\bz) \nonumber \\
        &= \int_{\msa \times \mathsf{B}} \int_{\mathbb{R}^d \times \mathbb{R}^p} \Pi_{\bfrho,\bfgamma}(\dd \tilde{\btheta},\dd \tilde{\bz})\int_{\mathbb{R}^p}\br{\prod_{i=1}^b R_{\rho_i,\gamma_i}(\tilde{\bz}_{i},\dd \tilde{\bz}_{i}^{(1)}|\tilde{\btheta})}\br{\prod_{i=1}^b R_{\rho_i,\gamma_i}^{N_i-1}(\tilde{\bz}_{i}^{(1)},\dd \bz_i|\tilde{\btheta})}\Pi_{\bfrho}(\dd\btheta|\bz) \nonumber \\
        &= \int_{\msa \times \mathsf{B}} \int_{\mathbb{R}^d \times \mathbb{R}^p} \br{\int_{\mathbb{R}^d \times \mathbb{R}^p} \Pi_{\bfrho,\bfgamma}(\dd \tilde{\btheta},\dd \tilde{\bz})\br{\prod_{i=1}^b R_{\rho_i,\gamma_i}(\tilde{\bz}_{i},\dd \tilde{\bz}_{i}^{(1)}|\tilde{\btheta})} \Pi_{\bfrho}(\dd \tilde{\btheta}^{(1)}|\tilde{\bz}_{i}^{(1)})} \nonumber\\
        &\times \br{\prod_{i=1}^b R_{\rho_i,\gamma_i}^{N_i-1}(\tilde{\bz}_{i}^{(1)},\dd \bz_i|\tilde{\btheta})}\Pi_{\bfrho}(\dd\btheta|\bz) \nonumber \\
        &= \int_{\msa \times \mathsf{B}} \int_{\mathbb{R}^d \times \mathbb{R}^p} \Pi_{\bfrho,\bfgamma}(\dd \tilde{\btheta}^{(1)},\dd \tilde{\bz}^{(1)}) \br{\prod_{i=1}^b R_{\rho_i,\gamma_i}^{N_i-1}(\tilde{\bz}_{i}^{(1)},\dd \bz_i|\tilde{\btheta}^{(1)})}\Pi_{\bfrho}(\dd\btheta|\bz) \nonumber \eqsp.
    \end{align}
    Using a straightforward induction, we finally get 
    \begin{align*}
        \int_{\msa \times \mathsf{B}} \int_{\mathbb{R}^d \times \mathbb{R}^p} \Pi_{\bfrho,\bfgamma}(\dd \tilde{\btheta},\dd \tilde{\bz})P_{\bfrho,\bfgamma,\bfN}((\tilde{\btheta},\tilde{\bz}),(\dd \btheta,\dd \bz)) = \int_{\msa \times \mathsf{B}} \Pi_{\bfrho,\bfgamma}(\dd \btheta,\dd \bz)\eqsp,
    \end{align*}
    which shows that $P_{\bfrho,\bfgamma,\bfN}$ leaves $\Pi_{\bfrho,\bfgamma}$ invariant.
		Since this stationary distribution is unique, we conclude that $\Pi_{\bfrho,\bfgamma,\bfN} = \Pi_{\bfrho,\bfgamma}$.
\end{proof}

We specify our result to the case where we take  a specific initial distribution. To define it, consider
\begin{equation}
  \label{eq:def_x_star_2}
  \text{
    $\B{x}^{\star} = ([\btheta^{\star}]^{\top},[\bz^{\star}]^{\top})^{\top}$, where $\btheta^{\star} = \argmin\{-\log\pi\}$ and $\bz^{\star} = ([\B{A}_1\btheta^{\star}]^{\top} , \cdots,[\B{A}_b\btheta^{\star}]^{\top})^{\top}$} \eqsp.
\end{equation}
We define the probability measure
\begin{equation}
\label{eq:_def_mu_star_supp_1} 
\mu_{\rhobf}^{\star} =   \updelta_{\B{z}^{\star}}  \otimes \Pi_{\rhobf}(\cdot|\bfz^{\star}) \eqsp.
\end{equation}
Note that sampling from $\mu_{\rhobf}^{\star}$ is straightforward  and simply consists in setting $\zbf_0 = \bfz^{\star}$ and $\bftheta_0= \bar{\B{B}}_0^{-1}\B{B}_0^{\top}\B{\tilde{D}}_{0}^{\half} \B{z}_0 + \bar{\B{B}}_0^{-\half} \xi $, where $\xi$ is a $d$-dimensional standard Gaussian random variable.
We now specify our result when using $\mu_{\rhobf}^{\star}$ as an initial distribution. 
Define the $\bz$-marginal under $\Pi_{\bfrho,\bfgamma}$ by
\begin{equation}
  \label{eq:def_pi_bz}
  \pi^{\bz}_{\bfrho,\bfgamma} = \int_{\Rd}\Pi_{\bfrho,\bfgamma}(\dd\btheta,\bz) \eqsp,
\end{equation}
and the transition kernel of the Markov chain $\{Z_{n}\}_{n \ge 0}$, for all $\bz \in \mathbb{R}^p$ and $\msb \in \mathcal{B}( \mathbb{R}^p)$, by
\begin{equation}
\label{eq:def:trans_kernel_z1}
    P_{\bfrho,\bfgamma,\bfN}^{\bz}(\bz,\msb) = \int_{\Rd} Q_{\bfrho,\bfgamma,\bfN}(\bz,\msb|\btheta)\Pi_{\rho}(\dd \btheta|\bz)\eqsp,
\end{equation}
where $\Pi_{\rho}(\cdot|\cdot)$ and $Q_{\bfrho,\bfgamma,\bfN}$ are defined in \eqref{eq:def:Pi_rho_cond} and \eqref{eq:Q_rho_gamma_N}, respectively.

\begin{proposition}\label{cor:convergence_rho_gamma_star_v2}
    Assume \Cref{ass:well_defined_density}-\Cref{ass:supp_fort_convex} and let $c>0$ and $\bfgamma = \{\gamma_i\}_{i=1}^b$, $\bfN\in(\N^*)^{b}$ such that \eqref{eq:Ngamma} is satisfied, for any $i\in[b]$, $N_i\gamma_i < 2/\txts\max_{i\in[b]}\{m_i + \tilde{M}_i\}$ and $\gamma_i<\tilde{M}_i^{-1}$. 
Then, for any integer $n\ge 1$, we have
\begin{multline*}
\wasserstein{}(\mu_{\rhobf}^{\star} P_{\bfrho,\bfgamma,\bfN}^{n}, \Pi_{\bfrho,\bfgamma})
\le 2^{\half}(1 - \min_{i\in[b]} \{N_i\gamma_i m_i\}/2)^{n-1}
\cdot (1+\|\bar{\B{B}}_0^{-1}\B{B}_0^{\top}\B{\tilde{D}}_{0}^{\half}\|^2)^{\half}\max_{i\in[b]}\{N_i\gamma_i\}^{\half}\\
\times  \defEns{\int_{\rset^d} \|\bfz_1 - \B{z}^{\star}\|_{\B{D}_{\bfN\bfgamma}^{-1}}^2 \pi^{\bz}_{\bfrho,\bfgamma}(\rmd\bz_1) + \int_{\rset^d} \|\bfz_1 - \B{z}^{\star}\|_{\B{D}_{\bfN\bfgamma}^{-1}}^2 P_{\bfrho,\bfgamma,\bfN}^{\bfz}(\bfz^{\star}, \rmd \bfz_1)}^{\half}\eqsp,
\end{multline*}
where $\bar{\B{B}}_0,\B{B}_0,\B{\tilde{D}}_{0}$ are defined in \eqref{eq:def_B_bar_B}-\eqref{eq:def_projection}.
\end{proposition}
\begin{proof}
Consider for $n \in\nsets$,  $X_{n}=(\theta_{n}^{\top}, Z_{n}^{\top})^{\top}, \tilde{X}_{n}=(\tilde{\theta}_{n}^{\top}, \Zc_{n}^{\top})^{\top}$ defined in \eqref{eq:cont_coupling} with $X_0$ distributed according to $\mu_{\rhobf}^{\star}$ and $\tilde{X}_{0}$ distributed according to $\Pi_{\rhobf,\gammabf}$.
Combining \Cref{lem:contraction_V0}, \Cref{lem:bound:cont_T1} and \Cref{lem:contrac}, we have for $n \ge 1$,
\begin{align*}
\|\Zc_{n+1}-\Zb_{n+1}\|_{\B{D}_{\bfN\bfgamma}^{-1}}
&\le  (1-\min_{i\in[b]} \{N_i \gamma_i m_i\}/2)\|\Zc_{n}-\Zb_{n}\|_{\B{D}_{\bfN\bfgamma}^{-1}}\eqsp.
\end{align*}
Thereby, for any $n\ge 1$, we obtain by induction
\begin{equation}\label{eq:bound:induction_z}
\|\Zc_{n}-\Zb_{n}\|_{\B{D}_{\bfN\bfgamma}^{-1}}
\le (1-\min_{i\in[b]} \{N_i \gamma_i m_i\}/2)^{n-1}\|\Zc_{1}-\Zb_{1}\|_{\B{D}_{\bfN\bfgamma}^{-1}}\eqsp.
\end{equation}
Using $\normn{\tilde{Z}_1-Z_1}^2_{\B{D}_{\bfN\bfgamma}^{-1}}\le 2\normn{\tilde{Z}_1-\bz^{\star}}_{\B{D}_{\bfN\bfgamma}^{-\half}}^{2}
	+2\normn{Z_{1}-\bz^{\star}}_{\B{D}_{\bfN\bfgamma}^{-\half}}^2$ combined with the definition of the Wasserstein distance and \Cref{lem:geo_decr_cont1} give
\begin{align}
	\nonumber
	\wasserstein{}(\mu_{\rhobf}^{\star} P_{\bfrho,\bfgamma,\bfN}^{n}, \Pi_{\bfrho,\bfgamma})
	&\le \E\br{\normn{\tilde{X}_n-X_n}^2}^{\half}\\ 
	\nonumber
	&\le (1+\|\bar{\B{B}}_0^{-1}\B{B}_0^{\top}\B{\tilde{D}}_{0}^{\half}\|^2)^{\half}\max_{i\in[b]}\{N_i\gamma_i\}^{\half}
	\E\brbigg{\normn{\tilde{Z}_n-Z_n}_{\B{D}_{\bfN\bfgamma}^{-\half}}^2}^{\half}\\ 
	\nonumber
	&\le  2^{\half}(1 - \min_{i\in[b]} \{N_i\gamma_i m_i\}/2)^{n-1}
	(1+\|\bar{\B{B}}_0^{-1}\B{B}_0^{\top}\B{\tilde{D}}_{0}^{\half}\|^2)^{\half}\max_{i\in[b]}\{N_i\gamma_i\}^{\half}\\
	\label{eq:bound:endproof}
	&\times \E\brbigg{\normn{\tilde{Z}_1-\bz^{\star}}_{\B{D}_{\bfN\bfgamma}^{-\half}}^{2}
	+\normn{Z_{1}-\bz^{\star}}_{\B{D}_{\bfN\bfgamma}^{-\half}}^2}^{\half}\eqsp.
\end{align}
Since $\tilde{X}_{0}$ is distributed according to the stationnary distribution $\Pi_{\rhobf,\gammabf}$,  $\tilde{X}_{1}$ also  and therefore $\Zc_1$ is distributed according to $\pi_{\rhobf,\bgamma}^{\bfz}$. Finally, by definition $\Zb_{1}$ has distribution $P_{\bfrho,\bfgamma,\bfN}^{\bfz}(\bfz^{\star},\cdot)$, therefore \eqref{eq:bound:endproof} completes the proof.
\end{proof}

\section{Proof of \Cref{prop:bound_bias_rho}}\label{sec:proof_prop1}


The proof of Proposition 3 stands for a generalization of \citet[Proposition 6]{Vono_Paulin_Doucet_2019} which only considered the specific case $\rho_i = \rho^2$ for $i \in [b]$.
This section is divided into two parts, the first gathers lemmas which allow us to upper bound the $\xi^2$-divergence between $\pi_{\bfrho}$ and $\pi$. Then, in the second subsection, we combine these results to control the Wasserstein distance $\wasserstein{}(\pi_{\bfrho},\pi)$ by showing that it is smaller than $\chi^2(\pi_{\bfrho}|\pi)$.
For any $\btheta \in \Rd$ and $\bfrho \in (\rset_+^*)^b$, define
\begin{align}
  \label{eq:def_u_i_rho_i}
    &U_i^{\rho_i}(\B{A}_i\btheta)=-\log\prbigg{\int_{\bz_i\in\R^d}\exp\acn{-U_i(\bz_i)-\|\bz_i-\B{A}_i\btheta\|^2/(2\rho_i)}\,\dd\bz_i/(2\pi\rho_i)^{d_i/2}}\eqsp, \\
  \label{eq:def_Ubar}
  &\overline{B}(\boldsymbol{\theta}) = \sum_{i=1}^{b}\rho_i\|\nabla U_i(\B{A}_i\boldsymbol{\theta})\|^2/2\eqsp,\\
  &\underline{B}(\boldsymbol{\theta})= \sum_{i=1}^{b}\ac{\rho_i\|\nabla U_i(\B{A}_i\boldsymbol{\theta})\|^2/[2(1+\rho_iM_i)] - d_i\log(1+\rho_iM_i)/2}
\end{align}
and consider
\begin{align*}
  &U(\btheta)=\sum_{i\in[b]}U_i(\B{A}_i\btheta)\eqsp,
  &U^{\bfrho}(\btheta)=\sum_{i\in[b]}U_i^{\rho_i}(\B{A}_i\btheta)\eqsp.
\end{align*}
\subsection{Technical lemmata}\label{prop1:technical_lemmata}
We start this subsection by \Cref{lem:UrhoUratiobnd} which allow us to bound the ratio between the integrals defined by $\txts\int_{\R^d}\exp\acn{-\sum_{i\in[b]}U_i^{\rho_i}(\B{A}_i\btheta)}$ and $\txts\int_{\R^d}\exp\acn{-\sum_{i\in[b]}U_i(\B{A}_i\btheta)}\,\dd\btheta$.
\begin{lemma}\label{lem:UrhoUratiobnd}
Assume \Cref{ass:well_defined_density}-\Cref{ass:supp_fort_convex}-\ref{ass:1} and let $\bfrho \in (\rset_+^*)^b$. Then, we have $\underline{B}(\btheta)\le U(\btheta)-U^{\bfrho}(\btheta)$, for any $\btheta \in \Rd$. 
If we assume in addition that for any $i \in [b]$, $U_i$ is convex, we have $U(\btheta)-U^{\bfrho}(\btheta)\le \overline{B}(\btheta)$, for any $\btheta \in \Rd$.
\end{lemma}
\begin{proof}
  The proof follows from the same lines as in \citet[Lemma 14]{Vono_Paulin_Doucet_2019}. In what follows, we give it for the sake of completeness. 
  First, note for any $\btheta\in \R^d$ and $i\in [b]$,
  \begin{equation}\label{eq:Uirhodiff}
    \exp\ac{U_i(\B{A}_i\btheta)-U_i^{\rho_i}(\B{A}_i\btheta)}
    = \int_{\mathbb{R}^{d_i}}\exp\pr{U_i(\B{A}_i\btheta)-U_i(\B{z}_i)-\|\B{z}_i-\B{A}_i\btheta\|^2/(2\rho_i)}\frac{\dd \B{z}_i}{(2\uppi \rho_i)^{d_i/2}}\eqsp.
  \end{equation}
Using \Cref{ass:supp_fort_convex}-\ref{ass:1}, and a second order Taylor expansion, for any $\btheta\in\R^d, i \in [b], \bz_i\in\R^{d_i}$, we have
\begin{align*}
      U_i(\B{A}_i\boldsymbol{\theta}) - U_i(\B{z}_i)
      &\ge \nabla U_i(\B{A}_i\boldsymbol{\theta})^{\top}(\B{A}_i\boldsymbol{\theta}-\B{z}_i)- M_i\|\B{A}_i\boldsymbol{\theta}-\B{z}_i\|^2/2 \eqsp.
  \end{align*}
  Hence, using \eqref{eq:Uirhodiff}, we have for any $\btheta\in \R^d$ and $i\in [b]$,
  \begin{align*}
      \exp\prbigg{\sum_{i=1}^{b} U_i(\B{A}_i\btheta)- U_i^{\rho_i}(\B{A}_i\btheta)}
      &\ge \prod_{i=1}^b\exp\prBig{\dfrac{\rho_i}{2(1+\rho_iM_i)}\norm{\nabla U_i(\B{A}_i\boldsymbol{\theta})}^2}\pr{1+\rho_iM_i}^{-d_i /2} \\
      &= \exp(\underline{B}(\boldsymbol{\theta})) \eqsp. 
  \end{align*}  
  Similarly, under the assumption that for any $i \in [b]$, $U_i$ is convex, the proof for the upper bound follows from the same lines using, for any $i \in [b]$, $\bftheta \in \rset^d$ and $\B{z}_i \in \rset^{d_i}$, that
 \begin{align*}
      U_i(\B{A}_i\boldsymbol{\theta}) - U_i(\B{z}_i) \le \nabla U_i(\B{A}_i\boldsymbol{\theta})^{\top}(\B{A}_i\boldsymbol{\theta}-\B{z}_i)\eqsp.
  \end{align*}
\end{proof}
\begin{lemma}
\label{lem:U_stronglyconvex}
Assume \Cref{ass:well_defined_density}-\Cref{ass:supp_fort_convex}.
Then, $U$ is $m_U$-strongly convex with $m_U = \lambda_{\min}(\sum_{i=1}^b m_i \B{A}_i^{\top} \B{A}_i)$.
\end{lemma}
\begin{proof}
Using by \Cref{ass:supp_fort_convex}-\ref{ass:1} that for any $i \in [b]$, $U_i$ is twice differentiable and by \Cref{ass:supp_fort_convex}-\ref{ass:2} the fact that for any $i \in [b]$, $U_i$ is $m_i$-strongly convex, we have for any $\btheta\in\Rd$ 
\begin{equation*}
  \grad^2 U(\btheta)=\sum_{i=1}^b \B{A}_i^{\top} \grad^2 U_i(\B{A}_i\btheta) \B{A}_i\succeq \sum_{i=1}^b m_i \B{A}_i^{\top} \B{A}_i\succeq \lambda_{\min}\prbigg{\sum_{i=1}^b m_i \B{A}_i^{\top} \B{A}_i} \B{I}_d=m_U \B{I}_d \eqsp.
\end{equation*}
\end{proof}

For any $\btheta \in \Rd$, define
\begin{equation}
  \label{eq:def_beta}
 \beta(\btheta)=\prbigg{\sum_{i=1}^{b}\rho_i\normbigg{\nabla U_i(\B{A}_i\boldsymbol{\theta})}^2}^{\half} \eqsp.
\end{equation}

\begin{lemma}
\label{lem:betaLipschitz}

Assume \Cref{ass:supp_fort_convex}-\ref{ass:1} and let $\bfrho \in (\rset_+^*)^b$.
Then $\beta$ is a Lipschitz function w.r.t. $\norm{\cdot}$, with Lipschitz constant
\begin{equation}\label{eq:Lbeta}
  L_{\beta}= \lambda_{\max}\prbigg{\sum_{i=1}^b\rho_i M_i^2 \B{A}_i^{\top} \B{A}_i}^{\half}\eqsp. 
\end{equation}
\end{lemma}
\begin{proof}
For any $\btheta_1, \btheta_2\in\R^{d}$, we have using
$
| (\sum_{i=1}^b a_i^2)^{\half} - (\sum_{i=1}^b b_i^2)^{\half} |
\le (\sum_{i=1}^b( a_i-b_i)^2)^{\half},
$
that
\begin{align*}
|\beta(\btheta_1) - \beta(\btheta_2)|
\le \prbigg{\sum_{i=1}^b \rho_i \| \nabla U_i(\B{A}_i \btheta_1) - \nabla U_i (\B{A}_i\btheta_2) \|^2}^{\half}
\le \prbigg{\sum_{i=1}^b \rho_i M_i^2 \| \B{A}_i(\btheta_1-\btheta_2) \|^2}^{\half},
\end{align*}
which completes the proof.
\end{proof}
Suppose \Cref{ass:supp_fort_convex}-\ref{ass:2} and for any $i\in [b]$, denote $\btheta_i^\star$ a minimiser of $\btheta \mapsto U_i(\B{A}_i\btheta)$.
\begin{lemma}
\label{lem:momgenbeta2}
Assume \Cref{ass:well_defined_density}-\Cref{ass:supp_fort_convex} and let  $\bfrho \in (\rset_+^*)^b$.
Then for any $s< m_U/(12L_{\beta}^2)$, where $L_{\beta}$ is defined in \eqref{eq:Lbeta}, we have
\begin{equation}
  \label{eq:lem:momgenbeta2}
\log \pi\brbig{\mathrm{e}^{s\{\beta^{2}-\pi[\beta^2]\}}}
\le
 8 s^{2} L_\beta^{4}/m_U^{2} + 4 s^{2}\{\pi[\beta]\}^{2} L_{\beta}^{2}/m_U \eqsp. 
\end{equation}
In addition,
\begin{equation}
  \label{eq:bound_epe_beta_square}
\pi(\beta^2) \le 2d L_\beta^2 /m_U+ 2 \sum_{i=1}^{b} \rho_i M_i^2\|\B{A}_{i}(\bthetastar-\btheta^\star_i)\|^2   \eqsp. 
\end{equation}
\end{lemma}
\begin{proof}
Using the decomposition
\[
\beta^{2}(\boldsymbol{\theta})-\defEnsLigne{\pi[\beta]}^2=(\beta(\boldsymbol{\theta})-\pi[\beta])^{2}+2 \pi[\beta](\beta(\boldsymbol{\theta})-\pi[\beta])
\]
and the Cauchy-Schwarz inequality imply, for any $s>0$,
\begin{equation}\label{eq:bound:expect_expo}
\pi\brbig{\mathrm{e}^{s\{\beta^{2}-\defEnsLigne{\pi[\beta]}^2\}}}
\le \acbig{\pi[\mathrm{e}^{2 s\{\beta-\pi[\beta]\}^{2}}]}^{\half}\cdot\acbig{{\pi}[\mathrm{e}^{4 s \pi[\beta] \{\beta-\pi[\beta]\}}]}^{\half}\eqsp.
\end{equation}
The proof consists in bounding the two terms in the right-hand sided. 
Since $\beta:\Rd\to\R$ is $L_\beta$-Lipschitz by \Cref{lem:betaLipschitz}, for any $0 \le s \le m_U/(12 L_{\beta}^2)$, using \citet[Lemma 16]{Vono_Paulin_Doucet_2019} and \Cref{lem:U_stronglyconvex} gives setting $\bar{\beta} = \beta-\pi[\beta]$, that
\begin{equation}
  \label{eq:1:proof:lem:momgenbeta2}
\pi\brbig{\exp\defEnsLigne{2 s\parentheseLigne{\bar{\beta}^{2}-\pi[\bar{\beta}^2]}}}
\le \exp\parentheseLigne{16 s^{2} L_\beta^{4}/m_U^{2}} \eqsp.
\end{equation}
In addition, using \citet[Proposition 5.4.1]{bakry2013analysis}, \Cref{lem:betaLipschitz} and \Cref{lem:U_stronglyconvex}, we get for any $s \ge 0$,
\begin{equation*}
  \label{eq:3}
\pi\brbig{\mathrm{e}^{4 s \pi[\beta](\beta-\pi[\beta])}} \le \mathrm{e}^{8 s^{2}\{\pi[\beta]\}^{2} L_{\beta}^{2}/m_U}\eqsp.
\end{equation*}
Plugging this result and \eqref{eq:1:proof:lem:momgenbeta2} into \eqref{eq:bound:expect_expo}, we get 
\[
\pi\brbig{\mathrm{e}^{s\{\beta^{2}-\defEnsLigne{\pi[\beta]}^2\}}}
\le
\exp\parentheseLigne{s \pi(\bar{\beta}^2)+ 8 s^{2} L_\beta^{4}/m_U^{2} + 4 s^{2}\{\pi[\beta]\}^{2} L_{\beta}^{2}/m_U} \eqsp.
\]
The proof of \eqref{eq:lem:momgenbeta2} follows using $\pi(\bar{\beta}^2) = \pi(\beta^2) - [\pi(\beta)]^2$ and rearranging terms. 
 
Using the Young inequality, \Cref{ass:supp_fort_convex}-\ref{ass:1},$\nabla U_i (\B{A}_i \thetabf_i^{\star}) = 0$, $\nabla U(\bftheta^{\star}) = 0$, we have
\begin{align*}
\pi(\beta^2) & = \int_{\Rd} \prbigg{\sum_{i=1}^{b} \rho_i \|\grad U_{i}(\B{A}_{i}\btheta)\|^2} \pi(\btheta)\,\dd \btheta \\
    & \le 2 \int_{\Rd}\prbigg{\sum_{i=1}^{b} \rho_i M_i^2\|\B{A}_{i}(\btheta-\bthetastar)\|^2} \pi(\btheta)\,\dd \btheta + 2 \sum_{i=1}^{b} \rho_i M_i^2\|\B{A}_{i}(\bthetastar-\btheta^\star_i)\|^2\\
    & \le 2\lambda_{\max}\prbigg{\sum_{i=1}^b\rho_i M_i^2 \B{A}_i^{\top} \B{A}_i} \int_{\Rd} \|\btheta-\bthetastar\|^2 \pi(\btheta) \,\dd \btheta+ 2 \sum_{i=1}^{b} \rho_i M_i^2\|\B{A}_{i}(\bthetastar-\btheta^\star_i)\|^2\\
    & \le 2d L_\beta^2 /m_U+ 2 \sum_{i=1}^{b} \rho_i M_i^2\|\B{A}_{i}(\bthetastar-\btheta^\star_i)\|^2 \eqsp,
\end{align*}
where we have used $\pi[\|\btheta-\bthetastar\|^2]\le d/m_U$ by \citet[Proposition 1 (ii)]{durmus2018high} and \Cref{lem:U_stronglyconvex}.
\end{proof}
\Cref{prop:pi_rho_proper} shows that $\pi_{\bfrho}(\cdot) = \int_{\rset^p} \Pi_{\bfrho}(\cdot,\bfz) \dd \bfz$ is well-defined and as such admits a finite normalising constant.
These two quantities are defined by
\begin{align}\label{eq:def_Z_pi_rho}
  &\rmZ_{\pi_{\bfrho}} = \int_{\R^d}\exp\acbigg{-\sum_{i\in[b]}U_i^{\rho_i}(\B{A}_i\btheta)}\,\dd\btheta\eqsp,
  &\pi_{\bfrho}(\cdot) = \exp\acbigg{-\sum_{i\in[b]}U_i^{\rho_i}(\B{A}_i\cdot)}/ \rmZ_{\pi_{\bfrho}} \eqsp. 
\end{align}
Finally, note that the following quantity $\rmZ_{\pi}$ is a normalising constant of $\pi$ associated with the potential $U$, \ie~$\pi = \rme^{-U}/\rmZ_{\pi}$,
\begin{equation}
  \label{eq:Z_pi}
\rmZ_{\pi} = \int_{\R^d} \exp\acbigg{-\sum_{i\in[b]}U_i(\B{A}_i\btheta)}\,\dd\btheta\eqsp.
\end{equation}
\begin{lemma}
\label{lem:ratioofnormalizingconstants}
Assume \Cref{ass:well_defined_density}-\Cref{ass:supp_fort_convex} and let $\bfrho \in (\rset_+^*)^b$. Suppose in addition that   $6 L_\beta^2\le m_U$ where $L_{\beta}$ is given in \eqref{eq:Lbeta}.
Then, we have
\[
 \log\pr{\mathrm{Z}_{\pi_{\bfrho}}/\mathrm{Z}_{\pi} }\le \acbigg{d L_\beta^2 /m_U+  \sum_{i=1}^{b} \rho_i M_i^2\|\B{A}_{i}(\bthetastar-\btheta^\star_i)\|^2}(1+2 L_{\beta}^{2}/m_U  ) + 2  L_\beta^{4}/m_U^{2}    \eqsp.
\]
\end{lemma}
\begin{proof}
From the definitions \eqref{eq:def_Z_pi_rho} and \eqref{eq:Z_pi}, we have
$
\mathrm{Z}_{\pi_{\bfrho}}/\mathrm{Z}_{\pi} = \int_{\mathbb{R}^d} \pi(\boldsymbol{\theta}) \exp\{\sum_{i=1}^{b} U_i(\B{A}_i\btheta)- U_i^{\rho_i}(\B{A}_i\btheta)\}\,\dd\boldsymbol{\theta}.
$
By \Cref{lem:UrhoUratiobnd}, we obtain
\begin{equation*}
    \mathrm{Z}_{\pi_{\bfrho}}/\mathrm{Z}_{\pi}
    \le 
    \int_{\mathbb{R}^d} 
    \pi(\btheta)\exp(\overline{B}(\btheta))
   \,\dd\boldsymbol{\theta} \eqsp.
\end{equation*}
Note that  $\overline{B} = \beta^2/2$ by \eqref{eq:def_Ubar}-\eqref{eq:def_beta},
hence using that $6 L_\beta^2\le m_U$, Lemma~\ref{lem:momgenbeta2} applied with $s=1/2$ shows that
\begin{equation*}
    \log \left(\int_{\mathbb{R}^d} \pi(\btheta)\exp(\overline{B}(\btheta))\,\dd\boldsymbol{\theta}\right)
    \le \pi[\beta^2]/2 + 2  L_\beta^{4}/m_U^{2} +  \{\pi[\beta]\}^{2} L_{\beta}^{2}/m_U  \eqsp.
\end{equation*}
Using \Cref{lem:momgenbeta2}-\eqref{eq:bound_epe_beta_square} and $\pi[\beta] \le \pi[\beta^2]$ concludes the proof.
\end{proof}
\subsection{Proof of \Cref{prop:bound_bias_rho}}
Based on the technical lemmas derived in \Cref{prop1:technical_lemmata}, we are now ready to bound the Wasserstein distance of order 2 between $\pi$ and $\pi_{\bfrho}$.
\begin{proof}[Proof of Proposition \Cref{prop:bound_bias_rho}]
  Let $\bfrho \in (\rset_+^*)^b$ such that $\max_{i \in [b]} \rho_i = \bar{\rho} \le \sigma_U^2/12$, where $\sigma_U^2=\|\B{A}^{\top}\B{A}\|\max_{i \in [b]}\{M_i^2\}/m_U$. Then, by definition of $L_{\beta}$ \eqref{eq:Lbeta}, we get 
  \begin{equation}
    \label{eq:proof_bias_pi_rho_condition_L_beta}
    12 L_{\beta}^2 \le m_U \eqsp.
  \end{equation}
  and \Cref{lem:momgenbeta2} can be applied for $s=1$ and \Cref{lem:ratioofnormalizingconstants} too.
By \Cref{lem:U_stronglyconvex}, $U=-\log \pi$ is $m_U$-strongly convex therefore $\pi$ satisfies a log-Sobolev inequality with constant $m_U$ \citep[Theorem 5.2]{Ledouxconcentrationofmeasure}.
Finally, \citet[Theorem 1]{Otto_Villani_2000} shows that $\pi$ satisfies for any $\nu \in \mathcal{P}_2(\Rd)$:
\begin{equation}
  \label{eq:proof_bias_pi_rho_bound_talagrand}
\wasserstein{}(\nu,\pi) \le \sqrt{(2/m_U)\mathrm{KL}(\nu|\pi)} \le \sqrt{(2/m_U)\chi^2(\pi_{\bfrho}|\pi)} \eqsp,
\end{equation}
where $\chi^2$ is the chi-square divergence and where we have used for the last inequality that $\mathrm{KL}(\pi_{\bfrho}|\pi) \le \chi^2(\pi_{\bfrho}|\pi)$ since 
for any $t >0$, $\log(t) \le t-1$. We now bound $\chi^2(\pi_{\bfrho}|\pi)$.
By \eqref{eq:def_Z_pi_rho} and \eqref{eq:Z_pi}, for any $\btheta \in \Rd$, consider the decomposition given by
\begin{align} \pi_{\bfrho}(\boldsymbol{\theta})/\pi(\boldsymbol{\theta}) - 1 &=
    (\Zrm_{\pi}/\Zrm_{\pi_{\bfrho}}) \exp\prbigg{\sum_{i=1}^{b} 
    \prbig{U_i(\B{A}_i\btheta)- U_i^{\rho_i}(\B{A}_i\btheta)}} - 1 \eqsp.
    \label{eq:ratio_pi_def}
\end{align}
In the sequel, we will both lower and upper bound \eqref{eq:ratio_pi_def} in order to upper bound $|1 - \pi_{\bfrho}(\boldsymbol{\theta})/\pi(\boldsymbol{\theta})|$.
Using the fact that for all $x \in \mathbb{R}, \exp(x) - 1 \ge x$, Lemmas \ref{lem:UrhoUratiobnd} and \ref{lem:ratioofnormalizingconstants} yield
\begin{align}\label{eq:ratio_lower_bound}
    &\pi_{\bfrho}(\boldsymbol{\theta})/\pi(\boldsymbol{\theta}) - 1 
    \ge \log\pr{\Zrm_{\pi}/\Zrm_{\pi_{\bfrho}}} + \sum_{i=1}^{b} \pr{U_i(\B{A}_i\btheta)- U_i^{\rho_i}(\B{A}_i\btheta)}\\
	\nonumber
    &\ge -\acbigg{d L_\beta^2 /m_U+  \sum_{i=1}^{b} \rho_i M_i^2\|\B{A}_{i}(\bthetastar-\btheta^\star_i)\|^2}(1+2 L_{\beta}^{2}/m_U  ) - 2  L_\beta^{4}/m_U^{2}  + \underline{B}(\btheta) \ge -A_1\eqsp,
\end{align}
where
\begin{multline*}
    A_1 = \acbigg{d L_\beta^2 /m_U+  \sum_{i=1}^{b} \rho_i M_i^2\|\B{A}_{i}(\bthetastar-\btheta^\star_i)\|^2}(1+2 L_{\beta}^{2}/m_U) \\
    + 2 L_\beta^{4}/m_U^{2}
    + \sum_{i=1}^{b}(d_i/2) \log(1+\rho_iM_i)\eqsp,
\end{multline*}
where we have used in the last inequality that $\underline{B}(\btheta) \ge -\sum_{i=1}^{b}(d_1/2) \log(1+\rho_iM_i)   $ by \eqref{eq:def_Ubar}.
In addition,
by \eqref{eq:def_Z_pi_rho} and  \eqref{eq:Z_pi}
$
\mathrm{Z}_{\pi_{\bfrho}}/\mathrm{Z}_{\pi} = \int_{\mathbb{R}^d} \pi(\boldsymbol{\theta}) \exp\{\sum_{i=1}^{b} U_i(\B{A}_i\btheta)- U_i^{\rho_i}(\B{A}_i\btheta)\}\,\dd\boldsymbol{\theta}.
$
, which implies by \Cref{lem:UrhoUratiobnd} and Jensen inequality
\begin{equation*}
    \mathrm{Z}_{\pi_{\bfrho}}/\mathrm{Z}_{\pi}
    \ge 
    \int_{\mathbb{R}^d} 
    \pi(\btheta)\exp(\underline{B}(\btheta))
   \,\dd\boldsymbol{\theta}  \ge \exp(\pi[\underline{B}]) \eqsp.
\end{equation*}
It follows by \eqref{eq:ratio_pi_def} that $\pi_{\bfrho}(\boldsymbol{\theta})/\pi(\boldsymbol{\theta}) - 1 \le
    \exp\parentheseLigne{\overline{B}(\btheta) -\pi\pr{\underline{B}}} - 1$.
Combining this result and  \eqref{eq:ratio_lower_bound}, it follows that the Pearson $\chi^2$-divergence between $\pi$ and $\pi_{\bfrho}$ can be upper bounded as  where
\begin{equation}
    \label{eq:proof_bias_pi_rho_bound_chi_2_1}
   \chi^2(\pi_{\bfrho}|\pi) \le \max(A_1^2,A_2) \eqsp, \qquad  A_2 = \int_{\Rd}\prbig{\exp\prn{\overline{B}(\btheta) -\pi\pr{\underline{B}}} - 1}^2\pi(\btheta)\,\dd \btheta \eqsp. \nonumber
\end{equation} 
We now provide an explicit bound for $A_2$. First by Jensen inequality, we have $\pi(\exp(\overline{B}))\ge\exp(\pi(\overline{B}))$ which implies that $\exp\prbig{-\pi\prn{\underline{B}}}\pi\brbig{\exp\prn{\overline{B}}}\ge\prod_{i=1}^b (1+\rho_iM_i)^{d_i/2}$ by \eqref{eq:def_Ubar}. Therefore, using  that $\overline{B} = \beta^2/2$ by \eqref{eq:def_Ubar}-\eqref{eq:def_beta} and \Cref{lem:momgenbeta2} with $s=1$ since \eqref{eq:proof_bias_pi_rho_condition_L_beta} holds, we get by \eqref{eq:def_Ubar},
\begin{align}
    &A_2 = \int_{\Rd}\prbig{\exp\prbig{\overline{B}(\btheta) -\pi\pr{\underline{B}}} - 1}^2\pi(\btheta)\,\dd \btheta \nonumber \\
    &= \exp\pr{-2\pi\pr{\underline{B}}}\pi\brbig{\exp\prbig{2\overline{B}}} \nonumber
    - 2\exp\prbig{-\pi\prbig{\underline{B}}}\pi\brbig{\exp\prbig{\overline{B}}}+ 1 \nonumber\\
    &\le  \prod_{i=1}^b (1+\rho_i M_i)^{d_i} \cdot \exp\parentheseLigne{-\pi\defEnsLigne{\sum_{i=1}^{b}(\rho_i/(1+\rho_i M_i))\normLigne{\nabla U_i(\B{A}_i\cdot)}^2}}\pi\brbig{\exp\prn{\beta^2}} \nonumber \\
    &\qquad - 2\prod_{i=1}^b (1+\rho_i M_i)^{d_i/2} + 1 \nonumber\\
    &\le \prod_{i=1}^b (1+\rho_i M_i)^{d_i} \cdot \exp\parentheseLigne{\pi\defEnsLigne{\sum_{i=1}^{b}(\rho_i^2M_i/(1+\rho_i M_i))\normLigne{\nabla U_i(\B{A}_i\cdot)}^2}} \nonumber \\
    &\times\exp\prbigg{8 L_\beta^{4}/m_U^{2} +4 \{2d L_\beta^2 /m_U+ 2 \sum_{i=1}^{b} \rho_i M_i^2\|\B{A}_{i}(\bthetastar-\btheta^\star_i)\|^2\} L_{\beta}^{2}/m_U } \\
    &\qquad- 2\prod_{i=1}^b (1+\rho_i M_i)^{d_i/2} + 1 \eqsp, \label{eq:proof_bias_pi_rho_A_2_1}
\end{align}
where we have used for the last inequality that for $\bftheta \in \rset^d$, $\beta(\bftheta)^2 -\sum_{i=1}^{b}(\rho_i/(1+\rho_i M_i))\normLigne{\nabla U_i(\B{A}_i\bftheta)}^2 =\sum_{i=1}^{b}(\rho_i^2M_i/(1+\rho_i M_i))\normLigne{\nabla U_i(\B{A}_i\bftheta)}^2$, $\pi[\beta]^2 \le \pi[\beta^2]$ by the Cauchy-Schwartz inequality and  \Cref{lem:momgenbeta2}-\eqref{eq:bound_epe_beta_square}.
Similarly to the proof of \Cref{lem:momgenbeta2}-\eqref{eq:bound_epe_beta_square}, by \Cref{ass:supp_fort_convex}-\ref{ass:1}, $\nabla U_i (\B{A}_i \thetabf_i^{\star}) = 0$, $\nabla U(\bftheta^{\star}) = 0$, \citet[Proposition 1 (ii)]{durmus2018high} and \Cref{lem:U_stronglyconvex}, we have
\begin{align*}
  \pi\br{\sum_{i=1}^{b}(\rho_i^2M_i/(1+\rho_i M_i))\normLigne{\nabla U_i(\B{A}_i\cdot)}^2}
  &\le\pi\br{\sum_{i=1}^{b}\rho_i^2M_i\normLigne{\nabla U_i(\B{A}_i\cdot)}^2}\\
  \le 2d\lambda_{\max}\prBig{\sum_{i=1}^b\rho_i^2 M_i^3 \B{A}_i^{\top} \B{A}_i} & / m_U+ 2 \sum_{i=1}^{b} \rho_i^2 M_i^3\|\B{A}_{i}(\bthetastar-\btheta^\star_i)\|^2\eqsp. 
\end{align*}
Therefore, we get by \eqref{eq:proof_bias_pi_rho_A_2_1}
\begin{align}
& A_2\le A_3=\prod_{i=1}^b (1+\rho_i M_i)^{d_i} \exp\prbigg{2d  \lambda_{\max}\prBig{\sum_{i=1}^b\rho_i^2 M_i^3 \B{A}_i^{\top} \B{A}_i} /m_U + 2 \sum_{i=1}^{b} \rho_i^2 M_i^3\|\B{A}_{i}(\bthetastar-\btheta^\star_i)\|^2
} \nonumber \\
  &\exp\prbigg{8 L_\beta^{4}/m_U^{2} +8 \brBig{d L_\beta^2 /m_U+ \sum_{i=1}^{b} \rho_i M_i^2\|\B{A}_{i}(\bthetastar-\btheta^\star_i)\|^2} L_{\beta}^{2}/m_U} - 2\prod_{i=1}^b (1+\rho_i M_i)^{d_i/2} + 1 \eqsp. \label{eq:proof_bias_pi_rho_A_2_2}
\end{align}
It follows by \eqref{eq:proof_bias_pi_rho_bound_chi_2_1} and  \eqref{eq:proof_bias_pi_rho_bound_talagrand}  that 
\begin{align}
    \label{eq:bound_bias}
    \wasserstein{}(\pi_{\bfrho},\pi) &\le \sqrt{(2/m_U)\max(A_1^2,A_3)}\eqsp,
\end{align}
where $A_1$ and $A_3$ are given by \eqref{eq:ratio_lower_bound} and \eqref{eq:proof_bias_pi_rho_A_2_2} respectively. Using that $L_{\beta}^2 = \bigO(\bar{\rho})$ and an expansion of the bound as $\bar{\rho} \to 0$ completes the proof. 
\end{proof}

\section{Proof of \Cref{prop:bias_gamma} and \Cref{prop:bias_gamma_bis}}\label{sec:prop3_proof}


As in \Cref{sec:proof-proposition-2-1}, we assume in all this section that $\bfrho \in (\rset_+^*)^b$ is fixed.
For any $\bfgamma = (\gamma_1,\ldots,\gamma_b)\in(\R_+^*)^b$, we establish in this section explicit bounds on $\wasserstein{}(\pi_{\bfrho,\bfgamma,\bfN}, \pi_{\bfrho})$ where $\pi_{\bfrho}$ is given in \eqref{eq:target_density} and $\pi_{\bfrho,\bfgamma,\bfN}$ is the marginal distribution defined by 
\begin{equation*}
  \pi_{\bfrho,\bfgamma,\bfN}(\msa) = \Pi_{\bfrho,\bfgamma,\bfN}(\msa\times \rset^p)\eqsp, \qquad \msa \in \mcbb(\rset^d)\eqsp,
\end{equation*}
of the stationary probability measure $\Pi_{\bfrho,\bfgamma,\bfN}$ associated with the Markov chain  $(Z_n,\theta_{n})_{n\ge 0}$ defined in \Cref{algo:ULAwSG}. Note that in the case $\bfN = N(1,\ldots,1)$, this distribution is independent of $N$, see \Cref{cor:convergence_rho_gamma}. 
To this purpose, we define an ``ideal'' dynamics from which we cannot sample but which converges geometrically towards $\Pi_{\bfrho}$ under appropriate conditions.
The corresponding ideal process will play the same role as the Langevin dynamics for the study of the unadjusted Langevin algorithm \citep{durmus2018high}.
This dynamics is defined as follows.
Consider first for any $\btheta\in \Rd$, $i\in[b]$, the stochastic differential equation (SDE) defined by
\begin{equation}\label{eq:def:cont_process}
\dd \Yc_t^{i, \btheta} = - \nabla \tildeU_{i}(\Yc^{i, \btheta}_t)\,\dd t - \rho_{i}^{-1} \B{A}_{i} \btheta + \sqrt{2}\,\dd B_t^{i}\eqsp,
\end{equation}
where $(B_t^i)_{t\ge 0}$ is a $d_{i}$-dimensional Brownian motion and $\tildeU_{i}$ is defined in \eqref{eq:def:cont_Utilde}.
Note that under \Cref{ass:supp_fort_convex}-\ref{ass:1}, this SDE admits a unique strong solution \citep[Theorem (2.1) in Chapter IX]{revuz2013continuous}. 
Denote for any $i\in [b]$, the Markov semi-group associated to \eqref{eq:def:cont_process} by $(\tR_{\rho_{i},t}^{i})_{t\ge 0}$ defined for any $\tilde{\by}_0^i\in\R^{d_{i}}$, $t\ge 0$ and $\msb_{i}\in\mcb{\R^{d_{i}}}$ by
\begin{equation*}
  \label{eq:semigroup_langevin_{i}}
  \tR_{\rho_{i},t}^{i}(\tilde{\by}_0^i,\msb_{i}|\btheta)=\PP(\Yc^{i, \btheta,\tilde{\by}_0^i}_t\in \msb_{i})\eqsp,
\end{equation*}
where $(\Yc_t^{i, \btheta,\tilde{\by}_0^i})_{t\ge 0}$ is a solution of \eqref{eq:def:cont_process} with $\Yc_0^{i,\btheta,\tilde{\by}_0^i}=\tilde{\by}_0^i$. For any bounded measurable function $f_{i}:\R^{d_{i}}\to \R_+$, \Cref{lem:cont_measurability} shows the measurability of the function $(\btheta,\tilde{\by}_0^i)\mapsto \PE [f_{i}(\Yc_t^{i,\btheta,\tilde{\by}_0^i})]$ on $\Rd\times\R^{d_{i}}$ and therefore $\tR_{\rho_{i},t}^i$ is a conditional Markov kernel.
\begin{lemma}\label{lem:cont_measurability}
For any bounded measurable function $f_{i}:\R^{d_{i}}\to \R_+$ and function $f_{i}$ satisfying \Cref{ass:supp_fort_convex}-\ref{ass:1}, the mapping $ (\tilde{\btheta}_{0},\tilde{\by}_{0}^{i})\mapsto \PE [f_{i}(\Yc_t^{i,\tilde{\btheta}_{0},\tilde{\by}_0^i})]$ is  Borel measurable.
\end{lemma}
\begin{proof}
Consider the following stochastic differential equation 
\begin{equation*}
\begin{cases}
\dd \tilde{\theta}_{t}=\mathbf{0}_{d}\eqsp, \\
\dd \tilde{Y}_t^{i} = - \nabla \tildeU_{i}(\tilde{Y}^{i}_t)\,\dd t - \rho_{i}^{-1} \B{A}_{i}\tilde{\theta}_t + \sqrt{2}\,\dd B_t^{i}\eqsp.
\end{cases}
\end{equation*}
Using \citet[Theorem (2.4) in Chapter IX]{revuz2013continuous}, since $U_{i}$ satisfies \Cref{ass:supp_fort_convex}-\ref{ass:1}, there exists a unique  solution $(\tilde{X}_{t}^{\tilde{\bx}})_{t\ge 0}=(\tilde{\theta}_{t},\tilde{Y}_t^{i})_{t\ge 0}$ with initial condition $\tilde{\bx}=(\tilde{\btheta}_0^{\top},(\tilde{\by}_{0}^{i})^{\top})^{\top}\in\R^p$.
Then, the proof follows from \citet[Theorem (1.9) in Chapter IX]{revuz2013continuous} and the fact that $\tilde{Y}_t^{i}$ is the unique solution of \eqref{eq:def:cont_process} with $\btheta = \btheta_0$.
\end{proof}

Define
for any $\btheta \in \Rd$, $\bz = (\bz_{1}^{\top},\cdots,\bz_{b}^{\top})^{\top} \in \mathbb{R}^{p}$, and for $i \in [b]$, $\mathsf{B}_{i} \in \mathcal{B}(\mathbb{R}^{d_{i}})$, 
\begin{align*}
  \tQ_{\bfrho, \bfgamma}\pr{\bz,\mathsf{B}_1\times\cdots\times\mathsf{B}_b|\btheta} = \prod_{i=1}^b \tR_{\rho_{i}, N_{i}\gamma_{i}}^{i}(\bz_{i},\mathsf{B}_{i}|\btheta)\eqsp,
\end{align*}
and consider the Markov kernel defined, for any $\bx^{\top} = (\btheta^{\top},\bz^{\top})$ and $\msa \in \mcbb(\rset^d)$, $\msb \in \mathcal{B}( \mathbb{R}^p)$, by
\begin{equation}\label{eq:def:P_tilde}
  \tP_{\bfrho, \bfgamma}(\bx,\msa \times \msb) = \int_{\msb} \tQ_{\bfrho,\bfgamma}\pr{\bz,\dd {\tbfz}|\thetabf}\int_{\msa}\Pi_{\bfrho}(\dd \tbtheta|\tbfz)\,\eqsp,
\end{equation}
where $\Pi_{\bfrho}(\cdot | \tbfz)$ is defined in \eqref{eq:def:Pi_rho_cond}. 
Note that $P_{\bfrho, \bfgamma, \bfN}$ can be interpreted as a discretised version of $\tP_{\bfrho, \bfgamma}$ using the Euler-Maruyama scheme.

In the sequel, we first derive technical lemmata in \Cref{subsec:tech_lemmata_cont} that are used to prove both \Cref{prop:bias_gamma} and \Cref{prop:bias_gamma_bis}.
Based on these lemmata, we then prove each proposition in a dedicated section, namely \Cref{subsec:proof_cont1} and \Cref{subsec:proof_cont2}.

\subsection{Synchronous coupling and a first estimate}
\label{subsec:tech_lemmata_cont}

The main idea to prove \Cref{prop:bias_gamma} and \Cref{prop:bias_gamma_bis} is to define $(X_n,\tilde{X}_n)_{n \in \N}$ such that for any $n \in \N$, $(X_n,\tilde{X}_n)$ is a coupling between $\updelta_{\B{x}}P^n_{\bfrho,\bfgamma,\bfN}$ defined in \eqref{eq:P_rho_gamma_N} and $\updelta_{\B{\tilde{x}}}\tP^n_{\bfrho,\bfgamma}$, and satisfies
$$
\mathbb{E}\br{\norm{X_n-\tilde{X}_n}^2} \le c_1(\B{x},\tilde{\B{x}}) \mathrm{e}^{-c_2\min_{i \in [b]}\{\gamma_i m_i\}} + c_3 \gamma^{\alpha}\eqsp,
$$
where $c_2,c_3 > 0$ and $\alpha \in \{1,2\}$ depending if \Cref{ass:hessian_lipschitz} holds or not.
Conditioning with respect to $(X_0,\tilde{X}_0)$ with distribution $\updelta_{\bx} \otimes \Pi_{\bfrho}$, using the definition of the Wasserstein distance of order 2 and taking $n \rightarrow \infty$, we obtain
$$
W_2(\pi_{\bfrho},\pi_{\bfrho,\bfgamma,\bfN}) \le W_2(\Pi_{\bfrho},\Pi_{\bfrho,\bfgamma,\bfN}) \le \tilde{c}_3 \gamma^{\alpha} \eqsp,
$$
where $\tilde{c}_3 > 0$.
We now provide the rigourous construction of $(X_n,\tilde{X}_n)_{n \in \N}$.

Let $\{(B^{(i,n)}_t)_{t\ge 0}:i\in[b], n\in \N\}$ be independent random variables such that for any $i\in[b]$, the sequences $\{(B^{(i,n)}_{t})_{t\ge 0}:n\in\N\}$ are i.i.d. $d_{i}$-dimensional Brownian motions and let $(\xi_n)_{n\ge 0}$ be a sequence of i.i.d. standard $d$-dimensional Gaussian random variables independent of $\{(B_t^{(i, n)})_{t\ge 0}:i\in[b], n\in \N\}$.
Consider the stochastic process $(\tilde{X}_n)_{n\ge 0}$ on $\R^d \times \R^p$ starting from $\tilde{X}_{0}$ distributed according to $\Pi_{\bfrho}$ and defined by the recursion: for $n\in\N$, $i\in [b]$,
\begin{equation}\label{eq:def:cont_variables}
\tilde{X}_{n+1}=(\tilde{\theta}_{n+1}^{\top}, \Zc_{n+1}^{\top})^{\top}\eqsp,
\quad \Zc_{n+1}^i = \Yc^{(i, n)}_{N_{i} \gamma_{i}}\eqsp,
\quad \tilde{\theta}_{n+1} = \bar{\B{B}}_0^{-1}\B{B}_0^{\top}\B{\tilde{D}}_{0}^{\half} \Zc_{n+1} + \bar{\B{B}}_0^{-\half} \xi_{n+1}\eqsp,
\end{equation}
where $(\Yc_t^{(i, n)})_{t\ge 0}$, is a solution of \eqref{eq:def:cont_process} starting from $\Zc_n^i$ with parameter $\btheta\leftarrow \btheta_n$.
Similarly to the process $(X_n)_{n\in\nset}$ defined in \Cref{algo:ULAwSG}, the process $(\tX_n)_{n \in\nset}$ defines a homogeneous Markov chain.  Indeed, it is easy to show that for any $n \in \nset$ and measurable function $f : \rset^{p} \to \rset_+$, $\PE[f(\tZ_{n+1})|\tX_n] = \int_{\rset^p} f(\tilde{\bz})\tQ_{\bfrho,\gammabf}(\tZ_{n},\dd \zbf|\tilde{\theta}_n)$ and therefore $(\tX_n)_{n \in\nset}$ is associated with \eqref{eq:def:P_tilde}.
\begin{proposition}\label{prop:cont_X_law}
Assume \Cref{ass:well_defined_density}-\Cref{ass:supp_fort_convex}-\ref{ass:1}, and let $\bfN\in(\N^*)^{b}, \gammabf \in (\rset_+^*)^b$.
Then, the Markov kernel $\tP_{\bfrho, \bfgamma}$ defined in \eqref{eq:def:P_tilde} admits $\Pi_{\bfrho}$ as an invariant probability measure. 
\end{proposition}
\begin{proof}
By property of the Langevin diffusion defined in \eqref{eq:def:cont_process}, for all $\btheta_0\in\R^d$, the Markov kernel $\tQ_{\bfrho,\bfgamma}(\cdot|\btheta_0)$ admits $\Pi_{\bfrho}(\cdot|\btheta_0)$ as invariant measure, see \eg~\cite{Roberts1996} or \cite{kent:1978}. Thus, for any $\btheta_0\in\R^d$ and $\mathsf{B} \in \mathcal{B}(\R^{p})$, we have
\begin{equation}\label{eq:eq:int_pi_rho}
\int_{\mathsf{B}}\Pi_{\bfrho}(\bz_{1}|\btheta_0)\,\dd\bz_{1}
= \int_{\bz_0\in\R^p}\tQ_{\bfrho,\bfgamma}(\bz_0,\mathsf{B}|\btheta_0)\Pi_{\bfrho}(\bz_0|\btheta_0)\,\dd\bz_0\eqsp.
\end{equation}
Denote by $\pi_{\bfrho}^{\btheta},\pi_{\bfrho}^{\bz}$ the marginals under $\Pi_{\bfrho}$: $\pi_{\bfrho}^{\btheta}(\msa) = \Pi_{\bfrho}(\msa \times \rset^p)$, $\pi_{\bfrho}^{\bz}(\msb) = \Pi_{\bfrho}(\rset^d \times \msb)$, for $\msa \in \mcbb(\rset^d)$ and $\msb \in \mcbb(\rset^p)$, and
consider the Markov chain $(\tX_n)_{n\in\N}$ defined in \eqref{eq:def:cont_variables}. For any measurable function $f:\R^{d+p}\to\R_+$, the Fubini-Tonelli theorem gives 
\begin{align}
  \nonumber
\E[f(\tX_1)]
&=\int_{\R^{d+p}}\int_{\R^{d+p}}f(\bx_{1})\Pi_{\bfrho}(\btheta_1|\bz_1)\,\dd\btheta_1 \tQ_{\bfrho,\bfgamma}(\bz_0,\dd \bz_1|\btheta_0)\Pi_{\bfrho}(\btheta_0,\bz_0)\,\dd\btheta_0\,\dd\bz_0\\
  \nonumber
&=\int_{\R^{d}}\int_{\R^{p}}f(\bx_{1}) \Pi_{\bfrho}(\btheta_1|\bz_1)\int_{\R^{d}}\brbigg{\int_{\R^{p}}
\tQ_{\bfrho,\bfgamma}(\bz_0,\dd \bz_1|\btheta_0)\Pi_{\bfrho}(\bz_0|\btheta_0)\,\dd\bz_0}\pi_{\bfrho}^{\btheta}(\btheta_0)\,\dd\btheta_0\,\dd\btheta_1\\
\label{eq:eq:int_pi_rho_2}
&=\int_{\R^{d}}\int_{\R^{p}}f(\bx_{1}) \Pi_{\bfrho}(\btheta_1|\bz_1)\brbigg{\int_{\btheta_0\in\R^{d}}\Pi_{\bfrho}(\bz_{1}|\btheta_0)\pi_{\bfrho}^{\btheta}(\btheta_0)\,\dd\btheta_0}\,\dd\bz_{1}\,\dd\btheta_1\\
    \nonumber
&=\int_{\R^{d}}\int_{\R^{p}}f(\bx_{1}) \Pi_{\bfrho}(\btheta_1|\bz_1)\pi_{\bfrho}^{\bz}(\bz_1)\,\dd\bz_1\dd\btheta_1\\
       \nonumber
&=\int_{\R^{d+p}}f(\bx_1)\Pi_{\bfrho}(\btheta_1,\bz_1)\,\dd\bz_1\dd\btheta_1
= \E[f(\tX_0)]\eqsp,
\end{align}
where we have used \eqref{eq:eq:int_pi_rho} in \eqref{eq:eq:int_pi_rho_2}.
Therefore, $X_1$ has distribution $\Pi_{\bfrho}$ and the Markov kernel $\tP_{\bfrho,\bfgamma}$ admits $\Pi_{\bfrho}$ as a stationary distribution, which completes the proof.
\end{proof}
Define by induction the synchronous coupling $(X_n=(\theta_n, Z_n))_{n \ge 0}, (\tilde{X}_n =(\tilde{\theta}_n, \Zc_n))_{n \ge 0}$, starting from $(\theta_0, \Zb_0) = (\btheta, \bz)$, $(\tilde{\theta}_0, \tilde{\Zb}_0)$ distributed according to $\Pi_{\bfrho}$, for any $i\in [b]$ and $n\ge 0$, as
\begin{align}
  \label{eq:cont_coupling_2}
  &\Zc_{n+1}^{i} = \Yc_{N_{i}\gamma_{i}}^{(i, n)}\eqsp,
  &\tilde{\theta}_{n+1} = \bar{\B{B}}_0^{-1}\B{B}_0^{\top}\B{\tilde{D}}_{0}^{\half} \Zc_{n+1} + \bar{\B{B}}_0^{-\half} \xi_{n+1}\eqsp, \\
  \nonumber
  &\Zb_{n+1}^{i} = \Yb_{N_{i}\gamma_{i}}^{(i, n)}\eqsp,
  &\theta_{n+1} = \bar{\B{B}}_0^{-1}\B{B}_0^{\top}\B{\tilde{D}}_{0}^{\half} \Zb_{n+1} + \bar{\B{B}}_0^{-\half} \xi_{n+1}\eqsp,
\end{align}
where we consider for any $i\in[b], k\in\N$, for $t\in[k\gamma_{i},(k+1)\gamma_{i})$
\begin{equation}
  \label{eq:cont_coupling_Y_N}
\begin{aligned}
&\Yc_{t}^{(i, n)} = \Yc_{k \gamma_{i}}^{(i, n)} - \int_{k\gamma_{i}}^{t} \nabla \tildeU_{i}(\Yc^{(i, n)}_{l})\,\dd l +(t-k\gamma_{i})(\rho_{i})^{-1} \B{A}_{i} \tilde{\theta}_{n} +2^{\half}(B_{t}^{(i, n)}-B_{k\gamma_{i}}^{(i, n)})\eqsp,\\
&\Yb_{t}^{(i, n)} = \Yb_{k \gamma_{i}}^{(i, n)}  - (t-k\gamma_{i})\nabla \tildeU_{i}(\Yb_{k \gamma_{i}}^{(i, n)}) + (t-k\gamma_{i})(\rho_{i})^{-1} \B{A}_{i} \theta_{n} + 2^{\half}(B_{t}^{(i, n)}-B_{k\gamma_{i}}^{(i, n)})\eqsp.
\end{aligned}
\end{equation}
Let $\mathcal{G}_0 = \sigma(Z_0,\tilde{Z}_0,\theta_0,\tilde{\theta}_0)$, for any $n \in \mathbb{N}^{*}$, let 
\begin{equation}\label{eq:filtrationG}
\mathcal{G}_{n} = \sigma\{(Z_0,\tilde{Z}_0,\theta_0,\tilde{\theta}_0), (B^{(i,k)}_t)_{t\ge 0}:i\in[b], k \le n\}\eqsp,
\end{equation}
and for any $t\ge0$, let $\mathcal{H}_t^{(n)} = \sigma(\{(B^{(i,n)}_s)_{s\le t}:i\in[b]\})$, and
\begin{equation}
\label{eq:filtrationF}
\mathcal{F}_t^{(n)} \text{ the } \sigma\text{-field generated by } \mathcal{H}_t^{(n)} \text{ and } \mathcal{G}_{n-1}\eqsp.
\end{equation}

Note that $X_{n}$ and $\tX_{n}$ are distributed according to $\Pi_{\bfrho}\tP_{\bfrho}^{n}$ and $\updelta_{\tbfx}P_{\bfrho,\bfgamma,\bfN}^{n}$, respectively.
Hence, by definition of the Wasserstein distance of order 2, it follows since $\Pi_{\bfrho}\tP_{\bfrho}^{n} = \Pi_{\bfrho}$ by \Cref{prop:cont_X_law} that 
\begin{equation}\label{eq:W2X_def}
    \wasserstein{}(\Pi_{\bfrho},\updelta_{\bfx}P_{\bfrho,\bfgamma,\bfN}^{n}) \le \mathbb{E}\br{\|X_{n}-\tX_{n}\|^{2}}^{\half}\eqsp.
\end{equation}

We start this section by a first estimate on $\mathbb{E}\parentheseDeuxLigne{\|X_{n}-\tX_{n}\|^{2}}^{\half}$ and some technical results needed for the proof of \Cref{prop:bias_gamma} and \Cref{prop:bias_gamma_bis}.
The following result holds regarding the process $(\Yc_t^{(i,n)})_{t \in \R_+}$ defined, for any $i \in [b]$ and $n \in \N$, in \eqref{eq:cont_coupling_Y_N}.

\begin{lemma} \label{lem:bound_expec_z_x_cont}
Assume \Cref{ass:well_defined_density}-\Cref{ass:supp_fort_convex}.
For $i\in [b], n \in \N$, denote by $\bz^{i}_{n,\star}$ the unique minimiser of $\bz_{i}\in\R^{d_{i}} \mapsto U_{i}(\bz_{i}) + \|\bz_{i} - \B{A}_{i} \tilde{\theta}_n \|/(2\rho_{i})$. 
Then, for any $i\in[b], k \in \N$ and $n \in \N$,
\begin{align}\label{eq:bound:cont_sum_expec1}
\E^{\mathcal{G}_{n}}\brbig{\| \Yc_{k\gamma_{i}}^{(i, n)} - \bz_{n,\star}^{i} \|^2}
&\le\nofrac{d_{i}}{\tilde{m}_{i}}\eqsp.
\end{align}
where $\tilde{m}_{i}$ is defined in \eqref{eq:def_tilde_m_M}.
\end{lemma}

\begin{proof}
  Let $n \in \N$.
  By \citet[Proposition 1]{durmus2018high}, for $i\in[b]$ and $k \in \N$, we have
  \begin{equation}\label{eq:bound:cont_expec}
  \E^{\mathcal{F}_{k\gamma_{i}}^{(n)}}\| \Yc_{k\gamma_{i}}^{(i, n)} - \bz_{n,\star}^{i} \|^2
  \le \| \Zc_{n}^{i} -  \bz_{n,\star}^{i} \|^2 \mathrm{e}^{-2 k \gamma_{i} \tilde{m}_{i}} + (\nofrac{d_{i}}{\tilde{m}_{i}}) ( 1 - \mathrm{e}^{-2 k\gamma_{i} \tilde{m}_{i}})\eqsp.
  \end{equation}

  By \eqref{eq:cont_coupling_Y_N}, using \Cref{prop:cont_X_law} we get that $\tilde{X}_n$ has distribution $\Pi_{\bfrho}$, therefore given $\tilde{\theta}_n$, $\tilde{Z}_n$ has distribution $\Pi_{\bfrho}(\cdot |\tilde{\theta}_n)$. Then, using \eqref{eq:bound:cont_expec}, \citet[Proposition 1(ii)]{durmus2018high} combined with \Cref{ass:supp_fort_convex}, and since $(\tilde{Z}_n^1,\ldots,\tilde{Z}_n^{b})$ are independent given $\tilde{\theta}_n$, we get the stated result.
\end{proof}

\begin{lemma}\label{lem:geo_decr_cont}
Assume \Cref{ass:well_defined_density} and let $\bfN\in(\N^*)^{b}, \gammabf \in (\rset_+^*)^b$.
Then, for any $n \in \nset$, the random variables $X_{n}=(\theta_{n}^{\top}, \Zb_{n}^{\top})^{\top}, \tilde{X}_{n}=(\tilde{\theta}_{n}^{\top}, \Zc_{n}^{\top})^{\top}$ defined in \eqref{eq:cont_coupling_2} sastify
\[
\|\tilde{X}_{n+1}-X_{n+1}\|^2 \le (1 + \| \bar{\B{B}}_0^{-1}\B{B}_0^{\top}\B{\tilde{D}}_{0}^{\half} \|^2) \|\Zc_{n+1} - \Zb_{n+1} \|^2\eqsp,
\]
where $\bar{\B{B}}_0,\B{B}_0,\B{\tilde{D}}_{0}$ are defined in \eqref{eq:def_B_bar_B}-\eqref{eq:def_projection}.
\end{lemma}
\begin{proof}
The proof is similar to the proof of \Cref{lem:geo_decr} and is omitted.
\end{proof}
For any $k,n\in\N, s\in\R_+$ consider the $p\times p$ matrices defined by
\begin{align}
  \label{eq:def:cont_J}
&\B{J}(k,s) = \mathrm{diag}\pr{\1_{[N_1]}(k+1) \1_{[0,\gamma_1]}(s) \cdot\B{I}_{d_1}, \cdots,  \1_{[N_b]}(k+1) \1_{[0,\gamma_b]}(s) \cdot\B{I}_{d_b}}\eqsp, \\
  \label{eq:def:cont_H}
&\B{H}_{U,k}^{(n)} = \mathrm{diag}\Big(\gamma_1
\int_0^1 \nabla^2 U_1((1-s) \Yb_{k\gamma_1}^{(1, n)} + s \Yc_{k\gamma_1}^{(1, n)} )\,\dd s,&\\
  \nonumber
&\qquad\qquad\qquad\qquad\qquad\qquad\qquad\qquad\qquad\hdots, \gamma_b \int_0^1 \nabla^2 U_b ((1-s) \Yb_{k\gamma_b}^{(b, n)} +s\Yc_{k\gamma_b}^{(b, n)} )\,\dd s\Big)\eqsp, \\
  \label{eq:def:cont_C}
&\B{C}_{k}^{(n)} = \B{J}(k,0)(\B{D}_{\bfgamma/\bfrho} + \B{H}_{U,k}^{(n)})\eqsp,\\
  \label{eq:def:cont_M}
&\B{M}_{k+1}^{(n)} = (\B{I}_{p} - \B{C}_{0}^{(n)} )^{-1} \ldots (\B{I}_{p} - \B{C}_{k}^{(n)} )^{-1}\eqsp, \qquad \text{ with } \B{M}_0^{(n)} = \B{I}_p \eqsp.
\end{align}

Similarly to \eqref{eq:coupling_process_Y}, for $n,k\in\N$ and $i \in [b]$, consider $\B{C}_{k}^{(i, n)}$ corresponding to the $i$-th diagonal block of $\B{C}_{k}^{(n)}$ defined in \eqref{eq:def:cont_C}, \emph{i.e.}
\begin{equation}
  \label{eq:def_B_C_i_cont}
  \B{C}_{k}^{(i, n)}= \1_{[N_i]}(k+1)\gamma_i\ac{\rho_i^{-1}\B{I}_{d_i}+\int_0^1 \nabla^2 U_i((1-s) \Yb_{k\gamma_i}^{(i, n)} + s \Yc_{k\gamma_i}^{(i, n)} )\,\dd s}\in\R^{d_i\times d_i} \eqsp,
\end{equation}
where, for any $n \in \N$ and $i \in [b]$, $(\Yb_{k \gamma_i}^{(i, n)},\Yc_{k \gamma_i}^{(i, n)})_{k \in\N}$ is defined in \eqref{eq:cont_coupling_Y_N}.
\begin{lemma}\label{lem:C_{i}nvertible}
  Assume \Cref{ass:well_defined_density}-\Cref{ass:supp_fort_convex} and let $\bfgamma\in(\R_+^*)^b$ such that, for any $i\in[b], \gamma_{i}<1/\tilde{M}_{i}$.
  Then, for any $n,k \in \N$, the matrix $(\B{I}_p-\B{C}_{k}^{(n)})$ is invertible and in addition, for any $i\in[b]$, we have
  \[\|\B{I}_{d_{i}}-\B{C}^{(i,n)}_{k}\|\le 1-\gamma_{i} \tilde{m_{i}}\eqsp,\]
\end{lemma}
where $\B{C}^{(i,n)}_{k}$ is defined in \eqref{eq:def_B_C_i_cont}.
\begin{proof}
Let $i\in[b],n,k\in\N$. 
By \Cref{ass:supp_fort_convex}, we have $\|\nabla^2 U_{i}\|\le M_{i}$ which implies by \eqref{eq:def_B_C_i_cont} that $\|\B{C}^{(i,n)}_{k}\|\le \gamma_{i}\tilde{M}_{i}$.
Since $\gamma_{i}<1/\tilde{M}_{i}$, the matrix $\B{I}_p-\B{C}_{k}^{(i,n)}$ is invertible and so is $\B{I}_p-\B{C}_{k}^{(n)}$. 
In addition, following the same lines as the proof of  \Cref{lem:bound:S} implies $\|\B{I}_{d_{i}}-\B{C}^{(i,n)}_{k}\|\le \max\{\absn{1-\gamma_{i} \tilde{m_{i}}},\absn{1-\gamma_{i} \tilde{M}_{i}}\}=1-\gamma_{i} \tilde{m_{i}}$.
\end{proof}
For any $n,k\in\N, i\in[b]$, if $\gamma_{i}\in(0, 1/\tilde{M}_{i})$, \Cref{lem:C_{i}nvertible} shows the invertibility of the matrices $\B{I}_p-\B{C}_{k}^{(n)}$.
Therefore, $\B{M}_{\infty}^{(n)}$ is invertible and we can define
\begin{align}\label{eq:def:T}
 &\B{\mathrm{T}}_1^{(n)}
= [\B{M}_{\infty}^{(n)}]^{-1}
+ \sum_{k=0}^{\infty}[\B{M}_{\infty}^{(n)}]^{-1} \B{M}_{k+1}^{(n)}\B{J}(k,0)\B{D}_{\bfN}^{-\half}\B{D}_{\bfgamma/\bfrho}^{\half} \B{P}_{0} \B{D}_{\bfgamma/\bfrho}^{\half}\B{D}_{\bfN}^{\half}\eqsp, \\
\label{eq:def:T2}
& \B{\mathrm{T}}_{2}^{(n)}
=  \sum_{k=0}^{\infty}\acBig{[\B{M}_{\infty}^{(n)}]^{-1}\B{M}_{k+1}^{(n)} \B{D}_{\bfN\bfgamma}^{-\half} \int_0^{\plusinfty} \B{J}(k,l) [ \nabla \tildeU (\Yc_{k\bfgamma+l}^{(n)}) - \nabla \tildeU (\Yc_{k\bfgamma}^{(n)})]\,\dd l}\eqsp.
\end{align}
Using these matrices, we have the following result.
\begin{lemma}\label{lem:bound:cont_Z_{i}nduction1}
Assume \Cref{ass:well_defined_density}-\Cref{ass:supp_fort_convex} and let $\bfN\in(\N^*)^{b}, \gammabf \in (\rset_+^*)^b$ such that, for any $i\in[b], \gamma_{i}<1/\tilde{M}_{i}$.
Then, for any $n \ge 1$,
\begin{align}\label{eq:eq:T2_new_bound}
\B{D}_{\bfN\bfgamma}^{-\half}(\Zc_{n+1}-\Zb_{n+1}) = \B{\mathrm{T}}_{1}^{(n)}(\tilde{Z}_n - Z_n) - \B{\mathrm{T}}_{2}^{(n)}\eqsp,
\end{align}
where $(\Zb_{n}, \Zc_{n})_{n \in \N}$ is defined in \eqref{eq:cont_coupling_2} and $\B{D}_{\bfN\bfgamma}=\mathrm{diag}(N_1\gamma_1\B{I}_{d_1},\ldots,N_b\gamma_b\B{I}_{d_b})\in\R^{p\times p}$.
\end{lemma}
\begin{proof}
Let $i\in[b]$ and $n \ge 1$. 
Recall that $\tildeU_{i}$ is defined in \eqref{eq:def:cont_Utilde} and for $\bz\in\R^p$, denote $\tildeU(\bz)=\sum_{i=1}^b \tildeU_{i}(\bz_{i})$. For any $k\in\N$, we have
\[
 \nabla \tildeU_{i}(\Yc_{k\gamma_{i}}^{(i, n)}) - \nabla \tildeU_{i}(\Yb_{k\gamma_{i}}^{(i, n)})
= \brbigg{\int_0^1 \nabla^2 \tildeU_{i}((1-s) \Yb_{k\gamma_{i}}^{(i,n)} + s \Yc_{k\gamma_{i}}^{(i,n)})\,\dd s} (\Yc_{k\gamma_{i}}^{(i,n)} - \Yb_{k\gamma_{i}}^{(i,n)})\eqsp.
\]
For $k\ge 0$, it follows from \eqref{eq:cont_coupling_Y_N} that
\begin{equation}
  \label{eq:rec_Y_proof_cont_1}
\begin{aligned}
\Yc_{(k+1)\gamma_{i}}^{(i, n)} - \Yb_{(k+1)\gamma_{i}}^{(i, n)}
&= \prbigg{\B{I}_{d_{i}} - \gamma_{i} \int_0^1 \nabla^2 \tildeU_{i}((1-s) \Yb_{k\gamma_{i}}^{(i, n)} + s \Yc_{k\gamma_{i}}^{(i, n)})\,\dd s} (\Yc_{k\gamma_{i}}^{(i, n)} - \Yb_{k\gamma_{i}}^{(i, n)}) \\
&- \int_0^{\gamma_{i}}\big[\nabla \tildeU_{i} (\Yc_{k\gamma_{i}+l}^{(i, n)}) - \nabla \tildeU_{i} (\Yc_{k\gamma_{i}}^{(i, n)})\big]\,\dd l
+ (\gamma_{i}/\rho_{i}) \B{A}_{i} (\tilde{\theta}_{n} - \theta_{n})\eqsp.
\end{aligned}
\end{equation}
Consider the process $(\Ycr_{t}^{(n)}, \Ybr_{t}^{(n)})_{t\in\R_+}$ valued in $\R^p\times\R^p$ and defined for any $t\ge 0$ by
\begin{align}\label{eq:def:stoch_approx}
&\Ycr_t^{(n)}=\Yc_{\min(t, N_{i}\gamma_{i})}^{(n)}\eqsp,
&\Ybr_t^{(n)}=\Yb_{\min(t, N_{i}\gamma_{i})}^{(n)}\eqsp.
\end{align}
The process \eqref{eq:def:stoch_approx} is continuous with respect to $t$ and defined so that its component  $(\Ycr_t^{(i,n)},\Ybr_t^{(i,n)})$ equals $(\Yc_{t}^{i}, \Yb_{t}^{i})$ for $t\le N_{i}\gamma_{i}$ and is constant for $t > N_{i}\gamma_{i}$.
For $l\ge 0$, we write $(\Ycr_{k\bfgamma +l}^{(n)},\Ybr_{k\bfgamma +l}^{(n)}) = (\Ycr_{k\gamma_{i} + l}^{(i, n)}, \Ybr_{k\gamma_{i} + l}^{(i, n)})_{i \in [b]}\in\R^p\times\R^p$. 
Using the matrices defined in \eqref{eq:def:cont_M}, for $k\in\N$, we obtain
\begin{multline}\label{eq:cont_by_rec_N}
\Ycr_{(k+1)\bfgamma}^{(n)} - \Ybr_{(k+1)\bfgamma}^{(n)}
= (\B{I}_{p}-\B{C}_{k}^{(n)} ) (\Ycr_{k\bfgamma}^{(n)} - \Ybr_{k\bfgamma}^{(n)})
- \txts\int_0^{\infty} \B{J}(k,l) \big[ \nabla \tildeU (\Ycr_{k\bfgamma+ l}^{(n)}) - \nabla \tildeU (\Ycr_{k\bfgamma}^{(n)})\big]\,\dd l \\
+ \B{J}(k,0)\B{D}_{\bfgamma/\sqrt{\bfrho}} \B{P}_{0}\B{\tilde{D}}_{0}^{\half} (\Ycr_{0}^{(n)} - \Ybr_{0}^{(n)})\eqsp,
\end{multline}
where $\B{P}_0$ is defined in \eqref{eq:def_projection}.
Recall the matrix $\B{M}_{k}^{(n)}$ defined in \eqref{eq:def:cont_M} with $\B{M}_0^{(n)} = \B{I}_p$ and for $k\ge 1$,  
$
\B{M}_{k}^{(n)} = (\B{I}_{p} - \B{C}_{0}^{(n)} )^{-1} \ldots (\B{I}_{p} - \B{C}_{k-1}^{(n)} )^{-1}
$.
By multiplying \eqref{eq:cont_by_rec_N} by $\B{M}_{k+1}^{(n)}\B{D}_{\bfN\bfgamma}^{-\half}$, we have
\begin{align*}
\B{M}_{k+1}^{(n)}\B{D}_{\bfN\bfgamma}^{-\half} (\Ycr_{(k+1)\bfgamma}^{(n)} - \Ybr_{(k+1)\bfgamma}^{(n)})
&= \B{M}_{k}^{(n)} \B{D}_{\bfN\bfgamma}^{-\half}(\Ycr_{k\bfgamma}^{(n)} - \Ybr_{k\bfgamma}^{(n)}) \\
&- \B{M}_{k+1}^{(n)} \B{D}_{\bfN\bfgamma}^{-\half} \int_0^{\infty} \B{J}(k,l) \big[ \nabla \tildeU (\Yc_{k\bfgamma+ l}^{(n)}) - \nabla \tildeU (\Yc_{k\bfgamma}^{(n)})\big]\,\dd l \\
&+ \B{M}_{k+1}^{(n)}\B{J}(k,0)\B{D}_{\bfN}^{-\half}\B{D}_{\bfgamma/\bfrho}^{\half} \B{P}_{0}\B{\tilde{D}}_{0}^{\half} (\Ycr_{0}^{(n)} - \Ybr_{0}^{(n)})\eqsp.
\end{align*}
By  \eqref{eq:def:stoch_approx} and \eqref{eq:cont_coupling_2}, we have for $t \ge \max_{i\in[b]} \{\gamma_{i} N_{i}\}$, $(\tilde{Z}_{n+1},Z_{n+1}) = (\Ycr_{t},\Ybr_{t})$. Therefore, summing the previous expression over $k$, we get 
\begin{align*}
  \B{M}_{\infty}^{(n)} &\B{D}_{\bfN\bfgamma}^{-\half}(\Zc_{n+1}-\Zb_{n+1})
  = - \sum_{k=0}^{\infty} \B{M}_{k+1}^{(n)} \B{D}_{\bfN\bfgamma}^{-\half} \int_0^{\infty} \B{J}(k,l) [ \nabla \tildeU (\Yc_{k\bfgamma+ l}^{(n)}) - \nabla \tildeU (\Yc_{k\bfgamma}^{(n)})]\,\dd l \\
  &+\brbigg{\B{M}_0^{(n)} + \sum_{k=0}^{\infty} \B{M}_{k+1}^{(n)}\B{J}(k,0)\B{D}_{\bfN}^{-\half}\B{D}_{\bfgamma/\bfrho}^{\half}\B{P}_{0}\B{D}_{\bfgamma/\bfrho}^{\half}\B{D}_{\bfN}^{\half}}
  \B{D}_{\bfN\bfgamma}^{-\half}\cdot(\Zc_{n}-\Zb_{n})\eqsp.
\end{align*}
By \Cref{lem:C_{i}nvertible}, $\B{M}_{\infty}^{(n)}$ is invertible and the proof is concluded by multiplying the previous equality by $[\B{M}_{\infty}^{(n)}]^{-1}$. 
\end{proof}

Based on \Cref{lem:bound:cont_Z_{i}nduction1}, we have the following relation between $\| \Zc_{n+1} - \Zb_{n+1} \|^2$ and $\| \Zc_{n} - \Zb_{n} \|^2$.

\begin{lemma}\label{lem:bound:cont_Z_{i}nduction}
Assume \Cref{ass:well_defined_density}-\Cref{ass:supp_fort_convex} and let $\bfN\in(\N^*)^{b}, \gammabf \in (\rset_+^*)^b$ such that, for any $i\in[b], \gamma_{i}<1/\tilde{M}_{i}$.
Then, for any $\epsilon >0$ and $n \ge 1$,
\[
\| \Zc_{n+1} - \Zb_{n+1} \|_{\B{D}_{\bfN\bfgamma}^{-1}}^2
\le (1+2\epsilon) \normn{\B{\mathrm{T}}_1^{(n)}}^2 \| \Zc_{n} - \Zb_{n} \|_{\B{D}_{\bfN\bfgamma}^{-1}}^2
+ (1+1/\{2\epsilon\}) \normn{\B{\mathrm{T}}_{2}^{(n)}}^2\eqsp.
\]
where $(\Zb_{n}, \Zc_{n})_{n \in \N}$ is defined in \eqref{eq:cont_coupling_2} and $\B{D}_{\bfN\bfgamma}=\mathrm{diag}(N_1\gamma_1\B{I}_{d_1},\ldots,N_b\gamma_b\B{I}_{d_b})\in\R^{p\times p}$.
\end{lemma}
\begin{proof}
  The proof follows from \Cref{lem:bound:cont_Z_{i}nduction1} and by using the fact that for $\B{a}, \B{b} \in \R^{p}, \epsilon > 0$ we have $2\langle \B{a}, \B{b}\rangle \le 2\epsilon\|\B{a}\|^2 + (1/\acn{2\epsilon})\|\B{b}\|^2$.
\end{proof}

Similarly to \Cref{lem:bound:cont_T1}, we have the following result regarding the contracting term.

\begin{lemma}
\label{lem:bound:cont_T1_bis}
Assume \Cref{ass:well_defined_density}-\Cref{ass:supp_fort_convex} and let $\bfN\in(\N^*)^{b}, \gammabf \in (\rset_+^*)^b$ such that, for any $i \in [b]$, $\gamma_{i}<1/\tilde{M}_{i}$ and $N_{i}\gamma_{i}\le 2/(m_{i} + \tilde{M}_{i})$. Then, for any $n\ge0$, we have
\begin{align*}
\normn{\B{\mathrm{T}}_1^{(n)}} &\le  1 - \min_{i\in[b]} \{N_{i} \gamma_{i} m_{i}\} + r_{\bfgamma, \bfrho, \bfN}\eqsp,
\end{align*}
where $\B{\mathrm{T}}_1^{(n)}$ and $r_{\bfgamma, \bfrho, \bfN}$ are defined in \eqref{eq:def:T} and \eqref{eq:def:r}, respectively.
\end{lemma}
\begin{proof}
The proof is similar to the proof of \Cref{lem:bound:cont_T1} and therefore is omitted.
\end{proof}

In the next lemma, we upper bound the coefficient $r_{\bfgamma,\bfrho,\bfN}$ defined in \eqref{eq:def:r}. For this, we explicit a choice of $\bfN$ that we denote $\bfN^\star=(N_1^\star(\gamma_1),\ldots,N_b^\star(\gamma_b))\in(\N^*)^b$ defined for any $i\in[b]$, any $\gamma_{i}>0$, by
\begin{equation}\label{eq:eq:good_choice_N}
  N_{i}^{\star}(\gamma_{i}) = \big\lfloor m_{i}\min_{i\in[b]}\acn{m_{i}/\tilde{M}_{i}}^2/\prbig{20\gamma_{i}\tilde{M}_{i}^2\max_{i\in[b]}\acn{m_{i}/\tilde{M}_{i}}^2}\big\rfloor\eqsp,
\end{equation}
where $\tilde{M}_{i}=M_{i}+1/\rho_{i}$.
\begin{lemma}\label{lem:choice_N_condition_gamma}
Assume \Cref{ass:well_defined_density}-\Cref{ass:supp_fort_convex} and let $\bfgamma\in(\R_+^*)^b$ such that, for any $i\in[b]$,
\[
\gamma_{i}\le \frac{m_{i}}{40\tilde{M}_{i}^2}\pr{\frac{\min_{i\in[b]}\acn{\nofrac{m_{i}}{\tilde{M}_{i}}}}{\max_{i\in[b]}\acn{\nofrac{m_{i}}{\tilde{M}_{i}}}}}^2\eqsp.
\]
Then, for any $i\in[b]$, we have $N_{i}^\star(\gamma_{i})\in\N^*$ and 
\[
r_{\bfgamma,\bfrho,\bfN^\star} < \min_{i\in[b]}\{N_{i}^\star(\gamma_{i})\gamma_{i}m_{i}\}/2\eqsp,
\]
where $r_{\bfgamma,\bfrho,\bfN^\star}$ is defined in \eqref{eq:def:r}.
\end{lemma}
\begin{proof}
The assumption on $\gamma_{i}$ combined with the definition~\eqref{eq:eq:good_choice_N} of $N_{i}^\star(\gamma_{i})$ imply $N_{i}^\star(\gamma_{i})\ge 2$, using in addition $m_i \le M_i$, $\max_{i\in[b]}\{N_{i}^\star(\gamma_{i})\gamma_{i}\tilde{M}_{i}\1_{N_{i}^\star(\gamma_{i})>1}\} \le 1/20$  and 
\begin{align}
  \nonumber
  \frac{1}{20}\prbigg{\frac{\min_{i\in[b]}\acn{m_{i}/\tilde{M}_{i}}}{\max_{i\in[b]}\acn{m_{i}/\tilde{M}_{i}}}}^2
  \ge \frac{N_{i}^\star(\gamma_{i})\gamma_{i}\tilde{M}_{i}^2}{m_{i}}
  &> \frac{1}{20}\prbigg{\frac{\min_{i\in[b]}\acn{m_{i}/\tilde{M}_{i}}}{\max_{i\in[b]}\acn{m_{i}/\tilde{M}_{i}}}}^2-\frac{\gamma_{i}\tilde{M}_{i}^2}{m_{i}}\\
  &\ge \frac{1}{40}\prbigg{\frac{\min_{i\in[b]}\acn{m_{i}/\tilde{M}_{i}}}{\max_{i\in[b]}\acn{m_{i}/\tilde{M}_{i}}}}^2\eqsp.\label{eq:eq:key}
\end{align}
Using the definition~\eqref{eq:def:r} of $r_{\bfgamma,\bfrho,\bfN}$, we have
$
r_{\bfgamma,\bfrho,\bfN} < 5 \max_{i\in[b]}\{N_{i}^\star(\gamma_{i})\gamma_{i}\tilde{M}_{i}\1_{N_{i}^\star(\gamma_{i})>1}\}^2.
$
Thus, plugging \eqref{eq:eq:key} in the previous inequality gives
\begin{equation}\label{eq:eq:key2}
r_{\bfgamma,\bfrho,\bfN} \le \max_{i \in[b]} \{m_i/ \tilde{M}_i\}^2 \max_{i\in[b]} \defEns{\frac{N_{i}^\star(\gamma_{i})\gamma_{i}\tilde{M}_{i}^2}{m_{i}}} < \frac{\min_{i\in[b]}\acn{m_{i}/\tilde{M}_{i}}^4}{80\max_{i\in[b]}\acn{m_{i}/\tilde{M}_{i}}^2}\eqsp.
\end{equation}
In addition,  \eqref{eq:eq:key} also shows that
\begin{equation}\label{eq:eq:key3}
\frac{1}{40}\prbigg{\frac{\min_{i\in[b]}\acn{m_{i}/\tilde{M}_{i}}}{\max_{i\in[b]}\acn{m_{i}/\tilde{M}_{i}}}}^2 \prbigg{\frac{m_{i}}{\tilde{M}_{i}}}^2
\le N_{i}^\star(\gamma_{i})\gamma_{i}m_{i}\eqsp.
\end{equation}
Therefore, combining \eqref{eq:eq:key2} and \eqref{eq:eq:key3} completes the proof.
\end{proof}

\subsection{Proof of \Cref{prop:bias_gamma}}
\label{subsec:proof_cont1}

We first give the formal statement of \Cref{prop:bias_gamma}.

\begin{proposition}\label{cor:bias_pi_rho_pirho_gamma}
Assume \Cref{ass:well_defined_density}-\Cref{ass:supp_fort_convex} and let $\gammabf\in (\rset_+^*)^b$, $\bfN \in (\N^*)^b$ such that for any $i \in [b]$, 
$
\txts\gamma_{i}\le \nofrac{m_{i}}{40\tilde{M}_{i}^2}(\min_{i\in[b]}\{m_{i}/\tilde{M}_{i}\}/\max_{i\in[b]}\{m_{i}/\tilde{M}_{i}\})^2$ and $N_{i} = \lfloor m_{i}\min_{i\in[b]}\acn{m_{i}/\tilde{M}_{i}}^2/(20\gamma_{i}\tilde{M}_{i}^2\max_{i\in[b]}\acn{m_{i}/\tilde{M}_{i}^2})\rfloor$.
Then, we have
\begin{multline*}
\wasserstein{}^{2}\prn{\Pi_{\bfrho,\bfgamma,\bfN},\Pi_{\bfrho}}
\le \frac{4(1+\|\bar{\B{B}}_0^{-1}\B{B}_0^{\top}\B{\tilde{D}}_{0}^{\half}\|^2) \max_{i\in[b]}\acn{\nofrac{m_{i}}{\tilde{M}_{i}^2}}}{5\min_{i\in[b]}\acn{\nofrac{m_{i}}{\tilde{M}_{i}}}^{2}\max_{i\in[b]}\acn{\nofrac{m_{i}}{\tilde{M}_{i}}}^2}\\
\times\sum_{i=1}^b d_{i}\gamma_{i}m_{i}(1 + \nofrac{\gamma_{i}^2\tilde{M}_{i}^2}{12}+ \gamma_{i}\tilde{M}_{i}^2/(2\tilde{m}_{i}))\eqsp,
\end{multline*}
where $\bar{\B{B}}_0,\B{B}_0,\B{\tilde{D}}_{0}$ are defined in \eqref{eq:def_B_bar_B}-\eqref{eq:def_projection}, and for any $i \in [b]$, $\tilde{m}_i$, $\tilde{M}_i$ are defined in \eqref{eq:def_tilde_m_M}.
\end{proposition}

By \Cref{lem:geo_decr_cont} and \Cref{lem:bound:cont_Z_{i}nduction}, we can note that the proof of \Cref{cor:bias_pi_rho_pirho_gamma} boils down to derive an upper bound on $\normn{\B{\mathrm{T}}_2^{(n)}}^2$ defined in \eqref{eq:def:T2} for $n \in \N$.
The following lemma provides such a bound.

\begin{lemma}\label{lem:bound:cont_expec_T2}
  Assume \Cref{ass:well_defined_density}-\Cref{ass:supp_fort_convex} and let $\bfN\in(\N^*)^{b}, \gammabf  \in (\rset_+^*)^b$ such that, for any $i\in[b]$, $\gamma_{i}< 1/\tilde{M}_{i}$. Then, for any $n \in \N$, we have
\begin{align*}
  {\E\br{\normn{\B{\mathrm{T}}_2^{(n)}}^2}}
&\le{\sum_{i=1}^b} d_{i}N_{i}\gamma_{i}^2\tilde{M}_{i}^2 \br{1 + \nofrac{\gamma_{i}^2\tilde{M}_{i}^2}{12}+ \gamma_{i}\tilde{M}_{i}^2/(2\tilde{m}_{i})}\eqsp,
\end{align*}
where $\tilde{m}_{i}, \tilde{M}_{i}, \mathrm{T}_{2}^{(n)}$ are defined in \eqref{eq:def_tilde_m_M} and \eqref{eq:def:T2}, respectively.
\end{lemma}
\begin{proof}
Let $n \in \mathbb{N}$.
Using \eqref{eq:def:cont_J}, we can write, for any $l \in \R_+$ and $k \in \N$, $\B{J}(k,l)$ as a block-diagonal matrix $\mathrm{diag}(\B{J}^1(k,l),\ldots,\B{J}^b(k,l))$ with $\B{J}^i(k,l) = \1_{[N_i]}(k+1) \1_{[0,\gamma_i]}(s) \cdot\B{I}_{d_i}$ for any $i \in [b]$.
By \eqref{eq:def:cont_M} and using for any $k \in \N$, that $[\B{M}_{\infty}^{(n)}]^{-1}\B{M}_{k+1}^{(n)} = \prod_{l=k+1}^{\infty}(\B{I}_{d_{i}}-\B{C}^{(i,n)}_l)$ is finite by \eqref{eq:def:cont_C}, we have
\begin{align}
\nonumber
\normn{\B{\mathrm{T}}_2^{(n)}}^2
&= \normBig{\sum_{k=0}^{\infty} [\B{M}_{\infty}^{(n)}]^{-1}\B{M}_{k+1}^{(n)} \B{D}_{\bfN\bfgamma}^{-\half} \int_0^{\infty} \B{J}(k,l) \big[\nabla \tildeU (\Yc_{k\bfgamma+l}^{(n)}) - \nabla \tildeU (\Yc_{k\bfgamma}^{(n)})\big]\,\dd l}^2 \\
&= \sum_{i=1}^b \frac{1}{N_{i}\gamma_{i}} \normBig{\sum_{k=0}^{\infty}\prod_{l=k+1}^{\infty}(\B{I}_{d_{i}}-\B{C}^{(i,n)}_l)\int_0^{\gamma_{i}} \B{J}^{i}(k,0) \big[\nabla \tildeU_{i} (\Yc_{k\gamma_{i}+l}^{(i, n)}) - \nabla \tildeU_{i} (\Yc_{k\gamma_{i}}^{(i, n)})\big]\,\dd l}^2 \label{eq:integral_bound_cont}\eqsp.
\end{align}
Since for any $i\in [b]$, $k\ge N_{i}$ we have $\B{J}^{i}(k,0)=\B{C}^{(i,n)}_l=\B{0}_{d_{i}\times d_{i}}$, \eqref{eq:integral_bound_cont} can be rewritten as
\begin{equation*}
\normn{\B{\mathrm{T}}_2^{(n)}}^2
= \sum_{i=1}^b \frac{1}{N_{i}\gamma_{i}} \normBig{\sum_{k=0}^{N_{i}-1}\prod_{l=k+1}^{N_{i}-1}(\B{I}_{d_{i}}-\B{C}^{(i,n)}_l)\int_0^{\gamma_{i}} \B{J}^{i}(k,0) \big[\nabla \tildeU_{i} (\Yc_{k\gamma_{i}+l}^{(i, n)}) - \nabla \tildeU_{i} (\Yc_{k\gamma_{i}}^{(i, n)})\big]\,\dd l}^2\eqsp,
\end{equation*}
and the Cauchy-Schwarz inequality gives
\begin{equation}\label{eq:bound:T2_after_CS}
\normn{\B{\mathrm{T}}_2^{(n)}}^2
\le \sum_{i=1}^b \frac{1}{\gamma_{i}}\pr{\sum_{k=0}^{N_{i}-1} \normBig{\prod_{l=k+1}^{N_{i}-1}(\B{I}_{d_{i}}-\B{C}^{(i,n)}_l)}^2 \normBig{\int_0^{\gamma_{i}} \br{\nabla \tildeU_{i} (\Yc_{k\gamma_{i}+l}^{(i, n)}) - \nabla \tildeU_{i} (\Yc_{k\gamma_{i}}^{(i, n)})}\,\dd l}^2}\eqsp.
\end{equation}
Since, for any $i\in[b]$, $\gamma_{i}\tilde{M}_{i}<1$, we get using \Cref{lem:C_{i}nvertible}, 
\begin{align*}
\norm{\prod_{l=k+1}^{N_{i}-1}(\B{I}_{d_{i}}-\B{C}^{(i,n)}_l)}^2
\le \{1-\gamma_{i} \tilde{m}_{i}\}^{2(N_{i}-k-1)}\eqsp.
\end{align*}
By combining \eqref{eq:bound:T2_after_CS} with the previous result and the Jensen inequality, we have
\begin{equation}\label{eq:bound:cont_T2}
\normn{\B{\mathrm{T}}_2^{(n)}}^2
\le \sum_{i=1}^b\sum_{k=0}^{N_{i}-1}\{1-\gamma_{i} \tilde{m}_{i}\}^{2(N_{i}-k-1)}\int_0^{\gamma_{i}} \norm{\nabla \tildeU_{i} (\Yc_{k\gamma_{i}+l}^{(i, n)}) - \nabla \tildeU_{i} (\Yc_{k\gamma_{i}}^{(i, n)})}^2\,\dd l\eqsp.
\end{equation}
For $i\in[b]$, using \citet[Lemma 21]{durmus2018high} applied to the potential $V_{i}^{\theta} : \by^{i} \mapsto U_{i}(\by^{i}) + \| \by^{i} - \B{A}_{i} \theta\|^2/(2\rho_{i})$ yields
\begin{align}\nonumber
\int_0^{\gamma_{i}} \E^{\mathcal{F}_{k\gamma_{i}}^{(n)}}\normn{\nabla \tildeU_{i} (\Yc_{k\gamma_{i}+l}^{(i, n)}) - \nabla \tildeU_{i} (\Yc_{k\gamma_{i}}^{(i, n)})}^2\,\dd l
= \int_0^{\gamma_{i}} \E^{\mathcal{F}_{k\gamma_{i}}^{(n)}}\normn{\nabla V_{i}^{\tilde{\theta}_n} (\Yc_{k\gamma_{i}+l}^{(i, n)}) - \nabla V_{i}^{\tilde{\theta}_n} (\Yc_{k\gamma_{i}}^{(i, n)})}^2\,\dd l \\
\label{eq:bound:cont_{i}nt}
\le \gamma_{i}^2\tilde{M}_{i}^2\br{d_{i} + \nofrac{d_{i}\gamma_{i}^2\tilde{M}_{i}^2}{12}+ (\nofrac{\gamma_{i}\tilde{M}_{i}^2}{2}) \| \Yc_{k\gamma_{i}}^{(i, n)} - \bz_{n,\star}^{i} \|^2}\eqsp,
\end{align}
where $\bz_{n,\star}^{i} = \arg\min_{\bz_i\in \R^{d_{i}}} V_{i}^{\tilde{\theta}_n}(\bz_i)$.

By \eqref{eq:bound:cont_{i}nt}, \eqref{eq:bound:cont_sum_expec1}, \Cref{lem:bound_expec_z_x_cont} and since $\max_{i\in[b]}\gamma_{i}\tilde{m}_{i}<1$, we get
\begin{multline*}
\sum_{i=1}^b \sum_{k=0}^{N_{i}-1}\{1-\gamma_{i} \tilde{m}_{i}\}^{2(N_{i}-k-1)}\int_0^{\gamma_{i}} \E\| \nabla \tildeU_{i} (\Yc_{k\gamma_{i}+l}^{(i, n)}) - \nabla \tildeU_{i} (\Yc_{k\gamma_{i}}^{(i, n)})\|^2\,\dd l\\
\le \sum_{i=1}^b d_{i}N_{i}\gamma_{i}^2\tilde{M}_{i}^2[1 + \nofrac{\gamma_{i}^2\tilde{M}_{i}^2}{12}+ \nofrac{\gamma_{i}\tilde{M}_{i}^2}{(2\tilde{m}_{i})}] \eqsp.
\end{multline*}
Combining this result with  \eqref{eq:bound:cont_T2} completes the proof.
\end{proof}

We can now combine \Cref{lem:bound:cont_expec_T2} and \Cref{lem:bound:cont_T1_bis} with \Cref{lem:bound:cont_Z_{i}nduction} to get the following bound.

\begin{lemma}\label{lem:bound:cont_expec_norm_spe}
Assume \Cref{ass:well_defined_density}-\Cref{ass:supp_fort_convex} and let $\bfN\in(\N^*)^{b}, \gammabf \in (\rset_+^*)^b$ such that, for any $i\in[b]$, $\gamma_{i}<1/\tilde{M}_{i}$, $N_{i}\gamma_{i}\le 2/(m_{i} + \tilde{M}_{i})$.
Suppose in addition $\upkappa_{\bfgamma,\bfrho,\bfN} = \min_{i\in[b]} \{N_{i}\gamma_{i} m_{i}\}-r_{\bfgamma,\bfrho,\bfN} \in \ooint{0,1}$, where $r_{\bfgamma,\bfrho,\bfN}$ is defined in \eqref{eq:def:r}.
Then, for $n \ge 1$, we have
\begin{multline*}\label{eq:cont_boud_expec_Z_final}
\E\brBig{\normn{\Zc_{n} - \Zb_{n}}_{\B{D}_{\bfN\bfgamma}^{-1}}^2}
\le (1-\upkappa_{\bfgamma,\bfrho,\bfN}+\upkappa_{\bfgamma,\bfrho,\bfN}^2/2)^{2(n-1)} \E\brBig{\normn{\Zc_{1} - \Zb_{1}}_{\B{D}_{\bfN\bfgamma}^{-1}}^2}\\
+ 2\upkappa_{\bfgamma,\bfrho,\bfN}^{-2}
\sum_{i=1}^b d_{i}N_{i}\gamma_{i}^2\tilde{M}_{i}^2 \prbigg{1 + \frac{\gamma_{i}^2\tilde{M}_{i}^2}{12}+ \frac{\gamma_{i}\tilde{M}_{i}^2}{2\tilde{m}_{i}}}\eqsp,
\end{multline*}
where, for any $i \in [b]$, $\tilde{M}_i$ and $\tilde{m}_i$ are defined in \eqref{eq:def_tilde_m_M}.
\end{lemma}
\begin{proof}
Taking expectation in \Cref{lem:bound:cont_Z_{i}nduction}, we get for any $n\in\N, \epsilon>0$ that 
\begin{equation*}
\E\br{\norm{\Zc_{n+1} - \Zb_{n+1}}_{\B{D}_{\bfN\bfgamma}^{-1}}^2}
\le (1+2\epsilon) \E\br{\normn{\B{\mathrm{T}}_1^{(n)}}^2 \normn{\Zc_{n} - \Zb_{n}}_{\B{D}_{\bfN\bfgamma}^{-1}}^2}
+ (1+1/\{2\epsilon\}) \E\br{\normn{\B{\mathrm{T}}_2^{(n)}}^2}\eqsp,
\end{equation*}
where $\B{\mathrm{T}}_1^{(n)}$ and $\B{\mathrm{T}}_2^{(n)}$ are defined in \eqref{eq:def:T} and \eqref{eq:def:T2}, respectively.
To ease notation, denote $\mathrm{B} = \sum_{i=1}^b d_{i}N_{i}\gamma_{i}^2\tilde{M}_{i}^2 (1 + \nofrac{\gamma_{i}^2\tilde{M}_{i}^2}{12}+ \gamma_{i}\tilde{M}_{i}^2/(2\tilde{m}_{i}))$. 
Using \Cref{lem:bound:cont_expec_T2}, we obtain for any $n\in\N, \epsilon>0$
\begin{equation}
\label{eq:cont_expec_Z_rec}
\E\br{\normn{\Zc_{n+1} - \Zb_{n+1}}_{\B{D}_{\bfN\bfgamma}^{-1}}^2}
\le (1+2\epsilon)\E\br{\normn{\mathrm{T}_1^{(n)}}^2\normn{\Zc_{n} - \Zb_{n}}_{\B{D}_{\bfN\bfgamma}^{-1}}^2}
+ (1+1/\{2\epsilon\}) \mathrm{B}\eqsp.
\end{equation}
In addition, \Cref{lem:bound:cont_T1_bis} implies that 
$\normn{\B{\mathrm{T}}_1^{(n)}}^2\le (1 -\upkappa_{\bfgamma,\bfrho,\bfN})^2$ almost surely. 
Therefore, taking  $\epsilon = (1-[1-\upkappa_{\bfgamma,\bfrho,\bfN}]^2)/(4 [1-\upkappa_{\bfgamma,\bfrho,\bfN}]^2)$, \eqref{eq:cont_expec_Z_rec} yields for any $n\ge 0$, 
\begin{equation*}
\E\br{\norm{\Zc_{n+1} - \Zb_{n+1}}_{\B{D}_{\bfN\bfgamma}^{-1}}^2}
\le \frac{1+(1-\upkappa_{\bfgamma,\bfrho,\bfN})^2}{2} \E\br{\norm{\Zc_{n} - \Zb_{n}}_{\B{D}_{\bfN\bfgamma}^{-1}}^2}
+ \frac{1+(1-\upkappa_{\bfgamma,\bfrho,\bfN})^2}{1-(1-\upkappa_{\bfgamma,\bfrho,\bfN})^2}\mathrm{B}\eqsp.
\end{equation*}
An easy induction implies for any $n \ge 1$, 
\begin{equation}\label{eq:bound:expec_rec_Z}
\E\brbig{\| \Zc_{n} - \Zb_{n} \|_{\B{D}_{\bfN\bfgamma}^{-1}}^2}
\le\prbigg{\frac{1+(1-\upkappa_{\bfgamma,\bfrho,\bfN})^2}{2}}^{n-1} \E\brbig{\| \Zc_{1} - \Zb_{1} \|_{\B{D}_{\bfN\bfgamma}^{-1}}^2}
+ 2\frac{1+(1-\upkappa_{\bfgamma,\bfrho,\bfN})^2}{(1-(1-\upkappa_{\bfgamma,\bfrho,\bfN})^2)^2}\mathrm{B}\eqsp.
\end{equation}
Since $\upkappa_{\bfgamma,\bfrho,\bfN}^2 = (\min_{i\in[b]}\{N_{i}\gamma_{i} m_{i}\}+r_{\bfgamma,\bfrho,\bfN})^{2}$ and using $\upkappa_{\bfgamma,\bfrho,\bfN}^2 \le 1$, we obtain
\begin{align*}
&(1+(1-\upkappa_{\bfgamma,\bfrho,\bfN})^2)/2
= 1-\upkappa_{\bfgamma,\bfrho,\bfN}+ \upkappa_{\bfgamma,\bfrho,\bfN}^2/2\eqsp,\\
&(1+(1-\upkappa_{\bfgamma,\bfrho,\bfN})^2)/(1-(1-\upkappa_{\bfgamma,\bfrho,\bfN})^2)^2
\le \upkappa_{\bfgamma,\bfrho,\bfN}^{-2} \eqsp. 
\end{align*}
Combining these inequalities with \eqref{eq:bound:expec_rec_Z} and \eqref{eq:cont_expec_Z_rec} completes the proof. 
\end{proof}
\begin{lemma}\label{lem:bound:cont_wass}
Assume \Cref{ass:well_defined_density}-\Cref{ass:supp_fort_convex} and let $\bfN\in(\N^*)^{b}, \gammabf\in (\rset_+^*)^b$ such that, for any $i\in[b]$, $\gamma_{i}<1/\tilde{M}_{i}, N_{i}\gamma_{i}\le 2/(m_{i} + \tilde{M}_{i})$ and $\upkappa_{\bfgamma,\bfrho,\bfN} = \min_{i\in[b]} \{N_{i}\gamma_{i} m_{i}\}-r_{\bfgamma,\bfrho,\bfN} \in \ooint{0,1}$, where $r_{\bfgamma,\bfrho,\bfN}$ is defined in \eqref{eq:def:r}. Then, for any $\bx \in \mathbb{R}^{d+p}$ and $n \ge 1$, we have
\begin{align*}
&\wasserstein{}^{2} (\updelta_{\bx} P_{\bfrho, \bfgamma,\bfN}^{n},\Pi_{\bfrho})\\
&\le (1- \upkappa_{\bfgamma,\bfrho,\bfN}+ \upkappa_{\bfgamma,\bfrho,\bfN}^2/2)^{2(n-1)} (1 + \|\bar{\B{B}}_0^{-1}\B{B}_0^{\top}\B{\tilde{D}}_{0}^{\half}\|^2)\max_{i\in[b]}\{N_{i}\gamma_{i}\} \E\brbig{\| \Zc_{1} - \Zb_{1} \|_{\B{D}_{\bfN\bfgamma}^{-1}}^2}\\
&\quad+ \frac{2(1 + \|\bar{\B{B}}_0^{-1}\B{B}_0^{\top}\B{\tilde{D}}_{0}^{\half}\|^2) \max_{i\in[b]} \{N_{i}\gamma_{i}\}}{\upkappa_{\bfgamma,\bfrho,\bfN}^{2}}
\sum_{i=1}^b d_{i}N_{i}\gamma_{i}^2\tilde{M}_{i}^2 [1 + \nofrac{\gamma_{i}^2\tilde{M}_{i}^2}{12}+ \gamma_{i}\tilde{M}_{i}^2/(2\tilde{m}_{i})]\eqsp,
\end{align*}
where $\bar{\B{B}}_0,\B{B}_0,\B{\tilde{D}}_{0}$ are defined in \eqref{eq:def_B_bar_B}-\eqref{eq:def_projection}, $P_{\bfrho, \bfgamma,\bfN}$ is defined in \eqref{eq:P_rho_gamma_N}, $(\tilde{Z}_n,Z_n)_{n \in \N}$ is defined in \eqref{eq:cont_coupling_2} and for any $i \in [b]$, $\tilde{M}_i$, $\tilde{m}_i$ are defined in \eqref{eq:def_tilde_m_M}.
\end{lemma}
\begin{proof}
By \Cref{lem:bound:cont_expec_norm_spe}, we have the following upper bound for $n \ge 1$,
\begin{multline*}
\E\brBig{\normn{\Zc_{n} - \Zb_{n}}_{\B{D}_{\bfN\bfgamma}^{-1}}^2}
\le (1-\upkappa_{\bfgamma,\bfrho,\bfN}+\upkappa_{\bfgamma,\bfrho,\bfN}^2/2)^{2(n-1)} \E\brBig{\normn{\Zc_{1} - \Zb_{1}}_{\B{D}_{\bfN\bfgamma}^{-1}}^2}\\
+ 2\upkappa_{\bfgamma,\bfrho,\bfN}^{-2}
\sum_{i=1}^b d_{i}N_{i}\gamma_{i}^2\tilde{M}_{i}^2 \prbigg{1 + \frac{\gamma_{i}^2\tilde{M}_{i}^2}{12}+ \frac{\gamma_{i}\tilde{M}_{i}^2}{2\tilde{m}_{i}}}\eqsp.
\end{multline*}
Using \eqref{eq:cont_coupling_2}, \Cref{lem:geo_decr_cont}, combined with the previous inequality, we get for any $n\ge 1, \bx\in\R^{d+p}$,
\begin{align*}
&\wasserstein{}^{2} (\Pi_{\bfrho}, \updelta_{\bx}P_{\bfrho, \bfgamma,\bfN}^{n})\\
&\le (1 + \|\bar{\B{B}}_0^{-1}\B{B}_0^{\top}\B{\tilde{D}}_{0}^{\half}\|^2) \E\brbig{\normn{\Zc_{n} - \Zb_{n}}^2} \\
&\le (1 + \|\bar{\B{B}}_0^{-1}\B{B}_0^{\top}\B{\tilde{D}}_{0}^{\half}\|^2) \max_{i\in[b]}\{N_{i}\gamma_{i}\} \E\brBig{\normn{\Zc_{n} - \Zb_{n}}_{\B{D}_{\bfN\bfgamma}^{-1}}^2} \\
&\le (1-\upkappa_{\bfgamma,\bfrho,\bfN}+\upkappa_{\bfgamma,\bfrho,\bfN}^2/2)^{2(n-1)} (1 + \|\bar{\B{B}}_0^{-1}\B{B}_0^{\top}\B{\tilde{D}}_{0}^{\half}\|^2)\max_{i\in[b]}\{N_{i}\gamma_{i}\}\E\brBig{\normn{\Zc_{1} - \Zb_{1}}_{\B{D}_{\bfN\bfgamma}^{-1}}^2} \\
&+ \frac{2(1 + \|\bar{\B{B}}_0^{-1}\B{B}_0^{\top}\B{\tilde{D}}_{0}^{\half}\|^2) \max_{i\in[b]}\{N_{i}\gamma_{i}\}}{\upkappa_{\bfgamma,\bfrho,\bfN}^2}
\sum_{i=1}^b d_{i}N_{i}\gamma_{i}^2\tilde{M}_{i}^2 \prBig{1 + \frac{\gamma_{i}^2\tilde{M}_{i}^2}{12}+ \frac{\gamma_{i}\tilde{M}_{i}^2}{2\tilde{m}_{i}}}\eqsp.
\end{align*}
  Hence the stated result.
\end{proof}

\paragraph{Proof of \Cref{prop:bias_gamma}/\Cref{cor:bias_pi_rho_pirho_gamma}.}

\begin{proof}

  Since for any $i\in[b]$, $\txts\gamma_{i}\le \nofrac{m_{i}}{40\tilde{M}_{i}^2}(\min_{i\in[b]}\{m_{i}/\tilde{M}_{i}\}/\max_{i\in[b]}\{m_{i}/\tilde{M}_{i}\})^2$, setting
  \begin{equation*}
  N_{i}^{\star}(\gamma_i) = \big\lfloor m_{i}\min_{i\in[b]}\acn{m_{i}/\tilde{M}_{i}}^2/\prbig{20\gamma_{i}\tilde{M}_{i}^2\max_{i\in[b]}\acn{m_{i}/\tilde{M}_{i}}^2}\big\rfloor
  \end{equation*}
implies $\upkappa_{\bfgamma,\bfrho,\bfN^{\star}} \in (0,1)$ by \Cref{lem:choice_N_condition_gamma}.
Thereby, letting $n$ tend towards infinity in \Cref{lem:bound:cont_wass} and using \Cref{cor:convergence_rho_gamma} conclude the proof.
\end{proof}

\subsection{Proof of \Cref{prop:bias_gamma_bis}}
\label{subsec:proof_cont2}

We first give the formal statement of \Cref{prop:bias_gamma_bis}.

\begin{proposition}\label{cor:bias_pi_rho_pirho_gamma_alternative}
Assume \Cref{ass:well_defined_density}-\Cref{ass:supp_fort_convex}-\Cref{ass:hessian_lipschitz} and let $\gammabf\in (\rset_+^*)^b$, $\bfN \in (\N^*)^b$ such that for any $i \in [b]$, 
$
\txts\gamma_{i}\le \nofrac{m_{i}}{40\tilde{M}_{i}^2}(\min_{i\in[b]}\{m_{i}/\tilde{M}_{i}\}/\max_{i\in[b]}\{m_{i}/\tilde{M}_{i}\})^2$ and $N_{i} = \lfloor m_{i}\min_{i\in[b]}\acn{m_{i}/\tilde{M}_{i}}^2/(20\gamma_{i}\tilde{M}_{i}^2\max_{i\in[b]}\acn{m_{i}/\tilde{M}_{i}^2})\rfloor$.
Then, we have
\begin{equation*}
\wasserstein{}^2\prn{\Pi_{\bfrho,\bfgamma,\bfN},\Pi_{\bfrho}}\le 4\prn{1+\normn{\bar{\B{B}}_0^{-1}\B{B}_0^{\top}\B{\tilde{D}}_{0}^{\half}}^2}
\frac{\max_{i\in[b]}\acn{m_{i}/\tilde{M}_{i}^2}}{\min_{i\in[b]}\acn{m_{i}/\tilde{M}_{i}}^2}
\mathscr{R}^{\star}(\bfgamma)\eqsp,
\end{equation*}
where setting $\mathfrak{f}_i = m_i/(20\tilde{M}_i)$,
\begin{align}
  \label{eq:1}
  \mathscr{R}^{\star}(\bfgamma)& = \sum_{i=1}^b \defEns{d_i \gamma_i^2\tilde{M}_i^2 + \frac{d_i\gamma_i^2 \mathfrak{f}_i}{\tilde{M}_i}\parenthese{d_iL^2_i + \frac{\tilde{M}_i^4}{\tilde{m}_i}} + d_i \gamma_i \tilde{M}_i \mathfrak{f}_i^3(1+\mathfrak{f}_i+\mathfrak{f}_{i}^2)}\eqsp,
\end{align}
$\bar{\B{B}}_0,\B{B}_0,\B{\tilde{D}}_{0}$ are defined in \eqref{eq:def_B_bar_B}-\eqref{eq:def_projection}, and for any $i \in [b]$, $\tilde{m}_i$, $\tilde{M}_i$ are defined in \eqref{eq:def_tilde_m_M}.
\end{proposition}

We provide the proof of \Cref{prop:bias_gamma_bis} in what follows.
Similarly to \Cref{lem:bound:cont_Z_{i}nduction} for the proof of \Cref{prop:bias_gamma}, we derive an explicit relation between $\normn{\Zc_{n+1}-\Zb_{n+1}}$ and $\normn{\Zc_{n}-\Zb_{n}}$.

\begin{lemma}\label{lem:bound:cont_expec_T2_better}
  Assume \Cref{ass:well_defined_density}-\Cref{ass:supp_fort_convex}-\Cref{ass:hessian_lipschitz} and let $\bfN\in(\N^*)^{b}, \gammabf \in (\rset_+^*)^b$ such that for any $i \in [b]$, $N_{i}\gamma_{i}\le 2/(m_{i} + \tilde{M}_{i})$ and $\gamma_{i}<1/\tilde{M}_{i}$. 
  Then, for $n \ge 1$, we have
  \begin{equation*}
  \nosqrt{\E\brbig{\normn{\Zc_{n+1}-\Zb_{n+1}}_{\B{D}_{\bfN\bfgamma}^{-1}}^2}}
  \le \prbig{1 - \min_{i\in[b]} \{N_{i} \gamma_{i} m_{i}\} + r_{\bfgamma, \bfrho, \bfN}}\E\brbig{\normn{\Zc_{n}-\Zb_{n}}_{\B{D}_{\bfN\bfgamma}^{-1}}^2}^{\half}+ \mathscr{R}(\bfgamma,\bfN)^\half \eqsp,
\end{equation*}
where
\begin{equation}
    \label{eq:def_scrR}
\begin{aligned}
\mathscr{R}(\bfgamma,\bfN) &=   \sum_{i=1}^{b} d_iN_{i}\gamma_i^3( d_{i} L_{i}^{2}+ \tilde{M}_{i}^4/\tilde{m}_{i} ) + \sum_{i =1}^b \pr{d_{i}\gamma_{i}^2\tilde{M}_{i}^2 + d_{i}N_i^3\gamma_{i}^4\tilde{M}_{i}^4} \\
  & \qquad \qquad + \sum_{i=1}^b d_i N_i^4 \gamma_i^5 \tilde{M}^5_i (1+N_i \gamma_i \tilde{M}_i) \eqsp,
\end{aligned}
\end{equation}
$(\tilde{Z}_n,Z_n)_{n \in \N}$ is defined in \eqref{eq:cont_coupling_2}, $r_{\bfgamma, \bfrho, \bfN}$ in \eqref{eq:def:r} and for any $i \in [b]$, $\tilde{m}_i$, $\tilde{M}_i$ are defined in \eqref{eq:def_tilde_m_M}.
\end{lemma}
\begin{proof}
Let $n\in\N$. 
For any $k\in\N$, recall that $\B{M}_{k}^{(n)}$ is defined in \eqref{eq:def:cont_M} and invertible by \Cref{lem:C_{i}nvertible}. 
Define
\begin{align*}
    &w_n=\B{D}_{\bfN\bfgamma}^{-\half}\prn{\tilde{Z}_n-Z_n}\eqsp.
\end{align*}
Under this notation, the result given in \Cref{lem:bound:cont_Z_{i}nduction1} can be rewritten as
\[
w_{n+1}=\B{\mathrm{T}}_1^{(n)} w_n - \B{\mathrm{T}}_2^{(n)} \eqsp,
\]
where $\B{\mathrm{T}}_1^{(n)}$ and $\B{\mathrm{T}}_2^{(n)}$ are defined in \eqref{eq:def:T} and \eqref{eq:def:T2}, respectively.
By the Minkowsky inequality and using \eqref{eq:filtrationG}, we have
\begin{equation}\label{eq:bound:u_n_rec}
\nosqrt{\E^{\mathcal{G}_{n}}\brbig{\normn{w_{n+1}}^2}}\le\nosqrt{\E^{\mathcal{G}_{n}}\brbig{\normn{\B{\mathrm{T}}_1^{(n)} w_{n}}^2}}+\nosqrt{\E^{\mathcal{G}_{n}}\brbig{\normn{\B{\mathrm{T}}_2^{(n)}}^2}}\eqsp.
\end{equation}
Since by \Cref{lem:bound:cont_T1_bis},
\begin{equation}\label{eq:def:prop3:K}
  \normn{\B{\mathrm{T}}_1^{(n)}}
  \le 1 - \min_{i\in[b]} \{N_{i} \gamma_{i} m_{i}\} + r_{\bfgamma, \bfrho, \bfN}\eqsp,
\end{equation}
it remains to bound $\E^{\mathcal{G}_{n}}[\|\B{\mathrm{T}}_2^{(n)}\|^2]$ to complete the proof.

For any $i\in[b]$, recall the function $V_{i}^{\theta_n}:\R^{d_{i}}\to\R$ defined for any $\by^i \in \R^{d_{i}}$ by $V_{i}^{\theta_n}(\by^i)=U_{i}(\by^i)+\normn{\by^i-\B{A}_{i}\theta_n}^2/(2\rho_{i})$. For any $i\in[b], k\in\N$, using the It\^{o}  formula, we have for $l \in[k\gamma_{i}, (k+1)\gamma_{i})$, 
\begin{multline}\label{eq:eq:ito}
\nabla V_{i} (\Yc_{k\gamma_{i}+l}^{(i, n)}) - \nabla V_{i} (\Yc_{k\gamma_{i}}^{(i, n)})
= \int_{k\gamma_{i}}^{k\gamma_{i}+l}\acbig{\nabla^{2} V_{i}^{\theta_n}(\Yc_{u}^{(i, n)}) \nabla V_{i}^{\theta_n} (\Yc_{u})+\vec{\Delta}(\nabla V_{i}^{\theta_n})(\Yc_{u}^{(i, n)})}\,\dd u\\
+ \sqrt{2} \int_{k\gamma_{i}}^{k\gamma_{i}+l} \nabla^{2} V_{i}^{\theta_n}(\Yc_{u}^{(i, n)})\,\dd B_{u}^i\eqsp.
\end{multline}
For any $i\in[b], k\in\N$, define
\begin{align*}
    & a_{1,k}^{(i,n)}=\1_{[N_{i}]}(k+1)\brn{\B{M}_{\infty}^{(i,n)}}^{-1}\B{M}_{k+1}^{(i,n)} \int_0^{\gamma_{i}}\int_{k\gamma_{i}}^{k\gamma_{i}+l}\nabla^{2} V_{i}^{\theta_n}(\Yc_{u}^{(i,n)}) \nabla V_{i}^{\theta_n}(\Yc_{u}^{(i,n)})\,\dd u\,\dd l\eqsp,\\
    & a_{2,k}^{(i,n)}=\1_{[N_{i}]}(k+1)\brn{\B{M}_{\infty}^{(i,n)}}^{-1}\B{M}_{k+1}^{(i,n)}\int_0^{\gamma_{i}}\int_{k\gamma_{i}}^{k\gamma_{i}+l} \vec{\Delta}(\nabla V_{i}^{\theta_n})(\Yc_{u}^{(i,n)})\,\dd u\,\dd l\eqsp,\\
    & a_{3,k}^{(i,n)}=\sqrt{2}\1_{[N_{i}]}(k+1)\brn{\B{M}_{\infty}^{(i,n)}}^{-1}\B{M}_{k+1}^{(i,n)}\int_0^{\gamma_{i}}\int_{k\gamma_{i}}^{k\gamma_{i}+l} \nabla^{2} V_{i}^{\theta_n}(\Yc_{u}^{(i,n)}) \,\dd B_{u}^{i}\,\dd l\eqsp.
\end{align*}
With these notation and by \eqref{eq:eq:ito}, we have
\begin{align}
\normn{T_2^{(n)}}^2 
&= \sum_{i\in[b]}\frac{1}{N_{i}\gamma_{i}}\normBig{\sum_{k\in\N}\acn{a_{1,k}^{(i,n)}+a_{2,k}^{(i,n)}+a_{3,k}^{(i,n)}}}^2\nonumber \\
&\le E_1 + E_2 + E_3\eqsp,\label{eq:bound:B_{n}_upper_bound}
\end{align}
where for any $j \in [3]$, $E_j = 3\sum_{i\in[b]}\|\sum_{k=0}^{N_{i}-1}a_{j,k}^{(i,n)}\|^2 / (N_i\gamma_{i}) $.
We now bound $\{E_j\}_{j \in [3]}$.

\paragraph{Upper bound on $E_1$.}

For any $i\in[b], k\in\N$, recall that we have $\brbig{\B{M}_{\infty}^{(i,n)}}^{-1}\B{M}_{k+1}^{(i,n)}=\prod_{l=k+1}^{\infty}\prn{\B{I}_{d_{i}}+\B{C}_l^{(i,n)}}$ where $\B{C}_l^{(i,n)}$ is defined in \eqref{eq:def:cont_C}. In addition, since we suppose for any $i\in[b]$, that $\gamma_{i}\tilde{M}_{i}<1$, \Cref{lem:C_{i}nvertible} implies 
\begin{align*}
\normBig{\prod_{l=k+1}^{N_{i}-1}(\B{I}_{d_{i}}-\B{C}^{(i,n)}_l)}^2
\le \ac{1-\gamma_{i} \tilde{m}_{i}}^{2(N_{i}-k-1)}\eqsp.
\end{align*}
Combining this result with the Cauchy-Schwarz inequality, we obtain
\begin{equation}
  \label{eq:contt1}
  \frac{1}{N_{i}}\normbigg{\sum_{k=0}^{N_{i}-1}a_{1,k}^{(i,n)}}^2 \le \sum_{k=0}^{N_{i}-1}\normBig{\int_{0}^{\gamma_{i}}\int_{k\gamma_{i}}^{k\gamma_{i}+l}\nabla^{2} V_{i}^{\theta_n}(\Yc_{u}^{(i,n)}) \nabla V_{i}^{\theta_n} (\Yc_{u}^{(i,n)})\,\dd u\,\dd l}^2\eqsp.
\end{equation}
For $i\in[b]$, using the definition of $\bz_{n,\star}^{i}=\arg\min_{\by^i\in \R^{d_{i}}} V_{i}^{\theta_n}(\by^i)\in\R^{d_{i}}$, we have $\nabla V_{i}^{\theta_n}(\bz_{n,\star}^{i})=\B{0}_{d_i}$. Therefore, for $i\in[b], k\in\N$, conditioning with respect to $\mathcal{F}_{k\gamma_{i}}^{(n)}$ defined in \eqref{eq:filtrationF} and using the $\tilde{M}_{i}$-Lipschitz property of $V_{i}^{\theta_n}$ by \Cref{ass:supp_fort_convex} gives
\begin{align*}
\mathbb{E}^{\mathcal{F}_{k\gamma_{i}}^{(n)}}\brbig{\|\nabla^{2} V_{i}^{\theta_n}(\Yc_{u}^{(i, n)}) \nabla V_{i}^{\theta_n}(\Yc_{u}^{(i, n)})\|^{2}}
&\le \tilde{M}_{i}^2 \mathbb{E}^{\mathcal{F}_{k\gamma_{i}}^{(n)}}\brbig{\|\nabla V_{i}^{\theta_n}(\Yc_{u}^{(i, n)})-\nabla V_{i}^{\theta_n}(\bz_{n,\star}^{i})\|^{2}}\\
&\le \tilde{M}_{i}^4 \mathbb{E}^{\mathcal{F}_{k\gamma_{i}}^{(n)}}\brbig{\|\Yc_{u}^{(i, n)}-\bz_{n,\star}^{i}\|^2}\eqsp.
\end{align*}
For any $i\in[b], k\in\N$, combining this result with the Jensen inequality yields
\begin{align}\label{eq:bound:expec_first_term}
&\nonumber\E^{\mathcal{F}_{k\gamma_{i}}^{(n)}}\br{\normbigg{\int_{0}^{\gamma_{i}}\int_{k\gamma_{i}}^{k\gamma_{i}+l}\nabla^{2} V_{i}^{\theta_n}(\Yc_{u}^{(i, n)}) \nabla V_{i}^{\theta_n} (\Yc_{u}^{(i, n)})\,\dd u\,\dd l}}^2\\
&\le \gamma_{i}\int_{0}^{\gamma_{i}} l \int_{k\gamma_{i}}^{k\gamma_{i}+l} \mathbb{E}^{\mathcal{F}_{k\gamma_{i}}^{(n)}}\brbig{\|\nabla^{2} V_{i}^{\theta_n}(\Yc_{u}^{(i, n)}) \nabla V_{i}^{\theta_n}(\Yc_{u}^{(i, n)})\|^{2}}\,\dd u\,\dd l\nonumber\\
&\le \gamma_{i}\tilde{M}_{i}^{4} \int_{0}^{\gamma_{i}} l \int_{k\gamma_{i}}^{k\gamma_{i}+l} \mathbb{E}^{\mathcal{F}_{k\gamma_{i}}^{(n)}}\brbig{\|\Yc_{u}^{(i, n)}-\bz_{n,\star}^{i}\|^{2}}\,\dd u\,\dd l\eqsp.
\end{align}
By \Cref{lem:bound_expec_z_x_cont}, we have for any $i \in [b]$, $u \in \R_+$,
\begin{align}\label{eq:bound:cont_sum_expec}
\E^{\mathcal{G}_{n}}\brBig{\| \Yc_{u}^{(i, n)} - \bz_{n,\star}^{i} \|^2}
&\le\nofrac{d_{i}}{\tilde{m}_{i}}\eqsp.
\end{align}
Injecting this result in \eqref{eq:bound:expec_first_term} yields
\[
\E\br{\int_{0}^{\gamma_{i}} l \int_{k\gamma_{i}}^{k\gamma_{i}+l} \mathbb{E}^{\mathcal{F}_{k\gamma_{i}}^{(n)}}\brBig{\|\Yc_{u}^{(i, n)}-\bz_{n,\star}^{i}\|^{2}}\,\dd u\,\dd l}
\le d_{i} \gamma_{i}^3/(3 \tilde{m}_{i})\eqsp.
\]
Finally, this inequality, \eqref{eq:bound:expec_first_term} and \eqref{eq:contt1}, we get
\begin{equation}\label{eq:prop3:bound1}
  \E\br{E_1}
  \le \sum_{i=1}^{b}d_{i} N_{i}\gamma_{i}^3\tilde{M}_{i}^4/\tilde{m}_{i}\eqsp.
\end{equation}

\paragraph{Upper bound on $E_2$.}

Using the Cauchy-Schwarz inequality, we have
\begin{equation*}
  \frac{1}{N_{i}}\normbigg{\sum_{k=0}^{N_{i}-1}a_{2,k}^{(i,n)}}^2
  \le\sum_{k=0}^{N_{i}-1}\normBig{\int_{0}^{\gamma_{i}}\int_{k\gamma_{i}}^{k\gamma_{i}+l}\vec{\Delta}(\nabla V^{\theta_n}_{i})(\Yc_{u}^{(i,n)})\,\dd u\,\dd l}^2\eqsp.
\end{equation*}
By \Cref{ass:hessian_lipschitz}, we have for any $\bz_i\in\R^{d_{i}}$, $\|\vec{\Delta}(\nabla V^{\theta_n}_{i})(\bz_i)\|^{2} \le d_{i}^{2} L_{i}^{2}$.
Therefore, we obtain
\begin{align*}
  \norm{\int_{0}^{\gamma_{i}}\int_{k\gamma_{i}}^{k\gamma_{i}+l}\vec{\Delta}(\nabla V^{\theta_n}_{i})(\Yc_{u}^{(i,n)})\,\dd u\,\dd l}^2
  &\le\gamma_{i}\int_{0}^{\gamma_{i}} l \int_{k\gamma_{i}}^{k\gamma_{i}+l}\|\vec{\Delta}(\nabla V^{\theta_n}_{i})(\Yc_{u}^{(i, n)})\|^{2}\,\dd u\,\dd l\\
  &\le d_{i}^{2} \gamma_{i}^{4} L_{i}^{2}/3 \eqsp.
\end{align*}
Thus, we get
\begin{equation}\label{eq:prop3:bound2}
\expe{  E_2} \le \sum_{i=1}^{b} d_{i}^{2}N_{i}\gamma_{i}^{3}L_{i}^{2} \eqsp.
\end{equation}

\paragraph{Upper bound on $E_3$.}

For any $i\in[b], k\in\N$, define 
\begin{align*}
  & \Delta_{3,k}^{(i,n)}=\int_{0}^{\gamma_{i}}\int_{k\gamma_{i}}^{k\gamma_{i}+l}\nabla^2V^{\theta_n}_{i}\prn{\Yc_{u}^{(i,n)}}\,\dd B_u^{i}\,\dd l\eqsp.
\end{align*}
Using for any $i \in [b],k \in \N$, $[\B{M}_{\infty}^{(i,n)}]^{-1} \B{M}_{k+1}^{(i,n)} = \B{I}_{d_i} - \sum_{l=k+1}^{\infty} \B{C}_{l}^{(i,n)} + \B{R}_{k}^{(i,n)}$ where $\B{R}_{k}^{(i,n)}$ is defined in \eqref{eq:def:cont_R}, we have, for any $i\in[b], k\in\N$, 
\begin{align}
    \normbigg{\sum_{k=0}^{N_{i}-1}a_{3,k}^{(i,n)}}^2
    &= \norm{\sqrt{2}\sum_{k=0}^{N_i-1} \prod_{l=k+1}^{N_i}\br{\B{I}_{d_i} - \B{C}_l^{(i,n)}}\Delta_{3,k}^{(i,n)}}^2 \nonumber\\
    &= 2 \sum_{k_1,k_2=0}^{N_i-1}\langle\B{R}_{k_1}^{(i,n)}\Delta_{3,k_1}^{(i,n)},\B{R}_{k_2}^{(i,n)}\Delta_{3,k_2}^{(i,n)}\rangle + 2 \sum_{k_1,k_2=0}^{N_i-1}\langle\Delta_{3,k_1}^{(i,n)},\Delta_{3,k_2}^{(i,n)}\rangle \nonumber\\
    &+ 2 \sum_{k_1,k_2=0}^{N_i-1}\langle \sum_{l=k_1+1}^{N_i}\B{C}_l^{(i,n)}\Delta_{3,k_1}^{(i,n)},\sum_{l=k_2+1}^{N_i}\B{C}_l^{(i,n)}\Delta_{3,k_2}^{(i,n)}\rangle \nonumber\\
    &- 4 \sum_{k_1,k_2=0}^{N_i-1}\langle \sum_{l=k_1+1}^{N_i}\B{C}_l^{(i,n)}\Delta_{3,k_1}^{(i,n)},\Delta_{3,k_2}^{(i,n)}\rangle\nonumber
    + 4 \sum_{k_1,k_2=0}^{N_i-1}\langle \B{R}_{k_1}^{(i,n)}\Delta_{3,k_1}^{(i,n)},\Delta_{3,k_2}^{(i,n)}\rangle\\
    &- 4 \sum_{k_1,k_2=0}^{N_i-1}\langle \B{R}_{k_1}^{(i,n)}\Delta_{3,k_1}^{(i,n)},\sum_{l=k_2+1}^{N_i}\B{C}_l^{(i,n)}\Delta_{3,k_2}^{(i,n)}\rangle\eqsp. \label{eq:dev_E3}
\end{align}
We now control the quantities which appear in \eqref{eq:dev_E3}.
First, by \Cref{ass:supp_fort_convex}, for any $i\in[b], \bx^i, \by^i \in \mathbb{R}^{d_{i}}$, note that we have
\[
\|\nabla^{2} V_{i}^{\theta_n}(\bx^i) \by^i\| \le \tilde{M}_{i}\|\by^i\|\eqsp.
\]
By the Jensen inequality and the It\^{o} isometry, for any $k\in\N$, we get
\begin{align}
  \nonumber
\E^{\mathcal{F}_{k\gamma_{i}}^{(n)}}\br{\normn{\Delta_{3,k}^{(i,n)}}^2}
&=\E^{\mathcal{F}_{k\gamma_{i}}^{(n)}}\brbigg{\normBig{\int_0^{\gamma_{i}}\int_{k\gamma_{i}}^{k\gamma_{i}+l}\nabla^2V^{\theta_n}_{i}\prn{\Yc_{u}^{(i,n)}}\,\dd B_{u}^{i}\,\dd l}^2}\\
&\le \gamma_{i}\tilde{M}_{i}^2 \int_{0}^{\gamma_{i}}\E^{\mathcal{F}_{k\gamma_{i}}^{(n)}}\brbigg{\normBig{\int_{k\gamma_{i}}^{k\gamma_{i}+l}\,\dd B_{u}^i}^2}\,\dd l = d_{i} \gamma_{i}^3\tilde{M}_{i}^2/2 \eqsp.\label{eq:bound:b_3_norm}
\end{align}
In addition, since for $i\in[b]$, $(\int_{0}^{t} \nabla^{2} V_{i}^{\theta_n}(\Yc_{u}^{(i, n)})\,\dd B_{u}^i)_{t \ge 0}$ is a $(\mathcal{F}_{t}^{(n)})_{t \ge 0}$-martingale, for $ (k_1,k_2)\in\acn{0,\ldots,N_{i}-1}^2$ such that $k_1<k_2$, we obtain
\[
  \E^{\mathcal{G}_{n}}\brBig{\brbig{\Delta_{3,k_1}^{(i,n)}}^{\top}\Delta_{3,k_2}^{(i,n)}}
  =\E^{\mathcal{G}_{n}}\brBig{\E^{\mathcal{F}_{k_2\gamma_{i}}^{(n)}}\brbig{\Delta_{3,k_1}^{(i,n)\top}\Delta_{3,k_2}^{(i,n)}}}
  =0\eqsp.
\]
Therefore, 
$$
\sum_{k_1,k_2=0}^{N_i-1}\E^{\mathcal{G}_{n}}\br{\langle\Delta_{3,k_1}^{(i,n)},\Delta_{3,k_2}^{(i,n)}\rangle} = d_{i}N_i\gamma_{i}^3\tilde{M}_{i}^2/2\eqsp. 
$$
Second, since for any $i\in[b], l\in\N, \B{C}_{l}^{(i,n)}\in\R^{d_{i}\times d_{i}}$ is symmetric positive semi-definite, we have
\begin{align*}
\sum_{k_1,k_2=0}^{N_i-1}\langle \sum_{l=k_1+1}^{N_i}\B{C}_l^{(i,n)}\Delta_{3,k_1}^{(i,n)},\Delta_{3,k_2}^{(i,n)}\rangle
&= \ps{\defEns{\sum_{l=1}^N C_l }\sum_{k_1 = 0}^{l-1} \Delta_{3,k_1}}{\sum_{k_1 = 0}^{l-1} \Delta_{3,k_2}}\ge 0 \eqsp. 
\end{align*}
Third, using for any $i\in[b], l\in\N$, using $\normn{\B{C}_{l}^{(i,n)}}\le \gamma_i \tilde{M}_i$ by definition   \eqref{eq:def:cont_C} and \Cref{ass:supp_fort_convex} and combining the Cauchy-Schwarz inequality with \eqref{eq:bound:b_3_norm}, for any $i\in[b], (k_1,k_2)\in\acn{0,\ldots,N_{i}-1}^2$, we get
$$
\sum_{k_1,k_2=0}^{N_i-1}\E^{\mathcal{G}_{n}}\brBig{\langle \sum_{l=k_1+1}^{N_i}\B{C}_l^{(i,n)}\Delta_{3,k_1}^{(i,n)},\sum_{l=k_2+1}^{N_i}\B{C}_l^{(i,n)}\Delta_{3,k_2}^{(i,n)}\rangle} \le d_i N_i^4\gamma_i^5\tilde{M}_i^4/8\eqsp.
$$
Using \eqref{eq:bound:b_3_norm} again  and \Cref{lem:bound_T1_approx}, for $i\in[b]$, we obtain
\begin{align*}
  &\sum_{k_1,k_2=0}^{N_i-1}\PE^{\mathcal{G}_{n}}\brBig{\langle\B{R}_{k_1}^{(i,n)}\Delta_{3,k_1}^{(i,n)},\B{R}_{k_2}^{(i,n)}\Delta_{3,k_2}^{(i,n)}\rangle}
  \le (d_i \gamma_i^3 \tilde{M}_i^2/2) \sum_{k_1,k_2=0}^{N_i-1}\expe{\normn{\B{R}_{k_1}^{(i,n)}}\normn{\B{R}_{k_2}^{(i,n)}}}
  \\
  & \le ( d_i \gamma_i^3 \tilde{M}_i^2/2) \defEns{\sum_{k=0}^{N_i-1} (\exp[(N_i-k)\gamma_i \tilde{M}_i] - 1 -[(N_i-k)\gamma_i \tilde{M}_i]) }^2\\
  & \le ( d_i \gamma_i^3 \tilde{M}_i^2/2) \defEns{(\tilde{M}_i \gamma_i)^{-1} \int_0^{N_i\gamma_i\tilde{M}_i} \{\rme^t - 1 -t \} \rmd t }^2\\
  & \le \frac{(\rme^{N_{i}\gamma_{i}\tilde{M}_{i}}+1)^2}{288} d_iN_i^6 \gamma_i^7 \tilde{M}^6_i \eqsp.
\end{align*}
Similarly, we get
Moreover, using the Cauchy-Schwarz inequality, for any $i\in[b]$ we get
\begin{align*}
  \sum_{k_1,k_2=0}^{N_i-1}\PE[\langle \Delta_{k_1}^{(i,n)}, \B{R}_{k_2}^{(i,n)}\Delta_{k_2}^{(i,n)}\rangle]
  &\le \sum_{k_1,k_2=0}^{N_i-1}\expe{\normn{\Delta_{k_1}^{(i,n)}}\normn{\Delta_{k_2}^{(i,n)}}\normn{\B{R}_{k_2}^{(i,n)}}}\\
  &\le \frac{d_i N_i\gamma_i^3\tilde{M}_i^2}{24}(\rme^{N_i\gamma_i\tilde{M}_i}+1)N_i^3\gamma_i^2\tilde{M}_i^2\\
  &\le d_iN_i^4\gamma_i^5\tilde{M}_i^4\frac{\rme^{N_i\gamma_i\tilde{M}_i}+1}{24}\eqsp.
\end{align*}
In addition, for any $i\in[b]$, we have also
\begin{equation}\label{eq:6}
  \sum_{k_1,k_2=0}^{N_i-1}\expe{\langle \B{R}_{k_1}^{(i,n)}\Delta_{3,k_1}^{(i,n)},\sum_{l=k_2+1}^{N_i}\B{C}_l^{(i,n)}\Delta_{3,k_2}^{(i,n)}\rangle}
\le d_iN_i^5\gamma_i^6\tilde{M}_i^5\frac{\rme^{N_i\gamma_i\tilde{M}_i}+1}{24}\eqsp.
\end{equation}
For any $i\in[b], k\in\N$, regrouping the previous results and using that $N_i\gamma_i\tilde{M}_i\le 2$ give
\begin{align}
  \E[E_3] &\le \sum_{i =1}^b \{d_{i}N_i\gamma_{i}^2\tilde{M}_{i}^2 + d_{i}N_i^3\gamma_{i}^4\tilde{M}_{i}^4\}
  + \sum_{i=1}^b d_i N_i^4 \gamma_i^5 \tilde{M}^5_i (1+N_i \gamma_i \tilde{M}_i)\eqsp.  \label{eq:prop3:bound3}
\end{align}

\paragraph{Combination of our previous results.}

Injecting the three upper bounds \eqref{eq:prop3:bound1}, \eqref{eq:prop3:bound2}, \eqref{eq:prop3:bound3} in \eqref{eq:bound:B_{n}_upper_bound}, we get
\begin{multline}\label{eq:eq:prop3:bound:bn}
  \E\br{\normn{T_2^{(n)}}^2 }
  \le \sum_{i=1}^{b} d_iN_{i}\gamma_i^3( d_{i} L_{i}^{2} + \tilde{M}_{i}^4/\tilde{m}_{i} )
  + \sum_{i =1}^b \{d_{i}\gamma_{i}^2\tilde{M}_{i}^2 + d_{i}N_i^3\gamma_{i}^4\tilde{M}_{i}^4\} \\
  + \sum_{i=1}^b d_i N_i^4 \gamma_i^5 \tilde{M}^5_i (1+ N_i \gamma_i \tilde{M}_i)\eqsp.
\end{multline}
Using the recursion defined in \eqref{eq:bound:u_n_rec}, and combining the upper bounds derived in \eqref{eq:def:prop3:K} and \eqref{eq:eq:prop3:bound:bn} completes the proof.
\end{proof}
\begin{lemma}\label{lem:bound:cont_expec_norm_spe_alternative}
  Assume \Cref{ass:well_defined_density}-\Cref{ass:supp_fort_convex}-\Cref{ass:hessian_lipschitz} and let $\bfN\in(\N^*)^{b}, \gammabf \in (\rset_+^*)^b$ such that for any $i \in [b]$, $N_{i}\gamma_{i}\le 2/(m_{i} + \tilde{M}_{i})$, $\gamma_{i}<1/\tilde{M}_{i}$ and $\upkappa_{\bfgamma,\bfrho,\bfN} = \min_{i\in[b]}\{N_{i}\gamma_{i} m_{i}\}-r_{\bfgamma,\bfrho,\bfN} \in \left(0, 1\right)$, where $r_{\bfgamma,\bfrho,\bfN}$ is defined in \eqref{eq:def:r}.
  Then, for $n \ge 1$, we have
  \begin{multline*}
  \nosqrt{\E\brbig{\normn{\Zc_{n+1}-\Zb_{n+1}}_{\B{D}_{\bfN\bfgamma}^{-1}}^2}}
  \le (1-\upkappa_{\bfgamma, \bfrho, \bfN})^{n-1}\nosqrt{\E\brbig{\normn{\Zc_{1}-\Zb_{1}}_{\B{D}_{\bfN\bfgamma}^{-1}}^2}} 
  + \acn{\upkappa_{\bfgamma, \bfrho, \bfN}}^{-1} \mathscr{R}(\bfgamma,\bfN) \eqsp,
\end{multline*}
where $\mathscr{R}(\bfgamma,\bfN)$ is given in \eqref{eq:def_scrR}. 
\end{lemma}
\begin{proof}
The proof follows from \Cref{lem:bound:cont_expec_T2_better} combined with a straightforward induction.
\end{proof}

\paragraph{Proof of \Cref{prop:bias_gamma_bis}/\Cref{cor:bias_pi_rho_pirho_gamma_alternative}.}

\begin{proof}[Proof of \Cref{prop:bias_gamma_bis}/\Cref{cor:bias_pi_rho_pirho_gamma_alternative}.]
  For any $i\in[b]$, consider 
  \begin{equation*}
  N_{i}^{\star}(\gamma_i) = \big\lfloor m_{i}\min_{i\in[b]}\acn{m_{i}/\tilde{M}_{i}}^2/\prbig{20\gamma_{i}\tilde{M}_{i}^2\max_{i\in[b]}\acn{m_{i}/\tilde{M}_{i}}^2}\big\rfloor\eqsp.
\end{equation*}
By  \Cref{cor:convergence_rho_gamma} and \Cref{lem:choice_N_condition_gamma}, $P_{\bfrho, \bfgamma,\bfN}$ converges in $W_2$ to $\Pi_{\rhobf,\bfgamma}$. Therefore, using \eqref{eq:W2X_def}, \Cref{lem:geo_decr_cont} and  \Cref{lem:bound:cont_expec_norm_spe_alternative} and taking $n \to \plusinfty$,  we obtain
  \begin{equation}\label{eq:bound:cor_bias_alternative_v2}
    \wasserstein{}^2\prn{\Pi_{\bfrho,\bfgamma,\bfN},\Pi_{\bfrho}}\le 4\prn{1+\normn{\bar{\B{B}}_0^{-1}\B{B}_0^{\top}\B{\tilde{D}}_{0}^{\half}}^2}
\frac{\max_{i\in[b]}\acn{N_{i}^{\star}(\gamma_i) \gamma_i}}{\min_{i\in [b]}\acn{N_{i}^{\star}(\gamma_i) \gamma_i m_i}}\mathscr{R}(\bfgamma,\bfN^{\star}(\bfgamma))
\eqsp.
\end{equation}
  By definition of $N_{i}^{\star}(\gamma_i)$, we have $  \gamma_i \tilde{M}_i N_{i}^{\star}(\gamma_i) \leq \mathfrak{f}_i = m_i/(20\tilde{M}_i)$ which completes the proof upon using it in  \eqref{eq:bound:cor_bias_alternative_v2}.
\end{proof}


\section{Explicit mixing times}\label{sec:mixing_times}


This section aims at providing mixing times for DG-LMC with explicit dependencies w.r.t. the dimension $d$ and the prescribed precision $\varepsilon$.
We specify our result to the case where for any $i \in [b]$, $m_i = m$, $M_i = M$, $L_i=L$, $\rho_i = \rho$, $\gamma_i = \gamma$, $N_i = N$ and for the specific initial distribution
\begin{equation}
\label{eq:_def_mu_star_supp_2} 
\mu_{\rhobf}^{\star} =   \updelta_{\B{z}^{\star}}  \otimes \Pi_{\rhobf}(\cdot|\bfz^{\star}) \eqsp,
\end{equation}
where
\begin{equation}
  \label{eq:def_x_star_2}
  \text{
    $\B{x}^{\star} = ([\btheta^{\star}]^{\top},[\bz^{\star}]^{\top})^{\top}$, where $\btheta^{\star} = \argmin\{-\log\pi\}$ and $\bz^{\star} = ([\B{A}_1\btheta^{\star}]^{\top} , \cdots,[\B{A}_b\btheta^{\star}]^{\top})^{\top}$} \eqsp.
\end{equation}
Note that sampling from $\mu_{\rhobf}^{\star}$ is straightforward  and simply consists in setting $\zbf_0 = \bfz^{\star}$ and $\bftheta_0= \bar{\B{B}}_0^{-1}\B{B}_0^{\top}\B{\tilde{D}}_{0}^{\half} \B{z}_0 + \bar{\B{B}}_0^{-\half} \xi $, where $\xi$ is a $d$-dimensional standard Gaussian random variable.
Starting from this initialisation, we consider the marginal law of $\theta_n$ for $n \ge 1$ and denote it $\Gamma_{\B{x}^\star}^n$.
By \Cref{cor:convergence_rho_gamma}, since for any $i \in [b]$, $N_i=N$, the stationary distribution associated to $P_{\bfrho,\bfgamma,\bfN}$ is $\Pi_{\bfrho,\bfgamma} = \Pi_{\bfrho,\bfgamma,\mathbf{1}_b}$.
We build upon the natural decomposition of the bias:
\begin{align*}
\label{eq:error_decomp_supp}
\wasserstein{}(\Gamma_{\B{x}^{\star}}^n,\pi) \le \wasserstein{}(\mu^{\star}_{\rhobf} P_{\bfrho,\bfgamma,\bfN}^{n}, \Pi_{\bfrho,\bfgamma}) + \wasserstein{} (\Pi_{\bfrho,\bfgamma},\Pi_{\bfrho}) + \wasserstein{}(\pi_{\bfrho},\pi)\eqsp,
\end{align*}
where $\Pi_{\bfrho,\bfgamma}$, $\Pi_{\bfrho}$ and $\pi_{\bfrho}$ are defined in \Cref{prop:convergence_rho_gamma}, \eqref{eq:joint_density_AXDA} and \eqref{eq:theta_marginal}, respectively.
The following subsections focus on deriving conditions on $n_{\varepsilon}$, $\gamma_{\varepsilon}$, $N_{\varepsilon}$ and $\rho_{\varepsilon}$ to satisfy $\wasserstein{}(\Gamma_{\B{x}^{\star}}^{n_{\varepsilon}},\pi) \le \varepsilon$, where $\varepsilon > 0$.

\subsection{Lower bound on the number of iterations $n_{\varepsilon}$}

In this section, we derive a lower bound on $n_{\varepsilon}$ such that $\wasserstein{}(\mu^{\star}_{\rhobf} P_{\bfrho,\bfgamma,\bfN}^{n_{\varepsilon}}, \Pi_{\bfrho,\bfgamma}) \le \varepsilon/3$ following the result provided in
\Cref{cor:convergence_rho_gamma_star_v2}.
Recall that we define the $\bz$-marginal under $\Pi_{\bfrho,\bfgamma}$ by
\begin{equation}
  \label{eq:def_pi_bz_supp_2}
  \pi^{\bz}_{\bfrho,\bfgamma} = \int_{\Rd}\Pi_{\bfrho,\bfgamma}(\btheta,\bz)\,\dd \btheta \eqsp,
\end{equation}
and the transition kernel of the Markov chain $\{Z_{n}\}_{n \ge 0}$, for all $\bz \in \mathbb{R}^p$ and $\msb \in \mathcal{B}( \mathbb{R}^p)$, by
\begin{equation}
\label{eq:def:trans_kernel_z1_2}
    P_{\bfrho,\bfgamma,\bfN}^{\bz}(\bz,\msb) = \int_{\Rd} Q_{\bfrho,\bfgamma,\bfN}(\bz,\msb|\btheta)\Pi_{\rho}(\btheta|\bz)\,\dd \btheta \eqsp,
\end{equation}
where $\Pi_{\rho}(\cdot|\bz)$ and $Q_{\bfrho,\bfgamma,\bfN}$ are defined in \eqref{eq:def:Pi_rho_cond} and \eqref{eq:Q_rho_gamma_N}, respectively. 
In the case $\bfN = \mathbf{1}_b$, we simply denote  $P_{\bfrho,\bfgamma,\bfN}^{\bz}$ by $P_{\bfrho,\bfgamma}^{\bz}$.
We need to bound  in \Cref{cor:convergence_rho_gamma_star_v2} the factor
\begin{equation}
  \label{eq:5}
   \defEns{\int_{\rset^d} \|\bfz_1 - \B{z}^{\star}\|_{\B{D}_{\bfN\bfgamma}^{-1}}^2 \pi^{\bz}_{\bfrho,\bfgamma}(\dd \bz_1) + \int_{\rset^d} \|\bfz_1 - \B{z}^{\star}\|_{\B{D}_{\bfN\bfgamma}^{-1}}^2 P_{\bfrho,\bfgamma,\bfN}^{\bfz}(\bfz^{\star}, \rmd \bfz_1)}^{\half} \eqsp.
\end{equation}
Our next results provide such bounds. 

\begin{lemma}\label{lem:invariant}
    Assume \Cref{ass:well_defined_density}. Then, the transition kernel $P_{\bfrho,\bfgamma}^{\bz}$ leaves $\pi_{\bfrho,\bfgamma}^{\bz}$ invariant, that is $\pi_{\bfrho,\bfgamma}^{\bz}P_{\bfrho,\bfgamma}^{\bz} = \pi_{\bfrho,\bfgamma}^{\bz}$, where $\pi_{\bfrho,\bfgamma}^{\bz}$ is defined by \eqref{eq:def_pi_bz}. 
\end{lemma}
\begin{proof}

We have for any $\mathsf{B} \in \mathcal{B}(\mathbb{R}^p)$
    \begin{equation*}
        \int_{\mathsf{B}}\pi_{\bfrho,\bfgamma}^{\bz}(\dd \bz) = \int_{\mathsf{B}}\int_{\mathbb{R}^d} \Pi_{\bfrho,\bfgamma}(\dd \btheta,\dd \bz) = \int_{\mathsf{B}}\pi_{\bfrho,\bfgamma}^{\bz}(\dd \bz)\int_{\mathbb{R}^d} \Pi_{\bfrho,\bfgamma}(\dd \btheta|\bz) \eqsp.
    \end{equation*}
  Therefore, using the fact that $P_{\bfrho,\bfgamma}$ leaves $\Pi_{\bfrho,\bfgamma}$ invariant from \Cref{prop:convergence_N_1} and Fubini's theorem, we get
    \begin{align}
          \int_{\mathsf{B}}\pi_{\bfrho,\bfgamma}^{\bz}(\dd \bz) &=      \int_{\mathsf{B}}\int_{\mathbb{R}^d} \Pi_{\bfrho,\bfgamma}(\dd \btheta,\dd \bz) 
        = \int_{\mathsf{B}}\int_{\mathbb{R}^d} \int_{\mathbb{R}^d \times \mathbb{R}^p} \Pi_{\bfrho,\bfgamma}(\dd \tilde{\btheta},\dd \tilde{\bz})P_{\bfrho,\bfgamma}((\tilde{\btheta},\tilde{\bz}),(\dd \btheta,\dd \bz)) \nonumber\\
        &= \int_{\mathsf{B}}\int_{\mathbb{R}^d} \int_{\mathbb{R}^d \times \mathbb{R}^p} \Pi_{\bfrho,\bfgamma}(\dd \tilde{\btheta},\dd \tilde{\bz})Q_{\bfrho,\bfgamma}(\tilde{\bz},\dd \bz | \tilde{\btheta})\Pi_{\bfrho}(\btheta|\bz)\,\dd \btheta \nonumber\\
                                                                &= \int_{\mathsf{B}} \int_{\mathbb{R}^d \times \mathbb{R}^p} \Pi_{\bfrho,\bfgamma}(\dd \tilde{\btheta},\dd \tilde{\bz})Q_{\bfrho,\bfgamma}(\tilde{\bz},\dd \bz | \tilde{\btheta})\int_{\mathbb{R}^d}\Pi_{\bfrho}(\btheta|\bz)\,\dd \btheta\\
       &= \int_{\rset^d} \pi_{\bfrho,\bfgamma}^{\bz}(\dd \tilde{\bz}) P_{\bfrho,\bfgamma}^{\bz}(\tilde{\bz},\msb) \eqsp.\nonumber
    \end{align}
  \end{proof}

  For any $i \in [b]$, let $\btheta_i^{\star}$ a minimiser of $\btheta \mapsto U_i(\B{A}_i\btheta)$, and define
  \begin{equation}
    \label{eq:def_u_star}
    \B{u}^{\star} = ([\B{A}_1(\btheta^{\star}-\btheta_1^{\star})]^{\top},\cdots,[\B{A}_b(\btheta^{\star}-\btheta_b^{\star})]^{\top})^{\top}
  \end{equation}

\begin{lemma}
\label{lem:expectation_P}
Assume \Cref{ass:well_defined_density}-\Cref{ass:supp_fort_convex} and let $\bfN\in(\N^*)^{b}, \gammabf,\bfrho \in (\rset_+^*)^b$ such that, for any $i \in [b]$, $\gamma_i \le 2 / (m_i + M_i + 1/\rho_i)$ and denote $\bz^{\star} = ([\B{A}_1\btheta^{\star}]^{\top}, \cdots,[\B{A}_b\btheta^{\star}]^{\top})^{\top}$.
Then, for any $\bz \in \mathbb{R}^{p}$ and $\varepsilon >0$,
\begin{align*}
    \int_{\mathbb{R}^p}\norm{\tilde{\bz}-\bz^{\star}}_{\B{D}_{\bfN\bfgamma}^{-1}}^2 P_{\bfrho,\bfgamma}^{\bz}(\bz,\dd\tilde{\bz})
    &\le \min_{i\in [b]}\{N_i\}^{-1}\Big[ \kappa_{\bfgamma}^2(1+ 2\varepsilon)\norm{\B{z}-\bz^{\star}}_{\B{D}_{\gamma}^{-1}}^2 \\
    &+ (1 + 1/(2\varepsilon))\max_{i \in [b]}\{\gamma_i M_i^2\}\norm{\B{u}^{\star}}^2
    + \mathrm{Tr}(\B{D}_{\bfgamma/\bfrho}\B{P}_{0}) + 2\sum_{i=1}^b d_i\Big]\eqsp,
\end{align*}
where the transition kernel $P_{\bfrho,\bfgamma}^{\bz}$ is defined in \eqref{eq:def:trans_kernel_z1} with $\bfN= \mathbf{1}_b$.
\end{lemma}

\begin{proof}
Let $\gamma_i \le 2 / (m_i + M_i + 1/\rho_i)$ for any $i \in [b]$.
Let $\xi$ be a $d$-dimensional Gaussian random variable independent of $\{\eta^{i} : i\in [b]\}$ where for any $i \in [b]$, $\eta^{i}$ is a $d_i$-dimensional Gaussian random variable.
Let $\bz \in \mathbb{R}^p$ and $Z$ be the random variable distributed according to $\updelta_{\bz}P_{\bfrho,\bfgamma}^{\bz}$, and defined by 
\begin{align*}
    \theta &= \bar{\B{B}}_0^{-1}\B{B}_0^{\top}\B{\tilde{D}}_{0}^{\half} \B{z} + \bar{\B{B}}_0^{-\half} \xi \eqsp, 
\end{align*}
and for any $i \in [b]$,
\begin{align*}
    Z^i &= \left(1-\gamma_i/\rho_i\right) \bz_i - \gamma_i \nabla U_i (\bz_i) + \frac{\gamma_i}{\rho_i} \B{A}_i \theta + \sqrt{2\gamma_i} \eta^i \\
    &= \left(1-\gamma_i/\rho_i\right) \bz_i - \gamma_i \nabla U_i (\bz_i) + \frac{\gamma_i}{\rho_i}\B{A}_i\bar{\B{B}}_0^{-1}\B{B}_0^{\top}\B{\tilde{D}}_{0}^{\half}\B{z} + \frac{\gamma_i}{\rho_i} \B{A}_i\bar{\B{B}}_0^{-\half} \xi + \sqrt{2\gamma_i} \eta^i \\
    &= \left(1-\gamma_i/\rho_i\right) \bz_i - \gamma_i [\nabla U_i (\bz_i) - \nabla U_i (\B{A}_i\btheta^{\star})] - \gamma_i [\nabla U_i (\B{A}_i\btheta^{\star}) - \nabla U_i (\B{A}_i\btheta_i^{\star})] \\
    &+ \frac{\gamma_i}{\rho_i}\B{A}_i\bar{\B{B}}_0^{-1}\B{B}_0^{\top}\B{\tilde{D}}_{0}^{\half}\B{z}
    + \frac{\gamma_i}{\rho_i} \B{A}_i\bar{\B{B}}_0^{-\half} \xi + \sqrt{2\gamma_i} \eta^i\eqsp.
\end{align*}
Let
\begin{align}
  &\B{D}_U^{\star} = \mathrm{diag}\prBig{\gamma_1\int_0^1 \grad^2 U_1(\bz_1 + t(\B{A}_1\btheta^{\star}-\bz_1))\,\dd t, \cdots, \gamma_b\int_0^1 \grad^2 U_b(\bz_b + t(\B{A}_b\btheta^{\star}-\bz_b))\,\dd t}\eqsp, \nonumber\\
  &\tilde{\B{D}}_U^{\star} = \mathrm{diag}\prBig{\gamma_1\int_0^1 \grad^2 U_1(\B{A}_1\btheta^{\star} + t(\B{A}_1\btheta_1^{\star}-\B{A}_1\btheta^{\star}))\,\dd t, \cdots, \gamma_b\int_0^1 \grad^2 U_b(\B{A}_b\btheta^{\star} + t(\B{A}_b\btheta_b^{\star}-\B{A}_b\btheta^{\star}))\,\dd t}\eqsp.\label{eq:DU_star}
\end{align}
Since $\B{P}_{0}\B{D}_{\bfrho}^{-\half}\bz^{\star} = \B{D}_{\bfrho}^{-\half}\bz^{\star}$, it follows that
\begin{equation*}
    Z - \bz^{\star} = \br{\B{I}_p - \B{D}_U^{\star} - \B{D}_{\bfgamma}^{\half} \B{D}_{\bfgamma/\bfrho}^{\half} (\B{I}_{p}-\B{P}_{0}) \B{D}_{\bfrho}^{-\half}}(\B{z}-\bz^{\star}) - \tilde{\B{D}}_U^{\star}\B{u}^{\star} + \B{D}_{\bfgamma}^{\half}\B{D}_{\bfgamma/\bfrho}^{\half}\B{B}_0\bar{\B{B}}_0^{-\half}\xi + \B{D}_{2\bfgamma}^{\half}\eta \eqsp.
\end{equation*}
With the notation $\B{H} = \B{I}_p - \B{D}_U^{\star} - \B{D}_{\bfgamma}^{\half} \B{D}_{\bfgamma/\bfrho}^{\half} (\B{I}_{p}-\B{P}_{0}) \B{D}_{\bfrho}^{-\half}$, \eqref{eq:contraction}, and using the fact that for any $\varepsilon > 0$, $\B{a},\B{b} \in \mathbb{R}^d$, $|\langle\B{a},\B{b} \rangle| \le \varepsilon \norm{\B{a}}^2 + (4\varepsilon)^{-1}\norm{\B{b}}^2$, it follows, for any $\bz \in \mathbb{R}^p$, that
\begin{align*}
    &\int_{\mathbb{R}^p}\normbig{\tilde{\bz}-\bz^{\star}}_{\B{D}_{\bfgamma}^{-1}}^2 P_{\bfrho,\bfgamma}^{\bz}(\bz,\dd\tilde{\bz}) \\
    &= \int_{\mathbb{R}^p}\int_{\mathbb{R}^d}\normbig{\B{H}(\B{z}-\bz^{\star}) - \tilde{\B{D}}_U^{\star}\B{u}^{\star} + \B{D}_{\bfgamma}^{\half}\B{D}_{\bfgamma/\bfrho}^{\half}\B{B}_0\bar{\B{B}}_0^{-\half}\boldsymbol{\xi} + \B{D}_{2\bfgamma}^{\half}\boldsymbol{\eta}}_{\B{D}_{\bfgamma}^{-1}}^2 \upphi_d(\boldsymbol{\xi})\,\dd \boldsymbol{\xi}\upphi_p(\boldsymbol{\eta})\,\dd \boldsymbol{\eta} \nonumber\\
    &= \normbig{\B{H}(\B{z}-\bz^{\star}) - \tilde{\B{D}}_U^{\star}\B{u}^{\star}}_{\B{D}_{\bfgamma}^{-1}}^2 + \mathrm{Tr}(\B{D}_{\bfgamma/\bfrho}\B{P}_{0}) + 2\sum_{i=1}^b d_i\nonumber\\
    &\le \kappa_{\bfgamma}^2\norm{\B{z}-\bz^{\star}}_{\B{D}_{\gamma}^{-1}}^2 - 2\langle\B{H}(\B{z}-\bz^{\star}), \tilde{\B{D}}_U^{\star}\B{u}^{\star}\rangle_{\B{D}_{\bfgamma}^{-1}} + \normbig{\tilde{\B{D}}_U^{\star}\B{u}^{\star}}_{\B{D}_{\bfgamma}^{-1}}^2
    + \mathrm{Tr}(\B{D}_{\bfgamma/\bfrho}\B{P}_{0}) + 2\sum_{i=1}^b d_i \nonumber \\
    &\le \kappa_{\bfgamma}^2(1+ 2\varepsilon)\norm{\B{z}-\bz^{\star}}_{\B{D}_{\bfgamma}^{-1}}^2 + \Big(1 + \frac{1}{2\varepsilon}\Big)\max_{i \in [b]}\{\gamma_i M_i^2\}\norm{\B{u}^{\star}}^2
    + \mathrm{Tr}(\B{D}_{\bfgamma/\bfrho}\B{P}_{0}) + 2\sum_{i=1}^b d_i\eqsp. \nonumber
\end{align*}
\end{proof}

\begin{proposition}
\label{initsupp}
Assume \Cref{ass:well_defined_density}-\Cref{ass:supp_fort_convex} and let $\bfN\in(\N^*)^{b}, \gammabf,\bfrho \in (\rset_+^*)^b$ such that, for any $i \in [b]$, $\gamma_i \le 2 / (m_i + M_i + 1/\rho_i)$.
Then, we have
\begin{multline*}
\int_{\rset^d} \|\bfz_1 - \B{z}^{\star}\|_{\B{D}_{\bfN\bfgamma}^{-1}}^2 \pi^{\bz}_{\bfrho,\bfgamma}(\dd \bz_1)\\\le \min_{i\in [b]}\{N_i\}^{-1}\frac{2}{1 - \kappa_{\bfgamma}^2}\prBigg{\frac{1 + \kappa_{\bfgamma}^2}{1 - \kappa_{\bfgamma}^2}\max_{i \in [b]}\{\gamma_i M_i^2\}\norm{\B{u}^{\star}}^2 + \mathrm{Tr}(\B{D}_{\bfgamma/\bfrho}\B{P}_{0}) + 2\sum_{i=1}^b d_i}\eqsp,
\end{multline*}
with $\kappa_{\bfgamma}$ defined in \eqref{eq:def:kappa}.
\end{proposition}

\begin{proof}
    With the choice $\varepsilon = (1 - \kappa_{\bfgamma}^2) / (4 \kappa_{\bfgamma}^2)$ in \Cref{lem:expectation_P} and using \Cref{lem:invariant}, we have
    \begin{align*}
        \int_{\mathbb{R}^p} \norm{\tilde{\bz}-\bz^{\star}}^2_{\B{D}_{\bfgamma}^{-1}} \pi^{\bz}_{\bfrho,\bfgamma}(\dd\tilde{\B{z}}) \le &\frac{\kappa_{\bfgamma}^2 + 1}{2}\int_{\mathbb{R}^p}\norm{\B{z}-\bz^{\star}}_{\B{D}_{\gamma}^{-1}}^2\pi^{\bz}_{\bfrho,\bfgamma}(\dd\B{z}) + \frac{1 + \kappa_{\bfgamma}^2}{1 - \kappa_{\bfgamma}^2}\max_{i \in [b]}\{\gamma_i M_i^2\}\norm{\B{u}^{\star}}^2\\
        + \mathrm{Tr}(\B{D}_{\bfgamma/\bfrho}\B{P}_{0})
        + 2\sum_{i=1}^b d_i\eqsp. 
    \end{align*}
    Rearranging terms concludes the proof. 
\end{proof}

\begin{lemma}
\label{lem:expectation_P_N}
Assume \Cref{ass:well_defined_density}-\Cref{ass:supp_fort_convex} and let $\bfN\in(\N^*)^{b}, \gammabf,\bfrho \in (\rset_+^*)^b$ such that, for any $i \in [b]$, $N_i\gamma_i \le 2 / (m_i + M_i + 1/\rho_i), \gamma_{i}\tilde{M}_{i}<1$ and denote $\bz^{\star} = ([\B{A}_1\btheta^{\star}]^{\top}, \cdots,[\B{A}_b\btheta^{\star}]^{\top})^{\top}$.
Then, we have
\begin{align*}
    \int_{\mathbb{R}^p}\norm{\tilde{\bz}-\bz^{\star}}_{\B{D}_{\bfN\bfgamma}^{-1}}^2 P_{\bfrho,\bfgamma,\bfN}^{\bz}(\bz^\star,\dd\tilde{\bz})
    \le 2 \sum_{i=1}^b \gamma_i N_i \pr{1 + \mathrm{Tr}(\B{P}_0) / \rho_i} + 4\sum_{i=1}^b d_i\eqsp.
\end{align*}
where the transition kernel $P_{\bfrho,\bfgamma,\bfN}^{\bz}$ is defined in \eqref{eq:def:trans_kernel_z1}.
\end{lemma}

\begin{proof}
Let $\{(\eta^{i}_{k})_{k\ge 1}:i\in[b]\}$ be independent random variables such that for any $i\in[b]$, the sequences $\{(\eta^{i}_{k})_{k\ge 1}\}$ are i.i.d. $d_{i}$-dimensional Brownian motions and let $\xi$ a $d$-dimensional standard Gaussian random variable independent of $\{(\eta^{i}_{k})_{k\ge 1}:i\in[b]\}$.
Consider the stochastic process $(\Yb_{k})_{k\in\N}$ initialised for any $i \in [b]$ at $\Yb_{0}^{i} = \B{A}_i\btheta^\star$ and defined, for any $i \in[b], k\in\N$, by
\begin{equation}\label{eq:coupling_process_Y:end}
\begin{aligned}
\Yb_{k+1}^{i} = \Yb_{k}^{i} -\gamma_i\nabla \tildeU_i(\Yb_{k}^{i}) + (\gamma_i/\rho_i) \B{A}_i \theta + \sqrt{2\gamma_i}\eta_{k+1}^{i}\eqsp,
\end{aligned}
\end{equation}
where the potential $V_{i} = \by^{i} \mapsto U_{i}(\by^{i}) + \|\by^{i}\|^2/(2\rho_{i})$ and
\begin{equation}\label{eq:cont_coupling:end}
\begin{aligned}
    \theta = \bar{\B{B}}_0^{-1}\B{B}_0^{\top}\B{\tilde{D}}_{0}^{\half} \bz^\star + \bar{\B{B}}_0^{-\half} \xi\eqsp.
\end{aligned}
\end{equation}
In addition, we define the random variable $Z = (Z^1,\ldots,Z^b)$, for any $i \in [b]$, as
$$
\Zb^{i} = \Yb_{N_i}^{i}\eqsp.
$$
By definition, note that $\Zb$ is distributed according to $P_{\bfrho,\bfgamma,\bfN}^{\bz}(\bz^\star,\cdot)$.
Define the process $(\Ybr_{k}= \{\Ybr_{k}^{i}\}_{i=1}^b)_{k\in\N}$ valued in $\R^p\times\R^p$ defined for any $i\in[b]$, $k\ge 0$ by
\begin{align*}
\Ybr_k^{i}=\Yb_{\min(k, N_i)}^{i}\eqsp.
\end{align*}
and consider the following matrices defined, for any $k\in\N$, by
\begin{align}
&\B{H}_{U,k} = \mathrm{diag}\bigg(\gamma_1\int_0^1 \nabla^2 U_1((1-s) \Yb_{k}^{1} + s \bz^\star)\,\dd s,\nonumber \\
&\qquad\qquad\qquad\qquad\qquad\qquad\qquad\qquad\hdots, \gamma_b\int_0^1 \nabla^2 U_b ((1-s) \Yb_{k}^{b} + s \bz^\star )\,\dd s\bigg)\eqsp, \nonumber\\
&\B{J}(k) = \mathrm{diag}\pr{\1_{[N_1]}(k+1)\cdot\B{I}_{d_1},\cdots,\1_{[N_b]}(k+1)\cdot\B{I}_{d_b}}\eqsp,\label{eq:def:Jend}\\
&\B{C}_{k} = \B{J}(k)(\B{D}_{\bfgamma/\bfrho} + \B{H}_{U,k})\eqsp,\label{eq:def:Cend}\\
&\B{M}_{k+1} = (\B{I}_{p}-\B{C}_{0})^{-1} \ldots (\B{I}_{p} - \B{C}_{k} )^{-1}\eqsp, \qquad \text{ with } \B{M}_0 = \B{I}_p \eqsp.\label{eq:def:Mend}
\end{align}
Using these notation and \eqref{eq:coupling_process_Y:end}, for any $k\in\N$, we get
\begin{align*}
\Ybr_{k+1}-\bz^\star
=& (\B{I}_{p}-\B{C}_{k}) (\Ybr_{k}-\bz^\star)
+ \B{J}(k)\pr{\B{D}_{\bfgamma/\sqrt{\bfrho}}\B{B}_{0}\theta - \B{D}_{\bfgamma} \nabla V(\bz^\star)+\B{D}_{2\bfgamma}^{\half}\eta_{k+1}}\eqsp.
\end{align*}
Multiplying the previous equality by $\B{M}_{k+1}\B{D}_{\bfN\bfgamma}^{-\half}$, we obtain, for $k \ge 0$,
\begin{multline*}
\B{M}_{k+1}\B{D}_{\bfN\bfgamma}^{-\half} (\Ybr_{k+1} - \bz^\star)
= \B{M}_{k} \B{D}_{\bfN\bfgamma}^{-\half}(\Ybr_{k} - \bz^\star)\\
+ \B{M}_{k+1}\B{J}(k)\B{D}_{\bfN\bfgamma}^{-\half}\pr{\B{D}_{\bfgamma/\sqrt{\bfrho}}\B{B}_{0}\theta-\B{D}_{\bfgamma}\nabla V(\bz^\star)+\B{D}_{2\bfgamma}^{\half}\eta_{k+1}}\eqsp.
\end{multline*}
Summing the previous equality over $k\in\N$ gives
\begin{multline*}
\B{M}_{\infty} \B{D}_{\bfN\bfgamma}^{-\half}(\Ybr_{\bfN} - \bz^\star)
= \B{M}_{0}\B{D}_{\bfN\bfgamma}^{-\half}(\Ybr_{0} - \bz^\star)\\
+ \sum_{k=0}^{\infty} \B{M}_{k+1}\B{J}(k)\B{D}_{\bfN\bfgamma}^{-\half}\pr{\B{D}_{\bfgamma/\sqrt{\bfrho}}\B{B}_{0}\theta-\B{D}_{\bfgamma}\nabla V(\bz^\star)+\B{D}_{2\bfgamma}^{\half}\eta_{k+1}}\eqsp.
\end{multline*}
Multiplying the last equality by $[\B{M}_{\infty}]^{-1}$ and using the fact that $\Ybr_{0} = \bz^\star$, we get
\begin{multline}\label{eq:eq:rec:end}
    \B{D}_{\bfN\bfgamma}^{-\half}(\Zb-\bz_{\star})
    = \sum_{k=0}^{\infty} [\B{M}_{\infty}]^{-1}\B{M}_{k+1}\B{J}(k)\B{D}_{\bfN\bfgamma}^{-\half}\pr{\B{D}_{\bfgamma/\sqrt{\bfrho}}\B{B}_{0}\theta - \B{D}_{\bfgamma}\nabla V(\bz^\star)+\B{D}_{2\bfgamma}^{\half}\eta_{k+1}}\eqsp.
\end{multline}
Recall that $\B{P}_0 = \B{B}_{0} \bar{\B{B}}_{0}^{-1} \B{B}_{0}^{\top}$.
Hence, by \eqref{eq:cont_coupling:end} and using $\B{P}_{0}\B{D}_{\bfrho}^{-\half}\bz^{\star} = \B{D}_{\bfrho}^{-\half}\bz^{\star}$, we get
\[
    \B{D}_{\bfgamma/\sqrt{\bfrho}}\B{B}_{0}\theta- \B{D}_{\bfgamma}\nabla V(\bz^\star)
    = \B{D}_{\bfgamma/\sqrt{\bfrho}}\B{B}_{0}\bar{\B{B}}_0^{-\half} \xi - \B{D}_{\bfgamma}\nabla U(\bz^\star)\eqsp.
\]
Plugging this equality into \eqref{eq:eq:rec:end} yields
\begin{align}
    \nonumber
    \B{D}_{\bfN\bfgamma}^{-\half}(\Zb-\bz^{\star})
    &= -\sum_{k=0}^{\infty} [\B{M}_{\infty}]^{-1}\B{M}_{k+1}\B{J}(k)\B{D}_{\bfgamma/\bfN}^{\half}\nabla U(\bz^\star)\\
        \nonumber
    &+ \sum_{k=0}^{\infty} [\B{M}_{\infty}]^{-1}\B{M}_{k+1}\B{J}(k)\B{D}_{\bfgamma/(\bfN\bfrho)}^{\half} \B{B}_{0}\bar{\B{B}}_0^{-\half} \xi \\
        \label{eq:eq:end}
    &+ \sqrt{2}\sum_{k=0}^{\infty} [\B{M}_{\infty}]^{-1}\B{M}_{k+1}\B{J}(k)\B{D}_{\bfN}^{-\half}\eta_{k+1}\eqsp.
\end{align}
Recall that $[\B{M}_{\infty}]^{-1}\B{M}_{k+1} = (([\B{M}_{\infty}]^{-1}\B{M}_{k+1})^1,\ldots,([\B{M}_{\infty}]^{-1}\B{M}_{k+1})^b)$ is a block-diagonal matrix where, for any $i\in[b]$, $([\B{M}_{\infty}]^{-1}\B{M}_{k+1})^{i}=\prod_{l=k+1}^{\infty}\prn{\B{I}_{d_{i}}-\B{C}_l^{i}}$ where $\B{C}_l^{i}$ is defined in \eqref{eq:def:Cend}. 
In addition, since we suppose for any $i\in[b]$, that $\gamma_{i}\tilde{M}_{i}<1$, \Cref{lem:C_{i}nvertible} implies
\begin{align*}
\normbigg{\prod_{l=k+1}^{N_{i}-1}(\B{I}_{d_{i}}-\B{C}^{i}_l)}^2
\le \pr{1-\gamma_{i} \tilde{m}_{i}}^{2(N_{i}-k-1)}\eqsp.
\end{align*}
We now upper bound separately each term on the right-hand side of \eqref{eq:eq:end}.
First, using the Cauchy-Schwarz inequality, we have
\begin{align}
    \norm{\sum_{k=0}^{\infty} [\B{M}_{\infty}]^{-1}\B{M}_{k+1}\B{J}(k)\B{D}_{\bfgamma/\bfN}^{\half}}^2\nonumber
    &\le \sum_{i=1}^b (\gamma_i/N_i)\norm{\sum_{k=0}^{\infty} ([\B{M}_{\infty}]^{-1}\B{M}_{k+1})^i\B{J}^i(k)}^2 \nonumber\\
    &\le \sum_{i=1}^b (\gamma_i/N_i)\norm{\sum_{k=0}^{N_i-1} \prod_{l=k+1}^{N_i-1} \pr{\B{I}_{d_{i}}-\B{C}^{i}_l}}^2 \nonumber\\
    &\le \sum_{i=1}^b \gamma_i\sum_{k=0}^{N_i-1}\norm{\prod_{l=k+1}^{N_i-1} \pr{\B{I}_{d_{i}}-\B{C}^{i}_l}}^2 \nonumber\\
    &\le \sum_{i=1}^b \gamma_i\sum_{k=0}^{N_i-1}\pr{1-\gamma_{i} \tilde{m}_{i}}^{2(N_{i}-k-1)} \nonumber\\
    &\le \sum_{i=1}^b N_i\gamma_i \eqsp. \label{eq:pp1}
\end{align}
Second, using the same techniques as for the above inequality, we obtain
\begin{align}
    \norm{\sum_{k=0}^{\infty} [\B{M}_{\infty}]^{-1}\B{M}_{k+1}\B{J}(k)\B{D}_{\bfgamma/(\bfN\bfrho)}^{\half} \B{B}_{0}\bar{\B{B}}_0^{-\half}\xi}^2 &\le \sum_{i=1}^b \frac{N_i\gamma_i}{\rho_i}\norm{\B{B}_{0}\bar{\B{B}}_0^{-\half}\xi}^2 \label{eq:pp2}
\end{align} 
Finally, the third term can be upper-bounded as
\begin{align}
    \E\br{\norm{\sqrt{2}\sum_{k=0}^{\infty} [\B{M}_{\infty}]^{-1}\B{M}_{k+1}\B{J}(k)\B{D}_{\bfN}^{-\half}\eta_{k+1}}^2} &\le 2 \sum_{i=1}^b d_i\eqsp.\label{eq:pp3}
\end{align}
Combining \eqref{eq:eq:end}, \eqref{eq:pp1}, \eqref{eq:pp2} and \eqref{eq:pp3}, we get
\begin{align*}
    \int_{\mathbb{R}^p}\norm{\tilde{\bz}-\bz^{\star}}_{\B{D}_{\bfN\bfgamma}^{-1}}^2 P_{\bfrho,\bfgamma,\bfN}^{\bz}(\bz^\star,\dd\tilde{\bz})
    \le \sum_{i=1}^b \gamma_i N_i \pr{1 + \mathrm{Tr}(\B{P}_0) / \rho_i} + 2\sum_{i=1}^b d_i\eqsp.
\end{align*}
\end{proof}

Given $\varepsilon > 0$, we are now ready to provide a condition on the number of iterations $n_{\varepsilon}$ to achieve $\wasserstein{}(\mu^{\star}_{\rhobf} P_{\bfrho,\bfgamma,\bfN}^{n_\varepsilon}, \Pi_{\bfrho,\bfgamma}) \le \varepsilon/3$ in the case where for any $i \in [b]$, $m_i = m$, $M_i = M$, $\rho_i = \rho$, $\gamma_i = \gamma$ and $N_i = N$.
Define
\begin{multline*}
\mathsf{E}_0^2 = 9(1 + \|\bar{\B{B}}_0^{-1}\B{B}_0^{\top}\B{\tilde{D}}_{0}^{\half}\|) 2N \gamma \Bigg[\frac{2}{N(1 - \kappa_{\bfgamma}^2)}\bigg(\frac{1 + \kappa_{\bfgamma}^2}{1 - \kappa_{\bfgamma}^2}\cdot \gamma M^2\norm{\B{u}^{\star}}^2\\
+ (\gamma/\rho)\mathrm{Tr}(\B{P}_{0})
+ 2\sum_{i=1}^b d_i\bigg) + 2 b\gamma N \pr{1 + \mathrm{Tr}(\B{P}_0) / \rho} + 4\sum_{i=1}^b d_i \Bigg]\eqsp.
\end{multline*}
\begin{theorem}
\label{lem:mixing_time}
Assume \Cref{ass:well_defined_density}-\Cref{ass:supp_fort_convex} and assume that for any $i \in [b]$, $m_i=m$ and $M_i=M$.
In addition, let $\bfN = N\mathbf{1}_b, \gammabf = \gamma\mathbf{1}_b,\bfrho = \rho\mathbf{1}_b$, $\rho > 0,\gamma>0, N\ge1$, such that $\gamma < 1/\tilde{M}$, $N\gamma < 2/(m + \tilde{M})$, and \eqref{eq:Ngamma} is satisfied.
Then, for any $\varepsilon > 0$, any
$$
n_{\varepsilon} \ge 2\log\pr{\mathsf{E}_0/\varepsilon}/(N\gamma m),
$$
we have, $\wasserstein{}(\mu^{\star}_{\rhobf} P_{\bfrho,\bfgamma,\bfN}^{n_\varepsilon}, \Pi_{\bfrho,\bfgamma})
\le \varepsilon/3$.
\end{theorem}
\begin{proof}
By some algebra and using $1/\log(1/(1-x))\le1/x$ for $0< x <1$, the proof directly follows from \Cref{cor:convergence_rho_gamma_star_v2} combined with \Cref{initsupp} and \Cref{lem:expectation_P_N}.
\end{proof}

\subsection{Upper bound on the tolerance parameter $\bfrho_{\varepsilon}$}

Define
\begin{align*}
    R_0 &= 2 \sigma^2_U\Big(d \sigma_U^2 + \sum_{i=1}^b M_i^2 \|\B{A}_i(\btheta^\star-\btheta^\star_i)\|^2\Big) + 2 \sigma^4_U \eqsp, \\
    R_1 &= d \sigma_U^2 + \sum_{i=1}^b M_i^2 \|\B{A}_i(\btheta^\star-\btheta^\star_i)\|^2 + \sum_{i=1}^bd_iM_i/2\\
    R_2 &= 2d\max_{i\in[b]}\{M_i\}\sigma_U^2 + 2\sum_{i=1}^b M_i^3 \|\B{A}_i(\btheta^\star-\btheta^\star_i)\|^2 + 8\sigma_U^4 + 8\sigma_U^2\Big[2d\sigma_U^2 \\
    &+ 2\sum_{i=1}^b M_i^2 \|\B{A}_i(\btheta^\star-\btheta^\star_i)\|^2\Big]\eqsp.
\end{align*}

Recall that $\bar{\rho} = \max_{i\in[b]}\{\rho_i\}$.
Then, the following result holds.

\begin{lemma}
    \label{lem:mixing_rho}
    Assume \Cref{ass:well_defined_density}-\Cref{ass:supp_fort_convex}.
    For any $\varepsilon > 0$, let $\bfrho_{\varepsilon} \in (\R_+^*)^b$ such that
    \begin{align*}
    \bar{\rho}_{\varepsilon} \le &\frac{-R_1 + \sqrt{R_1^2 + 4R_0\varepsilon m_U^{\half}/(3\sqrt{2})}}{2R_0} \wedge \frac{\varepsilon\sqrt{m_U}}{3\sqrt{2}\sqrt{R_2 +[R_2/(12\sigma_U^2) + \sum_{i=1}^b d_i M_i]^2}} \\
    &\wedge \frac{1}{12 \sigma_U^2}\eqsp \wedge \frac{-\sum_{i=1}^bd_iM_i + \sqrt{(\sum_{i=1}^bd_iM_i)^2 + 6R_2}}{2R_2} \eqsp.
    \end{align*}
    Then, $W_2(\pi_{\bfrho_{\varepsilon}},\pi) \le \varepsilon/3$.
\end{lemma}
\begin{proof}
    Let $\varepsilon > 0$.
    From \eqref{eq:bound_bias}, for any $\bar{\rho} \le 1/(12 \sigma_U^2)$, $W_2(\pi_{\bfrho},\pi) \le \sqrt{\frac{2}{ m_U}}\max(A_1,A_3^{\half})\eqsp$, where $A_1,A_3$ are defined in \eqref{eq:ratio_lower_bound} and \eqref{eq:proof_bias_pi_rho_A_2_2} respectively.
    This implies that $W_2(\pi_{\bfrho},\pi) \le \varepsilon/3$ is verified if $\max(A_1,A_3^{\half}) \le \varepsilon\sqrt{m_U}/(3\sqrt{2})$.
    First, $A_1 \le \varepsilon\sqrt{m_U}/(3\sqrt{2})$ holds if
    \begin{equation}
        \label{eq:rho_1}
        \bar{\rho} \le \frac{-R_1 + \sqrt{R_1^2 + 4R_0\varepsilon m_U^{\half}/(3\sqrt{2})}}{2R_0} \wedge \frac{1}{12 \sigma_U^2}\eqsp.
    \end{equation}
    We now focus on $A_3$. 
    Using the fact that for any $x \in \mathbb{R}, \mathrm{e}^x \ge x + 1$, we have $2\prod_{i=1}^b(1+\rho_iM_i)^{d_i} \ge 2 + \sum_{i=1}^bd_i\log(1+\rho_iM_i)$ and therefore
    \begin{align*}
         A_3 & \le \exp\Big(\bar{\rho}^2R_2 + \sum_{i=1}^b d_i\log(1+\rho_i M_i)\Big) - 1 -\sum_{i=1}^b d_i\log(1+\rho_iM_i)\eqsp.
    \end{align*}
    Since $\sum_{i=1}^b d_i\log(1+\rho_i M_i) \le \bar{\rho} \sum_{i=1}^b d_iM_i$, $\bar{\rho}^2R_2 + \sum_{i=1}^b d_i\log(1+\rho_i M_i) \le 3/2$ holds for 
    \begin{equation}
        \label{eq:rho_2}
        \bar{\rho} \le \frac{-\sum_{i=1}^bd_iM_i + \sqrt{(\sum_{i=1}^bd_iM_i)^2 + 6R_2}}{2R_2}\eqsp.
    \end{equation}
    Since for any $x \le 3/2, \mathrm{e}^x \le 1 + x + x^2$ and using the fact that $\bar{\rho} \le 1/(12\sigma_U^2)$, it follows that 
    \begin{align*}
        A_3 \le \bar{\rho}^2R_2 + \Big(\bar{\rho}^2R_2 + \bar{\rho}\sum_{i=1}^b d_i M_i)\Big)^2 \le \bar{\rho}^2\br{B_1 +\Big(\frac{R_2}{12\sigma_U^2} + \sum_{i=1}^b d_i M_i\Big)^2}\eqsp.
    \end{align*}
    Hence $A_3^{\half}\le \varepsilon\sqrt{m_U}/(3\sqrt{2})$ holds under \eqref{eq:rho_2} and
    \begin{equation}
        \label{eq:rho_3}
        \bar{\rho} \le \frac{\varepsilon\sqrt{m_U}}{3\sqrt{2}\sqrt{R_2 +\Big(\frac{R_2}{12\sigma_U^2} + \sum_{i=1}^b d_i M_i\Big)^2}}\eqsp.
    \end{equation}
    The proof is concluded by combining \eqref{eq:rho_1}, \eqref{eq:rho_2} and \eqref{eq:rho_3}.
\end{proof}

\subsection{Upper bound on the step-size $\bfgamma_{\varepsilon}$ and number of local iteration $\bfN_{\varepsilon}$}

Based on \Cref{cor:bias_pi_rho_pirho_gamma} or \Cref{cor:bias_pi_rho_pirho_gamma_alternative}, we now determine an upper bound on $\bfgamma_{\varepsilon}$ to ensure $\wasserstein{} (\Pi_{\bfrho}, \Pi_{\bfrho,\bfgamma_{\varepsilon}})\le \varepsilon/3$ in the case $\bfN = N\mathbf{1}_b, \gammabf = \gamma\mathbf{1}_b,\bfrho = \rho\mathbf{1}_b$ where $\rho > 0,\gamma>0, N\ge1$.
The following results hold depending if \Cref{ass:hessian_lipschitz} is considered.
Define 
\begin{align}
    &C_{\rho} = \frac{4\tilde{M}^2(1 + \|\bar{\B{B}}_0^{-1}\B{B}_0^{\top}\B{\tilde{D}}_{0}^{\half}\|^2)}{5 m}\eqsp, \label{Crho}\\
    &C_0 = (\tilde{M}^2/2)\br{\tilde{M}/\tilde{m} + 1/6}\sum_{i=1}^b d_i \eqsp, \quad C_1 = \sum_{i=1}^b d_i\eqsp, \quad C_2 = \varepsilon^2/(9C_{\rho})\eqsp.\nonumber
\end{align}

\begin{lemma}
    \label{lemma:Ng1}
    Assume \Cref{ass:well_defined_density}-\Cref{ass:supp_fort_convex} and assume for any $i \in [b]$, $m_i=m$ and $M_i=M$.
    In addition, let $\rho,\gamma_{\varepsilon}>0$ and $N_{\varepsilon}\ge1$ such that $\rhobf = \rho \mathbf{1}_b$, $\gammabf_{\varepsilon} = \gamma_{\varepsilon} \mathbf{1}_b$, $\bfN_{\varepsilon} = N_{\varepsilon} \mathbf{1}_b$ and $\varepsilon > 0$ satisfying
    \begin{align}
        \label{eq:choice_gamma}
        \gamma_{\varepsilon} &\le \frac{-C_1 + \sqrt{C_1^2 + 4C_0C_2}}{2C_0} \wedge \frac{m}{40\tilde{M}^2}\eqsp.
    \end{align}
    Then $\wasserstein{} (\Pi_{\bfrho}, \Pi_{\bfrho,\bfgamma_{\varepsilon}}) \le \varepsilon/3$.
\end{lemma}

\begin{proof}
    Let $\varepsilon > 0$. 
    By \Cref{cor:bias_pi_rho_pirho_gamma}, note that $\wasserstein{}^{2} (\Pi_{\bfrho}, \Pi_{\bfrho,\bfgamma_{\varepsilon}})\le \varepsilon^2/9$ is satisfied if
    $$
    C_0\gamma_{\varepsilon}^2 + C_1\gamma_{\varepsilon} \le C_2\eqsp.
    $$
    This inequality is satisfied under the choice \eqref{eq:choice_gamma}.
\end{proof}

We now provide a condition on $\bfN$ and $\bfgamma$ when \Cref{ass:hessian_lipschitz} is considered.

\begin{lemma}
    \label{lemma:Ng2}
    Assume \Cref{ass:well_defined_density}-\Cref{ass:supp_fort_convex} and assume for any $i \in [b]$, $m_i=m$, $M_i=M$ and $L_i=L$.
    In addition, let $\rho,\gamma_{\varepsilon}>0$ and $N_{\varepsilon}\ge1$ such that $\rhobf = \rho \mathbf{1}_b$, $\gammabf_{\varepsilon} = \gamma_{\varepsilon} \mathbf{1}_b$, $\bfN_{\varepsilon} = N_{\varepsilon} \mathbf{1}_b$ and $\varepsilon > 0$ satisfying
    \begin{align}
        \label{eq:choice_gamma}
        \gamma_{\varepsilon} \le &\frac{\varepsilon}{6b\sqrt{5\max_{i \in [b]}\{d_i \}C_\rho\tilde{M}^2[4 + (\max_{i \in [b]}\{d_i\}L^2m)/(20\tilde{M}^4)]}} \wedge \frac{m}{40\tilde{M}^2}\\
         & \quad \quad\wedge \frac{\varepsilon}{6b(5C_{\rho}\max_{i \in [b]}\{d_i\}m^3/\tilde{M}^2)}\eqsp,
    \end{align}
    where $C_{\rho}$ is defined in \eqref{Crho}.
    Then $\wasserstein{} (\Pi_{\bfrho}, \Pi_{\bfrho,\bfgamma_{\varepsilon}}) \le \varepsilon/3$.
\end{lemma}

\begin{proof}
    In \Cref{cor:bias_pi_rho_pirho_gamma_alternative}, we dissociate $R^\star(\bfgamma)$ into two contributions and the conditions we impose on $\gamma_{\varepsilon}$ ensure $\wasserstein{} (\Pi_{\bfrho}, \Pi_{\bfrho,\bfgamma_{\varepsilon}}) \le \varepsilon/3$.
    More precisely, we have $\sum_{i=1}^b d_i \gamma_i^2\tilde{M}_i^2 + \frac{d_i\gamma_i^2 \mathfrak{f}_i}{\tilde{M}_i}(d_iL^2_i + \frac{\tilde{M}_i^4}{\tilde{m}_i}) \le 2\varepsilon^2/9$ and  $\sum_{i=1}^b d_i \gamma_i \tilde{M}_i \mathfrak{f}_i^3(1+\mathfrak{f}_i+\mathfrak{f}_{i}^2) \le 2\varepsilon^2/9$ where $\mathfrak{f}_i < 1$ for any $i \in [b]$.
\end{proof}

\subsection{Discussion}

Let $\bfrho_{\varepsilon} = \rho_{\varepsilon} \mathbf{1}_b$ such that $\wasserstein{}(\pi_{\bfrho_{\varepsilon}},\pi) \le \varepsilon/3$.
From \Cref{lem:mixing_rho}, $\rho_{\varepsilon} = \mathcal{O}(\varepsilon/d)$ when $\varepsilon \rightarrow 0$ and $d \rightarrow \infty$.
Similarly, let $\bfgamma_{\varepsilon} = \gamma_{\varepsilon}\mathbf{1}_b$ such that $\wasserstein{} (\Pi_{\bfrho_{\varepsilon}}, \Pi_{\bfrho_{\varepsilon},\bfgamma_{\varepsilon}}) < \varepsilon/3$. 
Under \Cref{ass:well_defined_density}-\Cref{ass:supp_fort_convex}, we obtain by \Cref{lemma:Ng1} $\gamma_{\varepsilon}=\mathcal{O}(\varepsilon^4/d^3)$.
On the other hand, when \Cref{ass:hessian_lipschitz} is additionally assumed, we get by \Cref{lemma:Ng2} $\gamma_{\varepsilon}=\mathcal{O}(\varepsilon^2/d^2)$.
Finally, to apply \Cref{lem:mixing_time} for the previous choices $\gamma_{\varepsilon}$ and $\rho_{\varepsilon}$, we obtain for $\bfN_{\varepsilon} = N_{\varepsilon}\mathbf{1}_b$ the conditions $N_{\varepsilon}=\mathcal{O}(d/\varepsilon^2)$ and $N_{\varepsilon}=\mathcal{O}(1)$ under \Cref{ass:well_defined_density}-\Cref{ass:supp_fort_convex} and \Cref{ass:well_defined_density}-\Cref{ass:supp_fort_convex}-\Cref{ass:hessian_lipschitz}, respectively.
In both scenarios, \Cref{lem:mixing_time} implies $n_{\varespilon} = \mathcal{O}(d^2\log(d)/(\varepsilon^2 |\log(\varepsilon)|)$.
This concludes the results depicted in Table 1 in the main paper.




\end{document}